%% file: ham_linalg.tex
\documentclass[11pt]{article}

\usepackage{amssymb,amsmath,amsthm,amsfonts,dsfont}
\usepackage[numbers,comma,sort&compress]{natbib}
\usepackage[pagebackref,colorlinks,hypertexnames=false]{hyperref}
\hypersetup{
	pdfstartview={FitH},
	pdfnewwindow=true,
	colorlinks=true,
	linkcolor=blue,
	citecolor=blue,
	filecolor=blue,
	urlcolor=blue}

\renewcommand*{\backrefalt}[4]{%
\ifcase #1 %
No citations.%
\or
(Cited on page #2).%
\else
(Cited on pages #2).%
\fi
}

\usepackage{etoolbox}
\makeatletter
\patchcmd\NAT@citexnum{\let\NAT@last@num\NAT@num}{\MakeLinkTarget[cite]{}\Hy@backout{\@citeb\@extra@b@citeb}\let\NAT@last@num\NAT@num}{}{\fail}
\makeatother

\usepackage[capitalise]{cleveref}
\usepackage[all]{hypcap}
\usepackage{url}
\usepackage{graphicx}
\usepackage[font=small]{caption}
\usepackage[subrefformat=parens,labelformat=parens]{subcaption}
\usepackage{float}
\usepackage{footnote}
\usepackage{enumitem}
\usepackage[margin=1in]{geometry}

\newcommand{\abs}[1]{\left\lvert#1\right\rvert}
\newcommand{\norm}[1]{\left\lVert#1\right\rVert}

\newtheorem{theorem}{Theorem}
\newtheorem{lemma}{Lemma}

\newtheorem{proposition}[lemma]{Proposition}
\newtheorem{corollary}[lemma]{Corollary}
\theoremstyle{remark}

\theoremstyle{plain}
\newtheorem*{problem}{Problem}

\newcommand{\eq}[1]{\cref{eq:#1}}
\newcommand{\thm}[1]{\hyperref[thm:#1]{Theorem~\ref*{thm:#1}}}
\newcommand{\defn}[1]{\hyperref[defn:#1]{Definition~\ref*{defn:#1}}}
\newcommand{\lem}[1]{\hyperref[lem:#1]{Lemma~\ref*{lem:#1}}}
\newcommand{\prop}[1]{\hyperref[prop:#1]{Proposition~\ref*{prop:#1}}}
\newcommand{\fig}[1]{\hyperref[fig:#1]{Figure~\ref*{fig:#1}}}
\newcommand{\tab}[1]{\hyperref[tab:#1]{Table~\ref*{tab:#1}}}
\renewcommand{\sec}[1]{\hyperref[sec:#1]{Section~\ref*{sec:#1}}}
\newcommand{\append}[1]{\hyperref[append:#1]{Appendix~\ref*{append:#1}}}
\newcommand{\cor}[1]{\hyperref[cor:#1]{Corollary~\ref*{cor:#1}}}

\newcommand{\ket}[1]{|#1\rangle}
\newcommand{\bra}[1]{\langle#1|}
\newcommand{\ketbra}[2]{\ket{#1}\!\bra{#2}}

\usepackage{mathtools}

\DeclareFontFamily{U}{matha}{\hyphenchar\font45}
\DeclareFontShape{U}{matha}{m}{n}{
	<5> <6> <7> <8> <9> <10> gen * matha
	<10.95> matha10 <12> <14.4> <17.28> <20.74> <24.88> matha12
}{}
\DeclareSymbolFont{matha}{U}{matha}{m}{n}
\DeclareFontFamily{U}{mathx}{\hyphenchar\font45}
\DeclareFontShape{U}{mathx}{m}{n}{
	<5> <6> <7> <8> <9> <10>
	<10.95> <12> <14.4> <17.28> <20.74> <24.88>
	mathx10
}{}
\DeclareSymbolFont{mathx}{U}{mathx}{m}{n}
\DeclareMathSymbol{\obot}         {2}{matha}{"6B}
\DeclareMathSymbol{\bigobot}       {1}{mathx}{"CB}

\usepackage{nicematrix}
\usepackage{tikz}

\usepackage{colortbl}
\definecolor{Gray}{gray}{0.85}

\makeatletter
\def\newmaketag{%
  \def\maketag@@@##1{\hbox{\m@th\normalfont\normalsize##1}}%
  }
\makeatother

\usepackage{aligned-overset}

\usepackage[thicklines]{cancel}
\newcommand\Ccancel[2][black]{\renewcommand\CancelColor{\color{#1}}\cancel{#2}}

\usepackage{authblk}

\usepackage[breakable,skins]{tcolorbox}
\tcolorboxenvironment{theorem}{blanker,breakable=true,before skip=10pt,after skip=10pt}
\tcolorboxenvironment{proposition}{blanker,breakable=true,before skip=10pt,after skip=10pt}
\tcolorboxenvironment{corollary}{blanker,breakable=true,before skip=10pt,after skip=10pt}
\tcolorboxenvironment{lemma}{blanker,breakable=true,before skip=10pt,after skip=10pt}

\title{Quantum matrix arithmetics with Hamiltonian evolution}
\date{\vspace{-5mm}}
\author[1,2]{Christopher Kang}
\author[1]{Yuan Su\thanks{The second author is now at the \emph{AWS Center for Quantum Computing, Pasadena, CA 91106, USA}.}}
\affil[1]{Azure Quantum, Microsoft, Redmond, WA 98052, USA}
\affil[2]{Department of Computer Science, The University of Chicago, Chicago, IL 60637, USA}

\begin{document}
\maketitle

\newcommand{\circuitwidth}{1.05}

\begin{abstract}
The efficient implementation of matrix arithmetic operations underpins the speedups of many quantum algorithms. Standard circuit constructions rely on ancilla qubits and multi-qubit controlled gates, which do not well-align with the capabilities of quantum devices expected in the foreseeable future.

We develop a suite of methods to perform matrix arithmetics---with the result encoded in the off-diagonal blocks of a Hamiltonian---using Hamiltonian evolutions of input operators. We show how to maintain this \emph{Hamiltonian block encoding} after specifying all its entries, so that matrix operations can be composed one after another, and the entire quantum computation takes $\leq 2$ ancilla qubits.

We achieve this for matrix multiplication, matrix addition, matrix inversion, Hermitian conjugation, fractional scaling, integer scaling, complex phase scaling, as well as singular value transformation for both odd and even polynomials. We also present an overlap estimation algorithm to extract classical properties of Hamiltonian block encoded operators, analogous to the well known Hadamard test, at no extra cost of qubit.

Our Hamiltonian matrix multiplication uses the Lie group commutator product formula and its higher-order generalizations due to Childs and Wiebe. We prove a concrete error bound exactly matching the Baker-Campbell-Hausdorff series to third order, which is provably tight up to a single application of the triangle inequality. Our Hamiltonian singular value transformation employs a dominated polynomial approximation, where the approximation holds within the domain of interest, while the constructed polynomial is upper bounded by the target function over the entire unit interval.

When applied to quantum simulation, our methods inherit the commutator scaling of conventional product formulas and leverage the power of matrix arithmetics to reduce the cost of each simulation step. This commutator scaling can be substantially more efficient than that of alternative approaches when input matrices may be decomposed into nearly-commuting summands. To illustrate this feature, we describe a circuit for simulating a class of sum-of-squares Hamiltonians, attaining a commutator scaling in step count, while the gate cost per step remains comparable to that of more advanced algorithms. In particular, we apply this to the doubly factorized tensor hypercontracted Hamiltonians from recent studies of quantum chemistry, obtaining further improvements for initial states with a fixed number of particles. We achieve this with $1$ ancilla qubit.
\end{abstract}
\newpage
{
	\thispagestyle{empty}
	\clearpage\tableofcontents
	\thispagestyle{empty}
}
\newpage

\section{Introduction}
\label{sec:intro}
\input{intro.tex}

\section{Hamiltonian block encoding}
\label{sec:block}
\input{block.tex}

\section{Hamiltonian matrix multiplication}
\label{sec:multiply}
\input{multiply.tex}

\section{Hamiltonian singular value transformation}
\label{sec:qsvt}
\input{qsvt.tex}

\section{Hamiltonian overlap estimation}
\label{sec:overlap}
\input{overlap.tex}

\section{Sum-of-squares Hamiltonian simulation}
\label{sec:sos}
\input{sos.tex}

\section{Discussion}
\label{sec:discuss}
\input{discuss.tex}

\section*{Acknowledgements}
We thank Guang Hao Low and Yu Tong for helpful discussions.

C.K. is funded in part by the STAQ project under award NSF Phy-232580; in part by the US Department of Energy Office of Advanced Scientific Computing Research, Accelerated 
Research for Quantum Computing Program; and in part by the NSF Quantum Leap Challenge Institute for Hybrid Quantum Architectures and Networks (NSF Award 2016136), in part by the NSF National Virtual Quantum Laboratory program, in part based upon work supported by the U.S. Department of Energy, Office of Science, National Quantum 
Information Science Research Centers, and in part by the Army Research Office under Grant Number W911NF-23-1-0077. The views and conclusions contained in this document are those of the authors and should not be interpreted as representing the official policies, either expressed or implied, of the U.S. Government. The U.S. Government is authorized to reproduce and distribute reprints for Government purposes notwithstanding any copyright notation herein.

\newpage
\appendix
\section{Further analysis of Lie group commutator formulas}
\label{append:lie}
\input{lie.tex}

\section{Dominated approximation of composite functions}
\label{append:composite}
\input{composite.tex}

\section{Background on fermionic systems}
\label{append:fermionic}
\input{fermionic.tex}

\clearpage
\bibliographystyle{myhamsplain2}
\bibliography{ham_linalg.bib}

\end{document}

%% file: intro.tex
\subsection{Matrix arithmetics with unitary block encoding}
\label{sec:intro_unitary}
Quantum computers hold the promise of solving numerous problems faster than classical computers. One such example is the dynamical simulation of quantum many-body Hamiltonians~\cite{Lloyd96}---a problem that motivates Feynman~\cite{Fey82}, Manin~\cite{Manin80} and others to propose the idea of quantum computers. Due to the exponential growth of Hilbert space dimensions, Hamiltonian simulation is intractable on classical devices, but there exist efficient quantum algorithms that approximate the evolution $e^{-itH}$ using only a polynomial amount of resources. Other prominent examples include ground state preparation~\cite{Lin2020nearoptimalground,Ge19} and ground energy estimation, whose solutions provide insights for understanding complex chemical reactions and material properties~\cite{vonBurg21,Lee21}. These correspond to approximating the ground state projector of an underlying Hamiltonian $H$, which can be efficiently realized on a quantum computer assuming suitable initial states and energy gaps are available.

In practice, the Hamiltonian of interest $H$ is usually not a simple operator, and its time evolution $e^{-itH}$ cannot be directly performed on a quantum computer. Instead, $H$ can often be decomposed into a linear combination of elementary terms $H=\sum_jh_jH_j$, where $h_j$ are real coefficients and $H_j$ are Hermitian unitaries (such as tensor products of Pauli operators). Such a linear combination can be probabilistically implemented using the \emph{Linear Combination of Unitary} (LCU) technique of~\cite{ChildsWiebe12}, which can then be amplified and transformed to yield the desired output for Hamiltonian simulation and ground state preparation.

The above workflow was systematically extended by the \emph{Quantum Singular Value Transformation} (QSVT)~\cite{Gilyen2018singular} to realize more general matrix arithmetics on a quantum computer. As the singular values of a Hermitian operator agree with its eigenvalues in magnitude, QSVT can be applied to simulate Hamiltonian evolutions and prepare ground states, and its performance recovers or surpasses that of alternative methods. However, QSVT also supports operations like matrix multiplication that do not necessarily preserve Hermicity, which arise naturally in a host of applications such as solving systems of linear equations~\cite{Harrow2009} and differential equations~\cite{Berry_2014}, providing a unified methodology for developing quantum algorithms~\cite{2021MartynGrand}.

The input to QSVT is a \emph{unitary block encoding} of the form
\begin{equation}
    O_A=
    \begin{bmatrix}
        A & \boldsymbol\cdot\\
        \boldsymbol\cdot & \boldsymbol\cdot
    \end{bmatrix},
\end{equation}
where the desired matrix $A$ is encoded in the top-left block~\cite{CGJ19,Low17} with spectral norm $\norm{A}\leq1$. The remaining blocks satisfy the unitary constraint of $O_A$, but are otherwise unspecified and denoted by dots. The circuit construction of such a block encoding typically requires ancilla qubits and multi-qubit controlled quantum gates~\cite[Appendix G.4]{CMNRS18}. Note this is not a one-time cost: as block encodings are combined together by matrix arithmetics, fresh ancillas are required at each step of the computation, which may not well-align with the capabilities of quantum devices expected in the foreseeable future.

For problems such as quantum simulation and ground state preparation, the input $H$ is intrinsically Hermitian and one can instead encode it as a Hamiltonian evolution, leading to algorithms with cost comparable to QSVT using only a constant number of ancilla qubits and substantially simplified circuits~\cite{LinTong22,DongLinTong22,WangMcArdleBerta24,WangZhangYuWang23,Chakraborty25}. The ancilla requirement may be relaxed more generally using a recent uncomputation technique~\cite{Vasconcelos25}. However, that would need to be applied recursively to reduce the ancilla count to constant, incurring a query overhead exponential in the recursion depth. In particular, it is not suitable for Hamiltonian simulation (such as the simulation of sum-of-squares Hamiltonians studied here) where the number of required matrix operations increases polynomially with the problem size.

\subsection{Hamiltonian block encoding}
\label{sec:intro_block}
Here, we demonstrate how to perform common matrix operations---including matrix multiplication, matrix addition, matrix inversion, Hermitian conjugation, fractional scaling, integer scaling, complex phase scaling, and singular value transformation for both odd and even polynomials---using Hamiltonian evolutions of input operators. As matrix arithmetics serve as a fundamental primitive in developing quantum algorithms, the methods described here could lead to a generic resource reduction across many applications, beyond existing ad-hoc solutions for the Hermitian case.

To handle a generic matrix $A$, it is natural to consider its Hermitian dilation $\left[\begin{smallmatrix}
        0 & A^\dagger\\
        A & 0
\end{smallmatrix}\right]$ similar to prior work~\cite{Jordan09,berry2012black,Lloye21hamiltonianqsvt,Shang24Lindbladians,fang2025qubitefficientquantumalgorithmlinear}. We then define the \emph{Hamiltonian block encoding} through the operator
\begin{equation}
    E_A=\exp\left(-i
    \begin{bmatrix}
        0 & A^\dagger\\
        A & 0
    \end{bmatrix}\right).
\end{equation}
This can be seen as a quantum evolution under the Hamiltonian $X\otimes\Re(A)+Y\otimes\Im(A)$, where $X$, $Y$, $Z$ are Pauli matrices and $\Re(A)=\frac{A+A^\dagger}{2}$, $\Im(A)=\frac{A-A^\dagger}{2i}$.
Alternatively, one can also consider $\exp\left(-i
\left[
\begin{smallmatrix}
    \Re(A) & \Im(A)\\
    \Im(A) & -\Re(A)
\end{smallmatrix}
\right]\right)$, generated by the effective Hamiltonian $Z\otimes\Re(A)+X\otimes\Im(A)$.
This encoding consumes $1$ additional qubit beyond those supporting $A$. In practice, the Hamiltonian block encoding $E_A$ can be performed on a quantum computer when $A$ itself is a many-body Hamiltonian, 
{converted from a unitary block encoding $O_A$ using QSVT~\cite{Low2016Qubitization}~\cite[Theorem 5]{Low17},}
or constructed indirectly through operations like linear combination and matrix multiplication. For the purpose of generality, we will quantify the complexity of our algorithms in terms of the number of calls to $E_A$, which may be further bounded when the algorithms are instantiated in concrete applications, such as quantum simulation of the electronic structure Hamiltonians~\cite{vonBurg21,Lee21} to be discussed below.

Any quantum computation involving only unitary operations can be reformulated as a Hamiltonian evolution~\cite{Childs09}. Understanding such a formulation is theoretically interesting in its own right. However, Hamiltonian-based matrix arithmetics also offer unique benefits that are difficult to achieve otherwise.

\begin{enumerate}[label=(\roman*)]
    \item \textbf{Improved complexity scaling}: matrix addition can be performed by evolving individual summands using the Lie-Trotter-Suzuki formulas, with a complexity depending on the norm of nested commutators~\cite{CSTWZ19}. 
    Prior work has shown that this commutator scaling can be significantly smaller than that of alternative approaches, enabling asymptotic speedups~\cite{CSTWZ19,TrotterStep23}. 
    With new results on Hamiltonian matrix arithmetics, we show an analogous improvement for simulating a class of sum-of-squares models with more intricate coefficient tensors.
    \item \textbf{Reduced qubit consumption}: by coherently maintaining the Hamiltonian block encoding, one can complete the entire quantum computation with at most $2$ ancilla qubits. The ancilla consumption can often be reduced to $1$, as is the case with electronic structure Hamiltonian simulation. This contrasts with standard approaches based on QSVT where fresh ancillas are required each time unitary block encodings are combined.
    \item \textbf{More structured circuits}: implementing Hamiltonian-based matrix arithmetics often results in quantum circuits with structures resembling that of the encoded operators, preserving properties such as locality and symmetry. Meanwhile, the interactions between ancilla and system qubits are mostly controlled evolutions, making them more adaptable to digital-analog hybrid devices. Finally, Hamiltonian block-encodings avoid the need for multi-qubit controlled gates required by unitary block-encodings, which could introduce significant overhead. In our simulation of the electronic structure Hamiltonians, all the decomposed terms preserve the particle number, which we utilize to further reduce the cost scaling. 
\end{enumerate}

\begin{table}[htbp]
\centering
\resizebox{\textwidth}{!}{%
\renewcommand{\arraystretch}{1.3}
\begin{tabular}{c|c|c|c|c} 
 \textbf{Matrix arithmetic} & \textbf{Input} & \textbf{Output} & \textbf{Circuit} & \textbf{Complexity} \\
 \hline\hline
 Hermitian conjugation & $\exp\left(-i\begin{bmatrix}
     0 & A^\dagger\\
     A & 0
 \end{bmatrix}\right)$ & $\exp\left(-i\begin{bmatrix}
     0 & A\\
     A^\dagger & 0
 \end{bmatrix}\right)$ &\fig{conjugate} & \prop{herm_conjugate}\\\hline
Complex phase scaling & $\exp\left(-i\begin{bmatrix}
    0 & A^\dagger\\
    A & 0
\end{bmatrix}\right)$ & $\exp\left(-i\begin{bmatrix}
    0 & e^{-i\theta}A^\dagger\\
    e^{i\theta}A & 0
\end{bmatrix}\right)$ &\fig{complex_phase} & \prop{phase_scale}\\\hline
Integer scaling & $\exp\left(-i\begin{bmatrix}
    0 & A^\dagger\\
    A & 0
\end{bmatrix}\right)$ & $\exp\left(-i\begin{bmatrix}
    0 & nA^\dagger\\
    nA & 0
\end{bmatrix}\right)$ &\fig{integer} & \prop{integer_scale}\\\hline
Matrix addition & $\exp\left(-i\begin{bmatrix}
    0 & A_j^\dagger\\
    A_j & 0
\end{bmatrix}\right)$ & $\exp\left(-i\begin{bmatrix}
    0 & \sum_jA_j^\dagger\\
    \sum_jA_j & 0
\end{bmatrix}\right)$ &\fig{addition} & \prop{add}\\\hline
\cellcolor{Gray} Matrix multiplication & \cellcolor{Gray} \begin{tabular}{c}
$\exp\left(-i\begin{bmatrix}
    0 & A^\dagger & 0 \\
    A & 0 & 0 \\
    0 & 0 & 0 \\
\end{bmatrix}\right)$,\\$\exp\left(-i\begin{bmatrix}
    0 & B^\dagger & 0 \\
    B & 0 & 0 \\
    0 & 0 & 0 \\
\end{bmatrix}\right)$
\end{tabular}& \cellcolor{Gray} $\exp\left(-i\begin{bmatrix}
    0 & (AB)^\dagger & 0 \\
    AB & 0 & 0 \\
    0 & 0 & 0 \\
\end{bmatrix}\right)$ & \cellcolor{Gray} \fig{multiply} & \cellcolor{Gray} \thm{multiply}\\\hline
Hermitian matrix multiplication & \begin{tabular}{c}
$\exp\left(-i\begin{bmatrix}
    0 & A^\dagger\\
    A & 0
\end{bmatrix}\right)$,\\
$e^{-i\tau J}$,
$e^{-i\tau K}$
\end{tabular}& $\exp\left(-i\begin{bmatrix}
    0 & (JAK)^\dagger\\
    JAK & 0
\end{bmatrix}\right)$ & \fig{multiply_jk} &\prop{herm_multiply}\\\hline
Unitary matrix multiplication & \begin{tabular}{c} $\exp\left(-i\begin{bmatrix}
    0 & A^\dagger\\
    A & 0
\end{bmatrix}\right)$,\\
$U$, $V$
\end{tabular}& $\exp\left(-i\begin{bmatrix}
    0 & (UAV)^\dagger\\
    UAV & 0
\end{bmatrix}\right)$ &\fig{multiply_unitary} & \prop{unitary_multiply}\\\hline
\cellcolor{Gray} \begin{tabular}{c}
     Singular value transformation\\
     (odd case)
\end{tabular} & \cellcolor{Gray} $\exp\left(-i\begin{bmatrix}
    0 & A^\dagger\\
    A & 0
\end{bmatrix}\right)$ & \cellcolor{Gray} $\exp\left(-i\begin{bmatrix}
    0 & f_{\text{sv}}^\dagger(A)\\
    f_{\text{sv}}(A) & 0
\end{bmatrix}\right)$ & \cellcolor{Gray}\fig{ham_qsvt} & \cellcolor{Gray} \thm{qsvt_odd}\\\hline
\cellcolor{Gray} \begin{tabular}{c}
    Singular value transformation\\
    (Hermitian even case)
\end{tabular} & \cellcolor{Gray} $\exp\left(-i\begin{bmatrix}
    0 & H\\
    H & 0
\end{bmatrix}\right)$ & \cellcolor{Gray} $\exp\left(-i\begin{bmatrix}
    0 & f(H)\\
    f(H) & 0
\end{bmatrix}\right)$ & \cellcolor{Gray} & \cellcolor{Gray} \thm{qsvt_herm_even}\\\hline
\cellcolor{Gray} \begin{tabular}{c}
     Singular value transformation\\
     (even case)
\end{tabular} & \cellcolor{Gray}$\exp\left(-i\begin{bmatrix}
    0 & A^\dagger & 0 & 0\\
    A & 0 & 0 & 0\\
    0 & 0 & 0 & 0\\
    0 & 0 & 0 & 0\\
\end{bmatrix}\right)$ & \cellcolor{Gray}$\exp\left(-i\begin{bmatrix}
    0 & f_{\text{sv}}(A) & 0 & 0\\
    f_{\text{sv}}(A) & 0 & 0 & 0\\
    0 & 0 & 0 & 0\\
    0 & 0 & 0 & 0\\
\end{bmatrix}\right)$ & \cellcolor{Gray}\fig{ham_qsvt_even} & \cellcolor{Gray}\thm{qsvt_even}\\\hline
Matrix inversion & $\exp\left(-i\begin{bmatrix}
    0 & A^\dagger\\
    A & 0
\end{bmatrix}\right)$  & $\exp\left(-i\begin{bmatrix}
        0 & A^{-1\dagger}/\kappa\\
        A^{-1}/\kappa & 0
    \end{bmatrix}\right)$ & & \cor{inverse}\\\hline
Fractional scaling & $\exp\left(-i\begin{bmatrix}
    0 & A^\dagger\\
    A & 0
\end{bmatrix}\right)$ & $\exp\left(-i\begin{bmatrix}
    0 & \tau A^\dagger\\
    \tau A & 0
\end{bmatrix}\right)$ & & \cor{frac_scale}\\\hline
\cellcolor{Gray} Overlap estimation & \cellcolor{Gray} \begin{tabular}{c}
$\exp\left(-i\begin{bmatrix}
    0 & A^\dagger\\
    A & 0
\end{bmatrix}\right)$,\\
$\ket{\psi}$
\end{tabular}& \cellcolor{Gray} \begin{tabular}{c}
$\exp\left(-i\begin{bmatrix}
        0 & \frac{1}{2}\arcsin_{\text{sv}}^\dagger\left(A\right)\\
        \frac{1}{2}\arcsin_{\text{sv}}\left(A\right)& 0
    \end{bmatrix}\right)$,\\
    $\bra{\psi}A\ket{\psi}$
    \end{tabular}&\cellcolor{Gray} \fig{overlap} & \cellcolor{Gray} \thm{overlap}\\\hline
Green's function estimation & $\exp\left(-i\begin{bmatrix}
    0 & H\\
    H & 0
\end{bmatrix}\right)$ & $\exp\left(-i\begin{bmatrix}
    0 & \frac{1}{2}\arcsin\left(\frac{\eta H}{\eta^2+H^2}\right)\\
    \frac{1}{2}\arcsin\left(\frac{\eta H}{\eta^2+H^2}\right) & 0
\end{bmatrix}\right)$ & & \prop{green}\\\hline
\begin{tabular}{c}
     Sum-of-squares\\
     Hamiltonian simulation
\end{tabular} & $\exp\left(-i\begin{bmatrix}
    0 & A_{jk}^\dagger & 0\\
    A_{jk} & 0 & 0\\
    0 & 0 & 0
\end{bmatrix}\right)$ & $\exp\left(-i\begin{bmatrix}
    0 & \sum\limits_{k}\sum\limits_{u}A_{uk}^\dagger\sum\limits_{v}A_{vk} & 0\\
    \sum\limits_{k}\sum\limits_{u}A_{uk}^\dagger\sum\limits_{v}A_{vk} & 0 & 0\\
    0 & 0 & 0
\end{bmatrix}\right)$ & & \cor{sos}\\\hline
Hamiltonian squaring & $\exp\left(-i\begin{bmatrix}
    0 & H\\
    H & 0
\end{bmatrix}\right)$ & $\exp\left(-i\begin{bmatrix}
    0 & H^2\\
    H^2 & 0
\end{bmatrix}\right)$ & \fig{hamiltonian square} & \cor{square}
\end{tabular}
\renewcommand{\arraystretch}{1}
}
\caption{Summary of results on Hamiltonian-based matrix arithmetics, including main results on matrix multiplication, singular value transformation and overlap estimation, with the corresponding rows shaded in gray. Matrices are padded with zero blocks to illustrate the minimum space requirement for implementing our methods. By coherently maintaining the Hamiltonian block encoding, the entire quantum computation can be completed with at most $2$ ancillas and a query cost comparable to existing approaches. Moreover, if the computation does not involve the generic matrix multiplication (the first shaded row) and the singular value transformation for even polynomials and generic inputs (the fourth shaded row), then $1$ single ancilla qubit suffices. At the final stage, classical properties of the Hamiltonian block encoded operator can be estimated with the overlap estimation algorithm (the last shaded row) or the singular value estimation algorithm, at no extra cost of ancilla qubit.}
\label{tab:result}
\end{table}

See \tab{result} for a summary of our results on Hamiltonian-based matrix arithmetics. Among these results, Hermitian conjugation, integer scaling and complex phase scaling follow almost immediately from the definition of Hamiltonian block encoding, which we describe in \sec{block_conjugate} and \sec{block_scaling} respectively. Matrix addition, to be discussed in \sec{block_addition}, is also fairly straightforward to realize using known facts about the Lie-Trotter-Suzuki formulas~\cite{CSTWZ19}. In the following, we focus on Hamiltonian-based matrix multiplication, singular value transformation, and overlap estimation, which build upon new algorithmic ideas and lead to further results on matrix inversion, fractional scaling, and Hamiltonian squaring, as well as applications in simulating sum-of-squares Hamiltonians.

\subsection{Matrix multiplication}
\label{sec:intro_multiply}
Matrix multiplication serves as a fundamental primitive to realize more general arithmetic operations. Suppose that operators $A$ and $B$ are block encoded by unitaries $O_A$ and $O_B$ respectively, each using $1$ ancilla qubit. Then a unitary block encoding of their product $AB$ can be obtained as follows. We introduce an ancilla qubit to control $O_A$ and $O_B$, which increases the number of matrix blocks from $2\times 2$ to $4\times 4$. Concatenation then yields a block encoding with $AB$ at the top left corner, i.e.,
\begin{equation}
\NiceMatrixOptions
{
    custom-line = 
    {
        letter = I , 
        command = hline,
        width = \pgflinewidth
    }
}
    \begin{bNiceArray}{ccIcc}
        A & \boldsymbol\cdot & 0 & 0\\
        \boldsymbol\cdot & \boldsymbol\cdot & 0 & 0\\
        \hline
        0 & 0 & A & \boldsymbol\cdot\\
        0 & 0 & \boldsymbol\cdot & \boldsymbol\cdot
    \end{bNiceArray}
    \begin{bNiceArray}{ccIcc}
        B & 0 & \boldsymbol\cdot & 0\\
        0 & B & 0 & \boldsymbol\cdot\\
        \hline
        \boldsymbol\cdot & 0 & \boldsymbol\cdot & 0\\
        0 & \boldsymbol\cdot & 0 & \boldsymbol\cdot
    \end{bNiceArray}
    =\begin{bNiceArray}{ccIcc}
        AB & \boldsymbol\cdot & \boldsymbol\cdot & \boldsymbol\cdot\\
        \boldsymbol\cdot & \boldsymbol\cdot & \boldsymbol\cdot & \boldsymbol\cdot\\
        \hline
        \boldsymbol\cdot & \boldsymbol\cdot & \boldsymbol\cdot & \boldsymbol\cdot\\
        \boldsymbol\cdot & \boldsymbol\cdot & \boldsymbol\cdot & \boldsymbol\cdot
    \end{bNiceArray}.
\end{equation}
However, all the remaining $15$ matrix blocks are now unspecified and denoted by $\boldsymbol\cdot$, so to continue the computation, new blocks with zero matrix must be introduced, which requires additional fresh ancillas. Carrying out this construction recursively allows one to multiply matrices with an ancilla count growing linearly in the number of multiplicands. This linear dependence on ancilla count can be made logarithmic as shown by~\cite{LW18}, and the underlying idea can be traced back to the study of operator theory in the 1950s~\cite{egervary1954contractive} on ``power dilations''~\cite{thompson1982doubly,Shalit}.

We now explain how to perform matrix multiplication within Hamiltonian block encoding using $\leq2$ ancilla qubits. Our goal here is to produce $E_{AB}=\exp\left(-i\left[\begin{smallmatrix}
    0 & (AB)^\dagger\\
    AB & 0
\end{smallmatrix}\right]\right)$ by querying the Hamiltonian evolution of individual multiplicands $E_{A}=\exp\left(-i\left[\begin{smallmatrix}
    0 & A^\dagger\\
    A & 0
\end{smallmatrix}\right]\right)$ and $E_{ B}=\exp\left(-i\left[\begin{smallmatrix}
    0 & B^\dagger\\
    B & 0
\end{smallmatrix}\right]\right)$.
{We achieve this using the Lie group commutator product formulas, generalizing a previous construction of Zhao et al.~\cite{zhao2021compiling}.}
Historically, these formulas were used in quantum computing to refine operator approximations in the construction of Solovay-Kitaev iterates~\cite{DawsonNielsen07}, as well as to create higher-order terms in the engineering of Hamiltonians~\cite{Park2017,YuAn22commutator,crane2024hybrid,kang2025leveraging}.

In the lowest order, we prove the following error bound for the Lie group commutator product formula
\begin{equation}
    \norm{e^{-i\tau J}e^{-i\tau K}e^{i\tau J}e^{i\tau K}-e^{-\tau ^2[J,K]}}\leq\frac{\tau ^3}{2}\norm{[J,[J,K]]}+\frac{\tau ^3}{2}\norm{[K,[K,J]]},
\end{equation}
where $J$, $K$ are Hermitian and $\tau \geq0$ without loss of generality. This matches the Baker-Campbell-Hausdorff (BCH) expansion to third order (\append{lie}), and is provably tight up to a single application of the triangle inequality. In particular, this tightens a recent result of~\cite[Lemma 9]{Gluza2024doublebracket} which instead reads $\norm{e^{-i\tau J}e^{-i\tau K}e^{i\tau J}e^{i\tau K}-e^{-\tau ^2[J,K]}}\leq\tau^3\norm{[J,[J,K]]}+\tau^3\norm{[K,[K,J]]}$. Hence, we can approximate the evolution of $[J,K]=JK-KJ$ using exponentials of individual $J$ and $K$, and the approximation is accurate for a short time period. To evolve for a longer time, we divide the evolution into $r$ steps and implement each using the Lie group commutator formula. The approximation quality can then be systematically improved using higher-order formulas introduced by Childs and Wiebe~\cite{Childs13commutator,YuAn22commutator}.

By definition, a Lie group commutator formula generates the evolution of a matrix product and its reverse-ordered product with an equal weighting. This can be applied to perform Hamiltonian-based matrix multiplication as follows. We introduce another ancilla qubit to dilate the Hamiltonian and employ the group commutator formula to approximate
\vspace{-3mm}

\begin{footnotesize}
\newmaketag
\begin{align}
    &\exp\left(-i\tau
    \begin{bmatrix}
        0 & A^\dagger & 0 & 0\\
        A & 0 & 0 & 0\\
        0 & 0 & 0 & 0\\
        0 & 0 & 0 & 0
    \end{bmatrix}\right)
    \exp\left(-i\tau
    \begin{bmatrix}
        0 & 0 & B & 0\\
        0 & 0 & 0 & 0\\
        B^\dagger & 0 & 0 & 0\\
        0 & 0 & 0 & 0
    \end{bmatrix}\right)
    \exp\left(i\tau
    \begin{bmatrix}
        0 & A^\dagger & 0 & 0\\
        A & 0 & 0 & 0\\
        0 & 0 & 0 & 0\\
        0 & 0 & 0 & 0
    \end{bmatrix}\right)
    \exp\left(i\tau
    \begin{bmatrix}
        0 & 0 & B & 0\\
        0 & 0 & 0 & 0\\
        B^\dagger & 0 & 0 & 0\\
        0 & 0 & 0 & 0
    \end{bmatrix}\right) \notag\\
    &=
    \exp\left(-i\tau^2
    \begin{bmatrix}
        0 & 0 & 0 & 0\\
        0 & 0 & -iAB & 0\\
        0 & i(AB)^\dagger & 0 & 0\\
        0 & 0 & 0 & 0
    \end{bmatrix}\right)
    +\mathbf{O}\left(\tau^3\left(\norm{A}+\norm{B}\right)^3\right).
\end{align}
\end{footnotesize}%
Here, each multiplicand can be further converted into a standard Hamiltonian block encoding using simple corrections such as the Pauli-$X$ gate, the SWAP gate, the phase gate $S$ and $S^\dagger$. Crucially, the Hamiltonian block encoding is preserved after matrix multiplication, so no fresh ancilla is needed to perform multiple rounds of multiplications. 
{This is equivalent to the method of~\cite{zhao2021compiling} in the lowest order case. However, the error dependence can be made negligible using higher-order formulas~\cite{Childs13commutator}.} 
This construction is illustrated in \fig{multiply_intro} and detailed in \sec{multiply}.

\begin{figure}[t]
	\centering
\includegraphics[scale=0.75]{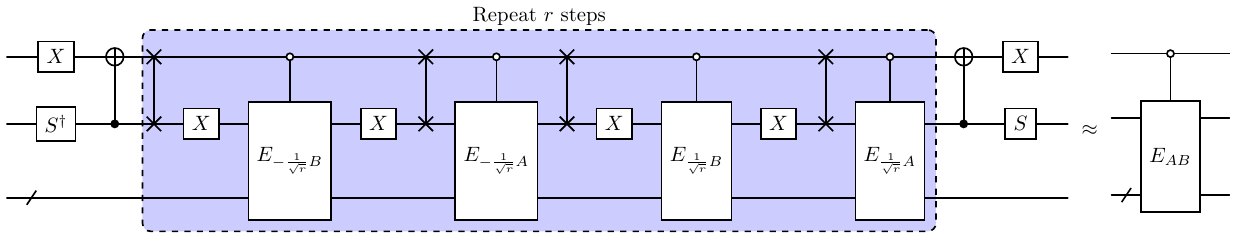}
\caption{Quantum circuit for Hamiltonian matrix multiplication. All circuit diagrams in the paper, including the present one, were prepared using Quantikz~\cite{kay2018tutorial}.}
\label{fig:multiply_intro}
\end{figure}

The above discussion handles the most general case where multiplicands $A$ and $B$ are both arbitrary operators. In many applications, we are promised that at least one (say $B$) is Hermitian. We can then directly evolve a Hamiltonian where $B$ is encoded in the diagonal block, yielding an alternative construction
\begin{equation}
\begin{aligned}
    &\exp\left(-i\tau
    \begin{bmatrix}
        0 & A^\dagger\\
        A & 0
    \end{bmatrix}\right)
    \exp\left(-i\tau
    \begin{bmatrix}
        B & 0\\
        0 & 0
    \end{bmatrix}\right)
    \exp\left(i\tau
    \begin{bmatrix}
        0 & A^\dagger\\
        A & 0
    \end{bmatrix}\right)
    \exp\left(i\tau
    \begin{bmatrix}
        B & 0\\
        0 & 0
    \end{bmatrix}\right)\\
    &=\exp\left(-i\tau^2
    \begin{bmatrix}
        0 & i(AB)^\dagger\\
        -iAB & 0
    \end{bmatrix}\right)
    +\mathbf{O}\left(\tau^3\left(\norm{A}+\norm{B}\right)^3\right).
\end{aligned}
\end{equation}
An analogous scheme holds when $A$ is Hermitian. This uses only $1$ ancilla qubit and thus improves over the generic case discussed above. And the error scaling can again be lowered using higher-order formulas.

Finally, we observe that it is fairly straightforward to multiply unitary matrices with Hamiltonian block encoding. Specifically, given unitaries $U$ and $V$, we have
\begin{equation}
    \begin{bmatrix}
        V^\dagger & 0\\
        0 & U
    \end{bmatrix}
    \exp\left(-i
    \begin{bmatrix}
        0 & A^\dagger\\
        A & 0
    \end{bmatrix}\right)
    \begin{bmatrix}
        V & 0\\
        0 & U^\dagger
    \end{bmatrix}
    =\exp\left(-i
    \begin{bmatrix}
        0 & (UAV)^\dagger\\
        UAV & 0
    \end{bmatrix}\right).
\end{equation}
This multiplication can be performed exactly, using no additional qubit beyond those supporting $A$ plus the single ancilla defining the Hamiltonian block encoding.

\subsection{Singular value transformation}
\label{sec:intro_qsvt}
Given a unitary block encoding of $A$ with the singular value decomposition $A=U\Sigma V^\dagger$, the goal of QSVT is to apply a polynomial $f$ to its singular values, with the result $f_{\text{sv}}(A)$ encoded by another unitary block encoding. Here $f_{\text{sv}}(A)=Uf\left(\Sigma\right)V^\dagger=f_{\text{sv}}^\dagger\left(A^\dagger\right)$ if $f$ is an odd polynomial and $f_{\text{sv}}(A)=Vf\left(\Sigma\right)V^\dagger=f_{\text{sv}}^\dagger(A)$ if $f$ is even. Crucially, QSVT has a query complexity depending only on degree of the target polynomial, and thus remains efficient even when the dimensions of $A$ are exponentially large. By constructing appropriate polynomial approximations, one can apply QSVT to solve various problems of interest with near-optimal query complexity~\cite{Montanaro2024quantum}, achieving a grand unification of quantum algorithms~\cite{2021MartynGrand}.

Many variants of QSVT have been developed  since its introduction in~\cite{LYC16,LC17,Low2016Qubitization,Gilyen2018singular}. Here, we review two versions that are of most relevance to our results. The first variant due to Lloyd et al.~\cite{Lloye21hamiltonianqsvt} is a Hamiltonian QSVT that achieves $\exp\left(-i\left[\begin{smallmatrix}
        0 & A^\dagger\\
        A & 0
    \end{smallmatrix}\right]\right)\mapsto
    \left[\begin{smallmatrix}
        p_{\text{sv}}\left(\cos_{\text{sv}}(A)\right) & i\sin_{\text{sv}}(A^\dagger)q_{\text{sv}}\left(\cos_{\text{sv}}(A^\dagger)\right)\\
        i\sin_{\text{sv}}\left(A\right)q_{\text{sv}}^*\left(\cos_{\text{sv}}(A)\right) & p_{\text{sv}}^*\left(\cos_{\text{sv}}(A^\dagger)\right)
    \end{smallmatrix}\right]$,
where $p$ and $q$ are complex polynomials of opposite parities, satisfying the normalization constraint $\abs{p(x)}^2+(1-x^2)\abs{q(x)}^2=1$ for all $-1\leq x\leq 1$.
Consider the special case where the input matrix $H$ is Hermitian and only one target polynomial $p$ is of interest. Then,~\cite[Theorem 4]{Gilyen2018singular} shows the existence of an \emph{implicit} complementary polynomial $q$ such that $\abs{p(x)}^2+(1-x^2)\abs{q(x)}^2=1$, among all possible choices of $q$ for which the constraint is satisfied. This gives the transformation
\begin{equation}
    \exp\left(-i\begin{bmatrix}
        0 & H\\
        H & 0
    \end{bmatrix}\right)
    \mapsto
    \begin{bmatrix}
        p(\cos(H)) & \cdot\\
        \cdot & \cdot
    \end{bmatrix},
\end{equation}
with the outcome encoded at the top left corner and all other blocks unspecified,
which is utilized by more recent work such as~\cite{WangZhangYuWang23,DongLinTong22} to reduce the quantum resource requirement for solving Hamiltonian-related problems like ground energy estimation.
Assuming $H$ is Hermitian, the second variant of QSVT~\cite{MotlaghWiebe24} performs $\exp\left(-i\left[\begin{smallmatrix}
        H & 0\\
        0 & 0
    \end{smallmatrix}\right]\right)
    \mapsto
    \left[\begin{smallmatrix}
        p\left(e^{-iH}\right) & \cdot\\
        q\left(e^{-iH}\right) & \cdot
    \end{smallmatrix}\right]$,
where $p$ and $q$ are complex polynomials satisfying the normalization constraint $\abs{p(z)}^2+\abs{q(z)}^2=1$ on the unit circle $\abs{z}=1$, and the unspecified blocks are denoted by dots. This reduces to
\begin{equation}
    \exp\left(-i\begin{bmatrix}
        H & 0\\
        0 & 0
    \end{bmatrix}\right)
    \mapsto
    \begin{bmatrix}
        p\left(e^{-iH}\right) & \cdot\\
        \cdot & \cdot
    \end{bmatrix}
\end{equation}
when only the polynomial $p$ is of interest.

Whichever version of QSVT one picks, only $1$ out of the $4$ matrix blocks encodes the desired outcome. This unspecified behavior precludes the minimization of ancilla usage and is incompatible with the previously described multiplication subroutine. Instead, we show how to realize the transformation
\begin{equation}
    \exp\left(-i\begin{bmatrix}
        0 & A^\dagger\\
        A & 0
    \end{bmatrix}\right)
    \mapsto
    \exp\left(-i\begin{bmatrix}
        0 & f_{\text{sv}}^\dagger\left(A\right)\\
        f_{\text{sv}}\left(A \right)& 0
    \end{bmatrix}\right).
\end{equation}
That is, we demand both the input and output to be Hamiltonian block encodings, so that this transformation can be properly composed with other Hamiltonian-based matrix operations. Our choice thus defines the behavior of all $4$ matrix blocks, and construction of the implicit complementary polynomial from~\cite[Theorem 4]{Gilyen2018singular} is not directly applicable. Of course, one can always implement polynomials $p$ and $q$ separately and combine them with LCU, but doing so necessitates additional fresh ancilla qubits, defeating the purpose of Hamiltonian block encoding. To address this, we solve the following dominated approximation problem.

\begin{problem}[Dominated polynomial approximation]
Given $\epsilon>0$, constant $0<\xi\leq\frac{\pi}{2}$, and function $f:\left[-\frac{\pi}{2},\frac{\pi}{2}\right]\rightarrow\mathbb{R}$, find real polynomials $p$ and $q$ such that the following hold:
\begin{equation}
\begin{aligned}
    &\abs{p(x)-\sin(f(\arcsin(x)))}\leq\epsilon,\qquad&&\forall x\in\left[-\sin\left(\frac{\pi}{2}-\xi\right),\sin\left(\frac{\pi}{2}-\xi\right)\right],\\
    &\abs{q(x)-\frac{\cos(f(\arcsin(x)))}{\sqrt{1-x^2}}}\leq\epsilon,\qquad&&\forall x\in\left[-\sin\left(\frac{\pi}{2}-\xi\right),\sin\left(\frac{\pi}{2}-\xi\right)\right],\\
    &p^2(x)+(1-x^2)q^2(x)\leq1+\epsilon,\qquad&&\forall x\in[-1,1].
\end{aligned}
\end{equation}
\end{problem}
\noindent Let us interpret this problem. In the setting of unitary block encoding, the target function $\sin(f(\arcsin(\cdot))):[-1,1]\rightarrow[-1,1]$ preserves the unit interval, whereas the input operator has singular values from $[-1+\xi,1-\xi]$ for $0<\xi\leq1$. To implement it with QSVT, one thus seeks a single polynomial $p$ with maximum error $\norm{p-\sin(f(\arcsin(\cdot)))}_{\max,[-1+\xi,1-\xi]}\leq\epsilon$, so that its action on the input operator mimics that of the target function. Moreover, one requires that $\norm{p}_{\max,[-1,1]}\leq1+\epsilon$ everywhere on the unit interval, so that an implicit complementary polynomial can be constructed via~\cite[Theorem 4]{Gilyen2018singular}. $q$ is obtained separately from a similar construction. This leads to a \emph{bounded polynomial approximation} problem, wherein the approximation holds over the domain of interest, while the constructed polynomial is bounded by $1$ over the unit interval. Our \emph{dominated approximation} problem is different, as we impose the stronger requirement $p^2(x)+(1-x^2)q^2(x)\leq1+\epsilon$ to hold \emph{simultaneously} for $p$ and $q$ throughout the entire unit interval $x\in[-1,1]$. 

We now explain how to realize Hamiltonian QSVT from a dominated polynomial approximation. We first modify QSVT to perform 
\begin{align}
    \exp\left(-i\begin{bmatrix}
        0 & A^\dagger\\
        A & 0
    \end{bmatrix}\right)\mapsto
    \begin{bmatrix}
        p_{\text{sv}}\left(\sin_{\text{sv}}(A^\dagger)\right) & iq_{\text{sv}}\left(\sin_{\text{sv}}(A)\right)\cos_{\text{sv}}(A)\\
        iq_{\text{sv}}^*\left(\sin_{\text{sv}}(A^\dagger)\right)\cos_{\text{sv}}\left(A^\dagger\right) & p_{\text{sv}}^*\left(\sin_{\text{sv}}(A)\right)
    \end{bmatrix}.
\end{align}
Compared to the version from~\cite{Lloye21hamiltonianqsvt}, note that we have replaced the cosine function by sine in the top-left block. This treatment avoids the singularity of $\arccos(x)$ near $x\approx\cos(0)=1$ for a generic input operator with singular values close to $0$, unlike prior results focusing on the Hermitian case~\cite{DongLinTong22} where this singularity can be circumvented by simply shifting the spectra. Now, given a target polynomial $f:\left[-\frac{\pi}{2},\frac{\pi}{2}\right]\rightarrow\mathbb{R}$ and a Hamiltonian block encoding of $A$ with $\norm{A}\leq\frac{\pi}{2}-\xi<\frac{\pi}{2}$, we find polynomials $p$ and $q$ that are simultaneous approximations of $\sin(f(\arcsin(x)))$ and $\frac{\cos(f(\arcsin(x)))}{\sqrt{1-x^2}}$ over $\left[-\sin\left(\frac{\pi}{2}-\xi\right),\sin\left(\frac{\pi}{2}-\xi\right)\right]$, while still satisfying the dominated condition throughout the unit interval.
The output naturally satisfies the normalization constraint of QSVT and approximates the Hamiltonian block encoding of $f_{\text{sv}}(A)$. Moreover, if $f$ is an odd polynomial, then $\sin(f(\arcsin(x)))$ is odd and $\frac{\cos(f(\arcsin(x)))}{\sqrt{1-x^2}}$ is even, so the parity constraint of QSVT is also guaranteed. As a result, we obtain a Hamiltonian QSVT algorithm for odd polynomials.

Implementing even polynomials is somewhat more involved. This is because when $f$ is even, $\sin(f(\arcsin(x)))$ and $\frac{\cos(f(\arcsin(x)))}{\sqrt{1-x^2}}$ are both even and hence the parity constraint of QSVT can never be fulfilled. 
In fact, it is generally not possible to replace even polynomials by their odd approximants due to potential singularities near $x=0$.
To overcome this limitation, we first consider a Hermitian input operator, which can be suitably shifted~\cite{FullyCoherent23} so that all its eigenvalues are positive. Concurrently, the domain of the target polynomial is also shifted. So to perform an even polynomial on the input operator, it suffices to implement an odd extension of the shifted polynomial on the shifted operator---the latter of which can be realized by our Hamiltonian QSVT for odd functions. See~\fig{even} for an illustration of this shifting. The general case can then be reduced to the Hermitian case by first constructing a Hamiltonian block encoding of $\left[\begin{smallmatrix}
    0 & A^\dagger\\
    A & 0
\end{smallmatrix}\right]$ for an arbitrary input matrix $A$. See \sec{qsvt} for details.
We note that our degree dependence almost quadratically improves that of recent result~\cite[Subroutine 2]{Odake23} based on randomization, while our error dependence exponentially improves that of the same subroutine, using $1$ single ancilla qubit.

\begin{figure}[t]
	\centering
    \begin{subfigure}[t]{.495\textwidth}
        \centering
        \includegraphics[scale=.95]{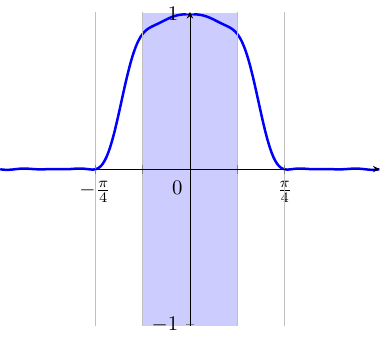}
        \caption{}
    \end{subfigure}%
    ~
    \begin{subfigure}[t]{.495\textwidth}
        \centering
        \includegraphics[scale=.95]{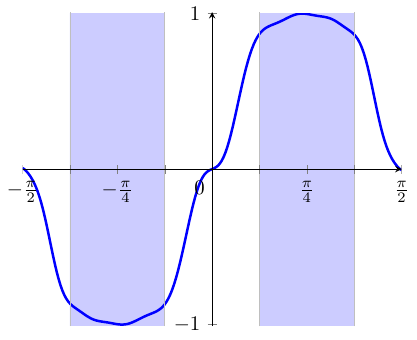}
        \caption{}
    \end{subfigure}%
\caption{Illustration of performing Hamiltonian QSVT for even functions via shifting and odd extension. 
Subfigure (a) illustrates some target function over a domain strictly enclosed by $[-\frac{\pi}{4}, \frac{\pi}{4}]$, with the desired approximation region shaded in blue. Subfigure (b) illustrates the constructed odd extension that approximates the initial even function on a shifted region.
}
\label{fig:even}
\end{figure}

Starting from a given function $f$, one can get dominated approximations of $\sin(f(\arcsin(x)))$ and $\frac{\cos(f(\arcsin(x)))}{\sqrt{1-x^2}}$ with degrees at most logarithmic factors larger than that of approximations of $f$. For instance, this happens if one approximates $f(x)$, $\sin(x)$, $\cos(x)$, $\arcsin(x)$, $\frac{1}{\sqrt{1-x^2}}$ separately, and concatenates the resulting polynomials. However, such a logarithmic overhead can often be avoided by handling the composite functions as a whole, giving $\abs{p(x)-\sin(f(\arcsin(x)))}\leq\epsilon$, $\abs{q(x)-\frac{\cos(f(\arcsin(x)))}{\sqrt{1-x^2}}}\leq\epsilon$ when $x\in\left[-\sin\left(\frac{\pi}{2}-\xi\right),\sin\left(\frac{\pi}{2}-\xi\right)\right]$, and $\abs{p(x)}\leq\abs{\sin(f(\arcsin(x)))}+\epsilon$, $\abs{q(x)}\leq\frac{\abs{\cos(f(\arcsin(x)))}+\epsilon}{\sqrt{1-x^2}}$ for all $x\in[-1,1]$. We achieve this by analyzing behavior of the composite functions in the complex plane around the unit interval~\cite{TangTian24}. See~\append{composite} for details. 
{This allows us to perform Hamiltonian QSVT with a query complexity essentially the same as that of unitary QSVT.}
As immediate corollaries, we show how to perform matrix inversion and fractional scaling within Hamiltonian block encoding.

\subsection{Overlap estimation}
\label{sec:intro_overlap}
Classical properties of an underlying quantum system can be accessed through quantum measurements. A common example is the quantum overlap estimation problem, which can be formulated using the language of block encoding as follows. Given a unitary block encoding of an operator $A$ along with a quantum state $\ket{\psi}$, the goal is to estimate the overlap $\bra{\psi}A\ket{\psi}$ which is generally a complex number with absolute value at most $1$. This can be solved with accuracy $\epsilon$ and failure probability $p_{\text{fail}}$ using the well-known \emph{Hadamard test}, by introducing another ancilla qubit and repeatedly measuring a controlled version of the unitary block encoding $\mathbf{O}\left(\frac{1}{\epsilon^2}\log\left(\frac{1}{p_{\text{fail}}}\right)\right)$ times. We now present an algorithm that estimates the overlap of a Hamiltonian block encoded operator, with comparable performance and no extra ancilla qubits.

Specifically, we assume that the input operator $A$ is Hamiltonian block encoded with $\norm{A}<1$ and singular value decomposition $A=U\Sigma V^\dagger$, whereas the quantum state $\ket{\psi}$ is prepared by a unitary. Then our algorithm proceeds as
$
    \exp\left(-i\left[\begin{smallmatrix}
        0 & A^\dagger\\
        A & 0
    \end{smallmatrix}\right]\right)
    \mapsto
    \exp\left(-i\left[\begin{smallmatrix}
        0 & f_{\text{sv}}^\dagger\left(A\right)\\
        f_{\text{sv}}\left(A \right)& 0
    \end{smallmatrix}\right]\right)
    \mapsto\bra{\psi}A\ket{\psi}
$.
Here, we perform Hamiltonian QSVT in the first step, with the target function $f$ to be specified momentarily. In the second step, we repeatedly measure the Hamiltonian block encoding and collect the following four types of measurement statistics:
\begin{enumerate}[label=(\roman*)]
    \item ($+Z$): we prepare the initial state $\ket{0}\ket{\psi}$ and project onto $\bra{0}\otimes I$, with probability
    \begin{equation}
        \norm{\left(\bra{0}\otimes I\right)\exp\left(-i\left[\begin{smallmatrix}
        0 & f_{\text{sv}}^\dagger\left(A\right)\\
        f_{\text{sv}}\left(A \right)& 0
    \end{smallmatrix}\right]\right)\ket{0}\ket{\psi}}^2
    =\norm{\cos\left(f\left(\Sigma\right)\right)V^\dagger\ket{\psi}}^2.
    \end{equation}
    \item ($-Z$): we prepare the initial state $\ket{1}\ket{\psi}$ and project onto $\bra{0}\otimes I$, with probability
    \begin{equation}
        \norm{\left(\bra{0}\otimes I\right)\exp\left(-i\left[\begin{smallmatrix}
        0 & f_{\text{sv}}^\dagger\left(A\right)\\
        f_{\text{sv}}\left(A \right)& 0
    \end{smallmatrix}\right]\right)\ket{1}\ket{\psi}}^2
    =\norm{\sin\left(f\left(\Sigma\right)\right)U^\dagger\ket{\psi}}^2.
    \end{equation}
    \item ($-X$): we prepare the initial state $\frac{\ket{0}-\ket{1}}{\sqrt{2}}\ket{\psi}$ and project onto $\bra{0}\otimes I$, with probability
    \begin{equation}
    \begin{aligned}
        &\norm{\left(\bra{0}\otimes I\right)\exp\left(-i\left[\begin{smallmatrix}
        0 & f_{\text{sv}}^\dagger\left(A\right)\\
        f_{\text{sv}}\left(A \right)& 0
    \end{smallmatrix}\right]\right)\frac{\ket{0}-\ket{1}}{\sqrt{2}}\ket{\psi}}^2\\
    &=\frac{1}{2}\left(\norm{\cos\left(f\left(\Sigma\right)\right)V^\dagger\ket{\psi}}^2
    +\norm{\sin\left(f\left(\Sigma\right)\right)U^\dagger\ket{\psi}}^2
    +\Im\left(\bra{\psi}U\sin\left(2f\left(\Sigma\right)\right)V^\dagger\ket{\psi}\right)\right).
    \end{aligned}
    \end{equation}
    \item ($+Y$): we prepare the initial state $\frac{\ket{0}+i\ket{1}}{\sqrt{2}}\ket{\psi}$ and project onto $\bra{0}\otimes I$, with probability
    \begin{equation}
    \begin{aligned}
        &\norm{\left(\bra{0}\otimes I\right)\exp\left(-i\left[\begin{smallmatrix}
        0 & f_{\text{sv}}^\dagger\left(A\right)\\
        f_{\text{sv}}\left(A \right)& 0
    \end{smallmatrix}\right]\right)\frac{\ket{0}+i\ket{1}}{\sqrt{2}}\ket{\psi}}^2\\
    &=\frac{1}{2}\left(\norm{\cos\left(f\left(\Sigma\right)\right)V^\dagger\ket{\psi}}^2
    +\norm{\sin\left(f\left(\Sigma\right)\right)U^\dagger\ket{\psi}}^2
    +\Re\left(\bra{\psi}U\sin\left(2f\left(\Sigma\right)\right)V^\dagger\ket{\psi}\right)\right).
    \end{aligned}
    \end{equation}
\end{enumerate}
Hence, we choose the odd function $f(x)=\frac{\arcsin(x)}{2}$
and estimate the overlap $\bra{\psi}A\ket{\psi}$ by combining estimates of its real and imaginary components. Similar to the Hadamard test, our method requires $\mathbf{O}\left(\frac{1}{\epsilon^2}\log\left(\frac{1}{p_{\text{fail}}}\right)\right)$ repeated measurements of quantum circuits, each with maximum query depth $\mathbf{O}\left(\log\left(\frac{1}{\epsilon}\right)\right)$. We present the algorithm and its analysis in detail in \sec{overlap}.

As an immediate application, we obtain a qubit-efficient quantum algorithm for estimating the Green's functions of many-body Hamiltonians~\cite{WangMcArdleBerta24,2021Yupreconditioned} which carry valuable spectroscopic information about the underlying quantum system. In this problem, the Hermitian part of the target operator is an odd function of the input Hamiltonian, and can thus be estimated by running our Hamiltonian QSVT and overlap estimation algorithm. On the other hand, the anti-Hermitian part is positive semidefinite up to a rescaling, and its overlap can be estimated by directly measuring a unitary block encoding. The resulting algorithm has a runtime matching that of the Fourier-based method~\cite{WangMcArdleBerta24,LinTong22}, using $1$ ancilla qubit throughout the entire computation.

\subsection{Simulating sum-of-squares Hamiltonians}
\label{sec:intro_sos}
As matrix arithmetic is a fundamental primitive in designing quantum algorithms, the techniques introduced here have the potential to reduce resource requirements across many applications. Specifically, when applied to quantum simulation, our methods inherit the commutator scaling of conventional product formulas, while utilizing the power of matrix arithmetics to improve the implementation of each Trotter step.

To illustrate this improvement, we consider the simulation of a class of \emph{sum-of-squares} Hamiltonians, which provide useful models for quantum systems of spins, fermions, and bosons~\cite{hastings2022perturbation}. Then, the target Hamiltonian has the form
\begin{equation}
    H=\sum_{k=1}^{n_K}H_k=\sum_{k=1}^{n_K}\left(\sum_{j_1=1}^{n_J}A_{j_1k}\right)^\dagger\left(\sum_{j_2=1}^{n_J}A_{j_2k}\right).
\end{equation}
For simplicity, we assume that each $A_{jk}$ is an elementary term whose Hamiltonian block encoding is available as input. In practice, $A_{jk}$ are monomials of spin, fermionic, and bosonic operators, and their Hamiltonian block encodings can be constructed through further matrix arithmetics. We will discuss this point below in the context of simulating sum-of-squares electronic structure.

When expanded in full, the target Hamiltonian contains $n_Kn_J^2$ terms. This suggests that each step of the simulation would have an implementation cost of $\mathbf{\Theta}\left(n_Kn_J^2\right)$ using conventional product formulas. We now explain how that can be improved using Hamiltonian-based matrix arithmetics.
To this end, let us start with a Hamiltonian block encoding of $A_{jk}$ which is available as input. Using the Lie-Trotter-Suzuki formulas, we combine the $n_J$ terms to get a Hamiltonian block encoding of $\sum_{j=1}^{n_J}A_{jk}$. Now we multiply the encoded operator and its Hermitian conjugate to obtain $H_k=\left(\sum_{j_1=1}^{n_J}A_{j_1k}\right)^\dagger\left(\sum_{j_2=1}^{n_J}A_{j_2k}\right)$ for each value of $k$. Finally, we take another summation of $n_K$ terms using product formulas. The resulting Hamiltonian block encoding is effectively a (bi-directional) controlled time evolution, and can thus be used in conjunction with quantum phase estimation to extract static properties of the system. The main stages are:
\begin{equation}
\begin{aligned}
    &\exp\left(-i\sqrt{\tau}\begin{bmatrix}
        0 & A_{jk}^\dagger\\
        A_{jk} & 0
    \end{bmatrix}\right)\\
    \overset{\text{Summation}}&{\longmapsto}\exp\left(-i\sqrt{\tau}\begin{bmatrix}
        0 & \sum_{j=1}^{n_J}A_{jk}^\dagger\\
        \sum_{j=1}^{n_J}A_{jk} & 0
    \end{bmatrix}\right)\\
    \overset{\text{Multiplication}}&{\underset{\dagger}{\longmapsto}}\exp\left(-i\tau\begin{bmatrix}
        0 & \left(\sum_{j_1=1}^{n_J}A_{j_1k}\right)^\dagger\left(\sum_{j_2=1}^{n_J}A_{j_2k}\right)\\
        \left(\sum_{j_1=1}^{n_J}A_{j_1k}\right)^\dagger\left(\sum_{j_2=1}^{n_J}A_{j_2k}\right) & 0
    \end{bmatrix}\right)\\
    \overset{\text{Summation}}&{\longmapsto}\exp\left(-i\tau\begin{bmatrix}
        0 & H\\
        H & 0
    \end{bmatrix}\right)
    \overset{\text{Redefinition}}{\longmapsto}\exp\left(-i\tau\begin{bmatrix}
        H & 0\\
        0 & -H
    \end{bmatrix}\right).
\end{aligned}
\end{equation}

To simulate $H$ for a long time $t$, we divide the evolution into $r$ steps and implement the above Hamiltonian matrix arithmetics with $\tau=\frac{t}{r}$. To ensure that the simulation error is at most $\epsilon$, it suffices to choose $r=\left(\alpha_{\text{comm}}+\alpha_{\infty}\right)t\left(\frac{n_K\alpha_{\infty}t}{\epsilon}\right)^{o(1)}$, where the effective normalization factors are $\alpha_{\text{comm}}=\sup_{p\in\mathbb{Z}_{\geq1}}\left(\sum_{k_1,\ldots,k_{p+1}=1}^{n_K}\norm{\left[H_{k_{p+1}},\ldots,\left[H_{k_2},H_{k_1}\right]\right]}\right)^{\frac{1}{p+1}}$ and the subdominant $\alpha_\infty=\max_{k=1,\ldots,n_K}\left(\sum_{j=1}^{n_J}\norm{A_{jk}}\right)^2$.
This simulation thus achieves the commutator scaling while the cost of each step is reduced to $\mathbf{O}\left(n_Kn_J\right)$, giving a total complexity of  
\begin{equation}
    \left(\alpha_{\text{comm}}+\alpha_{\infty}\right)t\left(\frac{n_K\alpha_{\infty}t}{\epsilon}\right)^{o(1)}n_Kn_J
\end{equation}
using $2$ ancilla qubits throughout. In contrast, other simulation algorithms such as qubitization~\cite{Low2016Qubitization} would have complexity $\mathbf{O}\left(\left(\alpha_1t+\log\left(\frac{1}{\epsilon}\right)\right)n_Kn_J\right)$ using $\mathbf{\Theta}\left(n_K\log\left(n_J\right)\right)$ ancillas, where $\alpha_1=\sum_{k=1}^{n_K}\left(\sum_{j=1}^{n_J}\norm{A_{jk}}\right)^2$.
Asymptotically, both $\alpha_{\text{comm}}$ and $\alpha_\infty$ fall below $\alpha_1$, often substantially so when $n_K$ is large but the terms nearly commute. See~\sec{sos} for further discussion.

As a concrete application, let us consider quantum simulation of the electronic structure Hamiltonians, whose solution offers valuable insights to the study of chemistry and materials science. Specifically, we focus on the Doubly Factorized Tensor HyperContracted (DFTHC) electronic structure Hamiltonians $H=\sum_{r=1}^{n_R}\sum_{c=1}^{n_C}H_{rc}^2
    =\sum_{r=1}^{n_R}\sum_{c=1}^{n_C}\left(w_{rc}I+\sum_{b=1}^{n_B}w_{rcb}\sum_{p,q=1}^nu_{rb,p}u_{rb,q}A_p^\dagger A_q\right)^2$,
with rank $n_R$, copies $n_C$, base size $n_B$ and number of spin orbitals $n$, where $A_p^\dagger$/$A_q$ are fermionic creation/annihilation operators, and $w$, $u$ are real coefficient tensors satisfying only the normalization condition $\norm{u_{rb,\cdot}}=1$, with a subdominant one-body term omitted for simplicity. This ansatz was recently introduced by~\cite{Low25}, which generalizes previous double factorizations with orthonormal coefficient tensors~\cite{PoulinTrotter,vonBurg21,Rocca2024,Oumarou2024acceleratingquantum,Peng2017} while capturing features of the tensor hypercontraction~\cite{Lee21}.

It is evident that a DFTHC Hamiltonian can be represented in the sum-of-squares form, and thus simulated using Hamiltonian matrix arithmetics as outlined above. However, several refinements can be made to improve the simulation performance. 
\begin{enumerate}[label=(\roman*)]
\item There is no need to implement each elementary term from $\sum_{p,q=1}^nu_{rb,p}u_{rb,q}A_p^\dagger A_q$ and take the summation. This is the product of two rotated Majorana operators up to a redefinition of the one-body terms and can be decomposed into $\mathbf{O}(n)$ Givens rotations~\cite[Lemma 8]{vonBurg21}.
\item Each $H_{rc}=w_{rc}I+\sum_{b=1}^{n_B}w_{rcb}\sum_{p,q=1}^nu_{rb,p}u_{rb,q}A_p^\dagger A_q$ is Hermitian. Instead of implementing $H_{rc}^\dagger H_{rc}$ using the generic matrix multiplication algorithm, we realize the Hamiltonian squaring $H_{rc}\mapsto H_{rc}^2$ with Hamiltonian QSVT.
\item Taking nested commutators of $H_{rc}$ yields a quadratic fermionic operator, whose coefficients are nested commutators of those of $H_{rc}$. Therefore, $\alpha_{\text{comm}}$ can be tightly bounded using eigenvalues of the coefficient tensors, and such a bound can be evaluated in polynomial time on a classical computer (\append{fermionic}).
\item All the Hamiltonian-based matrix arithmetic operations preserve the fermionic particle number. This allows us to only sum over $\eta$ fermionic modes when the initial state contains $\eta$ electrons, extending previous results for the planewave basis~\cite{Su2021nearlytight,McArdleCampbell22} and real-space grids~\cite{TrotterStep23}.
\end{enumerate}

Incorporating all the above improvements, we obtain a simulation algorithm with an effective normalization factor $\alpha_{\text{comm},\eta}$, depending on the largest $\eta$ eigenvalues of nested commutators of coefficient tensors in absolute value, better than that of the qubitization approach~\cite{vonBurg21,Lee21} but incomparable to that of the gap amplification method~\cite{Low25}.
Within each step, our algorithm implements a Hamiltonian evolution under the square of a sum of fermionic operators. In the special case where the coefficient tensors are perfectly orthonormal $\langle u_{rb_1,\cdot},u_{rb_2,\cdot}\rangle=\pmb{\delta}_{b_1,b_2}$, all operators $\sum_{p,q}u_{rb,p}u_{rb,q}A_p^\dagger A_q$ pairwise commute and can be simultaneously diagonalized. It then suffices to perform a parallelized version of thin QR decomposition~\cite[2.1.P28]{horn2012matrix} with $\mathbf{O}\left(n_Bn\right)$ Givens rotations~\cite{Kivlichan18}. However, when the coefficient tensors are not orthogonal, this circuit implementation is not directly applicable.
Our algorithm performs the square function using Hamiltonian QSVT, regaining the scaling $\sim n_Bn$ up to a logarithmic factor, leading to a Trotter step with total gate complexity $\sim n_Rn_Cn_Bn$. We note that this matches the total gate count of a single qubitization step in the double factorization case~\cite{vonBurg21} where $n_R,n_B=\mathbf{\Theta}(n)$ and $n_C=\mathbf{\Theta}(1)$, but falls short in Toffoli count or total gate count per step compared to the more recent approaches~\cite{Lee21,Low25}, owing to the more delicate use of Quantum Read-Only Memory (QROM) in the latter.

We achieve this with $1$ ancilla qubit.

We present preliminaries in \sec{block} on the Hamiltonian block encoding, and conclude in \sec{discuss} with a summary of our main contributions and open questions for future work.

%% file: block.tex
In this section, we provide preliminaries essential for understanding our results on Hamiltonian-based matrix arithmetics. We start by presenting an abstract definition of the Hamiltonian block encoding in \sec{block_abstract}. We then discuss several Hamiltonian matrix operations whose implementations either follow directly from the definition or are known from existing work. This includes Hermitian conjugation (\sec{block_conjugate}), complex phase scaling and integer scaling (\sec{block_scaling}), and matrix addition (\sec{block_addition}). We conclude in \sec{block_error} with an analysis of how error propagates between an operator and its Hamiltonian block encoding.

\subsection{Abstract definition}
\label{sec:block_abstract}

It is well known that any matrix $A$ can be dilated to a Hermitian operator of the form $\left[\begin{smallmatrix}
    0 & A^\dagger\\
    A & 0
\end{smallmatrix}\right]$. This allows us to construct a Hamiltonian block encoding of $A$ as
\begin{equation}
    \exp\left(-i
    \begin{bmatrix}
        0 & A^\dagger\\
        A & 0
    \end{bmatrix}\right)
    =e^{-i\left(X\otimes\Re(A)+Y\otimes\Im(A)\right)},
\end{equation}
where $X$ and $Y$ are Pauli matrices, and $A=\Re(A)+i\Im(A)$ is the unique decomposition of $A$ into its Hermitian and anti-Hermitian components $\Re(A)=\frac{A+A^\dagger}{2}$ and $\Im(A)=\frac{A-A^\dagger}{2i}$~\cite[Theorem 4.1.2]{horn2012matrix}. In using this matrix decomposition, we have implicitly assumed that $A$ is a square operator. We will impose this assumption throughout the paper to simplify the notation, though it is possible to handle nonsquare matrices by padding them with zero entries.

This is however not the only way to define Hamiltonian block encodings. For instance, one can also consider
\begin{equation}
    \exp\left(-i
    \begin{bmatrix}
        \Re(A) & \Im(A)\\
        \Im(A) & -\Re(A)
    \end{bmatrix}\right)
    =e^{-i\left(Z\otimes\Re(A)+X\otimes\Im(A)\right)},
\end{equation}
which relates to the previous construction through the change of basis $e^{-i\left(Z\otimes\Re(A)+X\otimes\Im(A)\right)}
    =\left(S\cdot\mathrm{Had}\otimes I\right)e^{-i\left(X\otimes\Re(A)+Y\otimes\Im(A)\right)}\left(\mathrm{Had}\cdot S^\dagger\otimes I\right)$
with $Z$ the Pauli matrix, $\mathrm{Had}=\frac{1}{\sqrt{2}}\left[\begin{smallmatrix}
    1 & 1\\
    1 & -1
\end{smallmatrix}\right]$ the Hadamard gate and $S=\left[\begin{smallmatrix}
    1 & 0\\
    0 & i
\end{smallmatrix}\right]$ the phase gate.
Moreover, we introduce the controlled version of Hamiltonian block encoding
\begin{equation}
    \exp\left(-i
    \begin{bmatrix}
        0 & A^\dagger & 0 & 0\\
        A & 0 & 0 & 0\\
        0 & 0 & 0 & 0\\
        0 & 0 & 0 & 0\\
    \end{bmatrix}\right)
    =\ketbra{0}{0}\otimes e^{-i\left(X\otimes\Re(A)+Y\otimes\Im(A)\right)}+\ketbra{1}{1}\otimes I\otimes I,
\end{equation}
which is useful for multiplying generic matrices with Hamiltonian evolution and Hamiltonian QSVT for even polynomials.

Abstractly, we define the Hamiltonian block encoding of $A$ by the evolution
\begin{equation}
    e^{-i\left(\Gamma^\dagger\otimes A+\Gamma\otimes A^\dagger\right)}.
\end{equation}
Here, $\Gamma^\dagger$ and $\Gamma$ are basis operators encoding the input matrix $A$ and its Hermitian conjugate, which are padded fermionic creation and annihilation operators up to a change of basis, satisfying the equivalent conditions listed below.
\begin{proposition}
\label{prop:fermionic_anticomm}
The following statements are equivalent for an operator $\Gamma$.
\begin{enumerate}
    \item \begin{enumerate}
        \item $\Gamma^2=0$;
        \item $\left(\Gamma^\dagger\Gamma\right)^2=\Gamma^\dagger\Gamma$ is an orthogonal projection.
    \end{enumerate}
    \item \begin{enumerate}
        \item $\Gamma^2=0$;
        \item $\Gamma\Gamma^\dagger\Gamma=\Gamma$.
    \end{enumerate}
    \item There exists a unitary $U$ such that
    \begin{equation}
        \Gamma=U
        \begin{bmatrix}
            0 & I & 0\\
            0 & 0 & 0\\
            0 & 0 & 0
        \end{bmatrix}
        U^\dagger.
    \end{equation}
\end{enumerate}
\end{proposition}
\begin{proof}
The implication $3\Rightarrow2\Rightarrow1$ follows from a direct verification. 

Assume Statement $1$ is true. Given an operator satisfying $\Gamma^2=0$, we invoke~\cite[3.4.P5]{horn2012matrix} to find a unitary $U_1$ such that
\begin{equation}
    \Gamma=U_1
    \begin{bmatrix}
        \begin{bmatrix}
            0 & \sigma_1\\
            0 & 0
        \end{bmatrix}& & &\\
        & \ddots & &\\
        & & \begin{bmatrix}
            0 & \sigma_r\\
            0 & 0
        \end{bmatrix} &\\
        & & & \begin{bmatrix}
            0 & &\\
            & \ddots &\\
            & & 0
        \end{bmatrix}
    \end{bmatrix}
    U_1^\dagger
    =U_1
    \begin{bmatrix}
        0 & 1 & & & & & &\\
        0 & 0 & & & & & &\\
        & & \ddots & & & & &\\
        & & & 0 & 1 & & &\\
        & & & 0 & 0 & & &\\
        & & & & &0 & &\\
        & & & & & &\ddots &\\
        & & & & & & &0\\
    \end{bmatrix}
    U_1^\dagger,
\end{equation}
where $r=\mathbf{Rank}(\Gamma)$ and $\sigma_1,\ldots,\sigma_r$ are all the positive singular values of $\Gamma$. Because $\left(\Gamma^\dagger\Gamma\right)^2=\Gamma^\dagger\Gamma$, $\Gamma$ only has singular values $1$ and $0$, leading to the above representation.

Now switching to the Dirac notation, we have
\begin{equation}
    \Gamma=U_1\sum_{j=0}^{r-1}\ketbra{2j}{2j+1}U_1^\dagger
    =U_1U_2\sum_{j=0}^{r-1}\ketbra{j}{j+r}U_2^\dagger U_1^\dagger,
\end{equation}
where $U_2$ is the permutation
\begin{equation}
    U_2=\sum_{j=0}^{r-1}\left(\ketbra{2j}{j}+\ketbra{2j+1}{j+r}\right).
\end{equation}
This gives the claimed matrix representation in Statement $3$.
\end{proof}

Alternatively, we can also define the Hamiltonian block encoding of $A=\Re(A)+i\Im(A)$ by the evolution
\begin{equation}
    e^{-i\left(\Gamma_\Re\otimes\Re(A)+\Gamma_\Im\otimes\Im(A)\right)}.
\end{equation}
Here, $\Gamma_{\Re}$ and $\Gamma_{\Im}$ are basis operators encoding the Hermitian and anti-Hermitian part of $A$ respectively, which are padded Pauli operators up to a change of basis, satisfying the equivalent conditions listed below. This is reminiscent of a more general canonical form for a family of diagonalizable anticommuting operators~\cite{KUMBASAR201279}~\cite[Section 7.1]{Rosenberg04}, but we present a proof of the special case for completeness.

\begin{proposition}
\label{prop:majorana_anticomm}
The following statements are equivalent for operators $\Gamma_{\Re}$ and $\Gamma_{\Im}$.
\begin{enumerate}
    \item $\Gamma_{\Re}$ and $\Gamma_{\Im}$ are both normal, satisfying
    \begin{enumerate}
        \item $\Gamma_{\Re}\Gamma_{\Im}=-\Gamma_{\Im}\Gamma_{\Re}$;
        \item $\Gamma_{\Re}^2=\Gamma_{\Im}^2$.
        \item $\Gamma_{\Re}^3=\Gamma_{\Re}$;
        \item $\Gamma_{\Im}^3=\Gamma_{\Im}$.
    \end{enumerate}
    \item $\Gamma_{\Re}$ and $\Gamma_{\Im}$ are both normal, satisfying
    \begin{enumerate}
        \item $\Gamma_{\Re}\Gamma_{\Im}=-\Gamma_{\Im}\Gamma_{\Re}$;
        \item $\Gamma_{\Re}^2\Gamma_{\Im}=\Gamma_{\Im}$;
        \item $\Gamma_{\Im}^2\Gamma_{\Re}=\Gamma_{\Re}$.
    \end{enumerate}
    \item There exists a unitary $U$ such that
    \begin{equation}
        \Gamma_{\Re}=U
        \begin{bmatrix}
            I & 0 & 0\\
            0 & -I & 0\\
            0 & 0 & 0
        \end{bmatrix}U^\dagger,\qquad
        \Gamma_{\Im}=U
        \begin{bmatrix}
            0 & I & 0\\
            I & 0 & 0\\
            0 & 0 & 0
        \end{bmatrix}U^\dagger,
    \end{equation}
    where matrix blocks are partitioned conformally.
\end{enumerate}
\end{proposition}
\begin{proof}
    We start with the implication $3\Rightarrow 2$ which follows from a direct verification. 

    Assume Statement $2$ is true. We have
    \begin{equation}
        \Gamma_{\Re}^2=\Gamma_{\Re}\left(\Gamma_{\Im}^2\Gamma_{\Re}\right)
        =\Gamma_{\Im}\left(\Gamma_{\Re}^2\Gamma_{\Im}\right)
        =\Gamma_{\Im}^2,
    \end{equation}
    and hence $1(b)$ is true. This immediately implies $1(c)$ and $1(d)$ as
    \begin{equation}
        \Gamma_{\Re}^3=\Gamma_{\Im}^2\Gamma_{\Re}=\Gamma_{\Re},\qquad
        \Gamma_{\Im}^3=\Gamma_{\Re}^2\Gamma_{\Im}=\Gamma_{\Im}.
    \end{equation}

    It remains to prove $1\Rightarrow 3$. Since $\Gamma_{\Re}$ and $\Gamma_{\Im}$ are normal, they can be unitarily diagonalized. All their eigenvalues $\lambda$ satisfy the constraint $\lambda^3=\lambda$ and must therefore be $\pm1$ or $0$. We thus deduce that $\Gamma_{\Re}$ and $\Gamma_{\Im}$ are actually Hermitian and that
\begin{equation}
    \Gamma_{\Re}=U_1
    \begin{bmatrix}
        I & 0 & 0\\
        0 & -I & 0\\
        0 & 0 & 0
    \end{bmatrix}U_1^\dagger
\end{equation}
for some unitary $U_1$, whereas the anticommutation relation forces
\begin{equation}
    \Gamma_{\Im}=U_1
    \begin{bmatrix}
        0 & B_{12} & 0\\
        B_{12}^\dagger & 0 & 0\\
        0 & 0 & B_{33}
    \end{bmatrix}U_1^\dagger.
\end{equation}
Here, $B_{33}$ is a square Hermitian operator, but $B_{12}$ is not known to be square for now.

Without loss of generality, suppose that $B_{12}$ has the singular value decomposition $B_{12}=V
\begin{bmatrix}
    \Sigma & 0
\end{bmatrix}W^\dagger$ for some nonnegative diagonal $\Sigma$ and unitaries $V$ and $W$, which transforms
\begin{equation}
    \begin{bmatrix}
        0 & B_{12} & 0\\
        B_{12}^\dagger & 0 & 0\\
        0 & 0 & B_{33}
    \end{bmatrix}
    =\begin{bmatrix}
        V & 0 & 0\\
        0 & W & 0\\
        0 & 0 & I
    \end{bmatrix}
    \begin{bmatrix}
        0 & \begin{bmatrix}
            \Sigma & 0
        \end{bmatrix} & 0\\
        \begin{bmatrix}
            \Sigma\\
            0
        \end{bmatrix} & 0 & 0\\
        0 & 0 & B_{33}
    \end{bmatrix}
    \begin{bmatrix}
        V^\dagger & 0 & 0\\
        0 & W^\dagger & 0\\
        0 & 0 & I
    \end{bmatrix}
\end{equation}
but preserves
\begin{equation}
    \begin{bmatrix}
        I & 0 & 0\\
        0 & -I & 0\\
        0 & 0 & 0
    \end{bmatrix}
    =\begin{bmatrix}
        V & 0 & 0\\
        0 & W & 0\\
        0 & 0 & I
    \end{bmatrix}
    \begin{bmatrix}
        I & 0 & 0\\
        0 & -I & 0\\
        0 & 0 & 0
    \end{bmatrix}
    \begin{bmatrix}
        V^\dagger & 0 & 0\\
        0 & W^\dagger & 0\\
        0 & 0 & I
    \end{bmatrix}.
\end{equation}
Applying $\Gamma_{\Re}^2=\Gamma_{\Im}^2$, we see that $B_{33}=0$, $\Sigma=I$ and that $B_{12}$ is in fact a square matrix. The proof is complete by letting $U=U_1
    \left[\begin{smallmatrix}
        V & 0 & 0\\
        0 & W & 0\\
        0 & 0 & I
    \end{smallmatrix}\right]$.
\end{proof}

Given basis operators $\Gamma_\Re$ and $\Gamma_\Im$ satisfying the above equivalent conditions, we have
\begin{equation}
\begin{aligned}
    \mathbf{Tr}\left(\Gamma_\Re\right)
    &=\mathbf{Tr}\left(\Gamma_\Im\right)
    =0,\\
    \mathbf{Tr}\left(\Gamma_\Re^2\right)
    &=\mathbf{Tr}\left(\Gamma_\Im^2\right)
    =\mathbf{Rank}\left(\Gamma_\Re\right)
    =\mathbf{Rank}\left(\Gamma_\Im\right).
\end{aligned}
\end{equation}
We require that the ancilla space can be expanded to have dimensions at least $2\mathbf{Rank}\left(\Gamma_\Re\right)
=2\mathbf{Rank}\left(\Gamma_\Im\right)$, so that Hamiltonian-based matrix multiplication (\thm{multiply}) and Hamiltonian QSVT for even polynomials and generic inputs (\thm{qsvt_even}) can be realized up to a change of basis.

As a concrete example, consider the set of Majorana operators $\{\Gamma_j\}$ satisfying $\Gamma_j\Gamma_k+\Gamma_k\Gamma_j=2\pmb{\delta}_{j,k}I$~\cite{Majorana20}.
This means $\left(\frac{\Gamma_j}{\sqrt{2}}\right)^2=I$ for all $j$ and $\frac{\Gamma_j}{\sqrt{2}}\frac{\Gamma_k}{\sqrt{2}}=-\frac{\Gamma_k}{\sqrt{2}}\frac{\Gamma_j}{\sqrt{2}}$ whenever $j\neq k$. We thus conclude that any two Majorana operators $\Gamma_j$ and $\Gamma_k$ can be used to define a Hamiltonian block encoding $e^{-\frac{i}{\sqrt{2}}\left(\Gamma_j\otimes\Re(A)+\Gamma_k\otimes\Im(A)\right)}$.
For notational convenience, we select $\exp\left(-i\left[\begin{smallmatrix}
    0 & A^\dagger\\
    A & 0
\end{smallmatrix}\right]\right)$ as the standard form and denote
\begin{equation}
    E_{A}
    =\exp\left(-i
    \begin{bmatrix}
        0 & A^\dagger\\
        A & 0
    \end{bmatrix}\right)
    =e^{-i\left(X\otimes\Re(A)+Y\otimes\Im(A)\right)}.
\end{equation}
We then quantify the complexity of our methods in terms of number of calls to $E_{A}$ and its controlled version. According to~\prop{fermionic_anticomm} and~\prop{majorana_anticomm}, any Hamiltonian block encoding can be converted to the standard form by applying a unitary conjugation on the ancilla system with no query overhead.

\subsection{Hermitian conjugation}
\label{sec:block_conjugate}
One elementary operation on the Hamiltonian block encoding is Hermitian conjugation. Specifically, supposing $A$ is an operator with the corresponding encoding $\exp\left(-i\left[\begin{smallmatrix}
    0 & A^\dagger\\
    A & 0
\end{smallmatrix}\right]\right)$, our goal is to realize the transformation
\begin{equation}
    \exp\left(-i
    \begin{bmatrix}
        0 & A^\dagger\\
        A & 0
    \end{bmatrix}\right)
    \mapsto
    \exp\left(-i
    \begin{bmatrix}
        0 & A\\
        A^\dagger & 0
    \end{bmatrix}\right)
\end{equation}
so the result is a Hamiltonian block encoding of $A^\dagger$.

To this end, we rewrite it in the Pauli basis as
\begin{equation}
    e^{-i\left(X\otimes\Re(A)+Y\otimes\Im(A)\right)}
    \mapsto e^{-i\left(X\otimes\Re(A)-Y\otimes\Im(A)\right)}.
\end{equation}
Our goal is then to find a unitary $U$ such that its unitary conjugation has the action
\begin{equation}
    UXU^\dagger=X,\qquad
    UYU^\dagger=-Y,
\end{equation}
which can be achieved by simply setting $U=X$. We thus obtain:

\begin{proposition}[Hermitian conjugation]
\label{prop:herm_conjugate}
Let $A$ be a matrix encoded by the Hamiltonian block encoding $E_{A}=\exp\left(-i\left[\begin{smallmatrix}
    0 & A^\dagger\\
    A & 0
\end{smallmatrix}\right]\right)$. Then the Hamiltonian block encoding
\begin{equation}
    E_{A^\dagger}
    =\exp\left(-i
    \begin{bmatrix}
        0 & A\\
        A^\dagger & 0
    \end{bmatrix}\right)
\end{equation}
can be constructed with zero error using $1$ query to $E_{A}$. See~\fig{conjugate} for the corresponding circuit diagram.
\end{proposition}

\begin{figure}[t]
	\centering
\includegraphics[scale=\circuitwidth]{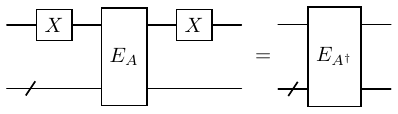}
\caption{Quantum circuit for Hermitian conjugation.}
\label{fig:conjugate}
\end{figure}

\subsection{Complex phase scaling and integer scaling}
\label{sec:block_scaling}
Next, we consider multiplying a Hamiltonian block encoded operator by a complex scalar. Specifically, given the Hamiltonian block encoding $\exp\left(-i\left[\begin{smallmatrix}
    0 & A^\dagger\\
    A & 0
\end{smallmatrix}\right]\right)$ and a phase angle $\theta\in\mathbb{R}$, our goal is to perform the transformation
\begin{equation}
    \exp\left(-i
    \begin{bmatrix}
        0 & A^\dagger\\
        A & 0
    \end{bmatrix}\right)
    \mapsto
    \exp\left(-i
    \begin{bmatrix}
        0 & e^{-i\theta}A^\dagger\\
        e^{i\theta}A & 0
    \end{bmatrix}\right),
\end{equation}
with the phased operator $e^{i\theta}A$ encoded by the resulting Hamiltonian evolution. This can be realized through the unitary conjugation
\begin{equation}
    \begin{bmatrix}
        e^{-i\frac{\theta}{2}} & 0\\
        0 & e^{i\frac{\theta}{2}}
    \end{bmatrix}\exp\left(-i
    \begin{bmatrix}
        0 & A^\dagger\\
        A & 0
    \end{bmatrix}\right)
    \begin{bmatrix}
        e^{i\frac{\theta}{2}} & 0\\
        0 & e^{-i\frac{\theta}{2}}
    \end{bmatrix}
    =\exp\left(-i
    \begin{bmatrix}
        0 & e^{-i\theta}A^\dagger\\
        e^{i\theta}A & 0
    \end{bmatrix}\right),
\end{equation}
giving:

\begin{proposition}[Complex phase scaling]
\label{prop:phase_scale}
Let $A$ be a matrix encoded by the Hamiltonian block encoding $E_{A}=\exp\left(-i\left[\begin{smallmatrix}
    0 & A^\dagger\\
    A & 0
\end{smallmatrix}\right]\right)$. For any $\theta\in\mathbb{R}$, the Hamiltonian block encoding
\begin{equation}
    E_{e^{i\theta}A}=\exp\left(-i
    \begin{bmatrix}
        0 & e^{-i\theta}A^\dagger\\
        e^{i\theta}A & 0
    \end{bmatrix}\right)
\end{equation}
can be constructed with zero error using $1$ query to $E_{A}$. See~\fig{complex_phase} for the corresponding circuit diagram.
\end{proposition}

\begin{figure}[t]
	\centering
\includegraphics[scale=\circuitwidth]{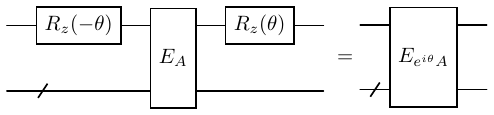}
\caption{Quantum circuit for complex phase scaling ($R_z(\theta)=e^{-i\frac{\theta}{2}Z}$).}
\label{fig:complex_phase}
\end{figure}

It is also fairly straightforward to perform integer scaling on a Hamiltonian block encoded operator:
\begin{equation}
    \exp\left(-i
    \begin{bmatrix}
        0 & A^\dagger\\
        A & 0
    \end{bmatrix}\right)
    \mapsto
    \exp\left(-i
    \begin{bmatrix}
        0 & nA^\dagger\\
        nA & 0
    \end{bmatrix}\right)
\end{equation}
with $n\in\mathbb{Z}_{\geq0}$ a nonnegative integer. This can be realized by simply repeating the original Hamiltonian block encoding $n$ times.

\begin{proposition}[Integer scaling]
\label{prop:integer_scale}
Let $A$ be a matrix encoded by the Hamiltonian block encoding $E_{A}=\exp\left(-i\left[\begin{smallmatrix}
    0 & A^\dagger\\
    A & 0
\end{smallmatrix}\right]\right)$. For any $n\in\mathbb{Z}_{\geq0}$, the Hamiltonian block encoding
\begin{equation}
    E_{nA}=\exp\left(-i
    \begin{bmatrix}
        0 & nA^\dagger\\
        nA & 0
    \end{bmatrix}\right)
\end{equation}
can be constructed with zero error using $n$ queries to $E_{A}$. See~\fig{integer} for the corresponding circuit diagram.
\end{proposition}

\begin{figure}[t]
	\centering
\includegraphics[scale=\circuitwidth]{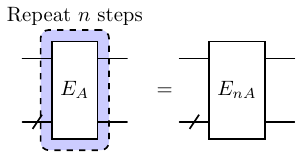}
\caption{Quantum circuit for integer scaling.}
\label{fig:integer}
\end{figure}

\subsection{Matrix addition}
\label{sec:block_addition}
We now consider adding Hamiltonian block encoded operators. Specifically, assume that we have Hamiltonian block encodings $\exp\left(-i\left[\begin{smallmatrix}
    0 & A_j^\dagger\\
    A_j & 0
\end{smallmatrix}\right]\right)$ for $j=1,\ldots,\gamma$. Then our goal is to perform the transformation
\begin{equation}
    \exp\left(-i
    \begin{bmatrix}
        0 & A_j^\dagger\\
        A_j & 0
    \end{bmatrix}\right)\ (j=1,\ldots,\gamma)
    \mapsto
    \exp\left(-i
    \begin{bmatrix}
        0 & \sum_{j=1}^\gamma A_j^\dagger\\
        \sum_{j=1}^\gamma A_j & 0
    \end{bmatrix}\right).
\end{equation}

This can be realized using higher-order product formulas, which are widely used in quantum simulation~\cite{BACS05}. Examples include the family of Lie-Trotter-Suzuki product formulas, which are defined recursively for Hermitian operators $H_1,\ldots,H_\gamma$ as
\begin{equation}
\begin{aligned}
    S_1(\tau )&=e^{-i\tau H_\gamma}\cdots e^{-i\tau H_1},\\
    S_2(\tau )&=e^{-i\frac{\tau }{2}H_1}\cdots e^{-i\frac{\tau }{2}H_\gamma}e^{-i\frac{\tau }{2}H_\gamma}\cdots e^{-i\frac{\tau }{2}H_1},\\
    S_{2k}(\tau )&=S_{2k-2}^2\left(u_k\tau \right)S_{2k-2}\left((1-4u_k)\tau \right)S_{2k-2}^2\left(u_k\tau \right),
\end{aligned}
\end{equation}
where $u_k=\frac{1}{4-4^{1/(2k-1)}}$ and $k=2,3,\ldots$ When expanded into Taylor series, the formula $S_{2k}(\tau)$ agrees with the ideal evolution $\exp\left(-i\tau\sum_{j=1}^\gamma H_j\right)$ up to order $\tau^{2k}$; hence we say $S_{2k}(\tau)$ is a ($2k$)th-order product formula. This provides a good approximation to the ideal evolution in the limit $\tau\rightarrow0$. To simulate for an arbitrary time $t$, we divide the evolution into $r$ steps and apply a product formula within each step for time $\tau=\frac{t}{r}$.

It was shown in~\cite{CSTWZ19} that product formulas can simulate Hamiltonian dynamics with a cost depending on the norm of nested commutators of individual summands. Specifically, for any Hermitian operators $H_1,\ldots,H_\gamma$ and a $p$th-order product formula $S_p(\tau)$ with $p\in\mathbb{Z}_{\geq1}$, the approximation
\begin{equation}
    \norm{S_p\left(\frac{t}{r}\right)^r-e^{-it\sum_{j=1}^\gamma H_j}}\leq\epsilon
\end{equation}
can be achieved with accuracy $\epsilon$ by choosing
\begin{equation}
    r=\mathbf{O}\left(\frac{\alpha_{\text{comm}}^{1+1/p}t^{1+1/p}}{\epsilon^{1/p}}\right),
\end{equation}
where
\begin{equation}
\label{eq:alpha_comm}
    \alpha_{\text{comm}}=\left(\sum_{j_1,\ldots,j_{p+1}=1}^{\gamma}\norm{\left[H_{j_{p+1}},\ldots,\left[H_{j_2},H_{j_1}\right]\right]}\right)^{\frac{1}{p+1}}.
\end{equation}
This commutator scaling holds in general as long as the product formula agrees with the ideal evolution up to order $\tau^p$. This covers the Lie-Trotter-Suzuki formulas introduced above but applies to other constructions as well, which may lead to simulations with smaller constant prefactors in their complexities.

In general, a $p$th-order product formula takes the form
\begin{equation}
\label{eq:def_pfp}
    S_p(\tau)=\prod_{l=1}^\upsilon\prod_{j=1}^\gamma e^{-ita_{l,j}H_{\pi_{l}(j)}},
\end{equation}
where $\upsilon$ is the total number of stages depending only on the value of $p$, $\pi_l$ is the permutation of terms within stage $l$, and $a_{l,j}$ are real coefficients such that all $\abs{a_{l,j}}\leq1$. Choosing a sufficiently large but constant value of $p$, we can simplify the scaling of number of Trotter steps as
\begin{equation}
    r=\frac{\alpha_{\text{comm}}^{1+o(1)}t^{1+o(1)}}{\epsilon^{o(1)}},
\end{equation}
where $o(1)$ represents a positive number that approaches $0$ as $p$ grows. Here, we have slightly abused the notation because the norm of nested commutators actually depends on the number of nesting layers $p$. So the scaling should be understood as
\begin{equation}
    \alpha_{\text{comm}}=\sup_{p\in\mathbb{Z}_{\geq1}}\left(\sum_{j_1,\ldots,j_{p+1}=1}^{\gamma}\norm{\left[H_{j_{p+1}},\ldots,\left[H_{j_2},H_{j_1}\right]\right]}\right)^{\frac{1}{p+1}}.
\end{equation}
We can use this to add Hamiltonian block encoded matrices as follows.

\begin{proposition}[Matrix addition]
\label{prop:add}
Let $A_j$ ($j=1,\ldots,\gamma$) be matrices encoded by the Hamiltonian block encodings $E_{A_j}=\exp\left(-i\left[\begin{smallmatrix}
    0 & A_j^\dagger\\
    A_j & 0
\end{smallmatrix}\right]\right)$ with $\norm{A_j}<\frac{\pi}{2}$. Then the Hamiltonian block encoding
\begin{equation}
    E_{\sum_{j=1}^\gamma A_j}=
    \exp\left(-i
    \begin{bmatrix}
        0 & \sum_{j=1}^\gamma A_j^\dagger\\
        \sum_{j=1}^\gamma A_j & 0
    \end{bmatrix}\right)
\end{equation}
can be constructed with accuracy $\epsilon$ using
\begin{equation}
    r=\frac{\alpha_{\text{comm}}^{1+o(1)}}{\epsilon^{o(1)}}\log(\gamma)
\end{equation}
queries to every $E_{A_j}$, where
\begin{equation}
    \alpha_{\text{comm}}=\sup_{p\in\mathbb{Z}_{\geq1}}\left(\sum_{j_1,\ldots,j_{p+1}=1}^{\gamma}\norm{\left[\begin{bmatrix}
        0 & A_{j_{p+1}}^\dagger\\
        A_{j_{p+1}} & 0
    \end{bmatrix},\ldots,\left[\begin{bmatrix}
        0 & A_{j_{2}}^\dagger\\
        A_{j_{2}} & 0
    \end{bmatrix},\begin{bmatrix}
        0 & A_{j_{1}}^\dagger\\
        A_{j_{1}} & 0
    \end{bmatrix}\right]\right]}\right)^{\frac{1}{p+1}}.
\end{equation}
See~\fig{addition} for the corresponding circuit diagram.
\end{proposition}

\begin{figure}[t]
	\centering
\includegraphics[scale=\circuitwidth]{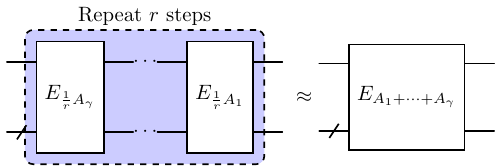}
\caption{Quantum circuit for Hamiltonian matrix addition. For illustration purposes, only the lowest-order formula is shown without the fractional scaling step.}
\label{fig:addition}
\end{figure}

The above result is established by dividing the desired Hamiltonian block encoding $E_{\sum_{j=1}^\gamma A_j}$ into $r$ steps, and simulating every step using a $p$th-order product formula $S_p\left(\frac{1}{r}\right)$, each making queries to $E_{\left(a_{l,\pi_l^{-1}(j)}/r\right)A_j}$ according to \eq{def_pfp}. These encodings may be constructed directly when $A_j$ are elementary terms. Alternatively, each effective evolution time $\abs{a_{l,\pi_l^{-1}(j)}}/r\leq 1$ is fractional, and the rescaled Hamiltonian block encodings can be constructed from $E_{A_j}$ via fractional scaling~\cor{frac_scale}, introducing a logarithmic overhead in the query complexity. 
Negative evolution times may be achieved via~\prop{phase_scale} at no overhead.

When applied to Hamiltonian block encodings, product formulas will have a cost depending on the nested commutators of $\left[\begin{smallmatrix}
    0 & A_j^\dagger\\
    A_j & 0
\end{smallmatrix}\right]$. These commutators can be computed as
\begin{equation}
    \left[\begin{bmatrix}
        0 & A_j^\dagger\\
        A_j & 0
    \end{bmatrix},\begin{bmatrix}
        0 & A_k^\dagger\\
        A_k & 0
    \end{bmatrix}\right]
    =\begin{bmatrix}
        A_j^\dagger A_k-A_k^\dagger A_j & 0\\
        0 & A_jA_k^\dagger-A_kA_j^\dagger
    \end{bmatrix}
    =\begin{bmatrix}
        \left[A_j,A_k\right]_\dagger & 0\\
        0 & \left[A_j^\dagger,A_k^\dagger\right]_\dagger
    \end{bmatrix},
\end{equation}
where the standard matrix commutator is replaced by
\begin{equation}
    \left[A_j,A_k\right]_\dagger=A_j^\dagger A_k-A_k^\dagger A_j.
\end{equation}
Alternatively, we may use the Pauli representation to evaluate 
\begin{equation}
\begin{aligned}
    &\left[X\otimes\Re(A_j)+Y\otimes\Im(A_j),X\otimes\Re(A_k)+Y\otimes\Im(A_k)\right]\\
    &=I\otimes\left(\left[\Re(A_j),\Re(A_k)\right]+\left[\Im(A_j),\Im(A_k)\right]\right)
    +iZ\otimes\left(\left\{\Re(A_j),\Im(A_k)\right\}-\left\{\Re(A_k),\Im(A_j)\right\}\right),
\end{aligned}
\end{equation}
where $\{\cdot,\cdot\}$ is the matrix anticommutator
\begin{equation}
    \left\{\Re(A_j),\Im(A_k)\right\}
    =\Re(A_j)\Im(A_k)+\Im(A_k)\Re(A_j).
\end{equation}
Both of these reduce to the standard commutator expression when the input operators $A_j$ are purely Hermitian.

\subsection{Error propagation}
\label{sec:block_error}
In a number of our matrix arithmetic results, the constructed operator is an approximation of the intended Hamiltonian block encoding. That is, we aim to construct a quantum algorithm $W$, which approximates the desired Hamiltonian block encoding of $A$ with accuracy $\epsilon$ in the sense that
\begin{equation}
    \norm{W-E_A}
    =\norm{W-\exp\left(-i\begin{bmatrix}
    0 & A^\dagger\\
    A & 0
\end{bmatrix}\right)}
    \leq\epsilon.
\end{equation}
However, if we know a priori that $W=E_B$ is also a Hamiltonian block encoding, and that $\norm{A},\norm{B}$ are sufficiently small, then $B$ is guaranteed to be an accurate operator approximation of $A$. This may be of use in applications such as ground energy estimation where one makes queries to the time evolution of normalized Hamiltonians.

\begin{lemma}[Perturbation of matrix exponential and matrix logarithm]
For matrices $J$ and $K$,
\begin{equation}
    \norm{e^{-iJ}-e^{-iK}}\leq\norm{J-K}.
\end{equation}
Conversely, if $\norm{J},\norm{K}\leq1-\xi<1$,
\begin{equation}
    \norm{J-K}\leq\frac{1}{\xi}\norm{e^{-iJ}-e^{-iK}}.
\end{equation}
\end{lemma}
\begin{proof}
The first claim follows from a standard estimate~\cite[Corollary 5]{CSTWZ19}:
\begin{equation}
    \norm{e^{-iJ}-e^{-iK}}
    =\norm{I-e^{iJ}e^{-iK}}
    =\norm{\int_{0}^{1}\mathrm{d}t\ e^{itJ}i(J-K)e^{-itK}}
    \leq\norm{J-K}.
\end{equation}

For the second claim, we start with the integral representation
\begin{equation}
    \ln(C)-\ln(D)
    =\int_{0}^{+\infty}\mathrm{d}t\
    (tI+C)^{-1}
    (C-D)
    (tI+D)^{-1}
\end{equation}
and the bound
\begin{equation}
    \norm{\ln(C)-\ln(D)}\leq\norm{C-D}\int_{0}^{+\infty}\mathrm{d}t\norm{(tI+C)^{-1}}\norm{(tI+D)^{-1}},
\end{equation}
which hold as long as $C$ and $D$ have no eigenvalues within $(-\infty,0]$~\cite[Eqs.\ (4.8) and (4.9)]{Gil2012}. To apply this for $C=e^{-iJ}$ and $D=e^{-iK}$, it suffices to require that $\norm{J},\norm{K}\leq1-\xi<1$. Indeed, 
\begin{equation}
    \norm{I-C}=\norm{I-e^{-iJ}}\leq\norm{J}\leq1-\xi<1,\qquad
    \norm{I-D}=\norm{I-e^{-iK}}\leq\norm{K}\leq1-\xi<1,
\end{equation}
which implies $\norm{\lambda I-C}\geq\abs{\lambda-1}-\norm{I-C}>0$, $\norm{\lambda I-D}\geq\abs{\lambda-1}-\norm{I-D}>0$ when $\lambda\leq0$.

Now, we have
\begin{equation}
\begin{aligned}
    \norm{(tI+C)^{-1}}
    &=\frac{1}{t+1}\norm{\frac{1}{I+\frac{C-I}{t+1}}}
    \leq\frac{1}{t+1}\left(1+\frac{1-\xi}{t+1}+\left(\frac{1-\xi}{t+1}\right)^2+\cdots\right)\\
    &=\frac{1}{t+1}\frac{1}{1-\frac{1-\xi}{t+1}}
    =\frac{1}{t+\xi},
\end{aligned}
\end{equation}
and similarly,
\begin{equation}
    \norm{(tI+D)^{-1}}\leq\frac{1}{t+\xi}.
\end{equation}
This gives the claimed bound as $\int_{0}^{+\infty}\mathrm{d}t\frac{1}{(t+\xi)^2}=\frac{1}{\xi}$.
\end{proof}

\begin{corollary}[Error propagation]
Let $A$ and $B$ be matrices encoded by the Hamiltonian block encodings $E_A=\exp\left(-i\left[\begin{smallmatrix}
    0 & A^\dagger\\
    A & 0
\end{smallmatrix}\right]\right)$ and $E_B=\exp\left(-i\left[\begin{smallmatrix}
    0 & B^\dagger\\
    B & 0
\end{smallmatrix}\right]\right)$. It holds that
\begin{equation}
    \norm{E_A-E_B}\leq\norm{A-B}.
\end{equation}
Conversely, if $\norm{A},\norm{B}<1$,
\begin{equation}
    \norm{A-B}=\mathbf{O}\left(\norm{E_A-E_B}\right).
\end{equation}
\end{corollary}

%% file: multiply.tex
In this section, we consider Hamiltonian-based matrix multiplication. We start by introducing Lie group commutator formulas and analyzing their performance in \sec{multiply_pf}. We then employ these formulas to multiply matrices within Hamiltonian block encoding, establishing our first main result~\thm{multiply} in \sec{multiply_multiply}. Finally, we describe simplified methods to multiply Hermitian and unitary operators in \sec{multiply_special}.

\subsection{Analysis of Lie group commutator product formulas}
\label{sec:multiply_pf}
Our method for multiplying generic matrices with Hamiltonian evolution relies on the Lie group commutator product formula and its higher-order generalizations~\cite{ChildsWiebe12}, which we now introduce.

Given Hermitian operators $J$ and $K$, the lowest order group commutator formula takes the form $e^{-i\tau J}e^{-i\tau K}e^{i\tau J}e^{i\tau K}$, which approximates $e^{-\tau^2[J,K]}$ when the evolution time $\tau$ is sufficiently small (assuming $\tau\geq0$ without loss of generality). Specifically, we show in~\append{lie_bch} that
\begin{equation}
    e^{-i\tau J}e^{-i\tau K}e^{i\tau J}e^{i\tau K}
    =\exp\left(-\tau^2[J,K]+\frac{i\tau^3}{2}[J,[J,K]]+\frac{i\tau^3}{2}[K,[J,K]]+\mathbf{O}\left(\tau^4\right)\right).
\end{equation}
Our goal is to prove the following concrete bound on $\norm{e^{-i\tau J}e^{-i\tau K}e^{i\tau J}e^{i\tau K}-e^{-\tau^2[J,K]}}$ that quantifies the approximation error.

\begin{lemma}[Commutator bound for second-order Lie group commutator formula]
\label{lem:commutator_tight}
Given Hermitian matrices $J$, $K$, and $\tau\geq0$,
\begin{equation}
    \norm{e^{-i\tau J}e^{-i\tau K}e^{i\tau J}e^{i\tau K}-e^{-\tau^2[J,K]}}\leq\frac{\tau^3}{2}\norm{[J,[J,K]]}+\frac{\tau^3}{2}\norm{[K,[K,J]]}.
\end{equation}
\end{lemma}
\begin{proof}
Note that the spectral norm is invariant under multiplication by unitary, so it suffices to estimate
\begin{equation}
    \norm{e^{-i\tau J}e^{-i\tau K}e^{i\tau J}e^{i\tau K}-e^{-\tau^2[J,K]}}
    =\norm{e^{i\tau J}e^{i\tau K}e^{\frac{\tau^2}{2}[J,K]}
    -e^{i\tau K}e^{i\tau J}e^{-\frac{\tau^2}{2}[J,K]}}.
\end{equation}
Observe that
\begin{equation}
    \frac{\mathrm{d}}{\mathrm{d}\tau}\left(e^{i\tau J}e^{i\tau K}e^{\frac{\tau^2}{2}[J,K]}\right)
    =\left(iJ+e^{i\tau J}iKe^{-i\tau J}+e^{i\tau J}e^{i\tau K}\tau [J,K]e^{-i\tau K}e^{-i\tau J}\right)\left(e^{i\tau J}e^{i\tau K}e^{\frac{\tau^2}{2}[J,K]}\right),
\end{equation}
and
\begin{equation}
    \frac{\mathrm{d}}{\mathrm{d}\tau}\left(e^{i\tau K}e^{i\tau J}e^{-\frac{\tau^2}{2}[J,K]}\right)
    =\left(iK+e^{i\tau K}iJe^{-i\tau K}-e^{i\tau K}e^{i\tau J}\tau [J,K]e^{-i\tau J}e^{-i\tau K}\right)\left(e^{i\tau K}e^{i\tau J}e^{-\frac{\tau^2}{2}[J,K]}\right).
\end{equation}
Thus, the distance between evolutions can be bounded by the generators as~\cite[Corollary 5]{CSTWZ19}:
\begin{equation}
    \norm{e^{i\tau J}e^{i\tau K}e^{\frac{\tau^2}{2}[J,K]}
    -e^{i\tau K}e^{i\tau J}e^{-\frac{\tau^2}{2}[J,K]}}
    \leq\int_0^\tau\mathrm{d}s\norm{R(s)},
\end{equation}
where
\begin{equation}
\begin{aligned}
    R(s)&=\left(iJ+e^{is J}(iK)e^{-is J}+e^{is J}e^{is K}(s [J,K])e^{-is K}e^{-is J}\right)\\
    &\quad-\left(iK+e^{is K}(iJ)e^{-is K}+e^{is K}e^{is J}(s[K,J])e^{-is J}e^{-is K}\right).
\end{aligned}
\end{equation}

At this point, one may proceed as~\cite{CSTWZ19} to Taylor expand each term to second order with the remainder in the integral form. However, doing so would produce a loose error bound with larger prefactors. We now explain how to expand $R(s)$ more carefully to avoid introducing redundant terms.
Let us start with the first line. First we expand the third term as
\begin{equation}
\begin{aligned}
    e^{is K}(s [J,K])e^{-is K}
    &=s [J,K]+s\int_0^s\mathrm{d}u\ e^{iu K}[iK, [J,K]]e^{-iu K}.\\
\end{aligned}
\end{equation}
When substituted back to $R(s)$, the integral remainder cannot be canceled, while
\begin{equation}
    e^{isJ}s[J,K]e^{-isJ}
    =s[J,K]+s\int_0^s\mathrm{d}u\ e^{iuJ}[iJ,[J,K]]e^{-iuJ}.
\end{equation}
Now for the second term in $R$, we have
\begin{equation}
\begin{aligned}
    e^{is J}iKe^{-is J}
    &=iK+\int_0^s\mathrm{d}u\ e^{iu J}[iJ,iK]e^{-iu J}\\
    &=iK+\int_0^s\mathrm{d}u\left(\int_0^u\mathrm{d}v\ e^{iv J}[iJ,[iJ,iK]]e^{-iv J}+[iJ,iK]\right)\\
    &=iK+s[iJ,iK]+\int_0^s\mathrm{d}v\int_v^s\mathrm{d}u\ e^{iv J}[iJ,[iJ,iK]]e^{-iv J}\\
    &=iK+s[iJ,iK]+s\int_0^s\mathrm{d}v\ e^{iv J}[iJ,[iJ,iK]]e^{-iv J}-\int_0^s\mathrm{d}v\ ve^{iv J}[iJ,[iJ,iK]]e^{-iv J}.\\
\end{aligned}
\end{equation}
Altogether, this gives
\CancelColor{\color{blue}}
\begin{equation}
\begin{aligned}
    &e^{is J}(iK)e^{-is J}+e^{is J}e^{is K}(s [J,K])e^{-is K}e^{-is J}\\
    &=iK+\Ccancel[blue]{s[iJ,iK]}+\Ccancel[blue]{s\int_0^s\mathrm{d}v\ e^{iv J}[iJ,[iJ,iK]]e^{-iv J}}-\int_0^s\mathrm{d}v\ ve^{iv J}[iJ,[iJ,iK]]e^{-iv J}\\
    &\quad+\Ccancel[blue]{s[J,K]}+\Ccancel[blue]{s\int_0^s\mathrm{d}u\ e^{iuJ}[iJ,[J,K]]e^{-iuJ}}
    +se^{isJ}\int_0^s\mathrm{d}u\ e^{iu K}[iK, [J,K]]e^{-iu K}e^{-isJ}\\
    &=iK-\int_0^s\mathrm{d}v\ ve^{iv J}[iJ,[iJ,iK]]e^{-iv J}+se^{isJ}\int_0^s\mathrm{d}u\ e^{iu K}[iK, [J,K]]e^{-iu K}e^{-isJ}.
\end{aligned}
\end{equation}
Similarly, for the second line of $R(s)$,
\begin{equation}
\begin{aligned}
    &e^{is K}(iJ)e^{-is K}+e^{is K}e^{is J}(s [K,J])e^{-is J}e^{-is K}\\
    &=iJ-\int_0^s\mathrm{d}v\ ve^{iv K}[iK,[iK,iJ]]e^{-iv K}+se^{isK}\int_0^s\mathrm{d}u\ e^{iu J}[iJ, [K,J]]e^{-iu J}e^{-isK}.
\end{aligned}
\end{equation}

We have now represented $R$ as
\begin{equation}
\begin{aligned}
    R(s)
    &=\left(iJ+iK-\int_0^s\mathrm{d}v\ ve^{iv J}[iJ,[iJ,iK]]e^{-iv J}+se^{isJ}\int_0^s\mathrm{d}u\ e^{iu K}[iK, [J,K]]e^{-iu K}e^{-isJ}\right)\\
    &\quad-\left(iK+iJ-\int_0^s\mathrm{d}v\ ve^{iv K}[iK,[iK,iJ]]e^{-iv K}+se^{isK}\int_0^s\mathrm{d}u\ e^{iu J}[iJ, [K,J]]e^{-iu J}e^{-isK}\right)\\
    &=-\int_0^s\mathrm{d}v\ ve^{iv J}[iJ,[iJ,iK]]e^{-iv J}+se^{isJ}\int_0^s\mathrm{d}u\ e^{iu K}[iK, [J,K]]e^{-iu K}e^{-isJ}\\
    &\quad+\int_0^s\mathrm{d}v\ ve^{iv K}[iK,[iK,iJ]]e^{-iv K}-se^{isK}\int_0^s\mathrm{d}u\ e^{iu J}[iJ, [K,J]]e^{-iu J}e^{-isK},\\
\end{aligned}
\end{equation}
which implies
\begin{equation}
    \norm{R(s)}\leq \frac{s^2}{2}\norm{[J,[J,K]]}+s^2\norm{[K,[J,K]]}
    +\frac{s^2}{2}\norm{[K,[K,J]]}+s^2\norm{[J,[K,J]]},
\end{equation}
and finally
\begin{equation}
    \norm{e^{-i\tau J}e^{-i\tau K}e^{i\tau J}e^{i\tau K}-e^{-\tau^2[J,K]}}
    \leq\int_0^\tau\mathrm{d}s\norm{R(s)}
    \leq\frac{\tau^3}{2}\norm{[J,[J,K]]}+\frac{\tau^3}{2}\norm{[K,[K,J]]}.
\end{equation}
\end{proof}

Note that our error bound $\norm{e^{-i\tau J}e^{-i\tau K}e^{i\tau J}e^{i\tau K}-e^{-\tau^2[J,K]}}
\leq\frac{\tau^3}{2}\norm{[J,[J,K]]}+\frac{\tau^3}{2}\norm{[K,[K,J]]}$ matches the third-order terms of the BCH expansion, and is thus provably tight up to a single application of the triangle inequality. In particular, this tightens a recent estimate~\cite[Lemma 9]{Gluza2024doublebracket} which instead reads $\norm{e^{-i\tau J}e^{-i\tau K}e^{i\tau J}e^{i\tau K}-e^{-\tau^2[J,K]}}
\leq\tau^3\norm{[J,[J,K]]}+\tau^3\norm{[K,[K,J]]}$.

The above approximation holds in the limit where $\tau\rightarrow0$. To evolve for a longer time $t$, we divide the evolution into $r$ steps and set $\tau=\sqrt{\frac{t}{r}}$. Denoting $M_2(\tau)=e^{-i\tau J}e^{-i\tau K}e^{i\tau J}e^{i\tau K}$, we have
\begin{equation}
    \norm{M_{2}\left(\sqrt{\frac{t}{r}}\right)-e^{-\frac{t}{r}[J,K]}}
    \leq\frac{t^{\frac{3}{2}}}{2r^{\frac{3}{2}}}\left(\norm{[J,[J,K]]}+\norm{[J,[J,K]]}\right),
\end{equation}
which implies that
\begin{equation}
    \norm{M_{2}^r\left(\sqrt{\frac{t}{r}}\right)-e^{-t[J,K]}}
    \leq r\norm{M_{2}\left(\sqrt{\frac{t}{r}}\right)-e^{-\frac{t}{r}[J,K]}}
    \leq\frac{t^{\frac{3}{2}}}{2r^{\frac{1}{2}}}\left(\norm{[J,[J,K]]}+\norm{[J,[J,K]]}\right).
\end{equation}
Hence to ensure that the error is at most $\epsilon$, it suffices to choose
\begin{equation}
    r=\left\lceil\frac{t^{3}}{4\epsilon^{2}}\left(\norm{[J,[J,K]]}+\norm{[J,[J,K]]}\right)^2\right\rceil.
\end{equation}
There is a subtlety regarding the sign of the evolution time $t$. The above analysis handles the case where $t\geq0$ is nonnegative. When $t<0$, one cannot apply the formula directly because matrix exponentials with $\sqrt{\frac{t}{r}}$ would lead to imaginary-time evolution. Instead, we should take the Hermitian conjugation.

\begin{corollary}[Step number for second-order Lie group commutator formula]
Let $J$, $K$ be Hermitian matrices, and $M_2(\tau)=e^{-i\tau J}e^{-i\tau K}e^{i\tau J}e^{i\tau K}$ be the second-order Lie group commutator product formula. For any evolution time $t\geq0$, the approximations
\begin{equation}
    \norm{M_{2}^r\left(\sqrt{\frac{t}{r}}\right)-e^{-t[J,K]}}\leq\epsilon,\qquad
    \norm{M_{2}^{\dagger r}\left(\sqrt{\frac{t}{r}}\right)-e^{t[J,K]}}\leq\epsilon
\end{equation}
can be achieved with accuracy $\epsilon$ by choosing
\begin{equation}
    r=\left\lceil\frac{t^{3}}{4\epsilon^{2}}\left(\norm{[J,[J,K]]}+\norm{[J,[J,K]]}\right)^2\right\rceil.
\end{equation}
\end{corollary}

The above asymptotic scaling can be systematically improved using higher-order group commutator formulas, such as the ones developed by Childs and Wiebe~\cite{ChildsWiebe12}. These formulas are characterized by the order condition $M_{p}(\tau)=e^{-\tau^2[J,K]}+\mathbf{O}\left(\tau^{p+1}\right)$ for $p\in\mathbb{Z}_{\geq2}$, which lead to 
\begin{equation}
    \norm{M_{p}^r\left(\sqrt{\frac{t}{r}}\right)-e^{-t[J,K]}}\leq\epsilon,\qquad
    \norm{M_{p}^{\dagger r}\left(\sqrt{\frac{t}{r}}\right)-e^{t[J,K]}}\leq\epsilon
\end{equation}
by choosing
\begin{equation}
    r=\mathbf{O}\left(\frac{\alpha_{\text{comm}}^{2+4/(p-1)}t^{1+2/(p-1)}}{\epsilon^{2/(p-1)}}\right),\qquad
    \alpha_{\text{comm}}=\left(\sum_{j_1,\ldots,j_{p+1}=1}^{2}\norm{\left[H_{j_{p+1}},\ldots,\left[H_{j_2},H_{j_1}\right]\right]}\right)^{\frac{1}{p+1}},
\end{equation}
with $H_1=J$ and $H_2=K$. We provide further analysis of these higher-order formulas in~\append{lie_higher}.

\subsection{Matrix multiplication}
\label{sec:multiply_multiply}
We now explain how to perform matrix multiplication with Hamiltonian block encodings. Specifically, suppose $E_{A}=\exp\left(-i\left[\begin{smallmatrix}
    0 & A^\dagger\\
    A & 0
\end{smallmatrix}\right]\right)$ and $E_{ B}=\exp\left(-i\left[\begin{smallmatrix}
    0 & B^\dagger\\
    B & 0
\end{smallmatrix}\right]\right)$ encode the input operators $A$ and $B$ respectively. Our goal is to produce $E_{AB}=\exp\left(-i\left[\begin{smallmatrix}
    0 & (AB)^\dagger\\
    AB & 0
\end{smallmatrix}\right]\right)$ which is a Hamiltonian block encoding of $AB$.

To this end, we first transform the Hamiltonian block encodings into
\begin{equation}
    \exp\left(-i\tau
    \begin{bmatrix}
        0 & A^\dagger & 0 & 0\\
        A & 0 & 0 & 0\\
        0 & 0 & 0 & 0\\
        0 & 0 & 0 & 0
    \end{bmatrix}\right),\qquad
    \exp\left(-i\tau
    \begin{bmatrix}
        0 & 0 & B & 0\\
        0 & 0 & 0 & 0\\
        B^\dagger & 0 & 0 & 0\\
        0 & 0 & 0 & 0
    \end{bmatrix}\right).
\end{equation}
Here, the operator $A$ remains in the standard form except that we introduce an additional ancilla qubit and perform a controlled version of Hamiltonian block encoding. As for $B$, we consider its controlled Hamiltonian block encoding and perform
\begin{equation}
\begin{aligned}
    &\left(\mathrm{SWAP}\otimes I\right)
    \left(I\otimes X\otimes I\right)
    \exp\left(-i\tau\begin{bmatrix}
        0 & B^\dagger & 0 & 0\\
        B & 0 & 0 & 0\\
        0 & 0 & 0 & 0\\
        0 & 0 & 0 & 0
    \end{bmatrix}\right)
    \left(I\otimes X\otimes I\right)
    \left(\mathrm{SWAP}\otimes I\right)\\
    &=\exp\left(-i\tau
    \begin{bmatrix}
        0 & 0 & B & 0\\
        0 & 0 & 0 & 0\\
        B^\dagger & 0 & 0 & 0\\
        0 & 0 & 0 & 0
    \end{bmatrix}\right).
\end{aligned}
\end{equation}

Observe that 
\begin{equation}
    \begin{bmatrix}
        0 & A^\dagger & 0\\
        A & 0 & 0\\
        0 & 0 & 0
    \end{bmatrix}
    \begin{bmatrix}
        0 & 0 & B\\
        0 & 0 & 0\\
        B^\dagger & 0 & 0
    \end{bmatrix}
    -\begin{bmatrix}
        0 & 0 & B\\
        0 & 0 & 0\\
        B^\dagger & 0 & 0
    \end{bmatrix}
    \begin{bmatrix}
        0 & A^\dagger & 0\\
        A & 0 & 0\\
        0 & 0 & 0
    \end{bmatrix}
    =\begin{bmatrix}
        0 & 0 & 0\\
        0 & 0 & AB\\
        0 & -B^\dagger A^\dagger & 0
    \end{bmatrix}.
\end{equation}
Hence, the second-order Lie group commutator product formula yields
\begin{footnotesize}
\newmaketag
\begin{align}
    &\exp\left(-i\tau
    \begin{bmatrix}
        0 & A^\dagger & 0 & 0\\
        A & 0 & 0 & 0\\
        0 & 0 & 0 & 0\\
        0 & 0 & 0 & 0
    \end{bmatrix}\right)
    \exp\left(-i\tau
    \begin{bmatrix}
        0 & 0 & B & 0\\
        0 & 0 & 0 & 0\\
        B^\dagger & 0 & 0 & 0\\
        0 & 0 & 0 & 0
    \end{bmatrix}\right)
    \exp\left(i\tau
    \begin{bmatrix}
        0 & A^\dagger & 0 & 0\\
        A & 0 & 0 & 0\\
        0 & 0 & 0 & 0\\
        0 & 0 & 0 & 0
    \end{bmatrix}\right)
    \exp\left(i\tau
    \begin{bmatrix}
        0 & 0 & B & 0\\
        0 & 0 & 0 & 0\\
        B^\dagger & 0 & 0 & 0\\
        0 & 0 & 0 & 0
    \end{bmatrix}\right)\notag\\
    &=
    \exp\left(-i\tau^2
    \begin{bmatrix}
        0 & 0 & 0 & 0\\
        0 & 0 & -iAB & 0\\
        0 & i(AB)^\dagger & 0 & 0\\
        0 & 0 & 0 & 0
    \end{bmatrix}\right)
    +\mathbf{O}\left(\tau^3\left(\norm{A}+\norm{B}\right)^3\right)
\end{align}
\end{footnotesize}%
for $\tau$ sufficiently small. As is explained in the previous subsection, the approximation accuracy can be systematically improved using higher-order group commutator formulas.

It remains to transform the output Hamiltonian block encoding back to the standard form. This can be achieved via the unitary conjugation
\begin{equation}
\begin{aligned}
    &\left(I\otimes XS^\dagger\otimes I\right)\left(X_{\mathbb{C}^2\otimes\mathbb{C}^2}^\dagger\otimes I\right)
    \exp\left(-i\tau^2
    \begin{bmatrix}
        0 & 0 & 0 & 0\\
        0 & 0 & -iAB & 0\\
        0 & i(AB)^\dagger & 0 & 0\\
        0 & 0 & 0 & 0
    \end{bmatrix}\right)
    \left(X_{\mathbb{C}^2\otimes\mathbb{C}^2}\otimes I\right)\left(I\otimes SX\otimes I\right)\\
    &=\exp\left(-i\tau^2
    \begin{bmatrix}
        0 & (AB)^\dagger & 0 & 0\\
        AB & 0 & 0 & 0\\
        0 & 0 & 0 & 0\\
        0 & 0 & 0 & 0
    \end{bmatrix}\right),
\end{aligned}
\end{equation}
where the cyclic shift operation $X_{\mathbb{C}^2\otimes\mathbb{C}^2}:\ket{00}\mapsto\ket{01}\mapsto\ket{10}\mapsto\ket{11}\mapsto\ket{00}$ can be further decomposed into
\begin{equation}
    X_{\mathbb{C}^2\otimes\mathbb{C}^2}
    =\mathrm{SWAP}\left(I\otimes X\right)\mathrm{CNOT}_{2,1}\left(I\otimes X\right)\mathrm{CNOT}_{1,2}
\end{equation}
for $\mathrm{CNOT}_{j,k}$ the CNOT gate with control qubit $j$ and target qubit $k$. Let us simplify this circuit as follows. Up to a global phase,
\begin{equation}
\begin{aligned}
    -iX_{\mathbb{C}^2\otimes\mathbb{C}^2}\left(I\otimes SX\right)
    &= X_{\mathbb{C}^2\otimes\mathbb{C}^2}\left(I\otimes XS^\dagger\right)\\
    &=\mathrm{SWAP}\left(I\otimes X\right)\mathrm{CNOT}_{2,1}\left(I\otimes X\right)\mathrm{CNOT}_{1,2}\left(I\otimes XS^\dagger\right)\\
    &=\mathrm{SWAP}\left(I\otimes X\right)\mathrm{CNOT}_{2,1}\mathrm{CNOT}_{1,2}\left(I\otimes S^\dagger\right)\\
    &=\begin{bmatrix}
        1 & 0 & 0 & 0\\
        0 & 0 & 1 & 0\\
        0 & 1 & 0 & 0\\
        0 & 0 & 0 & 1\\
    \end{bmatrix}
    \begin{bmatrix}
        0 & 1 & 0 & 0\\
        1 & 0 & 0 & 0\\
        0 & 0 & 0 & 1\\
        0 & 0 & 1 & 0
    \end{bmatrix}
    \begin{bmatrix}
        1 & 0 & 0 & 0\\
        0 & 0 & 0 & 1\\
        0 & 0 & 1 & 0\\
        0 & 1 & 0 & 0\\
    \end{bmatrix}
    \begin{bmatrix}
        1 & 0 & 0 & 0\\
        0 & 1 & 0 & 0\\
        0 & 0 & 0 & 1\\
        0 & 0 & 1 & 0\\
    \end{bmatrix}
    \left(I\otimes S^\dagger\right)\\
    &=\begin{bmatrix}
        0 & 0 & 1 & 0\\
        0 & 1 & 0 & 0\\
        1 & 0 & 0 & 0\\
        0 & 0 & 0 & 1
    \end{bmatrix}
    \left(I\otimes S^\dagger\right)\\
    &=\left(I\otimes\ketbra{1}{1}+X\otimes\ketbra{0}{0}\right)
    \left(I\otimes S^\dagger\right)
    =\mathrm{CNOT}_{2,1}\left(X\otimes S^\dagger\right).
\end{aligned}
\end{equation}
Hence,
\begin{equation}
\begin{aligned}
    &\left(X\otimes S\otimes I\right)\left(\mathrm{CNOT}_{2,1}\otimes I\right)
    \exp\left(-i\tau^2
    \begin{bmatrix}
        0 & 0 & 0 & 0\\
        0 & 0 & -iAB & 0\\
        0 & i(AB)^\dagger & 0 & 0\\
        0 & 0 & 0 & 0
    \end{bmatrix}\right)
    \left(\mathrm{CNOT}_{2,1}\otimes I\right)\left(X\otimes S^\dagger\otimes I\right)\\
    &=\exp\left(-i\tau^2
    \begin{bmatrix}
        0 & (AB)^\dagger & 0 & 0\\
        AB & 0 & 0 & 0\\
        0 & 0 & 0 & 0\\
        0 & 0 & 0 & 0
    \end{bmatrix}\right).
\end{aligned}
\end{equation}
Combining with the fractional scaling technique of~\cor{frac_scale}, we obtain:

\begin{theorem}[Matrix multiplication]
\label{thm:multiply}
Let $A$ and $B$ be matrices encoded by the Hamiltonian block encodings $E_{A}=\exp\left(-i\left[\begin{smallmatrix}
    0 & A^\dagger\\
    A & 0
\end{smallmatrix}\right]\right)$ and $E_{ B}=\exp\left(-i\left[\begin{smallmatrix}
    0 & B^\dagger\\
    B & 0
\end{smallmatrix}\right]\right)$ with $\norm{A},\norm{B}<\frac{\pi}{2}$. Then the controlled version of Hamiltonian block encoding
\begin{equation}
    E_{AB}=\exp\left(-i
    \begin{bmatrix}
        0 & (AB)^\dagger\\
        AB & 0
    \end{bmatrix}\right)
\end{equation}
can be constructed with accuracy $\epsilon$ using
\begin{equation}
    r=\frac{1}{\epsilon^{o(1)}}
\end{equation}
queries to controlled $E_{A}$ and $E_{B}$.
See~\fig{multiply} for the corresponding circuit diagram.
\end{theorem}

\begin{figure}[t]
	\centering
\includegraphics[scale=0.75]{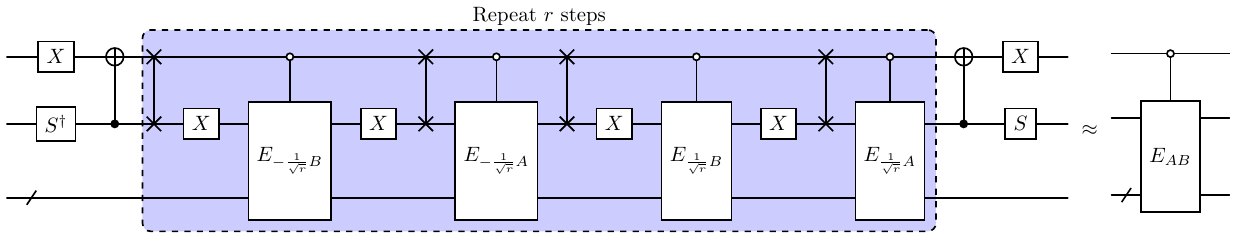}
\caption{Quantum circuit for Hamiltonian matrix multiplication. Observe that $2$ ancilla qubits are sufficient and these ancillae may be reused throughout the computation. For illustration purposes, only the lowest-order group commutator formula is shown without the fractional scaling step.}
\label{fig:multiply}
\end{figure}

\subsection{Multiplying Hermitian and unitary operators}
\label{sec:multiply_special}
Our result \thm{multiply} for multiplying generic matrices requires $2$ ancilla qubits. When at least one of the multiplicands is Hermitian, we can perform matrix multiplication using only $1$ ancilla qubit by directly evolving under the Hermitian input.

Consider first the case that $K$ is a Hermitian matrix and $A$ is arbitrary. Observe that
\begin{equation}
    \begin{bmatrix}
        0 & A^\dagger\\
        A & 0
    \end{bmatrix}
    \begin{bmatrix}
        K & 0\\
        0 & 0
    \end{bmatrix}
    -\begin{bmatrix}
        K & 0\\
        0 & 0
    \end{bmatrix}
    \begin{bmatrix}
        0 & A^\dagger\\
        A & 0
    \end{bmatrix}
    =\begin{bmatrix}
        0 & -KA^\dagger\\
        AK & 0
    \end{bmatrix},
\end{equation}
which yields 
\begin{equation}
\begin{aligned}
    &\exp\left(-i\tau
    \begin{bmatrix}
        0 & A^\dagger\\
        A & 0
    \end{bmatrix}\right)
    \exp\left(-i\tau
    \begin{bmatrix}
        K & 0\\
        0 & 0
    \end{bmatrix}\right)
    \exp\left(i\tau
    \begin{bmatrix}
        0 & A^\dagger\\
        A & 0
    \end{bmatrix}\right)
    \exp\left(i\tau
    \begin{bmatrix}
        K & 0\\
        0 & 0
    \end{bmatrix}\right)\\
    &=\exp\left(-i\tau^2
    \begin{bmatrix}
        0 & i(AK)^\dagger\\
        -iAK & 0
    \end{bmatrix}\right)
    +\mathbf{O}\left(\tau^3\left(\norm{A}+\norm{K}\right)^3\right).
\end{aligned}
\end{equation}
Here, the output Hamiltonian block encoding can be further converted into the standard form as
\begin{equation}
    \left(S\otimes I\right)\exp\left(-i\tau^2
    \begin{bmatrix}
        0 & i(AK)^\dagger\\
        -iAK & 0
    \end{bmatrix}\right)
    \left(S^\dagger\otimes I\right)
    =\exp\left(-i\tau^2
    \begin{bmatrix}
        0 & (AK)^\dagger\\
        AK & 0
    \end{bmatrix}\right).
\end{equation}
Similarly, when $J$ is Hermitian and $A$ is arbitrary, we have
\begin{equation}
    \begin{bmatrix}
        0 & 0\\
        0 & J
    \end{bmatrix}
    \begin{bmatrix}
        0 & A^\dagger\\
        A & 0
    \end{bmatrix}
    -\begin{bmatrix}
        0 & A^\dagger\\
        A & 0
    \end{bmatrix}
    \begin{bmatrix}
        0 & 0\\
        0 & J
    \end{bmatrix}
    =\begin{bmatrix}
        0 & -A^\dagger J\\
        JA & 0
    \end{bmatrix},
\end{equation}
which gives
\begin{equation}
\begin{aligned}
    &\exp\left(-i\tau
    \begin{bmatrix}
        0 & 0\\
        0 & J
    \end{bmatrix}\right)
    \exp\left(-i\tau
    \begin{bmatrix}
        0 & A^\dagger\\
        A & 0
    \end{bmatrix}\right)
    \exp\left(i\tau
    \begin{bmatrix}
        0 & 0\\
        0 & J
    \end{bmatrix}\right)
    \exp\left(i\tau
    \begin{bmatrix}
        0 & A^\dagger\\
        A & 0
    \end{bmatrix}\right)\\
    &=\exp\left(-i\tau^2
    \begin{bmatrix}
        0 & i(JA)^\dagger\\
        -iJA & 0
    \end{bmatrix}\right)
    +\mathbf{O}\left(\tau^3\left(\norm{A}+\norm{J}\right)^3\right)
\end{aligned}
\end{equation}
with the output convertible to the standard form through the unitary conjugation $(S\otimes I)(\cdot)(S^\dagger\otimes I)$. This gives:

\begin{proposition}[Matrix multiplication, Hermitian case]
\label{prop:herm_multiply}
Let $A$ be a matrix encoded by the Hamiltonian block encoding $E_{A}=\exp\left(-i\left[\begin{smallmatrix}
    0 & A^\dagger\\
    A & 0
\end{smallmatrix}\right]\right)$ with $\norm{A}<\frac{\pi}{2}$.
Let $J$ and $K$ be Hermitian matrices with $\norm{J},\norm{K}=\mathbf{O}(1)$ generating the Hamiltonian evolution $e^{-i\tau J}$ and $e^{-i\tau K}$.
Then the Hamiltonian block encoding
\begin{equation}
    E_{JAK}=\exp\left(-i
    \begin{bmatrix}
        0 & (JAK)^\dagger\\
        JAK & 0
    \end{bmatrix}\right)
\end{equation}
can be constructed with accuracy $\epsilon$ using
\begin{equation}
    r=\frac{1}{\epsilon^{o(1)}}
\end{equation}
queries to $E_{A}$, as well as $e^{-i\tau_k J}$ and $e^{-i\tau_l J}$ for effective times $\tau_k,\tau_l$ determined by the Lie group commutator formulas.
See~\fig{multiply_jk} for the corresponding circuit diagram.
\end{proposition}

\begin{figure}[t]
	\centering
    \begin{subfigure}[t]{\textwidth}
        \centering
        \includegraphics[scale=\circuitwidth]{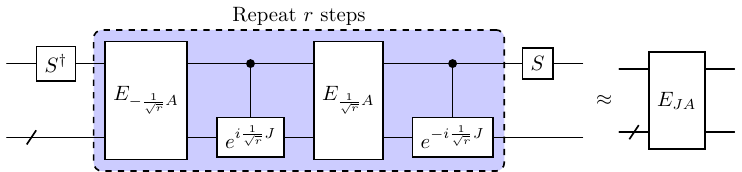}
        \caption{}
    \end{subfigure}%
    \\
    \begin{subfigure}[t]{\textwidth}
        \centering
        \includegraphics[scale=\circuitwidth]{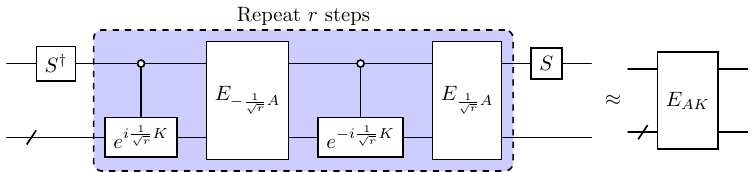}
        \caption{}
    \end{subfigure}%
\caption{Quantum circuit for multiplying Hermitian matrices. For illustration purposes, only the lowest-order group commutator formula is shown without the fractional scaling step.}
\label{fig:multiply_jk}
\end{figure}

Finally, we point out that one can multiply unitary matrices with Hamiltonian block encoding through a unitary conjugation. Specifically, if $U$ and $V$ are unitaries, then
\begin{equation}
    \begin{bmatrix}
        V^\dagger & 0\\
        0 & U
    \end{bmatrix}
    \exp\left(-i\tau
    \begin{bmatrix}
        0 & A^\dagger\\
        A & 0
    \end{bmatrix}\right)
    \begin{bmatrix}
        V & 0\\
        0 & U^\dagger
    \end{bmatrix}
    =\exp\left(-i\tau
    \begin{bmatrix}
        0 & (UAV)^\dagger\\
        UAV & 0
    \end{bmatrix}\right).
\end{equation}
This is an error-free operation and does not require any extra ancilla beyond the single one defining the Hamiltonian block encoding of $A$.

\begin{proposition}[Matrix multiplication, unitary case]
\label{prop:unitary_multiply}
Let $A$ be a matrix encoded by the Hamiltonian block encoding $E_{A}=\exp\left(-i\left[\begin{smallmatrix}
    0 & A^\dagger\\
    A & 0
\end{smallmatrix}\right]\right)$, and $U$, $V$ be unitaries.
Then the Hamiltonian block encoding
\begin{equation}
    E_{UAV}=\exp\left(-i
    \begin{bmatrix}
        0 & (UAV)^\dagger\\
        UAV & 0
    \end{bmatrix}\right)
\end{equation}
can be constructed with zero error using $1$ query to $E_{A}$, and $1$ query to controlled-$U$, $V$ and their inverses.
See~\fig{multiply_unitary} for the corresponding circuit diagram.
\end{proposition}

\begin{figure}[t]
	\centering
\includegraphics[scale=\circuitwidth]{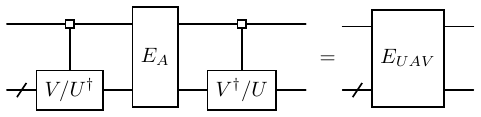}
\caption{Quantum circuit for multiplying unitary matrices.}
\label{fig:multiply_unitary}
\end{figure}

%% file: qsvt.tex
We now present our second main result on realizing singular value transformation within Hamiltonian block encoding. We start by solving the dominated polynomial approximation problem in~\sec{qsvt_dominate}. We then apply it to get a Hamiltonian QSVT algorithm for odd polynomials in~\sec{qsvt_odd}, establishing~\thm{qsvt_odd}. When the target polynomials are even, there is an intrinsic technical issue due to the parity constraint of QSVT. We show how to circumvent this by shifting the spectra, proving~\thm{qsvt_herm_even} and~\thm{qsvt_even} in~\sec{qsvt_even}. Finally, we apply Hamiltonian QSVT to implement Hamiltonian-based matrix inversion and fractional scaling in~\sec{qsvt_inverse_frac}.

\subsection{Dominated polynomial approximation}
\label{sec:qsvt_dominate}
The main goal of this subsection is to solve the following dominated approximation problem, which is essential to the construction of Hamiltonian QSVT as discussed in~\sec{intro}.

\begin{proposition}[Dominated polynomial approximation]
\label{prop:dominated}
Let $f$ be a real odd polynomial of degree $d$. For any $\epsilon>0$ and constant $0<\xi\leq\frac{\pi}{2}$, there exist a real odd polynomial $p(x)$ and even polynomial $q(x)$ such that
\begin{equation}
\begin{aligned}
    &\abs{p(x)-\sin(f(\arcsin(x)))}\leq\epsilon,\qquad&&\forall x\in\left[-\sin\left(\frac{\pi}{2}-\xi\right),\sin\left(\frac{\pi}{2}-\xi\right)\right],\\
    &\abs{q(x)-\frac{\cos(f(\arcsin(x)))}{\sqrt{1-x^2}}}\leq\epsilon,\qquad&&\forall x\in\left[-\sin\left(\frac{\pi}{2}-\xi\right),\sin\left(\frac{\pi}{2}-\xi\right)\right],\\
    &p^2(x)+(1-x^2)q^2(x)\leq1+\epsilon,\qquad&&\forall x\in[-1,1].
\end{aligned}
\end{equation}
Moreover, both $p$ and $q$ have the asymptotic degree
\begin{equation}
    \mathbf{O}\left(\left(d+\log\left(\frac{\norm{f}_{\max,\left[-\frac{\pi}{2},\frac{\pi}{2}\right]}}{\epsilon}\right)\right)
    \left(\norm{f}_{\max,\left[-\frac{\pi}{2},\frac{\pi}{2}\right]}+\log\left(\frac{1}{\epsilon}\right)\right)
    \right).
\end{equation}
\end{proposition}

We prove a weaker version of this result here, which solves the dominated approximation problem with asymptotic polynomial degree
\begin{equation}
    \mathbf{O}\left(d
    \log\left(\frac{d\norm{f}_{\max,\left[-\frac{\pi}{2},\frac{\pi}{2}\right]}}{\epsilon}\right)
    \left(\norm{f}_{\max,\left[-\frac{\pi}{2},\frac{\pi}{2}\right]}+\log\left(\frac{1}{\epsilon}\right)\right)
    \right).
\end{equation}
This is looser than the claimed asymptotic scaling by some logarithmic factors, but the analysis is significantly simpler.
Shaving-off the logarithmic factors requires analyzing the behavior of functions over the complex plane (\append{composite_stadium}), and the details are given in~\append{composite_qsvt}.

Proof of the weaker result uses the following standard tools of polynomial approximation.
\begin{lemma}[Dominated polynomial approximation of $\arcsin$ {\cite[Section 258]{fikhtengoltsfundamentals}}]
\label{lem:poly_arcsin}
For any $\epsilon_{\arcsin}>0$ and constant $0<\xi\leq\frac{\pi}{2}$, there exists a real odd polynomial $h_{\arcsin}$ such that
\begin{equation}
\begin{aligned}
    &\abs{h_{\arcsin}(x)-\arcsin(x)}\leq\epsilon_{\arcsin},\qquad&&\forall x\in\left[-\sin\left(\frac{\pi}{2}-\xi\right),\sin\left(\frac{\pi}{2}-\xi\right)\right],\\
    &\abs{h_{\arcsin}(x)}\leq\abs{\arcsin(x)},\qquad&&\forall x\in[-1,1],
\end{aligned}
\end{equation}
with degree 
\begin{equation}
    \mathbf{O}\left(\log\left(\frac{1}{\epsilon_{\arcsin}}\right)\right).
\end{equation}
\end{lemma}

\begin{lemma}[Dominated polynomial approximation of $\frac{1}{\sqrt{1-x^2}}$ {\cite[Section 258]{fikhtengoltsfundamentals}}]
\label{lem:poly_inv_sqrt}
For any $\epsilon_{\text{inv-sqrt}}>0$ and constant $0<\xi\leq\frac{\pi}{2}$, there exists a real even polynomial $h_{\text{inv-sqrt}}$ such that
\begin{equation}
\begin{aligned}
    &\abs{h_{\text{inv-sqrt}}(x)-\frac{1}{\sqrt{1-x^2}}}\leq\epsilon_{\text{inv-sqrt}},\qquad&&\forall x\in\left[-\sin\left(\frac{\pi}{2}-\xi\right),\sin\left(\frac{\pi}{2}-\xi\right)\right],\\
    &0\leq h_{\text{inv-sqrt}}(x)\leq\frac{1}{\sqrt{1-x^2}},\qquad&&\forall x\in[-1,1],
\end{aligned}
\end{equation}
with degree
\begin{equation}
    \mathbf{O}\left(\log\left(\frac{1}{\epsilon_{\text{inv-sqrt}}}\right)\right).
\end{equation}
\end{lemma}

\begin{lemma}[Polynomial approximation of $\sin$ and $\cos$ {\cite[Lemma 57]{Gilyen2018singular}}]
\label{lem:poly_trig}
For any $\epsilon_{\text{trig}}>0$ and $\alpha>0$, there exist a real odd polynomial $h_{\sin,\alpha}$ and even polynomial $h_{\cos,\alpha}$ such that
\begin{equation}
\begin{aligned}
    &\abs{h_{\sin,\alpha}(x)-\sin(\alpha x)}\leq\epsilon_{\text{trig}},\qquad&&\forall x\in[-1,1],\\
    &\abs{h_{\cos,\alpha}(x)-\cos(\alpha x)}\leq\epsilon_{\text{trig}},\qquad&&\forall x\in[-1,1],\\
\end{aligned}
\end{equation}
with degree
\begin{equation}
    \mathbf{O}\left(\alpha+\log\left(\frac{1}{\epsilon_{\text{trig}}}\right)\right).
\end{equation}
\end{lemma}

\begin{proof}[Proof of~\prop{dominated}]
We start by constructing the following polynomial approximations.
\begin{enumerate}[label=(\roman*)]
\item We construct $h_{\arcsin}$ from~\lem{poly_arcsin} with $\xi$ as given and $\epsilon_{\arcsin}$ to be determined.
\item We construct $h_{\text{inv-sqrt}}$ from~\lem{poly_inv_sqrt} with $\xi$ as given and $\epsilon_{\text{inv-sqrt}}$ to be determined.
\item We construct $h_{\sin,\alpha}$ and $h_{\cos,\alpha}$ from~\lem{poly_trig} with $\alpha=\norm{f}_{\max,\left[-\frac{\pi}{2},\frac{\pi}{2}\right]}$ and $\epsilon_{\text{trig}}$ to be determined.
\end{enumerate}
We then define
\begin{equation}
    p(x)=h_{\sin,\alpha}\left(\frac{1}{\alpha}f\left(h_{\arcsin}(x)\right)\right),\qquad
    q(x)=h_{\text{inv-sqrt}}(x)h_{\cos,\alpha}\left(\frac{1}{\alpha}f\left(h_{\arcsin}(x)\right)\right).
\end{equation}
Our goal is to show that $p(x)$ and $q(x)$ are dominated approximations of $\sin(f(\arcsin(x)))$ and $\frac{\cos(f(\arcsin(x)))}{\sqrt{1-x^2}}$, with suitably chosen accuracy parameters $\epsilon_{\arcsin}$, $\epsilon_{\text{inv-sqrt}}$, and $\epsilon_{\text{trig}}$.

Let us first verify the dominated condition. For any $x\in[-1,1]$, we have
\begin{equation}
\begin{aligned}
    &\abs{h_{\arcsin}(x)}\leq\abs{\arcsin(x)}\leq\frac{\pi}{2}\\
    &\Rightarrow\abs{\frac{1}{\alpha}f\left(h_{\arcsin}(x)\right)}\leq1\\
    &\Rightarrow\abs{h_{\sin,\alpha}\left(\frac{1}{\alpha}f\left(h_{\arcsin}(x)\right)\right)}\leq\abs{\sin\left(f\left(h_{\arcsin}(x)\right)\right)}+\epsilon_{\text{trig}}.
\end{aligned}
\end{equation}
Similarly,
\begin{equation}
\begin{aligned}
    &\abs{h_{\arcsin}(x)}\leq\abs{\arcsin(x)}\leq\frac{\pi}{2}\\
    &\Rightarrow\abs{\frac{1}{\alpha}f\left(h_{\arcsin}(x)\right)}\leq1\\
    &\Rightarrow\abs{h_{\cos,\alpha}\left(\frac{1}{\alpha}f\left(h_{\arcsin}(x)\right)\right)}\leq\abs{\cos\left(f\left(h_{\arcsin}(x)\right)\right)}+\epsilon_{\text{trig}}\\
    &\Rightarrow\abs{h_{\text{inv-sqrt}}(x)h_{\cos,\alpha}\left(\frac{1}{\alpha}f\left(h_{\arcsin}(x)\right)\right)}
    \leq\frac{1}{\sqrt{1-x^2}}\left(\abs{\cos\left(f\left(h_{\arcsin}(x)\right)\right)}+\epsilon_{\text{trig}}\right).
\end{aligned}
\end{equation}
Combining these two bounds, we obtain
\begin{equation}
\begin{aligned}
    &p^2(x)+(1-x^2)q^2(x)\\
    &=h_{\sin,\alpha}^2\left(\frac{1}{\alpha}f\left(h_{\arcsin}(x)\right)\right)
    +(1-x^2)h_{\text{inv-sqrt}}^2(x)h_{\cos,\alpha}^2\left(\frac{1}{\alpha}f\left(h_{\arcsin}(x)\right)\right)\\
    &\leq\left(\abs{\sin\left(f\left(h_{\arcsin}(x)\right)\right)}+\epsilon_{\text{trig}}\right)^2
    +\left(\abs{\cos\left(f\left(h_{\arcsin}(x)\right)\right)}+\epsilon_{\text{trig}}\right)^2\\
    &\leq1+2\sqrt{2}\epsilon_{\text{trig}}+2\epsilon_{\text{trig}}^2.
\end{aligned}
\end{equation}
Hence, the dominated condition is satisfied as long as $\epsilon_{\text{trig}}=\mathbf{O}(\epsilon)$.

For the remainder of the proof, we assume $x\in\left[-\sin\left(\frac{\pi}{2}-\xi\right),\sin\left(\frac{\pi}{2}-\xi\right)\right]$ and focus on the approximation condition. We have
\begin{equation}
\begin{aligned}
    &\abs{p(x)-\sin(f(\arcsin(x)))}\\
    &=\abs{h_{\sin,\alpha}\left(\frac{1}{\alpha}f\left(h_{\arcsin}(x)\right)\right)-\sin(f(\arcsin(x)))}\\
    &\leq\abs{h_{\sin,\alpha}\left(\frac{1}{\alpha}f\left(h_{\arcsin}(x)\right)\right)
    -\sin(f(h_{\arcsin}(x)))}
    +\abs{\sin\left(f\left(h_{\arcsin}(x)\right)\right)-\sin(f(\arcsin(x)))}\\
    &\leq\epsilon_{\text{trig}}+\abs{f\left(h_{\arcsin}(x)\right)-f(\arcsin(x))}\\
    &\leq\epsilon_{\text{trig}}+\norm{f'}_{\max,\left[-\frac{\pi}{2}+\xi,\frac{\pi}{2}-\xi\right]}\epsilon_{\arcsin}
    =\mathbf{O}\left(\epsilon_{\text{trig}}+d\norm{f}_{\max,\left[-\frac{\pi}{2},\frac{\pi}{2}\right]}\epsilon_{\arcsin}\right),
\end{aligned}
\end{equation}
where the last estimate follows from Bernstein's theorem~\cite[Eq.\ (12)]{Kalmykov21}, \cite[Lemma 29]{QEVP}
\begin{equation}
    \norm{f'}_{\max,\left[-\frac{\pi}{2}+\xi,\frac{\pi}{2}-\xi\right]}=\mathbf{O}\left(d\norm{f}_{\max,\left[-\frac{\pi}{2},\frac{\pi}{2}\right]}\right).
\end{equation}
Similarly,
\begin{equation}
\begin{aligned}
    &\abs{q(x)-\frac{1}{\sqrt{1-x^2}}\cos(f(\arcsin(x)))}\\
    &=\abs{h_{\text{inv-sqrt}}(x)h_{\cos,\alpha}\left(\frac{1}{\alpha}f\left(h_{\arcsin}(x)\right)\right)-\frac{1}{\sqrt{1-x^2}}\cos(f(\arcsin(x)))}\\
    &\leq\abs{h_{\text{inv-sqrt}}(x)h_{\cos,\alpha}\left(\frac{1}{\alpha}f\left(h_{\arcsin}(x)\right)\right)-\frac{1}{\sqrt{1-x^2}}h_{\cos,\alpha}\left(\frac{1}{\alpha}f\left(h_{\arcsin}(x)\right)\right)}\\
    &\quad+\abs{\frac{1}{\sqrt{1-x^2}}h_{\cos,\alpha}\left(\frac{1}{\alpha}f\left(h_{\arcsin}(x)\right)\right)-\frac{1}{\sqrt{1-x^2}}\cos(f(\arcsin(x)))}\\
    &\leq\abs{h_{\text{inv-sqrt}}(x)-\frac{1}{\sqrt{1-x^2}}}\left(1+\epsilon_{\text{trig}}\right)
    +\mathbf{O}\left(\epsilon_{\text{trig}}+d\norm{f}_{\max,\left[-\frac{\pi}{2},\frac{\pi}{2}\right]}\epsilon_{\arcsin}\right)\\
    &=\mathbf{O}\left(\epsilon_{\text{inv-sqrt}}+\epsilon_{\text{trig}}+d\norm{f}_{\max,\left[-\frac{\pi}{2},\frac{\pi}{2}\right]}\epsilon_{\arcsin}\right).
\end{aligned}
\end{equation}
Hence, the approximation condition is satisfied as long as $\epsilon_{\text{inv-sqrt}},\epsilon_{\text{trig}}=\mathbf{O}(\epsilon)$ and $\epsilon_{\arcsin}=\mathbf{O}\left(\epsilon/\left(d\norm{f}_{\max,\left[-\frac{\pi}{2},\frac{\pi}{2}\right]}\right)\right)$.

To summarize, $p$ is a real odd polynomial of degree
\begin{equation}
\begin{aligned}
    &\mathbf{O}\left(
    \left(\alpha+\log\left(\frac{1}{\epsilon_{\text{trig}}}\right)\right)d
    \log\left(\frac{1}{\epsilon_{\arcsin}}\right)\right)\\
    &=\mathbf{O}\left(d
    \left(\norm{f}_{\max,\left[-\frac{\pi}{2},\frac{\pi}{2}\right]}+\log\left(\frac{1}{\epsilon}\right)\right)
    \log\left(\frac{d\norm{f}_{\max,\left[-\frac{\pi}{2},\frac{\pi}{2}\right]}}{\epsilon}\right)\right),
\end{aligned}
\end{equation}
whereas $q$ is a real even polynomial of degree
\begin{equation}
\begin{aligned}
    &\mathbf{O}\left(
    \left(\alpha+\log\left(\frac{1}{\epsilon_{\text{trig}}}\right)\right)d
    \log\left(\frac{1}{\epsilon_{\arcsin}}\right)
    +\log\left(\frac{1}{\epsilon_{\text{inv-sqrt}}}\right)\right)\\
    &=\mathbf{O}\left(d
    \left(\norm{f}_{\max,\left[-\frac{\pi}{2},\frac{\pi}{2}\right]}+\log\left(\frac{1}{\epsilon}\right)\right)
    \log\left(\frac{d\norm{f}_{\max,\left[-\frac{\pi}{2},\frac{\pi}{2}\right]}}{\epsilon}\right)\right).
\end{aligned}
\end{equation}
This establishes the weaker bound on the asymptotic polynomial degree. The stronger bound follows by an analogous argument, with the dominated polynomial approximations constructed from analyzing composite functions on a Bernstein ellipse enclosing the unit interval.
See~\append{composite_stadium} and~\append{composite_qsvt} for details.
\end{proof}

\subsection{Singular value transformation for odd polynomials}
\label{sec:qsvt_odd}
In this and the next subsection, we will establish our main results on Hamiltonian quantum singular value transformation, with input and output both represented by Hamiltonian block encodings. 
Specifically, suppose we are given a Hamiltonian block encoding
\begin{equation}
    E_A=\exp\left(-i\begin{bmatrix}
        0 & A^\dagger\\
        A & 0
    \end{bmatrix}\right),
\end{equation}
where $A$ is a matrix with the singular value decomposition $A=U\Sigma V^\dagger$, and $\norm{A}\leq\frac{\pi}{2}-\xi<\frac{\pi}{2}$ for some constant $\xi$. For a real polynomial $f$, our desired output is 
\begin{equation}
    E_{f_{\text{sv}}(A)}=\exp\left(-i
    \begin{bmatrix}
        0 & f_{\text{sv}}^\dagger\left(A\right)\\
        f_{\text{sv}}\left(A\right) & 0
    \end{bmatrix}\right),
\end{equation}
where $f_{\text{sv}}\left(A\right)=Uf\left(\Sigma\right)V^\dagger$ if $f$ is odd, and $f_{\text{sv}}\left(A\right)=V f\left(\Sigma\right)V^\dagger$ if $f$ is even.

Let us start with a high-level description of the algorithm. To this end, we rewrite the input Hamiltonian block encoding as
\begin{equation}\label{eq: svd of hamiltonian block encoding}
\begin{aligned}
    E_A
    &=\exp\left(-i
    \begin{bmatrix}
        0 & V\Sigma U^\dagger\\
        U\Sigma V^\dagger & 0
    \end{bmatrix}\right)
    =\begin{bmatrix}
        V & 0\\
        0 & U
    \end{bmatrix}
    \exp\left(-i
    \begin{bmatrix}
        0 & \Sigma\\
        \Sigma & 0
    \end{bmatrix}\right)
    \begin{bmatrix}
        V^\dagger & 0\\
        0 & U^\dagger
    \end{bmatrix}\\
    &=\begin{bmatrix}
        V & 0\\
        0 & U
    \end{bmatrix}
    \begin{bmatrix}
        \cos\left(\Sigma\right) & -i\sin\left(\Sigma\right)\\
        -i\sin\left(\Sigma\right) & \cos\left(\Sigma\right)
    \end{bmatrix}
    \begin{bmatrix}
        V^\dagger & 0\\
        0 & U^\dagger
    \end{bmatrix}.\\
\end{aligned}
\end{equation}
Applying $iX$ on the ancilla qubit, this is then transformed into
\begin{equation}
\label{eq:qsvt_input_sin}
\begin{aligned}
    E_A\left(iX\otimes I\right)
    &=\begin{bmatrix}
        V & 0\\
        0 & U
    \end{bmatrix}
    \begin{bmatrix}
        \cos\left(\Sigma\right) & -i\sin\left(\Sigma\right)\\
        -i\sin\left(\Sigma\right) & \cos\left(\Sigma\right)
    \end{bmatrix}
    \left(iX\otimes I\right)
    \left(X\otimes I\right)
    \begin{bmatrix}
        V^\dagger & 0\\
        0 & U^\dagger
    \end{bmatrix}\left(X\otimes I\right)\\
    &=\begin{bmatrix}
        V & 0\\
        0 & U
    \end{bmatrix}
    \begin{bmatrix}
        \sin\left(\Sigma\right) & i\cos\left(\Sigma\right)\\
        i\cos\left(\Sigma\right) & \sin\left(\Sigma\right)
    \end{bmatrix}
    \begin{bmatrix}
        U^\dagger & 0\\
        0 & V^\dagger
    \end{bmatrix}.
\end{aligned}
\end{equation}
In the case where $A^\dagger =A$ is Hermitian, we have operator $X\otimes A$ in the exponent which commutes with $X\otimes I$. Hence, the application of $iX$ on the ancilla is equivalent to shifting the spectra of $A$ by $\frac{\pi}{2}$, i.e., $E_A(iX \otimes I)=E_AE_{-\frac{\pi}{2}} = E_{A - \frac{\pi}{2}}$.
Even when $A$ is not Hermitian, this modification still shifts the spectra by $\frac{\pi}{2}$ but now introduces a change of basis on the right.
To preserve this 2D block structure in QSVT, we would alternate between applications of this iterate and 
\begin{equation}
    \left(Z\otimes I\right)
    \left(-iX\otimes I\right)
    E_{-A}
    \left(Z\otimes I\right)
    =\begin{bmatrix}
        U &\\
        & V
    \end{bmatrix}
    \begin{bmatrix}
        \sin(\Sigma) & i\cos(\Sigma)\\
        i\cos(\Sigma) & \sin(\Sigma)
    \end{bmatrix}
    \begin{bmatrix}
        V^\dagger &\\
        & U^\dagger
    \end{bmatrix}.
\end{equation}
Here, $E_{-A}=\exp\left(-i\left[\begin{smallmatrix}
    0 & -A^\dagger\\
    -A & 0
\end{smallmatrix}\right]\right)$ is a Hamiltonian block encoding of $-A=(-U)\Sigma V^\dagger$ and can be constructed from the original $E_A$ through a complex phase scaling (\prop{phase_scale}).

Restricted to each 2D subspace, the input operator has the action 
\begin{equation}
    \begin{bmatrix}
        \sin\left(\sigma\right) & i\cos\left(\sigma\right)\\
        i\cos\left(\sigma\right) & \sin\left(\sigma\right)
    \end{bmatrix}
    =\begin{bmatrix}
        x & i\sqrt{1-x^2}\\
        i\sqrt{1-x^2} & x
    \end{bmatrix},
\end{equation}
where we have denoted $x=\sin\left(\sigma\right)$. Then QSVT allows us to perform~\cite[Theorem 3]{Gilyen2018singular}
\begin{equation}
    \begin{bmatrix}
        p(x) & iq(x)\sqrt{1-x^2}\\
        iq^*(x)\sqrt{1-x^2} & p^*(x)
    \end{bmatrix}
\end{equation}
within each 2D subspace. Here, $p$ and $q$ are complex polynomials whose degrees correspond to the number of queries to $E_A$. Furthermore, $p$ and $q$ have different parities, and satisfy the normalization constraint
\begin{equation}
\label{eq:poly_normalize}
    \abs{p(x)}^2+(1-x^2)\abs{q(x)}^2=1,\qquad \forall x\in[-1,1].
\end{equation}
To realize the target polynomial $f$, we thus seek polynomials $p$ and $q$ simultaneously such that, in addition to the parity and normalization condition, the following approximation condition is satisfied
\begin{equation}
\label{eq:poly_approx}
    p(x)\approx\sin(f(\arcsin(x))),\qquad
    q(x)\approx\frac{\cos(f(\arcsin(x)))}{\sqrt{1-x^2}},\qquad
    \forall x\in\left[-\sin\left(\frac{\pi}{2}-\xi\right),\sin\left(\frac{\pi}{2}-\xi\right)\right],
\end{equation}
We will show how these conditions can all be fulfilled for an odd $f$ using the dominated approximation result of~\prop{dominated}. 

Assuming for now this is achievable, the output of QSVT is then
\begin{equation}
    \begin{bmatrix}
        V & 0\\
        0 & U
    \end{bmatrix}
    \begin{bmatrix}
        \sin\left(f\left(\Sigma\right)\right) & i\cos\left(f\left(\Sigma\right)\right)\\
        i\cos\left(f\left(\Sigma\right)\right) & \sin\left(f\left(\Sigma\right)\right)
    \end{bmatrix}
    \begin{bmatrix}
        U^\dagger & 0\\
        0 & V^\dagger
    \end{bmatrix}
\end{equation}
up to an arbitrarily small error. A final application of $-iX$ on the ancilla qubit transforms this into
\begin{equation}
\begin{aligned}
    &\begin{bmatrix}
        V & 0\\
        0 & U
    \end{bmatrix}
    \begin{bmatrix}
        \sin\left(f\left(\Sigma\right)\right) & i\cos\left(f\left(\Sigma\right)\right)\\
        i\cos\left(f\left(\Sigma\right)\right) & \sin\left(f\left(\Sigma\right)\right)
    \end{bmatrix}
    \begin{bmatrix}
        U^\dagger & 0\\
        0 & V^\dagger
    \end{bmatrix}\left(-iX\otimes I\right)\\
    &=\begin{bmatrix}
        V & 0\\
        0 & U
    \end{bmatrix}
    \begin{bmatrix}
        \sin\left(f\left(\Sigma\right)\right) & i\cos\left(f\left(\Sigma\right)\right)\\
        i\cos\left(f\left(\Sigma\right)\right) & \sin\left(f\left(\Sigma\right)\right)
    \end{bmatrix}
    \left(-iX\otimes I\right)
    \left(X\otimes I\right)
    \begin{bmatrix}
        U^\dagger & 0\\
        0 & V^\dagger
    \end{bmatrix}\left(X\otimes I\right)\\
    &=\begin{bmatrix}
        V & 0\\
        0 & U
    \end{bmatrix}
    \begin{bmatrix}
        \cos\left(f\left(\Sigma\right)\right) & -i\sin\left(f\left(\Sigma\right)\right)\\
        -i\sin\left(f\left(\Sigma\right)\right) & \cos\left(f\left(\Sigma\right)\right)
    \end{bmatrix}
    \begin{bmatrix}
        V^\dagger & 0\\
        0 & U^\dagger
    \end{bmatrix}\\
    &=\begin{bmatrix}
        V & 0\\
        0 & U
    \end{bmatrix}
    \exp\left(-i\begin{bmatrix}
        0 & f\left(\Sigma\right)\\
        f\left(\Sigma\right) & 0
    \end{bmatrix}\right)
    \begin{bmatrix}
        V^\dagger & 0\\
        0 & U^\dagger
    \end{bmatrix}
    =\exp\left(-i
    \begin{bmatrix}
        0 & f_{\text{sv}}^\dagger\left(A\right)\\
        f_{\text{sv}}\left(A\right) & 0
    \end{bmatrix}\right).
\end{aligned}
\end{equation}
This is the desired output for Hamiltonian QSVT.

Observe that our above Hamiltonian QSVT differs from that described in~\cite{Lloye21hamiltonianqsvt}. Specifically, we have shifted $A$'s spectra by applying $iX$ to the ancilla, whereas~\cite{Lloye21hamiltonianqsvt} does not, yielding the dominated polynomial approximation problem where $p(x) \approx \sqrt{1 - f(\arccos(x))^2}, q(x) \approx \frac{f(\arccos(x))}{\sqrt{1-x^2}}$. Here, $\arccos(x)$ and $\frac{1}{\sqrt{1-x^2}}$ both have singularities at $x=\pm1$. To avoid this issue, \cite{Lloye21hamiltonianqsvt} imposed an additional requirement where the minimum singular value of $A$ is gapped away from $\sigma=\arccos(1)=0$. This gap requirement can be satisfied for a Hermitian input by shifting the spectra~\cite{DongLinTong22}, but no workaround exists in general. In contrast, $\arcsin $ also has a singularity at $\pm 1$, but $\sigma=\arcsin(1)=\frac{\pi}{2}$,  so it suffices to provide an upper bound on $\norm{A}$ to set an appropriate approximation region for dominated approximation, which introduces at most a constant factor overhead in many QSVT applications.

We now state the main theorem and explain the usage of dominated polynomial approximation in its derivation.

\begin{theorem}[Singular value transformation, odd case]
\label{thm:qsvt_odd}
Let $A$ be a matrix encoded by the Hamiltonian block encoding $E_{A}=\exp\left(-i\left[\begin{smallmatrix}
    0 & A^\dagger\\
    A & 0
\end{smallmatrix}\right]\right)$ with $\norm{A}<\frac{\pi}{2}$. Let $f$ be a real odd polynomial of degree $d$. Then the Hamiltonian block encoding
\begin{equation}
    E_{f_{\text{sv}}\left(A\right)}=\exp\left(-i
    \begin{bmatrix}
        0 & f_{\text{sv}}^\dagger\left(A\right)\\
        f_{\text{sv}}\left(A\right) & 0
    \end{bmatrix}\right)
\end{equation}
can be constructed with accuracy $\epsilon$ using
\begin{equation}
    \mathbf{O}\left(\left(d+\log\left(\frac{\norm{f}_{\max,\left[-\frac{\pi}{2},\frac{\pi}{2}\right]}}{\epsilon}\right)\right)
    \left(\norm{f}_{\max,\left[-\frac{\pi}{2},\frac{\pi}{2}\right]}+\log\left(\frac{1}{\epsilon}\right)\right)
    \right)
\end{equation}
queries to $E_{A}$.
See~\fig{ham_qsvt} for the corresponding circuit diagram.
\end{theorem}

\begin{figure}[t]
	\centering
\includegraphics[scale=\circuitwidth]{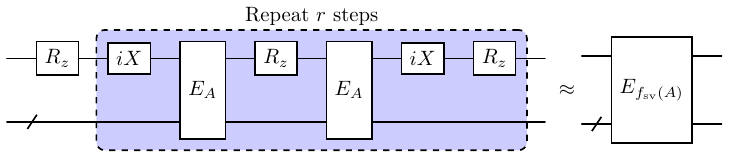}
\caption{Quantum circuit for Hamiltonian singular value transformation with odd polynomials.}
\label{fig:ham_qsvt}
\end{figure}

\begin{proof}
Suppose that $\norm{A}\leq\frac{\pi}{2}-\xi<\frac{\pi}{2}$ for some constant $0<\xi\leq\frac{\pi}{2}$. We apply~\prop{dominated} to construct a real odd polynomial $p$ and even polynomial $q$ such that
\begin{equation}
\begin{aligned}
    &\abs{p(x)-\sin(f(\arcsin(x)))}\leq\epsilon,\qquad&&\forall x\in\left[-\sin\left(\frac{\pi}{2}-\xi\right),\sin\left(\frac{\pi}{2}-\xi\right)\right],\\
    &\abs{q(x)-\frac{\cos(f(\arcsin(x)))}{\sqrt{1-x^2}}}\leq\epsilon,\qquad&&\forall x\in\left[-\sin\left(\frac{\pi}{2}-\xi\right),\sin\left(\frac{\pi}{2}-\xi\right)\right],\\
    &p^2(x)+(1-x^2)q^2(x)\leq1+\epsilon,\qquad&&\forall x\in[-1,1].
\end{aligned}
\end{equation}
Moreover, $p$ and $q$ have the asymptotic degree
\begin{equation}
    \mathbf{O}\left(\left(d+\log\left(\frac{\norm{f}_{\max,\left[-\frac{\pi}{2},\frac{\pi}{2}\right]}}{\epsilon}\right)\right)
    \left(\norm{f}_{\max,\left[-\frac{\pi}{2},\frac{\pi}{2}\right]}+\log\left(\frac{1}{\epsilon}\right)\right)
    \right).
\end{equation}

These polynomials already satisfy the parity constraint. To fulfill the normalization condition, we normalize them as
\begin{equation}
    p_1(x)=\frac{p(x)}{\sqrt{1+\epsilon}},\qquad
    q_1(x)=\frac{q(x)}{\sqrt{1+\epsilon}},
\end{equation}
which leads to the behavior
\begin{equation}
\begin{aligned}
    &\abs{p_1(x)-\sin(f(\arcsin(x)))}\leq\epsilon,\qquad&&\forall x\in\left[-\sin\left(\frac{\pi}{2}-\xi\right),\sin\left(\frac{\pi}{2}-\xi\right)\right],\\
    &\abs{q_1(x)-\frac{\cos(f(\arcsin(x)))}{\sqrt{1-x^2}}}\leq\epsilon,\qquad&&\forall x\in\left[-\sin\left(\frac{\pi}{2}-\xi\right),\sin\left(\frac{\pi}{2}-\xi\right)\right],\\
    &p_1^2(x)+(1-x^2)q_1^2(x)\leq1,\qquad&&\forall x\in[-1,1],
\end{aligned}
\end{equation}
while the polynomial degree remains asymptotically the same.
Indeed, the third claim follows directly from the definition. So we only verify the first two claims.
Using the fact that $\abs{p(x)}\leq\sqrt{1+\epsilon}$, we have
\begin{equation}
\begin{aligned}
    &\abs{p_1(x)-p(x)}
    =\frac{\sqrt{1+\epsilon}-1}{\sqrt{1+\epsilon}}\abs{p(x)}
    \leq\sqrt{1+\epsilon}-1
    \leq\frac{\epsilon}{2}\\
    \Rightarrow
    &\abs{p_1(x)-\sin(f(\arcsin(x)))}\leq\frac{3}{2}\epsilon.
\end{aligned}
\end{equation}
Similarly, since $\abs{q(x)}\leq\sqrt{\frac{1+\epsilon}{1-x^2}}\leq\frac{\sqrt{1+\epsilon}}{\cos\left(\frac{\pi}{2}-\xi\right)}$,
it holds
\begin{equation}
\begin{aligned}
    &\abs{q_1(x)-q(x)}
    =\frac{\sqrt{1+\epsilon}-1}{\sqrt{1+\epsilon}}\abs{q(x)}
    \leq\frac{\sqrt{1+\epsilon}-1}{\cos\left(\frac{\pi}{2}-\xi\right)}
    \leq\frac{\epsilon}{2\cos\left(\frac{\pi}{2}-\xi\right)}\\
    \Rightarrow&\abs{q_1(x)-\frac{\cos(f(\arcsin(x)))}{\sqrt{1-x^2}}}\leq\frac{\epsilon}{2\cos\left(\frac{\pi}{2}-\xi\right)}+\epsilon.
\end{aligned}
\end{equation}
The claimed behavior can then be achieved by rescaling $\epsilon$ without affecting the asymptotic degree of $p_1$ and $q_1$.

We are now in a position to apply~\cite[Theorem 5]{Gilyen2018singular}, which gives a complex odd polynomial $p_2$ and a complex even polynomial $q_2$ of the same asymptotic degree,
such that
\begin{equation}
    p_1(x)=\Re\left(p_2(x)\right),\qquad
    q_1(x)=\Re\left(q_2(x)\right).
\end{equation}
Moreover, the normalization constraint is satisfied over the entire unit interval:
\begin{equation}
    \abs{p_2(x)}^2+(1-x^2)\abs{q_2(x)}^2=1,\qquad \forall x\in[-1,1].
\end{equation}
Let us bound the error of using $p_2$ and $q_2$ in the implementation of Hamiltonian QSVT. For any $x\in\left[-\sin\left(\frac{\pi}{2}-\xi\right),\sin\left(\frac{\pi}{2}-\xi\right)\right]$, we have
\begin{equation}
\begin{aligned}
    \abs{p_1^2(x)+(1-x^2)q_1^2(x)-1}
    &=\abs{p_1^2(x)+(1-x^2)q_1^2(x)-\sin^2(f(\arcsin(x)))-\cos^2(f(\arcsin(x)))}\\
    &=\mathbf{O}(\epsilon).\\
\end{aligned}
\end{equation}
But because of the normalization condition $\abs{p_2(x)}^2+(1-x^2)\abs{q_2(x)}^2=1$, this means the imaginary components must be small:
\begin{equation}
    \Im^2\left(p_2(x)\right)+(1-x^2)\Im^2\left(q_2(x)\right)=\mathbf{O}(\epsilon).
\end{equation}
Restricted to each 2D subspace, we thus have 
\begin{equation}
\begin{aligned}
    &\norm{\begin{bmatrix}
        p_2(x) & iq_2(x)\sqrt{1-x^2}\\
        iq_2^*(x)\sqrt{1-x^2} & p_2^*(x)
    \end{bmatrix}
    -\begin{bmatrix}
        \sin(f(\arcsin(x))) & i\cos(f(\arcsin(x)))\\
        i\cos(f(\arcsin(x))) & \sin(f(\arcsin(x)))
    \end{bmatrix}}\\
    &\leq\norm{\begin{bmatrix}
        p_1(x)-\sin(f(\arcsin(x))) & iq_1(x)\sqrt{1-x^2}-i\cos(f(\arcsin(x)))\\
        iq_1(x)\sqrt{1-x^2}-i\cos(f(\arcsin(x))) & p_1(x)-\sin(f(\arcsin(x)))
    \end{bmatrix}}\\
    &\quad+\norm{\begin{bmatrix}
        i\Im\left(p_2(x)\right) & -\Im\left(q_2(x)\right)\sqrt{1-x^2}\\
        \Im\left(q_2(x)\right)\sqrt{1-x^2} & -i\Im\left(p_2(x)\right)
    \end{bmatrix}}
    =\mathbf{O}\left(\sqrt{\epsilon}\right),
\end{aligned}
\end{equation}
where the asymptotic estimate follows from the matrix distance formula in the Pauli basis~\cite[Lemma 21]{low2024quantum}:
\begin{equation}
    \norm{\begin{bmatrix}
        i\Im\left(p_2(x)\right) & -\Im\left(q_2(x)\right)\sqrt{1-x^2}\\
        \Im\left(q_2(x)\right)\sqrt{1-x^2} & -i\Im\left(p_2(x)\right)
    \end{bmatrix}}
    =\sqrt{\Im^2\left(p_2(x)\right)+\Im^2\left(q_2(x)\right)(1-x^2)}.
\end{equation}
By a final rescaling of $\epsilon$, we have completed the proof of all claims of the theorem.
\end{proof}

\subsection{Singular value transformation for even polynomials}
\label{sec:qsvt_even}
We now continue the discussion of Hamiltonian QSVT, but turn to real even polynomials $f$. In this case, both $\sin(f(\arcsin(x)))$ and $\frac{\cos(f(\arcsin(x)))}{\sqrt{1-x^2}}$ are even functions, so there is an intrinsic issue concerning the parity constraint of QSVT.

Note that if the input matrix has the norm bound $\norm{A}<\frac{\pi}{2}$, all its singular values satisfy $0\leq\sigma<\frac{\pi}{2}$ which implies $0\leq x=\sin(\sigma)<1$. Consequently, the functions $\sin(f(\arcsin(x)))$ and $\frac{\cos(f(\arcsin(x)))}{\sqrt{1-x^2}}$ only need to be approximated over $x\in[0,1)$, i.e., the function $f$ needs to be approximated only over $\sigma\in\left[0,\frac{\pi}{2}\right)$. Because of this, one may attempt to circumvent the parity issue by constructing an odd polynomial $h_{\text{odd}}$ that approximates the even function $f$ precisely over $\sigma\in\left[0,\frac{\pi}{2}\right)$. However, this approach would not work in general. Indeed, as $h_{\text{odd}}$ is an odd polynomial, we have $h_{\text{odd}}^{(2k)}(0) = 0$ for all $k\in\mathbb{Z}_{\geq0}$, whereas $f^{(2k)}(0)\neq0$ for some value of $k$ when $f$ is analytic but does not vanish in a neighborhood of $0$. This means a singularity at $0$ would occur if $f$ or its higher derivatives are non-vanishing, resulting in an inaccurate approximation.
We resolve this issue by shifting the domain of approximation away from $0$, thereby relaxing the requirements on the derivatives of $h_{\textnormal
{odd}}$. 

Let us first address this for Hermitian input matrices $H$. Assuming $H$ is properly normalized, we add multiples of identity so that the shifted operator only has positive spectra. Then performing even function $f$ on the original operator $H$ can be achieved by applying an odd extension of the shifted $f$ on the shifted $H$. 
This extension uses window functions and may be bundled with the dominated approximation result~\cor{approx_dominated_ellipse} whose proof also involves multiplying window functions. However, we leave the details for future work.

\begin{proposition}[Dominated polynomial extension]
\label{prop:dominated_ext}
Let $f$ be a real polynomial of degree $d$. For any $\epsilon_{\text{dom}}>0$ and constant $0<\xi\leq 1<b$, there exists a real polynomial $h_{\text{dom}}$ such that 
\begin{equation}
\begin{aligned}
    &\abs{h_{\text{dom}}(x)-f(x)}
    \leq\epsilon_{\text{dom}},\qquad&&x\in\left[-1+\xi,1-\xi\right],\\
    &\abs{h_{\text{dom}}(x)}\leq\abs{f(x)}+\epsilon_{\text{dom}},\qquad&&x\in\left[-1,1\right],\\
    &\abs{h_{\text{dom}}(x)}\leq\epsilon_{\text{dom}},\qquad&&x\in\left[-b,-1\right]\cup\left[1,b\right],
\end{aligned}
\end{equation}
with degree
\begin{equation}
    \mathbf{O}\left(\sqrt{b-1}\frac{b}{\xi}d+\frac{b}{\xi}\log\left(\frac{\norm{f}_{\max,\left[-1,1\right]}}{\epsilon_{\text{dom}}}\right)\right).
\end{equation}
Moreover, $h_{\text{dom}}$ has the same parity as $f$.
\end{proposition}
\begin{proof}
First, we have the following bound on $\norm{f}_{\max,[-b,b]}$ in terms of $\norm{f}_{\max,[-1,1]}$:
\begin{equation}
    \norm{f}_{\max,[-b,b]}\leq\exp\left(\sqrt{2(b-1)}d\right)\norm{f}_{\max,\left[-1,1\right]}.
\end{equation}
This is similar to the estimate in the proof of~\prop{poly_stadium}, and we reproduce part of that proof here for completeness. By Markov brothers' inequality~\cite[Eq.\ (33)]{Kalmykov21}, we can bound higher-order derivatives of $f$ as
\begin{equation}
    \norm{f^{(j)}}_{\max,\left[-1,1\right]}
    \leq\frac{d^{2j}}{(2j-1)!!}\norm{f}_{\max,\left[-1,1\right]}.
\end{equation}
Now without loss of generality, taking any $z\in[1,b]$ and applying Taylor's theorem, we get
\begin{equation}
\begin{aligned}
    \abs{f(z)}
    &=\abs{\sum_{j=0}^d\frac{f^{(j)}(1)}{j!}(z-1)^j}
    \leq\sum_{j=0}^d\frac{\abs{f^{(j)}(1)}}{j!}(b-1)^j
    \leq\sum_{j=0}^d\frac{d^{2j}}{(2j-1)!!}\frac{1}{j!}(b-1)^j\norm{f}_{\max,[-1,1]}\\
    &=\sum_{j=0}^d\frac{d^{2j}2^jj!}{(2j)!j!}(b-1)^j\norm{f}_{\max,[-1,1]}
    =\sum_{j=0}^d\frac{d^{2j}\sqrt{2(b-1)}^{2j}}{(2j)!}\norm{f}_{\max,[-1,1]}\\
    &\leq\cosh\left(\sqrt{2(b-1)}d\right)\norm{f}_{\max,[-1,1]}.
\end{aligned}
\end{equation}

To restrict the range of $f$, we multiply it with a smoothened rectangular window function. Specifically, we construct an even polynomial $w(x)$:
\begin{equation}
    w(x)\in
    \begin{cases}
        [1-\epsilon_{\text{rec}},1],\quad&x\in\left[-1+\xi,1-\xi\right],\\
        [-1,1],&x\in[-1,-1+\xi]\cup[1-\xi,1],\\
        [0,\epsilon_{\text{rec}}],&x\in\left[-b,-1\right]\cup\left[1,b\right],
    \end{cases}
\end{equation}
with the effective transition width $\frac{\xi}{b}$
and an asymptotic degree of~\cite[Lemma 29]{Gilyen2018singular}
\begin{equation}
    \mathbf{O}\left(\frac{b}{\xi}\log\left(\frac{1}{\epsilon_{\text{rec}}}\right)\right).
\end{equation}
We then let
\begin{equation}
    h_{\text{dom}}(x)=f(x)w(x).
\end{equation}
This already satisfies the first two requirements of $h_{\text{dom}}$ as long as
\begin{equation}
    \epsilon_{\text{rec}}=\mathbf{O}\left(\frac{\epsilon_{\text{dom}}}{\norm{f}_{\max,\left[-1+\xi,1-\xi\right]}}\right).
\end{equation}
To fulfill the third requirement, we require
\begin{equation}
    \epsilon_{\text{rec}}
    =\mathbf{O}\left(\frac{\epsilon_{\text{dom}}}{\exp\left(\sqrt{2(b-1)}d\right)\norm{f}_{\max,\left[-1,1\right]}}\right).
\end{equation}
This choice of $\epsilon_{\text{rec}}$ gives the polynomial extension $h_{\text{dom}}$ with the claimed behavior.
\end{proof}

We now state and prove the Hamiltonian QSVT result for even polynomials and Hermitian inputs, which has the same asymptotic query complexity as in the odd case.

\begin{theorem}[Singular value transformation, Hermitian even case]
\label{thm:qsvt_herm_even}
Let $H$ be a Hermitian matrix encoded by the Hamiltonian block encoding $E_{H}=\exp\left(-i\left[\begin{smallmatrix}
    0 & H\\
    H & 0
\end{smallmatrix}\right]\right)$ with $\norm{H}<\frac{\pi}{4}$. Let $f$ be a real even polynomial of degree $d$. Then the Hamiltonian block encoding
\begin{equation}
    E_{f(H)}=\exp\left(-i
    \begin{bmatrix}
        0 & f(H)\\
        f(H) & 0
    \end{bmatrix}\right)
\end{equation}
can be constructed with accuracy $\epsilon$ using
\begin{equation}
    \mathbf{O}\left(\left(d+\log\left(\frac{\norm{f}_{\max,\left[-\frac{\pi}{4},\frac{\pi}{4}\right]}}{\epsilon}\right)\right)
    \left(\norm{f}_{\max,\left[-\frac{\pi}{4},\frac{\pi}{4}\right]}+\log\left(\frac{1}{\epsilon}\right)\right)
    \right)
\end{equation}
queries to $E_{H}$.
\end{theorem}
\begin{proof}

Suppose that $\norm{H}\leq\frac{\pi}{4}-\xi<\frac{\pi}{4}$ for some constant $0<\xi\leq\frac{\pi}{4}$. This means we only need to approximate $\sin(f(\arcsin(x)))$ and $\frac{\cos(f(\arcsin(x)))}{\sqrt{1-x^2}}$ over $\left[-\sin\left(\frac{\pi}{4}-\xi\right),\sin\left(\frac{\pi}{4}-\xi\right)\right]$. Equivalently, we only need to approximate $f(x)$ over $\left[-\frac{\pi}{4}+\xi,\frac{\pi}{4}-\xi\right]$.

We first construct a dominated extension of $f$ using a rescaled version of~\prop{dominated_ext}. This produces a real even polynomial $h_{\text{dom}}$ with the behavior
\begin{equation}
\begin{aligned}
    &\abs{h_{\text{dom}}(x)-f(x)}
    \leq\epsilon,\qquad&&x\in\left[-\frac{\pi}{4}+\xi,\frac{\pi}{4}-\xi\right],\\
    &\abs{h_{\text{dom}}(x)}\leq\abs{f(x)}+\epsilon,\qquad&&x\in\left[-\frac{\pi}{4},\frac{\pi}{4}\right],\\
    &\abs{h_{\text{dom}}(x)}\leq\epsilon,\qquad&&x\in\left[-\frac{3\pi}{4},-\frac{\pi}{4}\right]\bigcup\left[\frac{\pi}{4},\frac{3\pi}{4}\right],
\end{aligned}
\end{equation}
and an asymptotic degree of
\begin{equation}
    \mathbf{O}\left(d+\log\left(\frac{\norm{f}_{\max,\left[-\frac{\pi}{4},\frac{\pi}{4}\right]}}{\epsilon}\right)\right).
\end{equation}

Then, we shift $h$ to the right by $\frac{\pi}{4}$ and take its odd component
\begin{equation}
    h_{\text{odd}}(x)=h_{\text{dom}}\left(x-\frac{\pi}{4}\right)-h_{\text{dom}}\left(-x-\frac{\pi}{4}\right),
\end{equation}
which has the behavior
\begin{equation}
\begin{aligned}
    &\abs{h_{\text{odd}}(x)-f\left(x-\frac{\pi}{4}\right)}
    \leq2\epsilon,\qquad&&x\in\left[\xi,\frac{\pi}{2}-\xi\right],\\
    &\abs{h_{\text{odd}}(x)}\leq\abs{f\left(x-\frac{\pi}{4}\right)}+2\epsilon,\qquad&&x\in\left[0,\frac{\pi}{2}\right],\\
\end{aligned}
\end{equation}
for $x\in\left[0,\frac{\pi}{2}\right]$, an odd parity $h_{\text{odd}}(-x)=-h_{\text{odd}}(x)$ and the same asymptotic degree as $h_{\text{dom}}$.
Correspondingly, we also shift the input Hamiltonian block encoding to
\begin{equation}
    E_{H+\frac{\pi}{4}I}=\exp\left(-i
    \begin{bmatrix}
        0 & H+\frac{\pi}{4}I\\
        H+\frac{\pi}{4}I & 0
    \end{bmatrix}\right).
\end{equation}
Since $H$ and $I$ commute, this shifting can be realized using the first-order Lie-Trotter formula with no Trotter error and no query overhead.

Now setting $H+\frac{\pi}{4}I$ as the input matrix with spectra enclosed by $\left[\xi,\frac{\pi}{2}-\xi\right]$, 
and $h_{\text{odd}}$ as the target real odd polynomial of degree $\mathbf{O}\left(d+\log\left(\frac{\norm{f}_{\max,\left[-\frac{\pi}{4},\frac{\pi}{4}\right]}}{\epsilon}\right)\right)$,
let us invoke the odd version of Hamiltonian QSVT in \thm{qsvt_odd} to get
\begin{equation}
    \exp\left(-i
    \begin{bmatrix}
        0 & h_{\text{odd,sv}}^\dagger\left(H+\frac{\pi}{4}I\right)\\
        h_{\text{odd,sv}}\left(H+\frac{\pi}{4}I\right) & 0
    \end{bmatrix}\right)
    =\exp\left(-i
    \begin{bmatrix}
        0 & h_{\text{odd}}\left(H+\frac{\pi}{4}I\right)\\
        h_{\text{odd}}\left(H+\frac{\pi}{4}I\right) & 0
    \end{bmatrix}\right)
\end{equation}
up to accuracy $\epsilon$.
To justify this identity, note that if $H=U\Lambda U^\dagger$ is the spectral decomposition of $H$, then $H+\frac{\pi}{4}I=U\left(\Lambda+\frac{\pi}{4}I\right) U^\dagger$ is the singular value decomposition of the positive semidefinite operator $H+\frac{\pi}{4}I$, which implies
\begin{equation}
    h_{\text{odd,sv}}\left(H+\frac{\pi}{4}I\right)
    =h_{\text{odd,sv}}\left(U\left(\Lambda+\frac{\pi}{4}I\right) U^\dagger\right)
    =Uh_{\text{odd}}\left(\Lambda+\frac{\pi}{4}I\right) U^\dagger
    =h_{\text{odd}}\left(H+\frac{\pi}{4}I\right),
\end{equation}
where
\begin{equation}
\begin{aligned}
    \norm{h_{\text{odd}}\left(H+\frac{\pi}{4}I\right)-f(H)}\leq2\epsilon.
\end{aligned}
\end{equation}
This has query complexity
\begin{equation}
\begin{aligned}
    &\mathbf{O}\left(\left(d+\log\left(\frac{\norm{f}_{\max,\left[-\frac{\pi}{4},\frac{\pi}{4}\right]}}{\epsilon}\right)
    +\log\left(\frac{\norm{h_{\text{odd}}}_{\max,\left[-\frac{\pi}{2},\frac{\pi}{2}\right]}}{\epsilon}\right)\right)
    \left(\norm{h_{\text{odd}}}_{\max,\left[-\frac{\pi}{2},\frac{\pi}{2}\right]}+\log\left(\frac{1}{\epsilon}\right)\right)
    \right)\\
    &=\mathbf{O}\left(\left(d+\log\left(\frac{\norm{f}_{\max,\left[-\frac{\pi}{4},\frac{\pi}{4}\right]}}{\epsilon}\right)\right)
    \left(\norm{f}_{\max,\left[-\frac{\pi}{4},\frac{\pi}{4}\right]}+\log\left(\frac{1}{\epsilon}\right)\right)
    \right)
\end{aligned}
\end{equation}
as expected.
\end{proof}

Now that we have the Hamiltonian QSVT algorithm for performing even polynomials $f(x)$ on Hermitian inputs, we handle a general input matrix $A$ by constructing a Hamiltonian block encoding of its Hermitian dilation $\left[\begin{smallmatrix}
    0 & A^\dagger\\
    A & 0
\end{smallmatrix}\right]$ and then applying $f(x)$. 

\begin{figure}[t]
	\centering
    \begin{subfigure}[t]{\textwidth}
        \centering
        \includegraphics[scale=1]{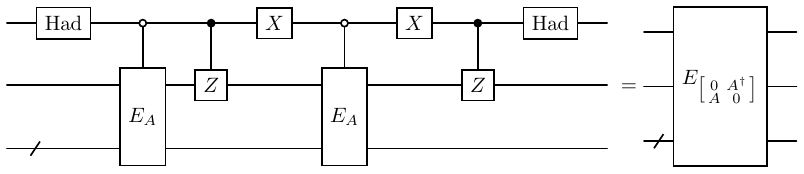}
        \caption{}
    \end{subfigure}%
    \\
    \begin{subfigure}[t]{\textwidth}
        \centering
        \includegraphics[scale=1]{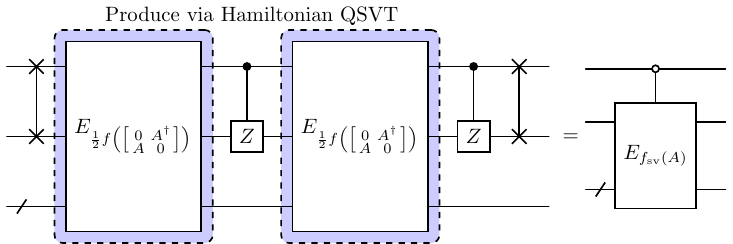}
        \caption{}
    \end{subfigure}%
\caption{Quantum circuit for Hamiltonian singular value transformation with even polynomials and generic input matrices.}
\label{fig:ham_qsvt_even}
\end{figure}

\begin{theorem}[Singular value transformation, even case]
\label{thm:qsvt_even}
Let $A$ be a matrix encoded by the Hamiltonian block encoding $E_{A}=\exp\left(-i\left[\begin{smallmatrix}
    0 & A^\dagger\\
    A & 0
\end{smallmatrix}\right]\right)$ with $\norm{A}<\frac{\pi}{4}$. Let $f$ be a real even polynomial of degree $d$. Then the controlled version of Hamiltonian block encoding
\begin{equation}
    E_{f_{\text{sv}}\left(A\right)}=\exp\left(-i
    \begin{bmatrix}
        0 & f_{\text{sv}}\left(A\right)\\
        f_{\text{sv}}\left(A\right) & 0
    \end{bmatrix}\right)
\end{equation}
can be constructed with accuracy $\epsilon$ using
\begin{equation}
    \mathbf{O}\left(\left(d+\log\left(\frac{\norm{f}_{\max,\left[-\frac{\pi}{4},\frac{\pi}{4}\right]}}{\epsilon}\right)\right)
    \left(\norm{f}_{\max,\left[-\frac{\pi}{4},\frac{\pi}{4}\right]}+\log\left(\frac{1}{\epsilon}\right)\right)
    \right)
\end{equation}
queries to controlled $E_{A}$.
See~\fig{ham_qsvt_even} for the corresponding circuit diagram.
\end{theorem}
\begin{proof}
We start with the Hamiltonian block encoding
$
    E_{\left[\begin{smallmatrix}
        0 & A^\dagger\\
        A & 0
    \end{smallmatrix}\right]}
    =\exp\left(-i\left[\begin{smallmatrix}
        0 & 0 & 0 & A^\dagger\\
        0 & 0 & A & 0\\
        0 & A^\dagger & 0 & 0\\
        A & 0 & 0 & 0
    \end{smallmatrix}\right]\right)
$,
which can be realized as follows:
\begin{equation}
\begin{aligned}
    \exp\left(-i\begin{bmatrix}
        0 & A^\dagger & 0 & 0\\
        A & 0 & 0 & 0\\
        0 & 0 & 0 & A^\dagger\\
        0 & 0 & A & 0
    \end{bmatrix}\right)
    \overset{\left(\text{C}Z\otimes I\right)(\cdot)\left(\text{C}Z\otimes I\right)}&{\longmapsto}
    \exp\left(-i\begin{bmatrix}
        0 & A^\dagger & 0 & 0\\
        A & 0 & 0 & 0\\
        0 & 0 & 0 & -A^\dagger\\
        0 & 0 & -A & 0
    \end{bmatrix}\right)\\
    \overset{\left(\text{Had}\otimes I\otimes I\right)(\cdot)\left(\text{Had}\otimes I\otimes I\right)}&{\longmapsto}
    \exp\left(-i\begin{bmatrix}
        0 & 0 & 0 & A^\dagger\\
        0 & 0 & A & 0\\
        0 & A^\dagger & 0 & 0\\
        A & 0 & 0 & 0
    \end{bmatrix}\right),
\end{aligned}
\end{equation}
where $\text{C}Z=\ketbra{0}{0}\otimes I+\ketbra{1}{1}\otimes Z$ is the controlled Pauli-$Z$ operator acting on the ancilla qubits. Here, $\exp\left(-i\left[\begin{smallmatrix}
        0 & A^\dagger & 0 & 0\\
        A & 0 & 0 & 0\\
        0 & 0 & 0 & A^\dagger\\
        0 & 0 & A & 0
    \end{smallmatrix}\right]\right)=I\otimes E_A$ is simply the input Hamiltonian block encoding with trivial action on the first ancilla. Alternatively, it can be constructed by combining the controlled version $\exp\left(-i\left[\begin{smallmatrix}
        0 & A^\dagger & 0 & 0\\
        A & 0 & 0 & 0\\
        0 & 0 & 0 & 0\\
        0 & 0 & 0 & 0
    \end{smallmatrix}\right]\right)$ and
\begin{equation}
    \exp\left(-i\begin{bmatrix}
        0 & A^\dagger & 0 & 0\\
        A & 0 & 0 & 0\\
        0 & 0 & 0 & 0\\
        0 & 0 & 0 & 0
    \end{bmatrix}\right)
    \overset{\left(X\otimes I\otimes I\right)(\cdot)\left(X\otimes I\otimes I\right)}{\longmapsto}
    \exp\left(-i\begin{bmatrix}
        0 & 0 & 0 & 0\\
        0 & 0 & 0 & 0\\
        0 & 0 & 0 & A^\dagger\\
        0 & 0 & A & 0
    \end{bmatrix}\right)
\end{equation}
using the first-order Lie-Trotter formula with no Trotter error.

Note that $\left[\begin{smallmatrix}
        0 & A^\dagger\\
        A & 0
    \end{smallmatrix}\right]$ is a Hermitian matrix with spectral norm $\norm{\left[\begin{smallmatrix}
        0 & A^\dagger\\
        A & 0
    \end{smallmatrix}\right]}=\norm{A}<\frac{\pi}{4}$.
We can thus invoke~\thm{qsvt_herm_even} with the even polynomial $\frac{1}{2}f$ to get
\begin{equation}
    E_{\frac{1}{2}f\left[\begin{smallmatrix}
        0 & A^\dagger\\
        A & 0
    \end{smallmatrix}\right]}
    =E_{\left[\begin{smallmatrix}
        \frac{1}{2}f_{\text{sv}}(A) & 0\\
        0 & \frac{1}{2}f_{\text{sv}}\left(A^\dagger\right)
    \end{smallmatrix}\right]}
    =\exp\left(-i\begin{bmatrix}
        0 & 0 & \frac{1}{2}f_{\text{sv}}(A) & 0\\
        0 & 0 & 0 & \frac{1}{2}f_{\text{sv}}\left(A^\dagger\right)\\
        \frac{1}{2}f_{\text{sv}}(A) & 0 & 0 & 0\\
        0 & \frac{1}{2}f_{\text{sv}}\left(A^\dagger\right) & 0 & 0
    \end{bmatrix}\right),
\end{equation}
where the first equality follows from the observation that
\begin{equation}
    \begin{bmatrix}
        0 & A^\dagger\\
        A & 0
    \end{bmatrix}^d
    =\begin{bmatrix}
        V\Sigma^dV^\dagger & 0\\
        0 & U\Sigma^dU^\dagger
    \end{bmatrix}
\end{equation}
when $d$ is an even integer and $A$ has the singular value decomposition $A=U\Sigma V^\dagger$.

We now combine this with
\begin{equation}
\begin{aligned}
    &\exp\left(-i\begin{bmatrix}
        0 & 0 & \frac{1}{2}f_{\text{sv}}(A) & 0\\
        0 & 0 & 0 & \frac{1}{2}f_{\text{sv}}\left(A^\dagger\right)\\
        \frac{1}{2}f_{\text{sv}}(A) & 0 & 0 & 0\\
        0 & \frac{1}{2}f_{\text{sv}}\left(A^\dagger\right) & 0 & 0
    \end{bmatrix}\right)\\
    \overset{\left(\text{C}Z\otimes I\right)(\cdot)\left(\text{C}Z\otimes I\right)}&{\longmapsto}
    \exp\left(-i\begin{bmatrix}
        0 & 0 & \frac{1}{2}f_{\text{sv}}(A) & 0\\
        0 & 0 & 0 & -\frac{1}{2}f_{\text{sv}}\left(A^\dagger\right)\\
        \frac{1}{2}f_{\text{sv}}(A) & 0 & 0 & 0\\
        0 & -\frac{1}{2}f_{\text{sv}}\left(A^\dagger\right) & 0 & 0
    \end{bmatrix}\right)
\end{aligned}
\end{equation}
using the first-order Trotter formula, obtaining $\exp\left(-i\left[\begin{smallmatrix}
        0 & 0 & f_{\text{sv}}(A) & 0\\
        0 & 0 & 0 & 0\\
        f_{\text{sv}}(A) & 0 & 0 & 0\\
        0 & 0 & 0 & 0
    \end{smallmatrix}\right]\right)$ exactly with no Trotter error. Finally, the result Hamiltonian block encoding can be converted to the standard form via
\begin{equation}
    \exp\left(-i\begin{bmatrix}
        0 & 0 & f_{\text{sv}}(A) & 0\\
        0 & 0 & 0 & 0\\
        f_{\text{sv}}(A) & 0 & 0 & 0\\
        0 & 0 & 0 & 0
    \end{bmatrix}\right)
    \overset{\left(\mathrm{SWAP}\otimes I\right)(\cdot)\left(\mathrm{SWAP}\otimes I\right)}{\longmapsto}
    \exp\left(-i\begin{bmatrix}
        0 & f_{\text{sv}}(A) & 0 & 0\\
        f_{\text{sv}}(A) & 0 & 0 & 0\\
        0 & 0 & 0 & 0\\
        0 & 0 & 0 & 0
    \end{bmatrix}\right).
\end{equation}
\end{proof}

\addtocounter{lemma}{1}

\subsection{Matrix inversion and fractional scaling}
\label{sec:qsvt_inverse_frac}
As immediate consequences of Hamiltonian QSVT, we obtain methods for inverting and rescaling matrices with Hamiltonian evolution.

Matrix inversion is a fundamental operation in quantum linear algebra, which is closely related to the quantum linear system problem that has been extensively studied~\cite{Childs2015LinearSystems,Ambainis2012VTAA,Subasi2019QLSPadiabatic,Costa2021linearsystems,Dalzell2024shortcut,low2024quantum} following the seminal work of Harrow, Hassidim and Lloyd~\cite{Harrow2009}. Within Hamiltonian block encoding, we assume that the input matrix $A$ is given by $E_{A}=\exp\left(-i\left[\begin{smallmatrix}
    0 & A^\dagger\\
    A & 0
\end{smallmatrix}\right]\right)$ with $\norm{A}<\frac{\pi}{2}$. Our goal is to produce $E_{A^{-1}/\kappa}=\exp\left(-i\left[\begin{smallmatrix}
        0 & A^{-1\dagger}/\kappa\\
        A^{-1}/\kappa & 0
\end{smallmatrix}\right]\right)$, where $\kappa\geq\norm{A^{-1}}$ is a known upper bound on the norm of inverse matrix. These conditions together imply that $A$ has singular values between $\big[\frac{1}{\kappa},\frac{\pi}{2}\big)$.

In light of the Hamiltonian QSVT algorithm from~\thm{qsvt_odd}, we can realize matrix inversion by first constructing an odd polynomial approximation $h(x)\approx\frac{1}{\kappa x}$ for $x\in\big[\frac{1}{\kappa},\frac{\pi}{2}\big)$~\cite[Corollary 69]{Gilyen2018singular}, followed by a dominated polynomial approximation of $p(x)\approx\sin(h(\arcsin(x)))$ and $q(x)\approx\frac{\cos(h(\arcsin(x)))}{\sqrt{1-x^2}}$ using~\prop{dominated}. However, this approach would introduce additional overhead to the query complexity. Instead, we directly construct the following dominated approximation of $\sin\left(\frac{1}{\kappa\arcsin(x)}\right)$ and $\frac{\cos\left(\frac{1}{\kappa\arcsin(x)}\right)}{\sqrt{1-x^2}}$, whose proof is deferred to~\append{composite_inverse}.

\begin{proposition}[Dominated approximation for Hamiltonian matrix inversion]
For any $\epsilon>0$, $\kappa>0$, and constant $0<\xi\leq\frac{\pi}{2}$, there exist a real odd polynomial $p(x)$ and even polynomial $q(x)$ such that
\begin{equation}
\begin{aligned}
    &\abs{p(x)-\sin\left(\frac{1}{\kappa\arcsin(x)}\right)}\leq\epsilon,\quad&&\forall x\in\left[-\sin\left(\frac{\pi}{2}-\xi\right),-\sin\left(\frac{1}{\kappa}\right)\right]\bigcup\left[\sin\left(\frac{1}{\kappa}\right),\sin\left(\frac{\pi}{2}-\xi\right)\right],\\
    &\abs{q(x)-\frac{\cos\left(\frac{1}{\kappa\arcsin(x)}\right)}{\sqrt{1-x^2}}}\leq\epsilon,\quad&&\forall x\in\left[-\sin\left(\frac{\pi}{2}-\xi\right),-\sin\left(\frac{1}{\kappa}\right)\right]\bigcup\left[\sin\left(\frac{1}{\kappa}\right),\sin\left(\frac{\pi}{2}-\xi\right)\right],\\
    &p^2(x)+(1-x^2)q^2(x)\leq1+\epsilon,\quad&&\forall x\in[-1,1].
\end{aligned}
\end{equation}
Moreover, both $p$ and $q$ have the asymptotic degree
\begin{equation}
    \mathbf{O}\left(\kappa\log\left(\frac{\kappa}{\epsilon}\right)\right).
\end{equation}
\end{proposition}

\begin{corollary}[Matrix inversion]
\label{cor:inverse}
Let $A$ be a matrix encoded by the Hamiltonian block encoding $E_{A}=\exp\left(-i\left[\begin{smallmatrix}
    0 & A^\dagger\\
    A & 0
\end{smallmatrix}\right]\right)$ with $\norm{A}<\frac{\pi}{2}$. Let $\kappa\geq\norm{A^{-1}}$ be a norm upper bound on the inverse matrix. Then the Hamiltonian block encoding
\begin{equation}
    E_{A^{-1}/\kappa}=\exp\left(-i\begin{bmatrix}
        0 & \frac{A^{-1\dagger}}{\kappa}\\
        \frac{A^{-1}}{\kappa} & 0
    \end{bmatrix}\right)
\end{equation}
can be constructed with accuracy $\epsilon$ using
\begin{equation}
    \mathbf{O}\left(\kappa\log\left(\frac{\kappa}{\epsilon}\right)\right)
\end{equation}
queries to $E_{A}$.
\end{corollary}

We also consider rescaling a Hamiltonian block encoding by a fractional number $0<\tau<1$, which has already been used in the derivation of~\prop{add},~\thm{multiply} and~\prop{herm_multiply}. Concretely, we assume that the input matrix $A$ is given by $E_{A}=\exp\left(-i\left[\begin{smallmatrix}
    0 & A^\dagger\\
    A & 0
\end{smallmatrix}\right]\right)$ with $\norm{A}<\frac{\pi}{2}$. Our goal is to produce $E_{\tau A}=\exp\left(-i\left[\begin{smallmatrix}
        0 & \tau A^{\dagger}\\
        \tau A & 0
\end{smallmatrix}\right]\right)$.  This requires a dominated approximation $p(x)\approx\sin(\tau\arcsin(x))$ and $q(x)\approx\frac{\cos(\tau\arcsin(x))}{\sqrt{1-x^2}}$, whose construction resembles the bounded approximation of fractional query functions~\cite[Corollary 72]{Gilyen2018singular} and~\cite[Corollary 25]{TangTian24}, but requires only $1$ ancilla qubit to implement. 

In practice, fractional scaling can often be realized via a direct implementation of $E_{\tau A}$, thereby avoiding a logarithmic overhead. For instance, when the input $A$ is a linear combination of Pauli operators and its Hamiltonian block encoding $E_A$ is constructed by the Lie-Trotter-Suzuki formulas, the scaling factor $\tau$ can be introduced by single-qubit rotation gates. Similarly, if an algorithm makes calls to $E_{\tau f_{\text{sv}}(A)}$, it suffices to absorb $\tau$ into the polynomial approximation of $f$ and solve the dominated approximation problem for each value of $\tau$.

The rescaling result is previewed below and proved in~\append{composite_frac}. 

\begin{proposition}[Dominated approximation for Hamiltonian fractional scaling]
For any $\epsilon>0$, $0<\tau<1$, and constant $0<\xi\leq\frac{\pi}{2}$, there exist a real odd polynomial $p(x)$ and even polynomial $q(x)$ such that
\begin{equation}
\begin{aligned}
    &\abs{p(x)-\sin\left(\tau\arcsin(x)\right)}\leq\epsilon,\qquad&&\forall x\in\left[-\sin\left(\frac{\pi}{2}-\xi\right),\sin\left(\frac{\pi}{2}-\xi\right)\right],\\
    &\abs{q(x)-\frac{\cos\left(\tau\arcsin(x)\right)}{\sqrt{1-x^2}}}\leq\epsilon,\qquad&&\forall x\in\left[-\sin\left(\frac{\pi}{2}-\xi\right),\sin\left(\frac{\pi}{2}-\xi\right)\right],\\
    &p^2(x)+(1-x^2)q^2(x)\leq1+\epsilon,\qquad&&\forall x\in[-1,1].
\end{aligned}
\end{equation}
Moreover, both $p$ and $q$ have the asymptotic degree
\begin{equation}
    \mathbf{O}\left(\log\left(\frac{1}{\epsilon}\right)\right).
\end{equation}
\end{proposition}

\begin{corollary}[Fractional scaling]
\label{cor:frac_scale}
Let $A$ be a matrix encoded by the Hamiltonian block encoding $E_{A}=\exp\left(-i\left[\begin{smallmatrix}
    0 & A^\dagger\\
    A & 0
\end{smallmatrix}\right]\right)$ with $\norm{A}<\frac{\pi}{2}$. For any $0<\tau<1$, the Hamiltonian block encoding
\begin{equation}
    E_{\tau A}=\exp\left(-i\begin{bmatrix}
        0 & \tau A^\dagger\\
        \tau A & 0
    \end{bmatrix}\right)
\end{equation}
can be constructed with accuracy $\epsilon$ using
\begin{equation}
    \mathbf{O}\left(\log\left(\frac{1}{\epsilon}\right)\right)
\end{equation}
queries to $E_A$.
\end{corollary}

%% file: overlap.tex
In this section, we consider Hamiltonian overlap estimation, which extracts classical properties of Hamiltonian block encoded operators. We describe and analyze the overlap estimation algorithm in~\sec{overlap_est}, proving the main result~\thm{overlap}. We then discuss an application in~\sec{overlap_green} for estimating the Green's functions of many-body quantum systems.

\subsection{Overlap estimation}
\label{sec:overlap_est}
Let $A$ be a matrix and $\ket{\psi}$ be a quantum state. Then, the goal of overlap estimation is to estimate the (complex) number $\bra{\psi}A\ket{\psi}$ to an arbitrary precision with a sufficiently large confidence. In the case where $A$ is accessible through a unitary block encoding $O_A$, overlap estimation can be realized by the well-known Hadamard test~\cite[Appendix D]{2021Yupreconditioned}. Here, we describe an overlap estimation algorithm for Hamiltonian block encoded operators with a comparable complexity, at no extra cost of qubits.

At a high level, the Hadamard test uses a controlled version of unitary $O_A$ to implement $\frac{1}{2} (I \pm O_A)$ and $\frac{1}{2} (I \pm i O_A)$, so that $\bra{\psi} A \ket{\psi}$ can be expressed as a linear combination of $\norm{\frac{1}{2} (I + O_A) \ket{\psi}}^2$ and $ \norm{\frac{1}{2} (I + iO_A) \ket{\psi}}^2$. Similarly, one could estimate $\bra{\psi} A \ket{\psi}$ by converting $E_A$ into a controlled unitary block encoding of $A$ via~\cite[Corollary 71]{Gilyen2018singular}, then applying the Hadamard test, at the cost of an additional ancilla qubit. But, $E_A$ is already a Hamiltonian dilation. We exploit this to perform overlap estimation without requiring any additional ancilla qubits.

Our overlap estimation algorithm uses the following sampling based amplitude estimation as a subroutine.

\begin{lemma}[Sampling based amplitude estimation]
Let $\ket{\psi}$ be a quantum state, $U$ be a unitary and $\Pi$ be an orthogonal projection. Then for any $\epsilon>0$ and $0<p_{\text{fail}}\leq1$, there exists a quantum algorithm that outputs $y$ with failure probability
\begin{equation}
    \mathbf{P}\left(\abs{y-\norm{\Pi U\ket{0}}}\geq\epsilon\right)<p_{\text{fail}}
\end{equation}
using
\begin{equation}
    \mathbf{O}\left(\frac{1}{\epsilon^2}\log\left(\frac{1}{p_{\text{fail}}}\right)\right)
\end{equation}
samples of a quantum circuit, each making $1$ query to $\ket{\psi}$, $U$, and $\Pi$.
\end{lemma}
\noindent Note that the above sampling-based amplitude estimation can be replaced by phase estimation~\cite{brassard2002quantum}, which reduces the total query complexity from $\mathbf{O}\left(\frac{1}{\epsilon^2}\log\left(\frac{1}{p_{\text{fail}}}\right)\right)$ to $\mathbf{O}\left(\frac{1}{\epsilon}\log\left(\frac{1}{p_{\text{fail}}}\right)\right)$ but increases the maximum query depth accordingly. The analysis of this modified variant is routine and will hence be omitted.

We now give a high-level description of the Hamiltonian-based overlap estimation algorithm.
Let $A$ be Hamiltonian block encoded as $E_A=\exp\left(-i\left[\begin{smallmatrix}
    0 & A^\dagger\\
    A & 0
\end{smallmatrix}\right]\right)$ with $\norm{A}<1$ and $A=U\Sigma V^\dagger$ be its singular value decomposition. We start by performing the Hamiltonian QSVT
\begin{equation}
    \exp\left(-i\begin{bmatrix}
        0 & A^\dagger\\
        A & 0
    \end{bmatrix}\right)
    \mapsto\exp\left(-i\begin{bmatrix}
        0 & f_{\text{sv}}^\dagger\left(A\right)\\
        f_{\text{sv}}(A) & 0
    \end{bmatrix}\right)
\end{equation}
with the target function $f$ to be specified momentarily. We then apply the transformed Hamiltonian block encoding to $4$ different initial states, perform repeated measurements and collect the measurement statistics.

First we prepare the initial state $\ket{0}\ket{\psi}$, and project onto $\bra{0}\otimes I$ after transforming the Hamiltonian block encoding. The probability of this event is given by
\begin{equation}
\begin{aligned}
    &\norm{\left(\bra{0}\otimes I\right)\exp\left(-i\begin{bmatrix}
        0 & f_{\text{sv}}^\dagger\left(A\right)\\
        f_{\text{sv}}\left(A\right)& 0
    \end{bmatrix}\right)\ket{0}\ket{\psi}}^2\\
    &=\norm{\left(\bra{0}\otimes I\right)\exp\left(-i\begin{bmatrix}
        0 & Vf(\Sigma)U^\dagger\\
        Uf(\Sigma)V^\dagger& 0
    \end{bmatrix}\right)\ket{0}\ket{\psi}}^2\\
    &=\norm{\left(\bra{0}\otimes I\right)\begin{bmatrix}
        V & 0\\
        0 & U
    \end{bmatrix}\begin{bmatrix}
        \cos\left(f(\Sigma)\right) & -i\sin\left(f(\Sigma)\right)\\
        -i\sin\left(f(\Sigma)\right) & \cos\left(f(\Sigma)\right)
    \end{bmatrix}\begin{bmatrix}
        V^\dagger & 0\\
        0 & U^\dagger
    \end{bmatrix}\ket{0}\ket{\psi}}^2\\
    &=\norm{\cos\left(f(\Sigma)\right)V^\dagger\ket{\psi}}^2.
\end{aligned}
\end{equation}
Then we prepare the initial state $\ket{1}\ket{\psi}$ and project onto $\bra{0}\otimes I$, with probability
\begin{equation}
\begin{aligned}
    &\norm{\left(\bra{0}\otimes I\right)\exp\left(-i\begin{bmatrix}
        0 & f_{\text{sv}}^\dagger\left(A\right)\\
        f_{\text{sv}}\left(A\right)& 0
    \end{bmatrix}\right)\ket{1}\ket{\psi}}^2\\
    &=\norm{\left(\bra{0}\otimes I\right)\exp\left(-i\begin{bmatrix}
        0 & Vf(\Sigma)U^\dagger\\
        Uf(\Sigma)V^\dagger& 0
    \end{bmatrix}\right)\ket{1}\ket{\psi}}^2\\
    &=\norm{\left(\bra{0}\otimes I\right)\begin{bmatrix}
        V & 0\\
        0 & U
    \end{bmatrix}\begin{bmatrix}
        \cos\left(f(\Sigma)\right) & -i\sin\left(f(\Sigma)\right)\\
        -i\sin\left(f(\Sigma)\right) & \cos\left(f(\Sigma)\right)
    \end{bmatrix}\begin{bmatrix}
        V^\dagger & 0\\
        0 & U^\dagger
    \end{bmatrix}\ket{1}\ket{\psi}}^2\\
    &=\norm{\sin\left(f(\Sigma)\right)U^\dagger\ket{\psi}}^2.
\end{aligned}
\end{equation}
Next we prepare the initial state $\frac{\ket{0}-\ket{1}}{\sqrt{2}}\ket{\psi}$ and project onto $\bra{0}\otimes I$, with probability
\begin{equation}
\begin{aligned}
    &\norm{\left(\bra{0}\otimes I\right)\exp\left(-i\begin{bmatrix}
        0 & f_{\text{sv}}^\dagger\left(A\right)\\
        f_{\text{sv}}\left(A\right)& 0
    \end{bmatrix}\right)\frac{\ket{0}-\ket{1}}{\sqrt{2}}\ket{\psi}}^2\\
    &=\norm{\left(\bra{0}\otimes I\right)\exp\left(-i\begin{bmatrix}
        0 & Vf(\Sigma)U^\dagger\\
        Uf(\Sigma)V^\dagger& 0
    \end{bmatrix}\right)\frac{\ket{0}-\ket{1}}{\sqrt{2}}\ket{\psi}}^2\\
    &=\norm{\left(\bra{0}\otimes I\right)\begin{bmatrix}
        V & 0\\
        0 & U
    \end{bmatrix}\begin{bmatrix}
        \cos\left(f(\Sigma)\right) & -i\sin\left(f(\Sigma)\right)\\
        -i\sin\left(f(\Sigma)\right) & \cos\left(f(\Sigma)\right)
    \end{bmatrix}\begin{bmatrix}
        V^\dagger & 0\\
        0 & U^\dagger
    \end{bmatrix}\frac{\ket{0}-\ket{1}}{\sqrt{2}}\ket{\psi}}^2\\
    &=\norm{\frac{1}{\sqrt{2}}\left(\cos\left(f(\Sigma)\right)V^\dagger+i\sin\left(f(\Sigma)\right)U^\dagger\right)\ket{\psi}}^2\\
    &=\frac{1}{2}\left(\norm{\cos\left(f\left(\Sigma\right)\right)V^\dagger\ket{\psi}}^2
    +\norm{\sin\left(f\left(\Sigma\right)\right)U^\dagger\ket{\psi}}^2
    +\Im\left(\bra{\psi}U\sin\left(2f\left(\Sigma\right)\right)V^\dagger\ket{\psi}\right)\right).
\end{aligned}
\end{equation}
Finally we prepare the initial state $\frac{\ket{0}+i\ket{1}}{\sqrt{2}}\ket{\psi}$ and project onto $\bra{0}\otimes I$, with probability
\begin{equation}
\begin{aligned}
    &\norm{\left(\bra{0}\otimes I\right)\exp\left(-i\begin{bmatrix}
        0 & f_{\text{sv}}^\dagger\left(A\right)\\
        f_{\text{sv}}\left(A\right)& 0
    \end{bmatrix}\right)\frac{\ket{0}+i\ket{1}}{\sqrt{2}}\ket{\psi}}^2\\
    &=\norm{\left(\bra{0}\otimes I\right)\exp\left(-i\begin{bmatrix}
        0 & Vf(\Sigma)U^\dagger\\
        Uf(\Sigma)V^\dagger& 0
    \end{bmatrix}\right)\frac{\ket{0}+i\ket{1}}{\sqrt{2}}\ket{\psi}}^2\\
    &=\norm{\left(\bra{0}\otimes I\right)\begin{bmatrix}
        V & 0\\
        0 & U
    \end{bmatrix}\begin{bmatrix}
        \cos\left(f(\Sigma)\right) & -i\sin\left(f(\Sigma)\right)\\
        -i\sin\left(f(\Sigma)\right) & \cos\left(f(\Sigma)\right)
    \end{bmatrix}\begin{bmatrix}
        V^\dagger & 0\\
        0 & U^\dagger
    \end{bmatrix}\frac{\ket{0}+i\ket{1}}{\sqrt{2}}\ket{\psi}}^2\\
    &=\norm{\frac{1}{\sqrt{2}}\left(\cos\left(f(\Sigma)\right)V^\dagger+\sin\left(f(\Sigma)\right)U^\dagger\right)\ket{\psi}}^2\\
    &=\frac{1}{2}\left(\norm{\cos\left(f\left(\Sigma\right)\right)V^\dagger\ket{\psi}}^2
    +\norm{\sin\left(f\left(\Sigma\right)\right)U^\dagger\ket{\psi}}^2
    +\Re\left(\bra{\psi}U\sin\left(2f\left(\Sigma\right)\right)V^\dagger\ket{\psi}\right)\right).
\end{aligned}
\end{equation}

Estimates of the above $4$ probabilities can be linearly recombined to yield an estimation of the complex number $\bra{\psi}U\sin\left(2f\left(\Sigma\right)\right)V^\dagger\ket{\psi}$. To get the desired overlap, we thus seek a polynomial $f\approx\frac{\arcsin}{2}$. Equivalently, we need a dominated approximation of $\sin\left(\frac{1}{2}\arcsin\arcsin(x)\right)$ and $\frac{\cos\left(\frac{1}{2}\arcsin\arcsin(x)\right)}{\sqrt{1-x^2}}$. We solve this dominated polynomial approximation problem in~\append{composite_overlap} and preview the result below.

\begin{proposition}[Dominated approximation for Hamiltonian overlap estimation]
\label{prop:dominated_overlap}
For any $\epsilon>0$ and constant $0<\xi\leq1$, there exist a real odd polynomial $p(x)$ and even polynomial $q(x)$ such that
\begin{equation}
\begin{aligned}
    &\abs{p(x)-\sin\left(\frac{1}{2}\arcsin\arcsin(x)\right)}\leq\epsilon,\qquad&&\forall x\in\left[-\sin\left(1-\xi\right),\sin\left(1-\xi\right)\right],\\
    &\abs{q(x)-\frac{\cos\left(\frac{1}{2}\arcsin\arcsin(x)\right)}{\sqrt{1-x^2}}}\leq\epsilon,\qquad&&\forall x\in\left[-\sin\left(1-\xi\right),\sin\left(1-\xi\right)\right],\\
    &p^2(x)+(1-x^2)q^2(x)\leq1+\epsilon,\qquad&&\forall x\in[-1,1].
\end{aligned}
\end{equation}
Moreover, both $p$ and $q$ have the asymptotic degree
\begin{equation}
    \mathbf{O}\left(\log\left(\frac{1}{\epsilon}\right)\right).
\end{equation}
\end{proposition}

\begin{theorem}[Overlap estimation]
\label{thm:overlap}
Let $A$ be a matrix encoded by the Hamiltonian block encoding $E_{A}=\exp\left(-i\left[\begin{smallmatrix}
    0 & A^\dagger\\
    A & 0
\end{smallmatrix}\right]\right)$ with $\norm{A}<1$, and $\ket{\psi}$ be a quantum state. Then, the overlap $\bra{\psi}A\ket{\psi}$ can be estimated with accuracy $\epsilon$ and success probability $1-p_{\text{fail}}$:
\begin{equation}
    \mathbf{P}\left(\abs{y-\bra{\psi}A\ket{\psi}}\geq\epsilon\right)<p_{\text{fail}},
\end{equation}
using
\begin{equation}
    \mathbf{O}\left(\frac{1}{\epsilon^2}\log\left(\frac{1}{p_{\text{fail}}}\right)\right)
\end{equation}
samples of quantum circuits, each making 
\begin{equation}
    \mathbf{O}\left(\log\left(\frac{1}{\epsilon}\right)\right)
\end{equation}
queries to $E_{A}$ and $1$ query to $\ket{\psi}$. See \fig{overlap} for the corresponding circuit diagram.
\end{theorem}
\begin{proof}
Suppose that $\norm{A}\leq1-\xi<1$ for some constant $0<\xi\leq1$, with the singular value decomposition $A=U\Sigma V^\dagger$. We apply Hamiltonian QSVT (\thm{qsvt_odd}) with the dominated polynomial approximation implemented by~\prop{dominated_overlap}. This produces the Hamiltonian block encoding
\begin{equation}
    E_{\frac{\arcsin_{\text{sv}}(A)}{2}}
    =\exp\left(-i\begin{bmatrix}
        0 & \frac{1}{2}\arcsin_{\text{sv}}^\dagger\left(A\right)\\
        \frac{1}{2}\arcsin_{\text{sv}}\left(A\right)& 0
    \end{bmatrix}\right)
\end{equation}
with accuracy $\epsilon$. Consequently, the $4$ probabilities 
\begin{equation}
\begin{aligned}
    &\norm{\cos\left(\frac{\arcsin(\Sigma)}{2}\right)V^\dagger\ket{\psi}}^2,\qquad
    \norm{\sin\left(\frac{\arcsin(\Sigma)}{2}\right)U^\dagger\ket{\psi}}^2,\\
    &\frac{1}{2}\left(\norm{\cos\left(\frac{\arcsin(\Sigma)}{2}\right)V^\dagger\ket{\psi}}^2
    +\norm{\sin\left(\frac{\arcsin(\Sigma)}{2}\right)U^\dagger\ket{\psi}}^2
    +\Im\left(\bra{\psi}U\Sigma V^\dagger\ket{\psi}\right)\right),\\
    &\frac{1}{2}\left(\norm{\cos\left(\frac{\arcsin(\Sigma)}{2}\right)V^\dagger\ket{\psi}}^2
    +\norm{\sin\left(\frac{\arcsin(\Sigma)}{2}\right)U^\dagger\ket{\psi}}^2
    +\Re\left(\bra{\psi}U\Sigma V^\dagger\ket{\psi}\right)\right),
\end{aligned}
\end{equation}
can be produced with accuracy $\mathbf{O}(\epsilon)$, which also leads to an $\mathbf{O}(\epsilon)$ approximation of $\bra{\psi}U\Sigma V^\dagger\ket{\psi}=\bra{\psi}A\ket{\psi}$.
Moreover, each circuit makes
\begin{equation}
    \mathbf{O}\left(\log\left(\frac{1}{\epsilon}\right)\right)
\end{equation}
queries to $E_A$ and $1$ query to $\ket{\psi}$.

We now apply the sampling based amplitude estimation to estimate each amplitude with accuracy $\mathbf{O}(\epsilon)$ and failure probability $p_{\text{fail}}$, which translates to an estimation of $\bra{\psi}A\ket{\psi}$ with accuracy $\mathbf{O}(\epsilon)$ and failure probability $4p_{\text{fail}}$. The claim follows by a final rescaling of $\epsilon$ and $p_{\text{fail}}$.
\end{proof}

\begin{figure}[t]
	\centering
\includegraphics[scale=\circuitwidth]{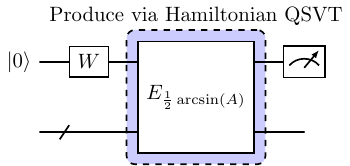}
\caption{Quantum circuit for overlap estimation. We select $W = I, X, (\mathrm{Had}\cdot X), (S\cdot\mathrm{Had})$ and perform norm estimation via measurement under $\left\{\ket{0},\ket{1}\right\}$.}
\label{fig:overlap}
\end{figure}

\subsection{Green's function estimation}
\label{sec:overlap_green}
Green's functions carry valuable spectroscopic information about the underlying quantum system and are widely applied in the study of quantum physics and chemistry~\cite{WangMcArdleBerta24,2021Yupreconditioned}. Specifically, the advanced and retarded Green's functions are defined as
\begin{equation}
\begin{aligned}
    [G]_{j,k,+}(z)&=\bra{\psi_0}A_j\left(z-(H-\lambda_0)\right)^{-1}A_k^\dagger\ket{\psi_0},\\
    [G]_{j,k,-}(z)&=\bra{\psi_0}A_k^\dagger\left(z+(H-\lambda_0)\right)^{-1}A_j\ket{\psi_0},\\
\end{aligned}
\end{equation}
where $A_j$ and $A_k^\dagger$ are the fermionic annihilation and creation operators, $\ket{\psi_0}$ is ground state of the target Hamiltonian $H$ with ground energy $\lambda_0$, and $z=\zeta+i\eta\in\mathbb{C}$ is a complex number such that $\lambda_0\pm z$ are not eigenvalues of $H$. This is satisfied when the broadening parameter $\eta\neq0$. Without loss of generality, we assume $\eta>0$ hereafter.

In calculating the Green's functions, our goal is to realize a complex-valued real-variable function of the form $x\mapsto\frac{1}{i\eta+x}$. We will present an implementation based on Hamiltonian QSVT using only one single ancilla qubit, recovering existing results based on the Fourier series expansion~\cite{WangMcArdleBerta24}.

Specifically, we consider the real and imaginary part of the function respectively:
\begin{equation}
\begin{aligned}
    \Re\left(\frac{1}{i\eta+x}\right)&=\frac{1}{2}\left(\frac{1}{i\eta+x}+\frac{1}{-i\eta+x}\right)
    =\frac{x}{\eta^2+x^2},\\
    \Im\left(\frac{1}{i\eta+x}\right)&=\frac{1}{2i}\left(\frac{1}{i\eta+x}-\frac{1}{-i\eta+x}\right)
    =\frac{-\eta}{\eta^2+x^2}.
\end{aligned}
\end{equation}
We then construct a dominated polynomial approximation for Hamiltonian QSVT, with the result previewed below and established in~\append{composite_green}.

\begin{proposition}[Dominated approximation for Green's function estimation]
\label{prop:dominated_green}
For any $\epsilon>0$, $\eta>0$, and constant $0<\xi\leq\frac{\pi}{2}$, there exist a real odd polynomial $p(x)$ and even polynomial $q(x)$ such that
\begin{equation}
\begin{aligned}
    &\abs{p(x)-\sin\left(\frac{1}{2}\arcsin\left(\frac{\eta\arcsin(x)}{\eta^2+\arcsin^2(x)}\right)\right)}\leq\epsilon,\quad&&\forall x\in\left[-\sin\left(\frac{\pi}{2}-\xi\right),\sin\left(\frac{\pi}{2}-\xi\right)\right],\\
    &\abs{q(x)-\frac{\cos\left(\frac{1}{2}\arcsin\left(\frac{\eta\arcsin(x)}{\eta^2+\arcsin^2(x)}\right)\right)}{\sqrt{1-x^2}}}\leq\epsilon,\quad&&\forall x\in\left[-\sin\left(\frac{\pi}{2}-\xi\right),\sin\left(\frac{\pi}{2}-\xi\right)\right],\\
    &p^2(x)+(1-x^2)q^2(x)\leq1+\epsilon,\quad&&\forall x\in[-1,1].
\end{aligned}
\end{equation}
Additionally, there exists a real even polynomial $h(x)$ such that
\begin{equation}
\begin{aligned}
    &\abs{h(x)-\frac{\eta}{\sqrt{\eta^2+\arcsin^2(x)}}}\leq\epsilon,\qquad&&\forall x\in\left[-\sin\left(\frac{\pi}{2}-\xi\right),\sin\left(\frac{\pi}{2}-\xi\right)\right],\\
    &\abs{h(x)}\leq1+\epsilon,\qquad&&\forall x\in[-1,1].
\end{aligned}
\end{equation}
Moreover, $p$, $q$, and $h$ all have the asymptotic degree
\begin{equation}
    \mathbf{O}\left(\frac{1}{\eta}\log\left(\frac{1}{\eta\epsilon}\right)\right).
\end{equation}
\end{proposition}

We now start to describe the quantum algorithm, focusing on the retarded Green's function without loss of generality. To this end, we first decompose the operator into its Hermitian and anti-Hermitian component
\begin{equation}
\begin{aligned}
    &\bra{\psi_0}A_k^\dagger\left(z+(H-\lambda_0)\right)^{-1}A_j\ket{\psi_0}\\
    &=\bra{\psi_0}A_k^\dagger\Re\left(\left(z+(H-\lambda_0)\right)^{-1}\right)A_j\ket{\psi_0}
    +i\bra{\psi_0}A_k^\dagger\Im\left(\left(z+(H-\lambda_0)\right)^{-1}\right)A_j\ket{\psi_0}.
\end{aligned}
\end{equation}
Moreover, each component can be further expanded as a combination of diagonal terms as
\begin{equation}
\begin{aligned}
    \bra{\psi_0}A_k^\dagger\Re\left(\left(z+(H-\lambda_0)\right)^{-1}\right)A_j\ket{\psi_0}
    &=\frac{1}{4}\bra{\psi_0}\left(A_j^\dagger+A_k^\dagger\right)\Re\left(\left(z+(H-\lambda_0)\right)^{-1}\right)\left(A_j+A_k\right)\ket{\psi_0}\\
    &\quad-\frac{1}{4}\bra{\psi_0}\left(A_j^\dagger-A_k^\dagger\right)\Re\left(\left(z+(H-\lambda_0)\right)^{-1}\right)\left(A_j-A_k\right)\ket{\psi_0}\\
    &\quad+i\frac{1}{4}\bra{\psi_0}\left(A_j^\dagger-iA_k^\dagger\right)\Re\left(\left(z+(H-\lambda_0)\right)^{-1}\right)\left(A_j+iA_k\right)\ket{\psi_0}\\
    &\quad-i\frac{1}{4}\bra{\psi_0}\left(A_j^\dagger+iA_k^\dagger\right)\Re\left(\left(z+(H-\lambda_0)\right)^{-1}\right)\left(A_j-iA_k\right)\ket{\psi_0},\\
\end{aligned}
\end{equation}
and
\begin{equation}
\begin{aligned}
    \bra{\psi_0}A_k^\dagger\Im\left(\left(z+(H-\lambda_0)\right)^{-1}\right)A_j\ket{\psi_0}
    &=\frac{1}{4}\bra{\psi_0}\left(A_j^\dagger+A_k^\dagger\right)\Im\left(\left(z+(H-\lambda_0)\right)^{-1}\right)\left(A_j+A_k\right)\ket{\psi_0}\\
    &\quad-\frac{1}{4}\bra{\psi_0}\left(A_j^\dagger-A_k^\dagger\right)\Im\left(\left(z+(H-\lambda_0)\right)^{-1}\right)\left(A_j-A_k\right)\ket{\psi_0}\\
    &\quad+i\frac{1}{4}\bra{\psi_0}\left(A_j^\dagger-iA_k^\dagger\right)\Im\left(\left(z+(H-\lambda_0)\right)^{-1}\right)\left(A_j+iA_k\right)\ket{\psi_0}\\
    &\quad-i\frac{1}{4}\bra{\psi_0}\left(A_j^\dagger+iA_k^\dagger\right)\Im\left(\left(z+(H-\lambda_0)\right)^{-1}\right)\left(A_j-iA_k\right)\ket{\psi_0}.\\
\end{aligned}
\end{equation}
Finally, the Hermitian and anti-Hermitian component can be re-expressed as
\begin{equation}
\begin{aligned}
    \Re\left(\left(z+(H-\lambda_0)\right)^{-1}\right)
    &=\frac{1}{2}\left(\left(i\eta+(H-\lambda_0+\zeta)\right)^{-1}+\left(-i\eta+(H-\lambda_0+\zeta)\right)^{-1}\right)\\
    &=\frac{1}{2}\left(i\eta+(H-\lambda_0+\zeta)\right)^{-1}
    2(H-\lambda_0+\zeta)
    \left(-i\eta+(H-\lambda_0+\zeta)\right)^{-1}\\
    &=\frac{H-\lambda_0+\zeta}{\eta^2+(H-\lambda_0+\zeta)^2},
\end{aligned}
\end{equation}
and
\begin{equation}
\begin{aligned}
    \Im\left(\left(z+(H-\lambda_0)\right)^{-1}\right)
    &=\frac{1}{2i}\left(\left(i\eta+(H-\lambda_0+\zeta)\right)^{-1}-\left(-i\eta+(H-\lambda_0+\zeta)\right)^{-1}\right)\\
    &=\frac{1}{2i}\left(i\eta+(H-\lambda_0+\zeta)\right)^{-1}
    (-2i\eta)
    \left(-i\eta+(H-\lambda_0+\zeta)\right)^{-1}\\
    &=\frac{-\eta}{\eta^2+(H-\lambda_0+\zeta)^2},
\end{aligned}
\end{equation}
where we have abbreviated $A^{-1}BC^{-1}$ as $\frac{B}{AC}$ when Hermitian operators $A,B,C$ pairwise commute.

The above analysis applies to a general fermionic Hamiltonian in second quantization. In the case of quantum chemistry simulation, the target Hamiltonian $H$ often takes the form
\begin{equation}
    H=\sum_{p,q}h_{pq}A_p^\dagger A_q+\sum_{p,q,r,s}h_{pqrs}A_p^\dagger A_q A_r^\dagger A_s,
\end{equation}
where the coefficients $h_{pq}$ and $h_{pqrs}$ are real due to the use of real basis functions. This means $H$ is a real symmetric matrix under the Jordan-Wigner representation, and so are its Hermitian and anti-Hermitian components $\frac{H-\lambda_0+\zeta}{\eta^2+(H-\lambda_0+\zeta)^2}$ and $\frac{-\eta}{\eta^2+(H-\lambda_0+\zeta)^2}$. Consequently, the diagonal expansions introduced above can be simplified to
\begin{equation}
\begin{aligned}
    \bra{\psi_0}A_k^\dagger\Re\left(\left(z+(H-\lambda_0)\right)^{-1}\right)A_j\ket{\psi_0}
    &=\frac{1}{2}\bra{\psi_0}\left(A_j^\dagger+A_k^\dagger\right)\Re\left(\left(z+(H-\lambda_0)\right)^{-1}\right)\left(A_j+A_k\right)\ket{\psi_0}\\
    &\quad-\bra{\psi_0}A_j^\dagger\Re\left(\left(z+(H-\lambda_0)\right)^{-1}\right)A_j\ket{\psi_0}\\
    &\quad-\bra{\psi_0}A_k^\dagger\Re\left(\left(z+(H-\lambda_0)\right)^{-1}\right)A_k\ket{\psi_0},
\end{aligned}
\end{equation}
and
\begin{equation}
\begin{aligned}
    \bra{\psi_0}A_k^\dagger\Im\left(\left(z+(H-\lambda_0)\right)^{-1}\right)A_j\ket{\psi_0}
    &=\frac{1}{2}\bra{\psi_0}\left(A_j^\dagger+A_k^\dagger\right)\Im\left(\left(z+(H-\lambda_0)\right)^{-1}\right)\left(A_j+A_k\right)\ket{\psi_0}\\
    &\quad-\bra{\psi_0}A_j^\dagger\Im\left(\left(z+(H-\lambda_0)\right)^{-1}\right)A_j\ket{\psi_0}\\
    &\quad-\bra{\psi_0}A_k^\dagger\Im\left(\left(z+(H-\lambda_0)\right)^{-1}\right)A_k\ket{\psi_0}.
\end{aligned}
\end{equation}

Let us first consider the Hermitian component of the Green's function. Without loss of generality, we explain how to estimate $\bra{\psi_0}\left(A_j^\dagger+A_k^\dagger\right)
    \frac{H-\lambda_0+\zeta}{\eta^2+(H-\lambda_0+\zeta)^2}
    \left(A_j+A_k\right)\ket{\psi_0}$.
To this end, we rewrite it as
\begin{equation}
\label{eq:green_herm_overlap0}
    \bra{\psi_0}\left(A_j^\dagger+A_k^\dagger\right)
    \frac{H-\lambda_0+\zeta}{\eta^2+(H-\lambda_0+\zeta)^2}
    \left(A_j+A_k\right)\ket{\psi_0}
    =\frac{4}{\widetilde{\eta}}
    \bra{\psi_0}\frac{A_j^\dagger+A_k^\dagger}{2}
    \frac{\widetilde{\eta}\widetilde{H}}{\widetilde{\eta}^2+\widetilde{H}^2}
    \frac{A_j+A_k}{2}\ket{\psi_0},
\end{equation}
where
\begin{equation}
    \widetilde{\eta}=\frac{\eta}{\alpha_H+\abs{\lambda_0}+\abs{\zeta}},\qquad
    \widetilde{H}=\frac{H-\lambda_0+\zeta}{\alpha_H+\abs{\lambda_0}+\abs{\zeta}},
\end{equation}
and $\alpha_H\geq\norm{H}$ is a known upper bound on the spectral norm of the Hamiltonian.
At a high level, we start by constructing a Hamiltonian block encoding of
\begin{equation}
    \exp\left(-i\frac{1}{2}\begin{bmatrix}
        0 & A_j^\dagger+A_k^\dagger\\
        A_j+A_k & 0
    \end{bmatrix}\right).
\end{equation}
Through inverse block encoding, we obtain a unitary block encoding close to
\begin{equation}
    \begin{bmatrix}
        \frac{1}{2}\left(A_j+A_k\right) & \cdot\\
        \cdot & \cdot
    \end{bmatrix}
\end{equation}
If we apply it to the initial state $\ket{0}\ket{\psi_0}$, the probability of obtaining $\bra{0}\otimes I$ is approximately
\begin{equation}
\label{eq:green_herm_overlap1}
    \norm{\frac{A_j+A_k}{2}\ket{\psi_0}}^2.
\end{equation}
Assuming this is true, the post-measurement state becomes $\frac{(A_j+A_k)\ket{\psi_0}}{\norm{(A_j+A_k)\ket{\psi_0}}}$.
Next, we construct a Hamiltonian block encoding close to
\begin{equation}
    \exp\left(-i\begin{bmatrix}
        0 & f\left(\frac{\widetilde{\eta}\widetilde{H}}{\widetilde{\eta}^2+\widetilde{H}^2}\right)\\
        f\left(\frac{\widetilde{\eta}\widetilde{H}}{\widetilde{\eta}^2+\widetilde{H}^2}\right) & 0
    \end{bmatrix}\right),
\end{equation}
for $f(x)\approx\frac{1}{2}\arcsin(x)$. Reusing the single ancilla qubit from the previous step, we apply the Hamiltonian block encoding to $\ket{\beta}\frac{(A_j+A_k)\ket{\psi_0}}{\norm{(A_j+A_k)\ket{\psi_0}}}$ with $\ket{\beta}=\ket{0},\ket{1},\frac{\ket{0}-\ket{1}}{\sqrt{2}},\frac{\ket{0}+i\ket{1}}{\sqrt{2}}$.
If we measure the ancilla qubit, the probability of projecting onto $\bra{0}\otimes I$ can be linearly combined to approximately yield
\begin{equation}
\label{eq:green_herm_overlap2}
    \frac{\bra{\psi_0}\left(A_j^\dagger+A_k^\dagger\right)}{\norm{\bra{\psi_0}\left(A_j^\dagger+A_k^\dagger\right)}}
    \frac{\widetilde{\eta}\widetilde{H}}{\widetilde{\eta}^2+\widetilde{H}^2}
    \frac{(A_j+A_k)\ket{\psi_0}}{\norm{(A_j+A_k)\ket{\psi_0}}}.
\end{equation}
Hence, with sufficiently many measurement statistics, we obtain an estimate of
\begin{equation}
\label{eq:green_herm_overlap3}
    \norm{\frac{A_j+A_k}{2}\ket{\psi_0}}^2\cdot
    \frac{\bra{\psi_0}\left(A_j^\dagger+A_k^\dagger\right)}{\norm{\bra{\psi_0}\left(A_j^\dagger+A_k^\dagger\right)}}
    \frac{\widetilde{\eta}\widetilde{H}}{\widetilde{\eta}^2+\widetilde{H}^2}
    \frac{(A_j+A_k)\ket{\psi_0}}{\norm{(A_j+A_k)\ket{\psi_0}}},
\end{equation}
which is our desired quantity up to a rescaling factor of $\frac{4}{\widetilde{\eta}}=\frac{4\left(\alpha_H+\abs{\lambda_0}+\abs{\zeta}\right)}{\eta}$.

Similarly, we also consider the anti-Hermitian component of the Green's function. Without loss of generality, we explain how to estimate $\bra{\psi_0}\left(A_j^\dagger+A_k^\dagger\right)
    \frac{-\eta}{\eta^2+(H-\lambda_0+\zeta)^2}
    \left(A_j+A_k\right)\ket{\psi_0}$.
To this end, we rewrite it as
\begin{equation}
    \bra{\psi_0}\left(A_j^\dagger+A_k^\dagger\right)
    \frac{-\eta}{\eta^2+(H-\lambda_0+\zeta)^2}
    \left(A_j+A_k\right)\ket{\psi_0}
    =\frac{-4}{\eta}\bra{\psi_0}\frac{A_j^\dagger+A_k^\dagger}{2}
    \frac{\widetilde{\eta}^2}{\widetilde{\eta}^2+\widetilde{H}^2}
    \frac{A_j+A_k}{2}\ket{\psi_0}.
\end{equation}
Proceeding as above, we prepare the state $\frac{(A_j+A_k)\ket{\psi_0}}{\norm{(A_j+A_k)\ket{\psi_0}}}$ using Hamiltonian block encoding with probability $\norm{\frac{A_j+A_k}{2}\ket{\psi_0}}^2$.
Next, reusing the single ancilla qubit, we construct an operator close to
\begin{equation}
\begin{bmatrix}
    \frac{\widetilde{\eta}}{\sqrt{\widetilde{\eta}^2+\widetilde{H}^2}} & \cdot\\
    \cdot & \cdot
\end{bmatrix}.
\end{equation}
Here, the entry of interest is an even function in $\widetilde{H}$, which can be implemented using~\cite{Lloye21hamiltonianqsvt} that outputs a unitary block encoding. Due to the positive semidefiniteness of the target operator, the probability of projecting onto $\bra{0}\otimes I$ with initial state $\ket{0}\frac{(A_j+A_k)\ket{\psi_0}}{\norm{(A_j+A_k)\ket{\psi_0}}}$ would yield
\begin{equation}
    \frac{\bra{\psi_0}\left(A_j^\dagger+A_k^\dagger\right)}{\norm{\bra{\psi_0}\left(A_j^\dagger+A_k^\dagger\right)}}
    \frac{\widetilde{\eta}^2}{\widetilde{\eta}^2+\widetilde{H}^2}
    \frac{(A_j+A_k)\ket{\psi_0}}{\norm{(A_j+A_k)\ket{\psi_0}}}.
\end{equation}
Hence, with sufficiently many measurement statistics, we obtain an estimate of
\begin{equation}
    \norm{\frac{A_j+A_k}{2}\ket{\psi_0}}^2\cdot
    \frac{\bra{\psi_0}\left(A_j^\dagger+A_k^\dagger\right)}{\norm{\bra{\psi_0}\left(A_j^\dagger+A_k^\dagger\right)}}
    \frac{\widetilde{\eta}^2}{\widetilde{\eta}^2+\widetilde{H}^2}
    \frac{(A_j+A_k)\ket{\psi_0}}{\norm{(A_j+A_k)\ket{\psi_0}}},
\end{equation}
which is our desired quantity up to a rescaling factor of $-\frac{4}{\eta}$. The remaining analysis proceeds as in the Hermitian case.

\begin{proposition}[Green's function estimation]
\label{prop:green}
Let $H$ be a Hamiltonian with ground state $\ket{\psi_0}$, ground energy $\lambda_0$ and norm upper bound $\alpha_H\geq\norm{H}$. For any $z=\xi+i\eta\in\mathbb{C}$ and $\epsilon>0$, the Green's functions 
\begin{equation}
\begin{aligned}
    [G]_{j,k,+}(z)&=\bra{\psi_0}A_j\left(z-(H-\lambda_0)\right)^{-1}A_k^\dagger\ket{\psi_0},\\
    [G]_{j,k,-}(z)&=\bra{\psi_0}A_k^\dagger\left(z+(H-\lambda_0)\right)^{-1}A_j\ket{\psi_0},\\
\end{aligned}
\end{equation}
can be estimated with accuracy $\epsilon$ and success probability $1-p_{\text{fail}}$ using
\begin{equation}
    \mathbf{O}\left(\frac{(\alpha_H+\abs{\lambda_0}+\abs{\xi})^2}{\eta^2\epsilon^2}\log\left(\frac{1}{p_{\text{fail}}}\right)\right)
\end{equation}
samples of quantum circuits, each making
\begin{equation}
    \mathbf{O}\left(\frac{\alpha_H+\abs{\lambda_0}+\abs{\xi}}{\abs{\eta}}\log\left(\frac{\alpha_H+\abs{\lambda_0}+\abs{\xi}}{\abs{\eta}\epsilon}\right)\right)
\end{equation}
queries to the Hamiltonian block encoding of $\frac{H-\lambda_0+\xi}{\alpha_H+\abs{\lambda_0}+\abs{\xi}}$ and
\begin{equation}
    \mathbf{O}\left(\log\left(\frac{\alpha_H+\abs{\lambda_0}+\abs{\xi}}{\abs{\eta}\epsilon}\right)\right)
\end{equation}
queries to the preparation of $\ket{\psi_0}$. The circuits use one ancilla qubit beyond the system qubits supporting $H$.
\end{proposition}
\begin{proof}
We first describe a method to construct the Hamiltonian block encoding of a linear combination of fermionic operators as required by the Green's function estimation algorithm.
Specifically, for a system with $n$ fermionic modes and any coefficients $c_j,c_k\in\mathbb C$, our goal is to obtain
\begin{equation}
    E_{c_jA_j+c_kA_k}=\exp\left(-i\begin{bmatrix}
        0 & c_j^*A_j^\dagger+c_k^*A_k^\dagger\\
        c_jA_j+c_kA_k & 0\\
    \end{bmatrix}\right),
\end{equation}
where $A^\dagger$ and $A$ are fermionic creation and annihilation operators and $1\leq j<k\leq n$.

To this end, we consider the Jordan-Wigner representation of fermionic operators which reads
\begin{equation}
    A_j^\dagger=\frac{1}{2}\left(X_j-iY_j\right)\otimes Z_{j-1}\otimes\ldots\otimes Z_1,\qquad
    A_j=\frac{1}{2}\left(X_j+iY_j\right)\otimes Z_{j-1}\otimes\ldots\otimes Z_1.
\end{equation}
This allows us to represent $c_jA_j+c_kA_k$ as a (complex) linear combination of Pauli strings, which can be further decomposed using the Lie-Trotter-Suzuki formulas. However, this approach would introduce Trotter error and increase the complexity. Instead, we now describe an error-free circuit implementation.
Note that 
\begin{equation}
\begin{aligned}
    c_jA_j+c_kA_k&=\frac{c_k}{2}\left(X_k+iY_k\right)\otimes Z_{k-1}\otimes\ldots\otimes Z_1+\frac{c_j}{2}\left(X_j+iY_j\right)\otimes Z_{j-1}\otimes\ldots\otimes Z_1\\
    &=\left[\frac{c_k}{2}\left(X_k+iY_k\right)\otimes Z_{k-1}\otimes\ldots\otimes Z_j+\frac{c_j}{2}\left(X_j+iY_j\right)\right]\otimes Z_{j-1}\otimes\ldots\otimes Z_1.
\end{aligned}
\end{equation}
Hence, after multiplying the unitary $Z_{j-1}\otimes\ldots\otimes Z_1$, it suffices to consider the Hamiltonian block encoding of $\frac{c_k}{2}\left(X_k+iY_k\right)\otimes Z_{k-1}\otimes\ldots\otimes Z_j+\frac{c_j}{2}\left(X_j+iY_j\right)$.

We now remove the Pauli-$Z$ string from the first term using the circuit identity
\begin{equation}
    \mathrm{CNOT}_{1,2}\left(Z_1\otimes Z_2\right)\mathrm{CNOT}_{1,2}
    =I_1\otimes Z_2,
\end{equation}
yielding
\begin{equation}
\begin{aligned}
    &\frac{c_k}{2}\left(X_k+iY_k\right)\otimes I\otimes\ldots\otimes I \otimes Z_{j+1}\otimes Z_j+\frac{c_j}{2}\left(X_j+iY_j\right)\\
    &=\mathrm{CNOT}_{j+2,j+1}\cdots\mathrm{CNOT}_{k-1,k-2}\\
    &\quad\cdot\left(\frac{c_k}{2}\left(X_k+iY_k\right)\otimes Z_{k-1}\otimes\ldots\otimes Z_{j+1}\otimes Z_j+\frac{c_j}{2}\left(X_j+iY_j\right)\right)\\
    &\quad\cdot\mathrm{CNOT}_{k-1,k-2}\cdots\mathrm{CNOT}_{j+2,j+1},
\end{aligned}
\end{equation}
where $\mathrm{CNOT}_{j,k}$ denotes a CNOT gate with qubit $j$ as the control and $k$ as the target. This reduces the original Hamiltonian block encoding problem to the construction of $4$-qubit unitary
\begin{small}
\begin{equation}
\newmaketag
    \exp\left(-i\frac{1}{2}\begin{bmatrix}
        0 & c_k^*\left(X_k-iY_k\right)\otimes Z_{j+1}\otimes Z_j+c_j^*\left(X_j-iY_j\right)\\
        c_k\left(X_k+iY_k\right)\otimes Z_{j+1}\otimes Z_j+c_j\left(X_j+iY_j\right) & 0
    \end{bmatrix}\right),
\end{equation}
\end{small}%
which can be realized with gate complexity $\mathbf{O}(1)$ through a direct gate synthesis.
Altogether, the above approach has a gate complexity $\mathbf{O}(k)$ to construct a Hamiltonian block encoding of linear combination of fermionic operators $c_jA_j+c_kA_k$ with arbitrary coefficients.

Let us now analyze the query complexity of the algorithm. Recall that we need to estimate the overlaps of both \eq{green_herm_overlap1} and \eq{green_herm_overlap2}. Suppose that we perform amplitude estimation with a sufficiently large confidence and accuracy $\epsilon_{\text{amp},1}$, $\epsilon_{\text{amp},2}$ respectively. Then the error in estimating the (normalized) overlap \eq{green_herm_overlap3} is at most $\mathbf{O}\left(\epsilon_{\text{amp},1}+\norm{\frac{A_j+A_k}{2}\ket{\psi_0}}^2\epsilon_{\text{amp},2}\right)$. To ensure that the error in the Hermitian component of the Green's function \eq{green_herm_overlap0} is at most $\epsilon$ after rescaling, we set $\epsilon_{\text{amp},1}=\mathbf{\Theta}\left(\frac{\abs{\eta}\epsilon}{\alpha_H+\abs{\lambda_0}+\abs{\xi}}\right)$. Now assuming $\epsilon$ is smaller than  \eq{green_herm_overlap0}, $\epsilon_{\text{amp},1}$ is asymptotically smaller than \eq{green_herm_overlap3}, which is in turn upper bounded by the probability $\norm{\frac{A_j+A_k}{2}\ket{\psi_0}}^2$ in \eq{green_herm_overlap1}. We thus obtain a constant multiplicative approximation of $\norm{\frac{A_j+A_k}{2}\ket{\psi_0}}^2$ as a byproduct. Using this constant multiplicative estimate, we choose $\epsilon_{\text{amp},2}=\mathbf{\Theta}\left(\frac{\abs{\eta}\epsilon}{\norm{\frac{A_j+A_k}{2}\ket{\psi_0}}(\alpha_H+\abs{\lambda_0}+\abs{\xi})}\right)$. The total sample complexity is then asymptotically 
\begin{align}
    \mathbf{O}\left(\frac{1}{\epsilon_{\text{amp},2}^2}\frac{1}{\norm{\frac{A_j+A_k}{2}\ket{\psi_0}}^2}\right)=\mathbf{O}\left(\frac{(\alpha_H+\abs{\lambda_0}+\abs{\xi})^2}{\eta^2\epsilon^2}\right).
\end{align}
This analysis covers the error introduced by the measurement statistics. 

Following a similar reasoning, we set the accuracy of block encodings to be $\mathbf{\Theta}\left(\frac{\abs{\eta}\epsilon}{\alpha_H+\abs{\lambda_0}+\abs{\xi}}\right)$. Then the inverse block encoding of $\frac{1}{2}(A_j+A_k)$ makes $\mathbf{O}\left(\log\left(\frac{\alpha_H+\abs{\lambda_0}+\abs{\xi}}{\abs{\eta}\epsilon}\right)\right)$ queries to the preparation of initial state $\ket{\psi_0}$. 
Hamiltonian QSVT makes $\mathbf{O}\left(\frac{\alpha_H+\abs{\lambda_0}+\abs{\xi}}{\abs{\eta}}\log\left(\frac{\alpha_H+\abs{\lambda_0}+\abs{\xi}}{\abs{\eta}\epsilon}\right)\right)$ queries to the Hamiltonian block encoding of $\frac{H-\lambda_0+\xi}{\alpha_H+\abs{\lambda_0}+\abs{\xi}}$, as per~\prop{dominated_green}. This completes the analysis of the Hermitian case.
The anti-Hermitian case can be handled in a similar manner.
\end{proof}

For presentational purpose, we have described the Green's function estimation algorithm under the assumption that both the ground state $\ket{\psi_0}$ and ground energy $\lambda_0$ are available a priori. However, it is fairly straightforward to handle the general case with an approximately prepared ground state and an estimate of ground energy. We refer the reader to~\cite{2021Yupreconditioned} for details.

%% file: sos.tex
We now describe applications of Hamiltonian-based matrix arithmetics to quantum simulation. Specifically, we describe a method to simulate a class of sum-of-squares Hamiltonians in~\sec{sos_sos}. When the summands are Hermitian, we present a Hamiltonian squaring subroutine in~\sec{sos_square} that significantly simplifies the simulation algorithm while reducing the ancilla qubit usage. We discuss the simulation of sum-of-squares electronic structure Hamiltonians in~\sec{sos_electron}.

\subsection{Generic sum-of-squares Hamiltonian simulation}
\label{sec:sos_sos}
A sum-of-squares Hamiltonian has the form
\begin{equation}
    H=\sum_{k=1}^{n_K}H_k=\sum_{k=1}^{n_K}\left(\sum_{j_1=1}^{n_J}A_{j_1k}\right)^\dagger\left(\sum_{j_2=1}^{n_J}A_{j_2k}\right).
\end{equation}
For simplicity, we assume that each $A_{jk}$ is an elementary term whose Hamiltonian block encoding is available as input. In practice, $A_{jk}$ are monomials of spin, fermionic, and bosonic operators, and their Hamiltonian block encodings can be constructed through further matrix arithmetics.

As discussed in~\sec{intro}, our algorithm proceeds as
\begin{equation}
\begin{aligned}
    &\exp\left(-i\sqrt{\tau}\begin{bmatrix}
        0 & A_{jk}^\dagger\\
        A_{jk} & 0
    \end{bmatrix}\right)\\
    \overset{\text{Summation}}&{\longmapsto}\exp\left(-i\sqrt{\tau}\begin{bmatrix}
        0 & \sum_{j=1}^{n_J}A_{jk}^\dagger\\
        \sum_{j=1}^{n_J}A_{jk} & 0
    \end{bmatrix}\right)\\
    \overset{\text{Multiplication}}&{\underset{\dagger}{\longmapsto}}\exp\left(-i\tau\begin{bmatrix}
        0 & \left(\sum_{j_1=1}^{n_J}A_{j_1k}\right)^\dagger\left(\sum_{j_2=1}^{n_J}A_{j_2k}\right)\\
        \left(\sum_{j_1=1}^{n_J}A_{j_1k}\right)^\dagger\left(\sum_{j_2=1}^{n_J}A_{j_2k}\right) & 0
    \end{bmatrix}\right)\\
    \overset{\text{Summation}}&{\longmapsto}\exp\left(-i\tau\begin{bmatrix}
        0 & H\\
        H & 0
    \end{bmatrix}\right)
    \overset{\text{Redefinition}}{\longmapsto}\exp\left(-i\tau\begin{bmatrix}
        H & 0\\
        0 & -H
    \end{bmatrix}\right).
\end{aligned}
\end{equation}
Here, we start with the Hamiltonian block encoding of each $A_{jk}$, and sum over $j$ using~\prop{add} to get a Hamiltonian block encoding of $\sum_j A_{jk}$. After taking the Hermitian conjugation using~\prop{herm_conjugate}, we invoke~\thm{multiply} to get a Hamiltonian block encoding of $H_k=\left(\sum_{j_1}A_{j_1k}\right)^\dagger\left(\sum_{j_2}A_{j_2k}\right)$. Finally, we apply~\prop{add} again to get a Hamiltonian evolution under $H$. Our goal is to choose a sufficiently small time step, such that the resulting simulation achieves the desired accuracy. This is handled by the following proposition.

\begin{corollary}[Generic sum-of-squares Hamiltonian simulation]
\label{cor:sos}
Given a sum-of-squares Hamiltonian $H=\sum_{k=1}^{n_K}H_k=\sum_{k=1}^{n_K}\left(\sum_{j_1=1}^{n_J}A_{j_1k}\right)^\dagger\left(\sum_{j_2=1}^{n_J}A_{j_2k}\right)$ with $\norm{A_{j,k}}<\frac{\pi}{2}$, accuracy $\epsilon>0$ and evolution time $t>0$, the controlled version of Hamiltonian block encoding
$E_{tH}=\exp\left(-it\left[\begin{smallmatrix}
    0 & H\\
    H & 0
\end{smallmatrix}\right]\right)$
can be constructed with accuracy $\epsilon$ using
\begin{equation}
    \left(\alpha_{\text{comm}}+\alpha_{\infty}\right)t\left(\frac{n_K\alpha_\infty t}{\epsilon}\right)^{o(1)}n_Kn_J
\end{equation}
queries to each controlled $E_{A_{jk}}=\exp\left(-i\left[\begin{smallmatrix}
    0 & A_{jk}^\dagger\\
    A_{jk} & 0
\end{smallmatrix}\right]\right)$, where
\begin{equation}
    \alpha_{\text{comm}}=\sup_{p\in\mathbb{Z}_{\geq1}}\left(\sum_{k_1,\ldots,k_{p+1}=1}^{n_K}\norm{\left[H_{k_{p+1}},\ldots,\left[H_{k_2},H_{k_1}\right]\right]}\right)^{\frac{1}{p+1}},\qquad
    \alpha_\infty=\max_{k=1,\ldots,n_K}\left(\sum_{j=1}^{n_J}\norm{A_{jk}}\right)^2.
\end{equation}
Moreover, the simulation uses $2$ ancilla qubits beyond those supporting $H$.
\end{corollary}
\begin{proof}
There are three sources of error in the simulation: the error from Hamiltonian block encoding the summation $\sum_j A_{jk}$, the error from Hamiltonian-based matrix multiplication $H_k=\left(\sum_{j_1}A_{j_1k}\right)^\dagger\left(\sum_{j_2}A_{j_2k}\right)$, and finally the error from Hamiltonian block encoding the summation $\sum_k\left(\sum_{j_1}A_{j_1k}\right)^\dagger\left(\sum_{j_2}A_{j_2k}\right)$. Dividing the simulation into $r$ steps, we will show how to choose $r$ to achieve an overall accuracy of $\epsilon$.

To this end, we use a $p$th-order formula $S_{p,k}\left(\sqrt{\frac{t}{r}}\right)$ to implement the Hamiltonian block encoding of $\sum_j A_{jk}$, obtaining
\begin{equation}
    \norm{S_{p,k}\left(\sqrt{\frac{t}{r}}\right)
    -\exp\left(-i\sqrt{\frac{t}{r}}\begin{bmatrix}
        0 & \sum_{j=1}^{n_J}A_{jk}^\dagger\\
        \sum_{j=1}^{n_J}A_{jk} & 0
    \end{bmatrix}\right)}
    =\mathbf{O}\left(\sqrt{\frac{\alpha_\infty t}{r}}^{p+1}\right).
\end{equation}
Assuming this is implemented perfectly, we then use a $p$th-order formula $M_{p,k}\left(\sqrt{\frac{t}{r}}\right)$ to realize the Hamiltonian-based matrix multiplication, giving
\begin{equation}
\begin{aligned}
    &\norm{M_{p,k}\left(\sqrt{\frac{t}{r}}\right)-
    \exp\left(-i\frac{t}{r}\begin{bmatrix}
        0 & \left(\sum_{j_1=1}^{n_J}A_{j_1k}\right)^\dagger\left(\sum_{j_2=1}^{n_J}A_{j_2k}\right)\\
        \left(\sum_{j_1=1}^{n_J}A_{j_1k}\right)^\dagger\left(\sum_{j_2=1}^{n_J}A_{j_2k}\right) & 0
    \end{bmatrix}\right)}\\
    &=\norm{M_{p,k}\left(\sqrt{\frac{t}{r}}\right)-
    \exp\left(-i\frac{t}{r}\begin{bmatrix}
        0 & H_k\\
        H_k & 0
    \end{bmatrix}\right)}
    =\mathbf{O}\left(\sqrt{\frac{\alpha_\infty t}{r}}^{p+1}\right).
\end{aligned}
\end{equation}
Assuming again that this is done perfectly, we finally use a $p$th-order formula $S_p\left(\frac{t}{r}\right)$ to construct the Hamiltonian block encoding of $\sum_kH_k$:
\begin{equation}
\begin{aligned}
    &\norm{S_p\left(\frac{t}{r}\right)
    -\exp\left(-i\frac{t}{r}\begin{bmatrix}
        0 & \sum_kH_k\\
        \sum_kH_k & 0
    \end{bmatrix}\right)}\\
    &=\norm{S_p\left(\frac{t}{r}\right)
    -\exp\left(-i\frac{t}{r}\begin{bmatrix}
        0 & H\\
        H & 0
    \end{bmatrix}\right)}
    =\mathbf{O}\left(\left(\frac{\alpha_{\text{comm}}t}{r}\right)^{p+1}\right).
\end{aligned}
\end{equation}

The total error of our simulation is thus asymptotically bounded by
\begin{equation}
    \mathbf{O}\left(rn_K\left(\frac{\alpha_\infty t}{r}\right)^{\frac{p+1}{2}}
    +r\left(\frac{\alpha_{\text{comm}}t}{r}\right)^{p+1}\right)
    =\mathbf{O}\left(\frac{n_K\alpha_\infty^{\frac{p+1}{2}}t^{\frac{p+1}{2}}}{r^{\frac{p-1}{2}}}
    +\frac{\alpha_{\text{comm}}^{p+1}t^{p+1}}{r^p}\right).
\end{equation}
Hence, to achieve a target accuracy of $\epsilon$, it suffices to take
\begin{equation}
    r=\mathbf{O}\left(\alpha_\infty t\left(\frac{n_K\alpha_\infty t}{\epsilon}\right)^{\frac{2}{p-1}}
    +\alpha_{\text{comm}}t\left(\frac{\alpha_{\text{comm}}t}{\epsilon}\right)^{\frac{1}{p}}\right)
    =\left(\alpha_{\text{comm}}+\alpha_{\infty}\right)t\left(\frac{n_K\alpha_\infty t}{\epsilon}\right)^{o(1)}
\end{equation}
for $p$ sufficiently large. The remaining claims follow from the definition of product formulas $S_{p,k}$, $M_{p,k}$, and $S_p$. 

This simulation makes fractional queries to each input Hamiltonian block encoding $E_{\tau A_{jk}}=\exp\left(-i\tau\left[\begin{smallmatrix}
    0 & A_{jk}^\dagger\\
    A_{jk} & 0
\end{smallmatrix}\right]\right)$, which may be further realized using~\cor{frac_scale}.
\end{proof}

\subsection{Hamiltonian squaring}
\label{sec:sos_square}
In our simulation of the generic sum-of-squares Hamiltonians, we have used the Hamiltonian-based matrix multiplication subroutine from~\thm{multiply}, which requires higher-order Lie group commutator formulas and $2$ ancilla qubits.

In this section, we present a significantly simplified algorithm for simulating Hermitian sum-of-squares Hamiltonians:
\begin{equation}
    H=\sum_{k=1}^{n_K}H_k
    =\sum_{k=1}^{n_K}\left(\sum_{j=1}^{n_J}H_{jk}\right)^2.
\end{equation}
where each $H_{jk}$ is Hermitian and its exponential can be implemented directly. The main technical ingredient is a Hamiltonian squaring subroutine which realizes Hamiltonian evolution under $H^2$ using evolutions of $H$. Since squaring is an even function, this can be realized using our Hamiltonian QSVT~\thm{qsvt_herm_even} for even polynomials and Hermitian inputs, although the detailed construction is somewhat convoluted. Here, we present a direct solution based on~\thm{qsvt_odd} by implementing the cubic function instead. This depends on the following dominated polynomial approximation, which is discussed further in~\append{composite_square}.

\begin{proposition}[Dominated approximation for Hamiltonian squaring]
For any $\epsilon>0$ and constant $0<\xi\leq\frac{\pi}{2}$, there exist a real odd polynomial $p(x)$ and even polynomial $q(x)$ such that
\begin{equation}
\begin{aligned}
    &\abs{p(x)-\sin\left(\arcsin^3(x)\right)}\leq\epsilon,\qquad&&\forall x\in\left[-\sin\left(\frac{\pi}{2}-\xi\right),\sin\left(\frac{\pi}{2}-\xi\right)\right],\\
    &\abs{q(x)-\frac{\cos\left(\arcsin^3(x)\right)}{\sqrt{1-x^2}}}\leq\epsilon,\qquad&&\forall x\in\left[-\sin\left(\frac{\pi}{2}-\xi\right),\sin\left(\frac{\pi}{2}-\xi\right)\right],\\
    &p^2(x)+(1-x^2)q^2(x)\leq1+\epsilon,\qquad&&\forall x\in[-1,1].
\end{aligned}
\end{equation}
Moreover, both $p$ and $q$ have the asymptotic degree
\begin{equation}
    \mathbf{O}\left(\log\left(\frac{1}{\epsilon}\right)\right).
\end{equation}
\end{proposition}

\begin{figure}[t]
	\centering
\includegraphics[scale=\circuitwidth]{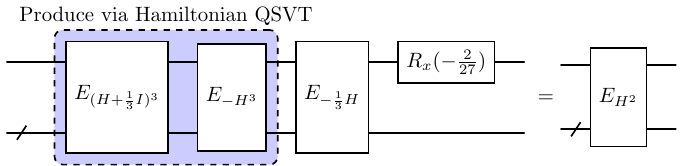}
\caption{Quantum circuit for Hamiltonian squaring ($R_x(\theta)=e^{-i\frac{\theta}{2}X}$).}
\label{fig:hamiltonian square}
\end{figure}

\begin{corollary}[Hamiltonian squaring]
\label{cor:square}
Let $H$ be a Hermitian matrix encoded by the Hamiltonian block encoding $E_H=\exp\left(-i\left[\begin{smallmatrix}
    0 & H\\
    H & 0
\end{smallmatrix}\right]\right)$ with $\norm{H}\leq1$. Then the Hamiltonian block encoding
\begin{equation}
    E_{H^2}=\exp\left(-i\begin{bmatrix}
        0 & H^2\\
        H^2 & 0
    \end{bmatrix}\right)
\end{equation}
can be constructed with accuracy $\epsilon$ using
\begin{equation}
    \mathbf{O}\left(\log\left(\frac{1}{\epsilon}\right)\right)
\end{equation}
queries to $E_H$. See \fig{hamiltonian square} for the corresponding circuit diagram.
\end{corollary}
\begin{proof}
We begin by constructing a Hamiltonian block encoding of the shifted operator $H+\frac{1}{3}I$:
\begin{equation}
    E_{H+\frac{1}{3}I}
    =\exp\left(-i\begin{bmatrix}
        0 & H+\frac{1}{3}I\\
        H+\frac{1}{3}I & 0
    \end{bmatrix}\right)
    =\exp\left(-i\begin{bmatrix}
        0 & H\\
        H & 0
    \end{bmatrix}\right)
    \exp\left(-i\frac{1}{3}\begin{bmatrix}
        0 & I\\
        I & 0
    \end{bmatrix}\right).
\end{equation}
Note that $H$ is a Hermitian matrix that commutes with the identity operator, so the above shifting can be implemented exactly using the first-order Lie Trotter formula with no Trotter error.

Here, $H+\frac{1}{3}I$ is a Hermitian matrix with eigenvalues $\left[-\frac{2}{3},\frac{4}{3}\right]$ strictly enclosed by $\left[-\frac{\pi}{2},\frac{\pi}{2}\right]$. So we can apply Hamiltonian QSVT to construct an $\epsilon$-approximate of
\begin{equation}
    E_{\left(H+\frac{1}{3}I\right)^3}
    =\exp\left(-i\begin{bmatrix}
        0 & \left(H+\frac{1}{3}I\right)^3\\
        \left(H+\frac{1}{3}I\right)^3 & 0
    \end{bmatrix}\right)
\end{equation}
with query complexity
\begin{equation}
    \mathbf{O}\left(\log\left(\frac{1}{\epsilon}\right)\right).
\end{equation}

This Hamiltonian block encoding contains the desired $H^2$ term, but also includes constant, linear, and cubic terms:
\begin{equation}
\begin{aligned}
    &\exp\left(-i\begin{bmatrix}
        0 & \left(H+\frac{1}{3}I\right)^3\\
        \left(H+\frac{1}{3}I\right)^3 & 0
    \end{bmatrix}\right)\\
    &=\exp\left(-i\begin{bmatrix}
        0 & H^3\\
        H^3 & 0
    \end{bmatrix}\right)
    \exp\left(-i\begin{bmatrix}
        0 & H^2\\
        H^2 & 0
    \end{bmatrix}\right)
    \exp\left(-i\frac{1}{3}\begin{bmatrix}
        0 & H\\
        H & 0
    \end{bmatrix}\right)
    \exp\left(-i\frac{1}{27}\begin{bmatrix}
        0 & I\\
        I & 0
    \end{bmatrix}\right).
\end{aligned}
\end{equation}
To cancel the unwanted terms, we implement each individual Hamiltonian block encoding, with the cubic case handled again by Hamiltonian QSVT.
Similar as before, the results can be combined using the first-order Lie-Trotter formula without Trotter error.
\end{proof}

\subsection{Sum-of-squares electronic structure Hamiltonian simulation}
\label{sec:sos_electron}
We now consider simulating sum-of-squares electronic structure Hamiltonians~\cite{Low25}, which can be expressed using the notation of~\append{fermionic_property} as
\begin{equation}
\begin{aligned}
    H=\sum_{r=1}^{n_R}\sum_{c=1}^{n_C}H_{rc}^2=\sum_{r=1}^{n_R}\sum_{c=1}^{n_C}\mathbf{Quad}^2(W_{rc}),\qquad
    [W_{rc}]_{p,q}=\sum_{b=1}^{n_B}w_{rcb}u_{rb,p}u_{rb,q},
\end{aligned}
\end{equation}
with rank $n_R$, copies $n_C$, base size $n_B$ and number of spin orbitals $n$, where $w$, $u$ are real-valued coefficient tensors satisfying only the normalization condition $\norm{u_{rb,\cdot}}=1$, and $\mathbf{Quad}\left(W\right)
=\sum_{p,q}\left[W\right]_{p,q}A_p^\dagger A_q$ is the fermionic quadratic operator. For simplicity, we have omitted a one-body term $\sum_{p,q=1}^nh_{pq}A_p^\dagger A_q$ as it does not dominate the complexity scaling. We have also dropped the shifting by identity matrix $w_{rc}I$ as its commutator with any operator is zero.

It is clear that these Hamiltonians are special cases of the generic sum-of-squares models studied in~\sec{sos_sos}, to which~\cor{sos} applies. However, as is discussed in~\sec{intro_sos}, the performance of quantum simulation can be improved by utilizing specific properties of the fermionic systems. For instance, instead of implementing elementary terms from $\mathbf{Quad}(\cdot)$ one by one, we decompose its exponential into $\mathbf{O}(n)$ Givens rotations following~\cite[Lemma 8]{vonBurg21}. Moreover, each $\mathbf{Quad}(W_{rc})$ is Hermitian, and the evolution under $\mathbf{Quad}^2(W_{rc})$ can be implemented via Hamiltonian squaring (\cor{square}). 

In the following, we evaluate the fermionic $\eta$-seminorm of nested commutators of the Hamiltonian terms $\alpha_{\text{comm},\eta}$. This $\eta$-seminorm, introduced in~\append{fermionic_seminorm}, can be seen as the spectral norm restricted to the subspace of $\eta$ particles. It determines normalization factor of the simulation algorithm from~\cor{sos} for input states with a fixed number of electrons. With notation same as~\append{fermionic_seminorm}, 
the $\eta$-seminorm of nested commutators (for a $p$th-order formula) is given by
\begin{equation}
    \left(\sum_{\substack{r_{p+1},c_{p+1},\ldots,\\r_{2},c_{2},r_{1},c_{1}}}\norm{\left[\mathbf{Quad}^2(W_{r_{p+1}c_{p+1}}),\ldots,\left[\mathbf{Quad}^2(W_{r_{2}c_{2}}),\mathbf{Quad}^2(W_{r_{1}c_{1}})\right]\right]}_\eta\right)^{\frac{1}{p+1}}.
\end{equation}
Let us analyze the combinatorial structure of these nested commutators. We start by observing the following property of commutator
\begin{equation}
    \left[J^2,K\right]
    =J^2K-KJ^2
    =J(JK-KJ)+(JK-KJ)J
    =\left\{J,\left[J,K\right]\right\}
\end{equation}
for arbitrary matrices $J$ and $K$ where $\{J,K\}=JK+KJ$ is the matrix anticommutator. This allows us to rewrite the nested commutator (for a fixed choice of $r_1,c_1,\ldots,r_{p+1},c_{p+1}$) as
\begin{equation}
\begin{aligned}
    &\left[\mathbf{Quad}^2(W_{r_{p+1}c_{p+1}}),\cdots,\left[\mathbf{Quad}^2(W_{r_{2}c_{2}}),\mathbf{Quad}^2(W_{r_{1}c_{1}})\right]\right]\\
    &=\left\{\mathbf{Quad}(W_{r_{p+1}c_{p+1}}),\left[\mathbf{Quad}(W_{r_{p+1}c_{p+1}}),\ldots,\left\{\mathbf{Quad}(W_{r_{2}c_{2}}),\left[\mathbf{Quad}(W_{r_{2}c_{2}}),\mathbf{Quad}^2(W_{r_{1}c_{1}})\right]\right\}\right]\right\}.
\end{aligned}
\end{equation}

To save space, we abbreviate $\mathbf{Quad}(W_{r_{j}c_{j}})$ as $\mathbf{Quad}_{j}$. Starting from the left-most operator, we exchange matrix commutators and anticommutators using the following rule
\begin{equation}
\begin{aligned}
    \left[J,\left\{K,L\right\}\right]
    &=J(KL+LK)-(KL+LK)J\\
    &=(JK-KJ)L+K(JL-LJ)
    +(JL-LJ)K+L(JK-KJ)\\
    &=\left\{[J,K],L\right\}
    +\left\{K,[J,L]\right\}.
\end{aligned}
\end{equation}
The first application of this commutation rule yields
\begin{equation}
\begin{aligned}
    &\left\{\mathbf{Quad}_{p+1},
    \left[\mathbf{Quad}_{p+1},
    \left\{\mathbf{Quad}_{p},
    \left[\mathbf{Quad}_{p},\cdots
    \right]
    \right\}
    \right]
    \right\}\\
    &=\left\{\mathbf{Quad}_{p+1},
    \left\{
    \left[\mathbf{Quad}_{p+1},\mathbf{Quad}_{p}
    \right],
    \left[\mathbf{Quad}_{p},\cdots
    \right]
    \right\}
    \right\}\\
    &\quad+\left\{\mathbf{Quad}_{p+1},
    \left\{\mathbf{Quad}_{p},
    \left[\mathbf{Quad}_{p+1},\left[\mathbf{Quad}_{p},\cdots
    \right]
    \right]
    \right\}
    \right\}.
\end{aligned}
\end{equation}
Continuing, the first term splits into
\begin{equation}
\begin{aligned}
    &\left\{\mathbf{Quad}_{p+1},
    \left\{
    \left[\mathbf{Quad}_{p+1},\mathbf{Quad}_{p}
    \right],
    \left[\mathbf{Quad}_{p},
    \left\{\mathbf{Quad}_{p-1},
    \left[\mathbf{Quad}_{p-1},\cdots
    \right]
    \right\}
    \right]
    \right\}
    \right\}\\
    &=\left\{\mathbf{Quad}_{p+1},
    \left\{
    \left[\mathbf{Quad}_{p+1},\mathbf{Quad}_{p}
    \right],
    \left\{\left[\mathbf{Quad}_{p},\mathbf{Quad}_{p-1}\right],
    \left[\mathbf{Quad}_{p-1},\cdots
    \right]
    \right\}
    \right\}
    \right\}\\
    &\quad+\left\{\mathbf{Quad}_{p+1},
    \left\{
    \left[\mathbf{Quad}_{p+1},\mathbf{Quad}_{p}
    \right],
    \left\{\mathbf{Quad}_{p-1},
    \left[\mathbf{Quad}_{p},\left[\mathbf{Quad}_{p-1},\cdots
    \right]\right]
    \right\}
    \right\}
    \right\},
\end{aligned}
\end{equation}
whereas the second term expands into
\begin{equation}
\begin{aligned}
    &\left\{\mathbf{Quad}_{p+1},
    \left\{\mathbf{Quad}_{p},
    \left[\mathbf{Quad}_{p+1},\left[\mathbf{Quad}_{p},
    \left\{\mathbf{Quad}_{p-1},
    \left[\mathbf{Quad}_{p-1},\cdots
    \right]
    \right\}
    \right]
    \right]
    \right\}
    \right\}\\
    &=\left\{\mathbf{Quad}_{p+1},
    \left\{\mathbf{Quad}_{p},
    \left\{\left[\mathbf{Quad}_{p+1},\left[\mathbf{Quad}_{p},\mathbf{Quad}_{p-1}\right]\right],
    \left[\mathbf{Quad}_{p-1},\cdots
    \right]
    \right\}
    \right\}
    \right\}\\
    &\quad+\left\{\mathbf{Quad}_{p+1},
    \left\{\mathbf{Quad}_{p},
    \left\{\left[\mathbf{Quad}_{p+1},\mathbf{Quad}_{p-1}\right],
    \left[\mathbf{Quad}_{p},\left[\mathbf{Quad}_{p-1},\cdots
    \right]
    \right]
    \right\}
    \right\}
    \right\}\\
    &\quad+\left\{\mathbf{Quad}_{p+1},
    \left\{\mathbf{Quad}_{p},
    \left\{\left[\mathbf{Quad}_{p},\mathbf{Quad}_{p-1}\right],
    \left[\mathbf{Quad}_{p+1},\left[\mathbf{Quad}_{p-1},\cdots
    \right]
    \right]
    \right\}
    \right\}
    \right\}\\
    &\quad+\left\{\mathbf{Quad}_{p+1},
    \left\{\mathbf{Quad}_{p},
    \left\{\mathbf{Quad}_{p-1},
    \left[\mathbf{Quad}_{p+1},\left[\mathbf{Quad}_{p},\left[\mathbf{Quad}_{p-1},\cdots
    \right]\right]\right]
    \right\}
    \right\}
    \right\}.\\
\end{aligned}
\end{equation}
We can then unwrap this by induction.

The result is a linear combination of nested anticommutators of $p+2$ operators, each consisting of nested commutators. Specifically, we have
\begin{equation}
\begin{aligned}
    &\left[\mathbf{Quad}^2(W_{r_{p+1}c_{p+1}}),\cdots,\left[\mathbf{Quad}^2(W_{r_{2}c_{2}}),\mathbf{Quad}^2(W_{r_{1}c_{1}})\right]\right]\\
    &=\left\{\mathbf{Quad}(W_{r_{p+1}c_{p+1}}),\left[\mathbf{Quad}(W_{r_{p+1}c_{p+1}}),\ldots,\left\{\mathbf{Quad}(W_{r_{2}c_{2}}),\left[\mathbf{Quad}(W_{r_{2}c_{2}}),\mathbf{Quad}^2(W_{r_{1}c_{1}})\right]\right\}\right]\right\}\\
    &=\frac{1}{2}\sum_{\mathcal{S}_{p+1},\mathcal{S}_{p},\ldots,\mathcal{S}_{1},\mathcal{S}_{0}}
    \left\{\mathbf{Quad}\left(W_{\mathcal{S}_{p+1}}\right),\left\{\mathbf{Quad}\left(W_{\mathcal{S}_{p}}\right),\ldots,\left\{\mathbf{Quad}\left(W_{\mathcal{S}_{1}}\right),\mathbf{Quad}\left(W_{\mathcal{S}_{0}}\right)\right\}\right\}\right\}.
\end{aligned}
\end{equation}
Here, the subsets $\mathcal{S}_0,\mathcal{S}_1,\ldots,\mathcal{S}_{p+1}$ form a monotonic partition of $\{2,\ldots,p+1\}$ satisfying the following constraints.
\begin{enumerate}
    \item Exhaustion: $\mathcal{S}_0\cup\mathcal{S}_1\cup\cdots\cup\mathcal{S}_{p+1}=\{2,\ldots,p+1\}$.
    \item Disjointness: $\mathcal{S}_j\cap\mathcal{S}_l$ is empty for all $j\neq l$.
    \item Monotonicity: $\mathcal{S}_j\subseteq\{j+1,\ldots,p+1\}$ for all $j\geq 1$ and $\mathcal{S}_0\subseteq\{2,\ldots,p+1\}$.
\end{enumerate}
Fixing a specific
\begin{equation}
    \mathcal{S}_j=
    \begin{cases}
        \left\{j+1\leq s_1\leq \cdots\leq s_{\#\mathcal{S}_j}\leq p+1\right\},\qquad&j\geq1,\\
        \left\{2\leq s_1\leq \cdots\leq s_{\#\mathcal{S}_0}\leq p+1\right\},\qquad&j=0,\\
    \end{cases}
\end{equation}
we define
\begin{equation}
    W_{\mathcal{S}_{j}}
    =\begin{cases}
        \left[\cdots,\left[W_{r_{s_2},c_{s_2}},\left[W_{r_{s_1},c_{s_1}},W_{r_j,c_j}\right]\right]\right],\qquad&j\geq 1,\\
        \left[\cdots,\left[W_{r_{s_2},c_{s_2}},\left[W_{r_{s_1},c_{s_1}},W_{r_1,c_1}\right]\right]\right],\qquad&j=0.\\
    \end{cases}
\end{equation}
Applying the triangle inequality, we finally obtain
\begin{equation}
\begin{aligned}
    \alpha_{\text{comm},\eta}
    &=\left(\sum_{\substack{r_{p+1},c_{p+1},\ldots,\\r_{2},c_{2},r_{1},c_{1}}}
    \sum_{\mathcal{S}_{p+1},\mathcal{S}_{p},\ldots,\mathcal{S}_{1},\mathcal{S}_{0}}
    \norm{\mathbf{Quad}\left(W_{\mathcal{S}_{p+1}}\right)}_\eta\cdots\norm{\mathbf{Quad}\left(W_{\mathcal{S}_{1}}\right)}_\eta\norm{\mathbf{Quad}\left(W_{\mathcal{S}_{0}}\right)}_\eta\right)^{\frac{1}{p+1}}\\
    &=\left(\sum_{\substack{r_{p+1},c_{p+1},\ldots,\\r_{2},c_{2},r_{1},c_{1}}}
    \sum_{\mathcal{S}_{p+1},\mathcal{S}_{p},\ldots,\mathcal{S}_{1},\mathcal{S}_{0}}
    \norm{W_{\mathcal{S}_{p+1}}}_\eta\cdots\norm{W_{\mathcal{S}_{1}}}_\eta\norm{W_{\mathcal{S}_{0}}}_\eta\right)^{\frac{1}{p+1}}.
\end{aligned}
\end{equation}
Correspondingly, the gate complexity of our simulation is
\begin{equation}
    \left(\alpha_{\text{comm},\eta}+\alpha_{\infty}\right)t\left(\frac{n\alpha_\infty t}{\epsilon}\right)^{o(1)}n_Rn_Cn_Bn,
\end{equation}
where
\begin{equation}
    \alpha_\infty=\max_{r,c}\left(\sum_b\abs{w_{rcb}}\right)^2.
\end{equation}
We achieve this with $1$ ancilla qubit.

Per the discussion in~\append{fermionic_seminorm}, $\alpha_{\text{comm},\eta}$ depends only on the coefficient tensors $W$, and can thus be computed in polynomial time on a classical computer. It can be seen as the effective normalization factor of our simulation algorithm for $\eta$-particle input states. We now explain how a loose upper bound of it reproduces the normalization factor of qubitization~\cite{Low2016Qubitization}. In particular, we identify three sources of looseness:
\begin{enumerate}
    \item Ignorance of particle number: if we were to ignore the number of electrons $\eta$ in the initial state, we would get the worst-case spectral norm bound
\begin{equation}
\begin{aligned}
    &\sum_{\substack{r_{p+1},c_{p+1},\ldots,\\r_{2},c_{2},r_{1},c_{1}}}
    \sum_{\mathcal{S}_{p+1},\mathcal{S}_{p},\ldots,\mathcal{S}_{1},\mathcal{S}_{0}}
    \norm{W_{\mathcal{S}_{p+1}}}_\eta\cdots\norm{W_{\mathcal{S}_{1}}}_\eta\norm{W_{\mathcal{S}_{0}}}_\eta\\
    &\leq\sum_{\substack{r_{p+1},c_{p+1},\ldots,\\r_{2},c_{2},r_{1},c_{1}}}
    \sum_{\mathcal{S}_{p+1},\mathcal{S}_{p},\ldots,\mathcal{S}_{1},\mathcal{S}_{0}}
    \norm{W_{\mathcal{S}_{p+1}}}\cdots\norm{W_{\mathcal{S}_{1}}}\norm{W_{\mathcal{S}_{0}}}.
\end{aligned}
\end{equation}
    \item Ignorance of commutators: if we were to drop all the commutator structures and apply the loose bound $\norm{[W,V]}\leq2\norm{W}\norm{V}$, we would have
\begin{equation}
\begin{aligned}
    \sum_{\substack{r_{p+1},c_{p+1},\ldots,\\r_{2},c_{2},r_{1},c_{1}}}
    \sum_{\mathcal{S}_{p+1},\mathcal{S}_{p},\ldots,\mathcal{S}_{1},\mathcal{S}_{0}}
    \norm{W_{\mathcal{S}_{p+1}}}\cdots\norm{W_{\mathcal{S}_{1}}}\norm{W_{\mathcal{S}_{0}}}
    &\leq2^{p}\left(\sum_{r,c}\norm{W_{rc}}^2\right)^{p+1}.
\end{aligned}
\end{equation}
    \item Worst-case bound of nonorthogonal tensors: if we further applied the triangle inequality to the factorization $W_{rc}=\sum_bw_{rcb}u_{rb}u_{rb}^\top$ as the sum of rank-$1$ mutually nonorthogonal projections, we would end up with
\begin{equation}
    \sum_{r,c}\norm{W_{rc}}^2
    \leq\sum_{r,c}\left(\sum_b\abs{w_{rcb}}\norm{u_{rb}u_{rb}^\top}\right)^2
    =\sum_{r,c}\left(\sum_b\abs{w_{rcb}}\right)^2.
\end{equation}
\end{enumerate}
This now becomes the $\alpha_1$ of qubitization.

%% file: discuss.tex
In this work, we have elucidated a framework for performing matrix arithmetics with Hamiltonian evolution, where matrices of interest are embedded in the off-diagonal blocks of a Hamiltonian. We show how to maintain this Hamiltonian block encoding after specifying all its entries, so that matrix operations can be composed one after another, and the entire quantum computation uses at most two ancilla qubits. We present a host of methods to manipulate these Hamiltonian block encodings, implementing matrix multiplication, matrix addition, matrix inversion, Hermitian conjugation, complex phase scaling, integer scaling, fractional scaling, and singular value transformation for both odd and even polynomials. We then apply our results to a range of problems, such as Green's function estimation and quantum simulation of the sum-of-squares Hamiltonian, matching or improving the state of the art.

We have demonstrated how to multiply generic matrices within Hamiltonian block encoding using the Lie group commutator formula and its higher-order generalizations. These formulas approximate the evolution under commutators of Hamiltonian terms, in contrast to the conventional Lie-Trotter-Suzuki formulas, which target the exponential of a sum of Hamiltonian terms. Historically, Lie group commutator formulas have played an important role in the study of Lie theory and have served as useful tools in proving the Solovay-Kitaev theorem and in engineering quantum many-body Hamiltonians. However, to the best of our knowledge, there have been very few algorithmic applications. Our Hamiltonian-based matrix multiplication algorithm relies essentially on these commutator product formulas, which motivates further investigation into improving their practical performance, as exemplified by the proof of our sharpened error bound in~\lem{commutator_tight} and the construction of new product formulas in recent work~\cite{Casas2025,YuAn22commutator}.

Our generic multiplication scheme uses $2$ ancilla qubits independent of the number of multiplicands. In fact, a single qutrit suffices: the abstract definition from \cref{sec:block_abstract} immediately shows that an ancilla space with dimension $\geq3$ enables a sufficiently large operator dilation for matrix multiplication. By allowing the multiplication to be approximate, we have thus surpassed the lower bound on the ancilla space dimensions, known from operator theory for more than four decades~\cite{thompson1982doubly}. Moreover, we have achieved this for generic multiplicands, going beyond recent result~\cite[Theorem 4]{Vasconcelos25} targeting only at a neighborhood of the identity operator. Thus, while our approach introduces a $\frac{1}{\epsilon^{o(1)}}$ query overhead and uses a slightly larger ancilla space for dilation, it enables a generic reduction in space complexity across many potential applications, which would be difficult to obtain otherwise.

We have developed multiple algorithms for Hamiltonian QSVT, with both the input and output represented by Hamiltonian block encodings. We show this can be naturally reduced to the dominated approximation problem, where the target functions $\sin(f(\arcsin(x)))\approx p(x)$ and $\frac{\cos(f(\arcsin(x)))}{\sqrt{1-x^2}}\approx q(x)$ are simultaneously approximated by polynomials over the domain of interest, while the dominated condition $p^2(x)+(1-x^2)q^2(x)\leq1$ is approximately satisfied throughout the entire unit interval. For even polynomials, there is an intrinsic challenge due to the parity constraint of QSVT.
Nevertheless, we show that for Hermitian inputs the Hamiltonian block encoding can be shifted such that its spectra lie sufficiently far away from $x = 0$, and the transformation can be approximated via an odd extension of the original polynomial. 
The generic case is then handled by considering the Hermitian dilation.
In constructing the polynomial approximations, we have analyzed the behavior of a variety of functions over complex regions surrounding the unit interval. It is plausible that our approximation results can be improved by exploiting further properties of the target function, as suggested in~\cite[Remark 22]{TangTian24}, or by employing randomized compilation~\cite{martyn2025halving}, but a systematic study of such improvements is left for future work. 

We have primarily focused on applications to Green's function estimation and sum-of-squares Hamiltonian simulation, but we expect that our methods are useful in a variety of other settings.
For instance, recent algorithms for solving differential equations~\cite{fang2025qubitefficientquantumalgorithmlinear} and simulating Lindbladian evolution for ground state preparation~\cite{lin2025dissipativepreparationmanybodyquantum, ding2025endtoendefficientquantumthermal, zhan2025rapidquantumgroundstate} can be formulated using Hamiltonian block encodings. These algorithms do not utilize the more complex matrix arithmetic operations we have studied, but may benefit from the new techniques for transforming Hamiltonians developed here. 
{We have shown how to compute state overlaps with respect to a Hamiltonian encoded matrix. Our current approach first transforms the Hamiltonian block encoding into a unitary block encoding, and then solves the overlap estimation problem. This would introduce a logarithmic query overhead if performed separately, which could be avoided by bundling the overlap estimation with any prior Hamiltonian QSVT as in the Green's function estimation~\prop{dominated_green}.}
Note that unitaries can be recovered by evolving their Hamiltonian block encodings for a constant time~\cite[Section VIII]{berry2012black}.

Notably, our methods can inherit tightened, state-dependent Trotter error bounds. We prepare and add Hamiltonian block encodings using conventional product formulas. Thus, adding terms with low Trotter error can lead to significantly improved gate complexities. While initial Trotter error bounds were incredibly loose~\cite{CMNRS18}, they have been systematically tightened via commutator-based~\cite{CSTWZ19} and state-based bounds~\cite{mizuta2025trotterizationsubstantiallyefficientlowenergy,Su2021nearlytight,ChenBrandao21,Zhao21,Hatomura22,Sahinoglu2021} (for quantum simulation in the low-energy subspace~\cite{Zlokapa2024hamiltonian,Gu2021fastforwarding} and beyond). 

Several other algorithmic paradigms also maintain the Hamiltonian block encoding and can therefore be used in conjunction with our methods. This includes the algorithm qDRIFT~\cite{Campbell18} based on randomization, the symmetry protection algorithm~\cite{Burgarth2022oneboundtorulethem,Tran21} based on Zeno subspace projection, and the linear Hamiltonian transformation algorithm~\cite{Odake24} in the Pauli basis. 
By applying these methods together, one can achieve additional improvements in deploying quantum linear algebra algorithms, uncovering speedups not known previously.

There are reasons for optimism regarding the implementability of our algorithms. First, our algorithms are qubit-efficient, requiring at most two additional ancilla qubits, and only one if the computation does not invoke the generic matrix multiplication subroutine. This is especially important for early fault-tolerant devices with few logical qubits~\cite{Katabarwa_2024}. Second, our circuits are significantly simpler to describe than those required for unitary block encodings, which employ complicated arithmetic and may require QROM. 
Simpler circuits may be easier to compile and implement on fault-tolerant quantum computers.  While our approach underperforms in $T$ complexity for specific quantum chemistry applications, recent advances in magic state cultivation have significantly reduced the cost of preparing magic states~\cite{gidney2024magic, ruiz2025unfolded, sahay2025foldtransversalsurfacecodecultivation}. 
Note that this underperformance can also be attributed in part to the fact that the ansatz Hamiltonians studied in those chemistry work are specifically tailored for unitary block encodings; developing new ansatze for Trotter-based methods leveraging Hamiltonian-based matrix arithmetics could substantially reduce the cost of simulation.
Furthermore, the Trotter steps required for Hamiltonian block encodings may be implemented via multi-target CNOTs and arbitrary angle rotations, subroutines that can both be optimized. Multi-target CNOTs can be implemented in constant depth via lattice surgery on surface codes~\cite{fowler2018low}.  Arbitrary angle $R_z$ gates can be implemented with constant depth in expectation via PAR rotations~\cite{Cody_Jones_2012,choi2023fault,akahoshi2024partially,litinski2019game,sun2025space}. 
We leave concrete resource estimates of our methods to future work.

%% file: lie.tex
In this appendix, we present auxiliary results on the Lie group commutator product formulas. In~\append{lie_bch}, we derive the BCH expansion of the second-order formula, which shows the tightness of our error bound in the main text. We then analyze higher-order group commutator formulas in \append{lie_higher}.

\subsection{Baker-Campbell-Hausdorff expansion}
\label{append:lie_bch}
Given Hermitian operators $J$, $K$ and $\tau\geq0$ without loss of generality, we now derive the BCH expansion of $e^{-i\tau J}e^{-i\tau K}e^{i\tau J}e^{i\tau K}$ to third order.
To this end, we begin with the expansion
\begin{equation}
\begin{aligned}
    e^{-i\tau J}e^{-i\tau K}
    &=\exp\left(-i\tau J-i\tau K-\frac{\tau^2}{2}[J,K]
    +\frac{i\tau^3}{12}[J,[J,K]]
    +\frac{i\tau^3}{12}[K,[K,J]]+\mathbf{O}\left(\tau^4\right)\right),\\
    e^{i\tau J}e^{i\tau K}
    &=\exp\left(i\tau J+i\tau K-\frac{\tau^2}{2}[J,K]
    -\frac{i\tau^3}{12}[J,[J,K]]
    -\frac{i\tau^3}{12}[K,[K,J]]+\mathbf{O}\left(\tau^4\right)\right),
\end{aligned}
\end{equation}
where the convergence is guaranteed if
\begin{equation}
    \tau\left(\norm{J}+\norm{K}\right)<\frac{\ln(2)}{2}.
\end{equation}

Assuming the above requirement on the convergence holds, let us bound the norm of operator in the exponent as~\cite{muger2019notes}
\begin{equation}
    \norm{\ln\left(e^{-i\tau J}e^{-i\tau K}\right)}
    =\norm{\sum_{l=1}^\infty\frac{(-1)^{l-1}}{l}\left(e^{-i\tau J}e^{-i\tau K}-I\right)^l}
    \leq\sum_{l=1}^\infty\frac{1}{l}\norm{e^{-i\tau J}e^{-i\tau K}-I}^l.
\end{equation}
Here,
\begin{equation}
    \norm{e^{-i\tau J}e^{-i\tau K}-I}
    =\norm{\sum_{\substack{j,k\geq0\\j+k>0}}\frac{\left(-i\tau J\right)^j\left(-i\tau K\right)^k}{j!k!}}
    \leq\sum_{\substack{j,k\geq0\\j+k>0}}\frac{\left(\tau \norm{J}\right)^j\left(\tau \norm{K}\right)^k}{j!k!}
    =e^{\tau\left(\norm{J}+\norm{K}\right)}-1,
\end{equation}
which gives
\begin{equation}
    \norm{\ln\left(e^{-i\tau J}e^{-i\tau K}\right)}
    \leq\sum_{l=1}^\infty\frac{1}{l}\left(e^{\tau\left(\norm{J}+\norm{K}\right)}-1\right)^l
    =-\ln\left(2-e^{\tau\left(\norm{J}+\norm{K}\right)}\right).
\end{equation}

Now, applying the BCH expansion again yields
\begin{equation}
\begin{aligned}
    &e^{-i\tau J}e^{-i\tau K}e^{i\tau J}e^{i\tau K}\\
    &=\exp\left(-\tau^2[J,K]
    +\frac{1}{2}\left[-i\tau J-i\tau K,-\frac{\tau^2}{2}[J,K]\right]
    +\frac{1}{2}\left[-\frac{\tau^2}{2}[J,K],i\tau J+i\tau K\right]
    +\mathbf{O}\left(\tau^4\right)\right)\\
    &=\exp\left(-\tau^2[J,K]+\frac{i\tau^3}{2}[J,[J,K]]+\frac{i\tau^3}{2}[K,[J,K]]+\mathbf{O}\left(\tau^4\right)\right).
\end{aligned}
\end{equation}
To ensure that the series converges, we require that
\begin{equation}
    -\ln\left(2-e^{\tau\left(\norm{J}+\norm{K}\right)}\right)<\frac{\ln(2)}{4}\quad\Rightarrow\quad
    \tau\left(\norm{J}+\norm{K}\right)<\ln\left(2-\frac{1}{2^{1/4}}\right).
\end{equation}
Our explicit error bound in~\lem{commutator_tight} thus matches the third-order terms of the BCH expansion for $\tau$ sufficiently small, and is provably tight up to a single application of the triangle inequality.

\subsection{Analysis of higher-order formulas}
\label{append:lie_higher}
We have focused on the second-order Lie group commutator product formula in the main text. Here, we analyze higher-order group commutator formulas which achieve asymptotically better performance. Such formulas can be constructed recursively as follows~\cite{ChildsWiebe12}. We start with 
\begin{equation}
    M_2\left(\tau\right)=e^{-i\tau J}e^{-i\tau K}e^{i\tau J}e^{i\tau K},
\end{equation}
which approximates the exponential $e^{-\tau^2[J,K]}$ to second order. Then higher-order formulas can be defined for $k\in\mathbb{Z}_{\geq2}$ via the recursion
\begin{equation}
\label{eq:def_commpf}
\begin{aligned}
    M_{2k}\left(\tau\right)
    &=M_{2(k-1)}\left(\gamma_k\tau\right)
     M_{2(k-1)}\left(-\gamma_k\tau\right)\\
    &\quad\cdot M_{2(k-1)}^{-1}\left(\beta_k\tau\right)
     M_{2(k-1)}^{-1}\left(-\beta_k\tau\right)\\
    &\quad\cdot M_{2(k-1)}\left(\gamma_k\tau\right)
     M_{2(k-1)}\left(-\gamma_k\tau\right),
\end{aligned}
\end{equation}
where
\begin{equation}
    \beta_k=\sqrt{2v_k},\qquad
    \gamma_k=\sqrt{\frac{1}{4}+v_k},\qquad
    v_k=\frac{2^{\frac{1}{k}}}{4\left(2-2^{\frac{1}{k}}\right)}.
\end{equation}
It is shown in~\cite{ChildsWiebe12} that $M_{2k}\left(\tau\right)=e^{-\tau^2[J,K]}+\mathbf{O}\left(\tau^{2k+1}\right)$, where the limit is taken as $\tau\rightarrow0$.

More generally, a $p$th-order group commutator formula is characterized by the order condition $M_{p}\left(\tau\right)=e^{-\tau^2[J,K]}+\mathbf{O}\left(\tau^{p+1}\right)$. Following an analysis similar to~\cite{CSTWZ19}, we obtain the following commutator bound for a generic higher-order Lie group commutator product formula.

\begin{lemma}[Commutator bound for higher-order group commutator formulas]
Let $J$, $K$ be Hermitian matrices, $\tau\geq0$, and $M_{p}\left(\tau\right)=e^{-\tau^2[J,K]}+\mathbf{O}\left(\tau^{p+1}\right)$ be a $p$th-order group commutator product formula for some $p\in\mathbb{Z}_{\geq2}$. Then,
\begin{equation}
    \norm{M_{p}(\tau)-e^{-\tau^2[J,K]}}
    =\mathbf{O}\left(\alpha_{\text{comm}}^{p+1}\tau^{p+1}\right),
\end{equation}
where
\begin{equation}
    \alpha_{\text{comm}}=\left(\sum_{j_1,\ldots,j_{p+1}=1}^{2}\norm{\left[H_{j_{p+1}},\ldots,\left[H_{j_2},H_{j_1}\right]\right]}\right)^{\frac{1}{p+1}},
\end{equation}
with $H_1=J$ and $H_2=K$.
\end{lemma}
\begin{proof}
We start with the integral representation
\begin{equation}
\begin{aligned}
    M_{p}(\tau)-e^{-\tau^2[J,K]}
    &=e^{-\tau^2[J,K]}\left(e^{\tau^2[J,K]}M_{p}(\tau)-I\right)
    =\int_{0}^{\tau}\mathrm{d}\tau_1\ e^{-(\tau^2-\tau_1^2)[J,K]}R(\tau_1)M_{p}(\tau_1),\\
    R(\tau_1)&=2\tau_1[J,K]
    +\left(\frac{\mathrm{d}}{\mathrm{d}\tau_1}M_{p}(\tau_1)\right)M_{p}^{-1}(\tau_1),
\end{aligned}
\end{equation}
which follows from the fundamental theorem of calculus.

Differentiating both sides of $M_{p}(\tau)=e^{-\tau^2[J,K]}+\mathbf{O}\left(\tau^{p+1}\right)$, we obtain the order condition
\begin{equation}
    \frac{\mathrm{d}}{\mathrm{d}\tau_1}M_{p}(\tau_1)=-2\tau_1[J,K]e^{-\tau_1^2[J,K]}+\mathbf{O}\left(\tau_1^{p}\right),
\end{equation}
which implies
\begin{equation}
    \frac{\mathrm{d}}{\mathrm{d}\tau_1}M_{p}(\tau_1)=-2\tau_1[J,K]M_{p}(\tau)+\mathbf{O}\left(\tau_1^{p}\right),
\end{equation}
or equivalently,
\begin{equation}
    R(\tau_1)=2\tau_1[J,K]
    +\left(\frac{\mathrm{d}}{\mathrm{d}\tau_1}M_{p}(\tau_1)\right)M_{p}^{-1}(\tau_1)
    =\mathbf{O}\left(\tau_1^{p}\right).
\end{equation}
The remaining analysis proceeds as in~\cite{CSTWZ19}.
\end{proof}

\begin{corollary}[Step number for higher-order group commutator formulas]
Let $J$, $K$ be Hermitian matrices, and $M_{p}\left(\tau\right)=e^{-\tau^2[J,K]}+\mathbf{O}\left(\tau^{p+1}\right)$ be a $p$th-order group commutator product formula for some $p\in\mathbb{Z}_{\geq2}$. For any evolution time $t\geq0$, the approximations
\begin{equation}
    \norm{M_{p}^r\left(\sqrt{\frac{t}{r}}\right)-e^{-t[J,K]}}\leq\epsilon,\qquad
    \norm{M_{p}^{\dagger r}\left(\sqrt{\frac{t}{r}}\right)-e^{t[J,K]}}\leq\epsilon
\end{equation}
can be achieved with accuracy $\epsilon$ by choosing
\begin{equation}
    r=\mathbf{O}\left(\frac{\alpha_{\text{comm}}^{2+4/(p-1)}t^{1+2/(p-1)}}{\epsilon^{2/(p-1)}}\right),
\end{equation}
where
\begin{equation}
    \alpha_{\text{comm}}=\left(\sum_{j_1,\ldots,j_{p+1}=1}^{2}\norm{\left[H_{j_{p+1}},\ldots,\left[H_{j_2},H_{j_1}\right]\right]}\right)^{\frac{1}{p+1}}
\end{equation}
with $H_1=J$ and $H_2=K$.
\end{corollary}
\begin{proof}
Applying the commutator bound for time $\tau=\sqrt{\frac{t}{r}}$, we have
\begin{equation}
    \norm{M_{p}\left(\sqrt{\frac{t}{r}}\right)-e^{-\frac{t}{r}[J,K]}}
    =\mathbf{O}\left(\alpha_{\text{comm}}^{p+1}\frac{t^{\frac{p+1}{2}}}{r^{\frac{p+1}{2}}}\right),
\end{equation}
which implies that
\begin{equation}
    \norm{M_{p}^r\left(\sqrt{\frac{t}{r}}\right)-e^{-t[J,K]}}
    \leq r\norm{M_{p}\left(\sqrt{\frac{t}{r}}\right)-e^{-\frac{t}{r}[J,K]}}
    =\mathbf{O}\left(\alpha_{\text{comm}}^{p+1}\frac{t^{\frac{p+1}{2}}}{r^{\frac{p-1}{2}}}\right).
\end{equation}
To ensure that the error is at most $\epsilon$, it suffices to choose
\begin{equation}
    r=\mathbf{O}\left(\frac{\alpha_{\text{comm}}^{2+4/(p-1)}t^{1+2/(p-1)}}{\epsilon^{2/(p-1)}}\right)
\end{equation}
as claimed. This handles the case where the evolution time $t\geq0$ is nonnegative. When $t<0$, we can simply take the Hermitian conjugation.
\end{proof}

%% file: composite.tex
In this appendix, we present auxiliary results on approximating composite functions by polynomials. We start by reviewing a method to approximate functions analytic over Bernstein ellipses surrounding the unit interval (\append{composite_stadium}). We also introduce common geometric objects such as stadiums and disks to simplify the reasoning. We then derive dominated polynomial approximation results to realize Hamiltonian QSVT (\append{composite_qsvt}), Hamiltonian matrix inversion (\append{composite_inverse}), Hamiltonian fractional scaling (\append{composite_frac}), Hamiltonian overlap estimation (\append{composite_overlap}), Green's function estimation (\append{composite_green}), and Hamiltonian squaring (\append{composite_square}).

\subsection{Dominated polynomial approximation over the unit interval}
\label{append:composite_stadium}
To find polynomial approximations of a function $f$ defined over a real interval, it is often helpful to examine its behavior over a slightly larger region surrounding the real interval.
A common setting is that $f$ is analytic in a complex region containing the Bernstein ellipses~\cite{demanet2010chebyshev,Fawzi}. 

A Bernstein ellipse is an ellipse in the complex plane with foci $\pm 1$ and center $0$. It can be described in terms of the semi-major axis $a>1$, the semi-minor axis $b>0$, or the elliptical radius $\rho>1$. These parameters are related by
\begin{equation}
\begin{aligned}
    &a=\frac{\rho+\rho^{-1}}{2},\qquad
    &&b=\frac{\rho-\rho^{-1}}{2},\\
    &\rho=a+\sqrt{a^2-1},\qquad
    &&\rho=b+\sqrt{b^2+1},\\
    &a=\sqrt{b^2+1},\qquad
    &&b=\sqrt{a^2-1}.
\end{aligned}
\end{equation}
We will slightly abuse the notation, denoting the Bernstein ellipse as $\mathbf{Ellipse}\left(a=\cdot\right)$, $\mathbf{Ellipse}\left(b=\cdot\right)$, and $\mathbf{Ellipse}\left(\rho=\cdot\right)$, when values for the semi-major axis, semi-minor axis and elliptical radius are provided.
Whenever the context is clear, we will use $\mathbf{Ellipse}(\cdot)$ to also denote the region enclosed by the Bernstein ellipse. We will include the foci to describe a general Ellipse oriented along the real axis. For instance,
\begin{equation}
    \mathbf{Ellipse}\left(\left[c_0,c_1\right];a,b\right)
    =\left\{\frac{\left(\Re(z)-c_0\right)^2}{a^2}+\frac{\left(\Im(z)-c_1\right)^2}{b^2}=1\right\}.
\end{equation}
A direct calculation yields the following containment~\cite[Fact 30]{TangTian24}.
\begin{lemma}
For any $0<\delta\leq1$,
\begin{small}
\begin{equation}
\newmaketag
\begin{aligned}
    \mathbf{Ellipse}\left(a=1+\frac{\delta^2}{4}\right)
    \subseteq\mathbf{Ellipse}\left(\rho=1+\delta\right)&=\mathbf{Ellipse}\left(a=1+\frac{\delta^2}{2(1+\delta)}\right)
    \subseteq\mathbf{Ellipse}\left(a=1+\frac{\delta^2}{2}\right),\\
    \mathbf{Ellipse}\left(b=\frac{3\delta}{4}\right)
    \subseteq\mathbf{Ellipse}\left(\rho=1+\delta\right)&=\mathbf{Ellipse}\left(b=\delta-\frac{\delta^2}{2(1+\delta)}\right)
    \subseteq\mathbf{Ellipse}\left(b=\delta\right).\\
\end{aligned}
\end{equation}
\end{small}%
\end{lemma}

The following result from approximation theory (\cite[Theorem 8.1 and 8.2]{trefethen2019approximation} and~\cite{Fawzi}) shows how to construct polynomial approximations of a function over the unit interval, based on its behavior on the Bernstein ellipses. The polynomial used in the proof is actually the truncated Chebyshev expansion of $f$, although this fact is not needed in the remainder of our work.
\begin{lemma}[Polynomial approximation via Bernstein ellipses]
\label{lem:approx_ellipse}
Let $f$ be analytic over a complex region enclosing $\mathbf{Ellipse}(\rho)$ for some elliptical radius $\rho>1$, such that $f_{[-1,1]}$ is real. For every $d\in\mathbb{Z}_{\geq0}$, there exists a real polynomial $h$ of degree $d$, such that
\begin{equation}
    \norm{h-f}_{\max,[-1,1]}\leq\frac{2\norm{f}_{\max,\mathbf{Ellipse}(\rho)}}{\rho-1}\rho^{-d}.
\end{equation}
Moreover, $h$ has the same parity as $f$.
\end{lemma}

The polynomial constructed above fulfills the approximation condition for $x\in[-1,1]$. However, their value will blow up exponentially outside the unit interval. To remedy this problem, one can multiply it with a window function to get a bounded polynomial approximation~\cite[Theorem 21]{TangTian24}. However, recall that we need a dominated polynomial approximation to realize Hamiltonian QSVT, which is different from the setting of~\cite[Theorem 21]{TangTian24}.

\begin{corollary}[Dominated polynomial approximation via Bernstein ellipse]
\label{cor:approx_dominated_ellipse}
Let $f$ be analytic over a complex region enclosing $\mathbf{Ellipse}(\rho=1+\alpha)$ for some $\alpha>0$, such that $f_{[-1,1]}$ is real. For any $\epsilon_{\text{dom}}>0$ and $0<\xi\leq1<b$, there exists a real polynomial $h_{\text{dom}}$ satisfying
\begin{equation}
\begin{aligned}
    &\abs{h_{\text{dom}}(x)-f(x)}\leq\epsilon_{\text{dom}},\qquad&&\forall x\in\left[-1+\xi,1-\xi\right],\\
    &\abs{h_{\text{dom}}(x)}\leq\abs{f(x)}+\epsilon_{\text{dom}},\qquad&&\forall x\in\left[-1,1\right],\\
    &\abs{h_{\text{dom}}(x)}\leq\epsilon_{\text{dom}},\qquad&&\forall x\in\left[-b,-1\right]\cup\left[1,b\right],
\end{aligned}
\end{equation}
with the asymptotic degree
\begin{equation}
    \mathbf{O}\left(\frac{b}{\alpha^2}\log\left(\frac{b\norm{f}_{\max,\mathbf{Ellipse}(\rho=1+\alpha)}}{\alpha^2\epsilon_{\text{dom}}}\right)
    +\frac{b}{\xi}\log\left(\frac{\norm{f}_{\max,\mathbf{Ellipse}(\rho=1+\alpha)}}{\epsilon_{\text{dom}}}\right)\right).
\end{equation}
Moreover, $h_{\text{dom}}$ has the same parity as $f$.
\end{corollary}
\begin{proof}
We start with the polynomial $h_{\text{bnd}}$ constructed by~\cite[Theorem 21]{TangTian24}, which has the behavior
\begin{equation}
\begin{aligned}
    &\abs{h_{\text{bnd}}(x)-f(x)}\leq\epsilon_{\text{bnd}}\norm{f}_{\max,\mathbf{Ellipse}(\rho=1+\alpha)},\qquad&&\forall x\in\left[-1,1\right],\\
    &\abs{h_{\text{bnd}}(x)}\leq\norm{f}_{\max,\mathbf{Ellipse}(\rho=1+\alpha)},\qquad&&\forall x\in\left[-1-\delta,1+\delta\right],\\
    &\abs{h_{\text{bnd}}(x)}\leq\epsilon_{\text{bnd}}\norm{f}_{\max,\mathbf{Ellipse}(\rho=1+\alpha)},\qquad&&\forall x\in\left[-b,-1-\delta\right]\cup\left[1+\delta,b\right],
\end{aligned}
\end{equation}
with the asymptotic degree
\begin{equation}
    \mathbf{O}\left(\frac{b}{\alpha^2}\log\left(\frac{b}{\alpha^2\epsilon_{\text{bnd}}}\right)\right),
\end{equation}
where $\delta=\mathbf{\Theta}\left(\alpha^2\right)$ is sufficiently small.
The first condition can be further split into
\begin{equation}
\begin{aligned}
    &\abs{h_{\text{bnd}}(x)-f(x)}\leq\epsilon_{\text{bnd}}\norm{f}_{\max,\mathbf{Ellipse}(\rho=1+\alpha)},\qquad&&\forall x\in\left[-1+\xi,1-\xi\right],\\
    &\abs{h_{\text{bnd}}(x)}\leq\abs{f(x)}+\epsilon_{\text{bnd}}\norm{f}_{\max,\mathbf{Ellipse}(\rho=1+\alpha)},\qquad&&\forall x\in\left[-1,-1+\xi\right]\cup\left[1-\xi,1\right].\\
\end{aligned}
\end{equation}

To restrict the range of $h_{\text{bnd}}$ outside $[-1,1]$, we multiply it with another smoothened rectangular window function. Specifically, we construct an even polynomial $w(x)$:
\begin{equation}
    w(x)\in
    \begin{cases}
        \left[1-\epsilon_{\text{rec}},1\right],\quad&x\in\left[-1+\xi,1-\xi\right],\\
        [-1,1],&x\in\left[-1,-1+\xi\right]\cup\left[1-\xi,1\right],\\
        \left[0,\epsilon_{\text{rec}}\right],&x\in\left[-b,-1\right]\cup\left[1,b\right],
    \end{cases}
\end{equation}
with an effective transition width $\frac{\xi}{b}$
and an asymptotic degree of~\cite[Lemma 29]{Gilyen2018singular}
\begin{equation}
    \mathbf{O}\left(\frac{b}{\xi}\log\left(\frac{1}{\epsilon_{\text{rec}}}\right)\right).
\end{equation}
We then let
\begin{equation}
    h_{\text{dom}}(x)=h_{\text{bnd}}(x)w(x).
\end{equation}
The polynomial $h_{\text{dom}}$ has the claimed behavior as long as
\begin{equation}
    \epsilon_{\text{bnd}},\epsilon_{\text{rec}}=\mathbf{O}\left(\frac{\epsilon_{\text{dom}}}{\norm{f}_{\max,\mathbf{Ellipse}(\rho=1+\alpha)}}\right).
\end{equation}
\end{proof}

We now consider other geometric objects which can simplify the reasoning about the Bernstein ellipse. We first consider an annulus centered at $c$ with inner and outer radius $0\leq\ell_0\leq\ell_1\leq\infty$:
\begin{equation}
    \mathbf{Disk}\left(c;\left[\ell_0,\ell_1\right]\right)
    =\left\{\ell_0\leq\abs{z-c}\leq\ell_1\right\},
\end{equation}
which reduces to a circle when $\ell_0=\ell_1$ and an actual disk when $\ell_0=0$. These annuli satisfy the inversion rule
\begin{equation}
    \frac{1}{\mathbf{Disk}\left(0;\left[\ell_0,\ell_1\right]\right)}
    =\mathbf{Disk}\left(0;\left[\frac{1}{\ell_1},\frac{1}{\ell_0}\right]\right)
\end{equation}
assuming $\ell_0>0$.

We then consider stadiums~\cite{demanet2010chebyshev}.
A stadium encloses a region that is the Minkowski sum of a disk and a line segment. It can be equivalently characterized as the set of points having the same distance from a line segment. We will restrict to stadiums centered on the real axis, writing $\mathbf{Stadium}\left([c_0,c_1];\ell\right)$ to denote a stadium with centers $c_0$ and $c_1$ and radius $\ell>0$ of the semicircles. Explicitly,
\begin{equation}
\begin{aligned}
    &\mathbf{Stadium}\left([c_0,c_1];\ell\right)\\
    &=\left\{\Im(z)=\pm\ell,c_0\leq \Re(z)\leq c_1\right\}
    \bigcup\left\{\abs{z-c_0}=\ell,\Re(z)<c_0\right\}
    \bigcup\left\{\abs{z-c_1}=\ell,\Re(z)>c_1\right\}.
\end{aligned}
\end{equation}
Whenever the context is clear, we will use $\mathbf{Stadium}(\cdot)$ to also denote the region enclosed by the stadium curve.
These stadiums satisfy the rescaling rule
\begin{equation}
    \alpha\cdot\mathbf{Stadium}\left([-c,c];\ell\right)
    =\mathbf{Stadium}\left([-\alpha c,\alpha c];\alpha\ell\right)
\end{equation}
for $c,\ell,\alpha>0$.
The following lemma establishes a containment relation between Bernstein ellipses and stadiums. 
It follows by comparing derivatives of the ellipse and semicircles, but we omit the details as the calculation is elementary and tedious.
See~\fig{stadium} for a pictorial illustration of this relation.

\begin{lemma}
For any $\ell>0$,
\begin{equation}
    \mathbf{Ellipse}\left(b=\ell\right)
    \subseteq\mathbf{Stadium}\left([-1,1];\ell\right)
    \subseteq\mathbf{Ellipse}\left(a=1+\ell\right).
\end{equation}
\end{lemma}

\begin{figure}[t]
	\centering
\includegraphics[scale=0.8]{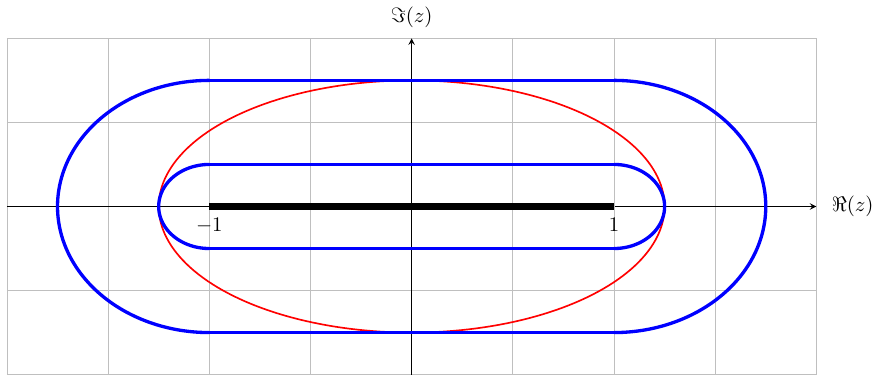}
\caption{Inscribing and circumscribing Bernstein ellipse with stadiums.}
\label{fig:stadium}
\end{figure}

\subsection{Dominated polynomial approximation for Hamiltonian QSVT}
\label{append:composite_qsvt}
Recall from~\sec{qsvt_dominate} that in Hamiltonian QSVT, we are given a real odd polynomial $f$ of degree $d$. For any $\epsilon>0$ and constant $0<\xi\leq\frac{\pi}{2}$, our goal is to find a real odd polynomial $p$ and even polynomial $q$ that solves the following dominated approximation problem:
\begin{equation}
\begin{aligned}
    &\abs{p(x)-\sin(f(\arcsin(x)))}\leq\epsilon,\qquad&&\forall x\in\left[-\sin\left(\frac{\pi}{2}-\xi\right),\sin\left(\frac{\pi}{2}-\xi\right)\right],\\
    &\abs{q(x)-\frac{\cos(f(\arcsin(x)))}{\sqrt{1-x^2}}}\leq\epsilon,\qquad&&\forall x\in\left[-\sin\left(\frac{\pi}{2}-\xi\right),\sin\left(\frac{\pi}{2}-\xi\right)\right],\\
    &p^2(x)+(1-x^2)q^2(x)\leq1+\epsilon,\qquad&&\forall x\in[-1,1].
\end{aligned}
\end{equation}

To this end, we define $1-\xi_1=\sin\left(\frac{\pi}{2}-\xi\right)$ and consider approximating the rescaled function $f\left(\arcsin\left(\left(1-\xi_1\right)y\right)\right)$ over the unit interval $y\in[-1,1]$. Since $\xi$ is constant, we can relax the Bernstein ellipse to a stadium and consider the behavior of this composite function over stadiums surrounding the unit interval.

\begin{proposition}
\label{prop:arcsin_stadium}
For any $0<\ell\leq\xi_1\leq 1$,
\begin{equation}
    \arcsin\left(\left(1-\xi_1\right)
    \mathbf{Stadium}\left([-1,1];\ell\right)
    \right)
    \subseteq\mathbf{Stadium}\left(\left[-\arcsin\left(1-\xi_1\right),\arcsin\left(1-\xi_1\right)\right];
    \ell_1\right),
\end{equation}
with $\ell_1
    =\frac{\left(1-\xi_1\right)}{\xi_1}\ell$.
\end{proposition}
\begin{proof}
It is clear that the unit interval $[-1,1]$ is mapped to $\left[-\arcsin\left(1-\xi_1\right),\arcsin\left(1-\xi_1\right)\right]$ under the function $\arcsin\left(\left(1-\xi_1\right)(\cdot)\right)$. So we only need to examine points within the stadium that are outside the unit interval.

To this end, take any complex $w\in\mathbf{Stadium}\left([-1,1];\ell\right)$ and find the closest real $y\in[-1,1]$, such that $\abs{w-y}\leq\ell$ by definition.
Then,
\begin{equation}
\begin{aligned}
    \abs{\arcsin\left(\left(1-\xi_1\right)w\right)
    -\arcsin\left(\left(1-\xi_1\right)y\right)}
    &\leq\frac{\abs{1-\xi_1}}{\sqrt{1-\left(1-\xi_1\right)^2(1+\ell)^2}}\ell
    \leq\frac{\abs{1-\xi_1}}{\sqrt{1-\left(1-\xi_1\right)^2\left(1+\xi_1\right)^2}}\ell\\
    &=\frac{\abs{1-\xi_1}}{\sqrt{2\xi_1^2-\xi_1^4}}\ell
    \leq\frac{\left(1-\xi_1\right)}{\xi_1}\ell,
\end{aligned}
\end{equation}
where we have differentiated $\arcsin$ in the first inequality and applied the assumption $\ell\leq\xi_1$ in the second inequality.
\end{proof}

Now, let us consider the action of $f$ on the stadium.
\begin{proposition}
\label{prop:poly_stadium}
Let $f$ be a real polynomial of degree $d$. For any $\ell_1>0$,
\begin{equation}
    \norm{f}_{\max,\mathbf{Stadium}(\left[-\frac{\pi}{2},\frac{\pi}{2}\right];\ell_1)}
    \leq\cosh\left(\sqrt{2\ell_1}d\right)\norm{f}_{\max,\left[-\frac{\pi}{2},\frac{\pi}{2}\right]}
    \leq\exp\left(\sqrt{2\ell_1}d\right)\norm{f}_{\max,\left[-\frac{\pi}{2},\frac{\pi}{2}\right]}.
\end{equation}
\end{proposition}
\begin{proof}
By Markov brothers' inequality~\cite[Eq.\ (33)]{Kalmykov21}, we can bound higher-order derivatives of $f$ as
\begin{equation}
    \norm{f^{(j)}}_{\max,\left[-\frac{\pi}{2},\frac{\pi}{2}\right]}
    \leq\frac{d^{2j}}{(2j-1)!!}\norm{f}_{\max,\left[-\frac{\pi}{2},\frac{\pi}{2}\right]}.
\end{equation}
Now take any complex $z\in\mathbf{Stadium}(\left[-\frac{\pi}{2},\frac{\pi}{2}\right];\ell_1)$ and find the closest real $x\in\left[-\frac{\pi}{2},\frac{\pi}{2}\right]$, such that $\abs{z-x}\leq\ell_1$. Applying Taylor's theorem,
\begin{equation}
\begin{aligned}
    \abs{f(z)}
    &=\abs{\sum_{j=0}^d\frac{f^{(j)}(x)}{j!}(z-x)^j}
    \leq\sum_{j=0}^d\frac{\abs{f^{(j)}(x)}}{j!}\ell_1^j
    \leq\sum_{j=0}^d\frac{d^{2j}}{(2j-1)!!}\frac{1}{j!}\ell_1^j\norm{f}_{\max,\left[-\frac{\pi}{2},\frac{\pi}{2}\right]}\\
    &=\sum_{j=0}^d\frac{d^{2j}2^jj!}{(2j)!j!}\ell_1^j\norm{f}_{\max,\left[-\frac{\pi}{2},\frac{\pi}{2}\right]}
    =\sum_{j=0}^d\frac{d^{2j}\sqrt{2\ell_1}^{2j}}{(2j)!}\norm{f}_{\max,\left[-\frac{\pi}{2},\frac{\pi}{2}\right]}\\
    &\leq\cosh\left(\sqrt{2\ell_1}d\right)\norm{f}_{\max,\left[-\frac{\pi}{2},\frac{\pi}{2}\right]}.
\end{aligned}
\end{equation}
\end{proof}

Summarizing the above discussion, we obtain
\begin{equation}
\begin{aligned}
    &\norm{f\left(\arcsin\left(\left(1-\xi_1\right)(\cdot)\right)\right)}_{\max,\mathbf{Ellipse}\left(\rho=1+\ell\right)}\\
    &\leq\norm{f\left(\arcsin\left(\left(1-\xi_1\right)(\cdot)\right)\right)}_{\max,\mathbf{Ellipse}\left(b=\ell\right)}\\
    &\leq\norm{f\left(\arcsin\left(\left(1-\xi_1\right)(\cdot)\right)\right)}_{\max,\mathbf{Stadium}\left([-1,1];\ell\right)}\\
    &\leq\norm{f}_{\max,\mathbf{Stadium}\left(\left[-\arcsin\left(1-\xi_1\right),\arcsin\left(1-\xi_1\right)\right];
    \ell_1\right)}\\
    &\leq\norm{f}_{\max,\mathbf{Stadium}\left(\left[-\frac{\pi}{2},\frac{\pi}{2}\right];
    \ell_1\right)}
    \leq\exp\left(\sqrt{2\ell_1}d\right)\norm{f}_{\max,\left[-\frac{\pi}{2},\frac{\pi}{2}\right]}.
\end{aligned}
\end{equation}
Now suppose we start with constant $0<\xi<\xi_2\leq\frac{\pi}{2}$, which gives $1>1-\xi_1=\sin\left(\frac{\pi}{2}-\xi\right)>1-\xi_3=\sin\left(\frac{\pi}{2}-\xi_2\right)\geq0$. By rescaling~\cor{approx_dominated_ellipse}, we obtain a dominated polynomial approximation $h_{f\arcsin,\text{dom}}$, which is a real odd polynomial with the behavior
\begin{equation}
\begin{aligned}
    &\abs{h_{f\arcsin,\text{dom}}(x)-f(\arcsin(x))}\leq\epsilon_{\text{dom}},\qquad&&\forall x\in\left[-\sin\left(\frac{\pi}{2}-\xi_2\right),\sin\left(\frac{\pi}{2}-\xi_2\right)\right],\\
    &\abs{h_{f\arcsin,\text{dom}}(x)}\leq\abs{f(\arcsin(x))}+\epsilon_{\text{dom}},\qquad&&\forall x\in\left[-\sin\left(\frac{\pi}{2}-\xi\right),\sin\left(\frac{\pi}{2}-\xi\right)\right],\\
    &\abs{h_{f\arcsin,\text{dom}}(x)}\leq\epsilon_{\text{dom}},\qquad&&\forall x\in\left[-1,-\sin\left(\frac{\pi}{2}-\xi\right)\right]\bigcup\left[\sin\left(\frac{\pi}{2}-\xi\right),1\right],
\end{aligned}
\end{equation}
and an asymptotic degree of
\begin{equation}
    \mathbf{O}\left(d+\log\left(\frac{\norm{f}_{\max,\left[-\frac{\pi}{2},\frac{\pi}{2}\right]}}{\epsilon_{\text{dom}}}\right)\right).
\end{equation}
That is, the polynomial $h_{f\arcsin,\text{dom}}(x)$ approximates $f(\arcsin(x))$ over a slightly smaller interval $\left[-\sin\left(\frac{\pi}{2}-\xi_2\right),\sin\left(\frac{\pi}{2}-\xi_2\right)\right]$, while its absolute value is bounded by $\abs{f}+\epsilon$ across the entire unit interval $x\in[-1,1]$.

The remaining proof proceeds as that of~\prop{dominated}, except we replace the composite function $f(\arcsin(\cdot))$ by our newly constructed polynomial $h_{f\arcsin,\text{dom}}$. In particular, we construct $h_{\text{inv-sqrt}}$, $h_{\sin,\alpha}$ and $h_{\cos,\alpha}$ as before and define
\begin{equation}
    p(x)=h_{\sin,\alpha}\left(\frac{1}{\alpha}h_{f\arcsin,\text{dom}}(x)\right),\qquad
    q(x)=h_{\text{inv-sqrt}}(x)h_{\cos,\alpha}\left(\frac{1}{\alpha}h_{f\arcsin,\text{dom}}(x)\right),
\end{equation}
where $\alpha=\norm{f}_{\max,\left[-\frac{\pi}{2},\frac{\pi}{2}\right]}+\epsilon_{\text{dom}}=\mathbf{O}\left(\norm{f}_{\max,\left[-\frac{\pi}{2},\frac{\pi}{2}\right]}\right)$.
Then the dominated condition is verified in the same way as before.

For the approximation condition, we assume $x\in\left[-\sin\left(\frac{\pi}{2}-\xi_2\right),\sin\left(\frac{\pi}{2}-\xi_2\right)\right]$ and compute
\begin{equation}
\begin{aligned}
    &\abs{p(x)-\sin(f(\arcsin(x)))}\\
    &=\abs{h_{\sin,\alpha}\left(\frac{1}{\alpha}h_{f\arcsin,\text{dom}}(x)\right)-\sin(f(\arcsin(x)))}\\
    &\leq\abs{h_{\sin,\alpha}\left(\frac{1}{\alpha}h_{f\arcsin,\text{dom}}(x)\right)
    -\sin(h_{f\arcsin,\text{dom}}(x))}
    +\abs{\sin\left(h_{f\arcsin,\text{dom}}(x)\right)-\sin(f(\arcsin(x)))}\\
    &\leq\epsilon_{\text{trig}}+\abs{h_{f\arcsin,\text{dom}}(x)-f(\arcsin(x))}
    \leq\epsilon_{\text{trig}}+\epsilon_{\text{dom}}.
\end{aligned}
\end{equation}
Similarly,
\begin{equation}
\begin{aligned}
    &\abs{q(x)-\frac{1}{\sqrt{1-x^2}}\cos(f(\arcsin(x)))}\\
    &=\abs{h_{\text{inv-sqrt}}(x)h_{\cos,\alpha}\left(\frac{1}{\alpha}h_{f\arcsin,\text{dom}}(x)\right)-\frac{1}{\sqrt{1-x^2}}\cos(f(\arcsin(x)))}\\
    &\leq\abs{h_{\text{inv-sqrt}}(x)h_{\cos,\alpha}\left(\frac{1}{\alpha}h_{f\arcsin,\text{dom}}(x)\right)-\frac{1}{\sqrt{1-x^2}}h_{\cos,\alpha}\left(\frac{1}{\alpha}h_{f\arcsin,\text{dom}}(x)\right)}\\
    &\quad+\abs{\frac{1}{\sqrt{1-x^2}}h_{\cos,\alpha}\left(\frac{1}{\alpha}h_{f\arcsin,\text{dom}}(x)\right)-\frac{1}{\sqrt{1-x^2}}\cos(f(\arcsin(x)))}\\
    &\leq\abs{h_{\text{inv-sqrt}}(x)-\frac{1}{\sqrt{1-x^2}}}\left(1+\epsilon_{\text{trig}}\right)
    +\mathbf{O}\left(\epsilon_{\text{trig}}+\epsilon_{\text{dom}}\right)
    =\mathbf{O}\left(\epsilon_{\text{inv-sqrt}}+\epsilon_{\text{trig}}+\epsilon_{\text{dom}}\right).
\end{aligned}
\end{equation}
Hence, the approximation condition is satisfied as long as $\epsilon_{\text{inv-sqrt}},\epsilon_{\text{trig}},\epsilon_{\text{dom}}=\mathbf{O}(\epsilon)$.

To summarize, $p$ is a real odd polynomial of degree
\begin{equation}
\begin{aligned}
    &\mathbf{O}\left(\left(d+\log\left(\frac{\norm{f}_{\max,\left[-\frac{\pi}{2},\frac{\pi}{2}\right]}}{\epsilon_{\text{dom}}}\right)\right)
    \left(\alpha+\log\left(\frac{1}{\epsilon_{\text{trig}}}\right)\right)\right)\\
    &=\mathbf{O}\left(\left(d+\log\left(\frac{\norm{f}_{\max,\left[-\frac{\pi}{2},\frac{\pi}{2}\right]}}{\epsilon}\right)\right)
    \left(\norm{f}_{\max,\left[-\frac{\pi}{2},\frac{\pi}{2}\right]}+\log\left(\frac{1}{\epsilon}\right)\right)
    \right),
\end{aligned}
\end{equation}
whereas $q$ is a real even polynomial of degree
\begin{equation}
\begin{aligned}
    &\mathbf{O}\left(\left(d+\log\left(\frac{\norm{f}_{\max,\left[-\frac{\pi}{2},\frac{\pi}{2}\right]}}{\epsilon_{\text{dom}}}\right)\right)
    \left(\alpha+\log\left(\frac{1}{\epsilon_{\text{trig}}}\right)\right)
    +\log\left(\frac{1}{\epsilon_{\text{inv-sqrt}}}\right)\right)\\
    &=\mathbf{O}\left(\left(d+\log\left(\frac{\norm{f}_{\max,\left[-\frac{\pi}{2},\frac{\pi}{2}\right]}}{\epsilon}\right)\right)
    \left(\norm{f}_{\max,\left[-\frac{\pi}{2},\frac{\pi}{2}\right]}+\log\left(\frac{1}{\epsilon}\right)\right)
    \right).
\end{aligned}
\end{equation}
This completes the proof of~\prop{dominated} on the existence of dominated polynomial approximation with the claimed asymptotic degree.

\subsection{Dominated polynomial approximation for Hamiltonian matrix inversion}
\label{append:composite_inverse}
Recall from~\sec{qsvt_inverse_frac} that in Hamiltonian-based matrix inversion, we aim to implement a normalized inverse function. Specifically, for any $\epsilon>0$, $\kappa>0$ and constant $0<\xi\leq\frac{\pi}{2}$, our goal is to find a real odd polynomial $p$ and even polynomial $q$ such that 
\begin{equation}
\begin{aligned}
    &\abs{p(x)-\sin\left(\frac{1}{\kappa\arcsin(x)}\right)}\leq\epsilon,\quad&&\forall x\in\left[-\sin\left(\frac{\pi}{2}-\xi\right),-\sin\left(\frac{1}{\kappa}\right)\right]\bigcup\left[\sin\left(\frac{1}{\kappa}\right),\sin\left(\frac{\pi}{2}-\xi\right)\right],\\
    &\abs{q(x)-\frac{\cos\left(\frac{1}{\kappa\arcsin(x)}\right)}{\sqrt{1-x^2}}}\leq\epsilon,\quad&&\forall x\in\left[-\sin\left(\frac{\pi}{2}-\xi\right),-\sin\left(\frac{1}{\kappa}\right)\right]\bigcup\left[\sin\left(\frac{1}{\kappa}\right),\sin\left(\frac{\pi}{2}-\xi\right)\right],\\
    &p^2(x)+(1-x^2)q^2(x)\leq1+\epsilon,\quad&&\forall x\in[-1,1].
\end{aligned}
\end{equation}

Without loss of generality, let us focus on $\sin\left(\frac{1}{\kappa\arcsin(x)}\right)$. Define
\begin{equation}
    \mathbf{Inv}_+(x)=
    \begin{cases}
        \frac{1}{x},\quad&x>0,\\
        0,\quad&x\leq0,
    \end{cases}
    \qquad
    \mathbf{Inv}_-(x)=
    \begin{cases}
        \frac{1}{x},\quad&x<0,\\
        0,\quad&x\geq0,
    \end{cases}
\end{equation}
to be the positive and negative branches of the inversion function. Then, we have
\begin{equation}
\begin{aligned}
    \sin\left(\frac{1}{\kappa\arcsin(x)}\right)
    &=\sin\left(\frac{\mathbf{Inv}_+\left(\arcsin(x)\right)}{\kappa}
    +\frac{\mathbf{Inv}_-\left(\arcsin(x)\right)}{\kappa}\right)\\
    &=\sin\left(\frac{\mathbf{Inv}_+\left(\arcsin(x)\right)}{\kappa}\right)
    \cos\left(\frac{\mathbf{Inv}_-\left(\arcsin(x)\right)}{\kappa}\right)\\
    &\quad+\cos\left(\frac{\mathbf{Inv}_+\left(\arcsin(x)\right)}{\kappa}\right)
    \sin\left(\frac{\mathbf{Inv}_-\left(\arcsin(x)\right)}{\kappa}\right)
\end{aligned}
\end{equation}
for $x\in\left[-\sin\left(\frac{\pi}{2}-\xi\right),-\sin\left(\frac{1}{\kappa}\right)\right]\bigcup\left[\sin\left(\frac{1}{\kappa}\right),\sin\left(\frac{\pi}{2}-\xi\right)\right]$. Without loss of generality, we will focus on $\cos\left(\frac{\mathbf{Inv}_+\left(\arcsin(x)\right)}{\kappa}\right)$.

Let us first consider approximating the function $\cos\left(\frac{\mathbf{Inv}_+\left(\arcsin(x)\right)}{\kappa}\right)=\cos\left(\frac{1}{\kappa\arcsin(x)}\right)$ only for $x\in\left[\sin\left(\frac{1}{\kappa}\right),\sin\left(\frac{\pi}{2}-\xi\right)\right]$. 
We apply the affine mapping
\begin{equation}
    x=\sin\left(\frac{1}{\kappa}\right)
    +\frac{y+1}{2}\left(\sin\left(\frac{\pi}{2}-\xi\right)-\sin\left(\frac{1}{\kappa}\right)\right),
\end{equation}
which gives the function
\begin{equation}
    \cos\left(\frac{1}{\kappa\arcsin\left(
    \sin\left(\frac{1}{\kappa}\right)
    +\frac{y+1}{2}\left(\sin\left(\frac{\pi}{2}-\xi\right)-\sin\left(\frac{1}{\kappa}\right)\right)\right)}\right)
\end{equation}
moving the unit interval $y\in[-1,1]$ to $x\in\left[\sin\left(\frac{1}{\kappa}\right),\sin\left(\frac{\pi}{2}-\xi\right)\right]$.

Suppose we start with the Bernstein ellipse $\mathbf{Ellipse}\left(\rho=1+\sqrt{\frac{c}{\kappa}}\right)$ with a tunable constant $c>0$ sufficiently small, which can be upper bounded by
\begin{equation}
\begin{aligned}
    \mathbf{Ellipse}\left(\rho=1+\sqrt{\frac{c}{\kappa}}\right)
    \subseteq\mathbf{Ellipse}\left([-1,1];a=1+\frac{c}{2\kappa},b=\sqrt{\frac{c}{\kappa}}\right).
\end{aligned}
\end{equation}
Then the affine mapping transforms it into
\begin{small}
\begin{equation}
\newmaketag
\begin{aligned}
    &\mathbf{Ellipse}\left(\left[\sin\left(\frac{1}{\kappa}\right),\sin\left(\frac{\pi}{2}-\xi\right)\right];
    a=\left(1+\frac{c}{2\kappa}\right)\frac{\sin\left(\frac{\pi}{2}-\xi\right)-\sin\left(\frac{1}{\kappa}\right)}{2},
    b=\sqrt{\frac{c}{\kappa}}\frac{\sin\left(\frac{\pi}{2}-\xi\right)-\sin\left(\frac{1}{\kappa}\right)}{2}\right)\\
    &=\left\{\frac{\left(\Re(z)-\frac{\sin\left(\frac{\pi}{2}-\xi\right)+\sin\left(\frac{1}{\kappa}\right)}{2}\right)^2}{\left(\left(1+\frac{c}{2\kappa}\right)\frac{\sin\left(\frac{\pi}{2}-\xi\right)-\sin\left(\frac{1}{\kappa}\right)}{2}\right)^2}
    +\frac{\Im^2(z)}{\frac{c}{\kappa}\left(\frac{\sin\left(\frac{\pi}{2}-\xi\right)-\sin\left(\frac{1}{\kappa}\right)}{2}\right)^2}
    =1\right\}.
\end{aligned}
\end{equation}
\end{small}%
Our main effort is devoted to analyzing the behavior of this ellipse under the action of 
$\arcsin$.
\begin{proposition}
\begin{footnotesize}
\begin{equation}
\newmaketag
\begin{aligned}
    &\arcsin\left(\mathbf{Ellipse}\left(\left[\sin\left(\frac{1}{\kappa}\right),\sin\left(\frac{\pi}{2}-\xi\right)\right];
    a=\left(1+\frac{c}{2\kappa}\right)\frac{\sin\left(\frac{\pi}{2}-\xi\right)-\sin\left(\frac{1}{\kappa}\right)}{2},
    b=\sqrt{\frac{c}{\kappa}}\frac{\sin\left(\frac{\pi}{2}-\xi\right)-\sin\left(\frac{1}{\kappa}\right)}{2}\right)\right)\\
    &\subseteq\mathbf{Disk}\left(0;\left[\mathbf{\Omega}\left(\frac{1}{\kappa}\right),\frac{\pi}{2}-\xi+\mathbf{O}\left(\frac{1}{\kappa}\right)\right]\right),
\end{aligned}
\end{equation}
\end{footnotesize}%
where $\mathbf{\Omega}$ is taken in the limit $c\rightarrow0$ and $\kappa\rightarrow\infty$.
\end{proposition}
\begin{proof}
To understand the asymptotic scaling, we consider the limit $c\rightarrow0$ and $\kappa\rightarrow\infty$, but assume $\xi$ is constant.
The left extreme point of the ellipse has $x$-coordinate
\begin{equation}
    x_{\text{left}}
    =-\left(1+\frac{c}{2\kappa}\right)\frac{\sin\left(\frac{\pi}{2}-\xi\right)-\sin\left(\frac{1}{\kappa}\right)}{2}
    +\frac{\sin\left(\frac{\pi}{2}-\xi\right)+\sin\left(\frac{1}{\kappa}\right)}{2}
    \rightarrow\sin\left(\frac{1}{\kappa}\right)
\end{equation}
and the right extreme point has $x$-coordinate
\begin{equation}
    x_{\text{right}}
    =\left(1+\frac{c}{2\kappa}\right)\frac{\sin\left(\frac{\pi}{2}-\xi\right)-\sin\left(\frac{1}{\kappa}\right)}{2}
    +\frac{\sin\left(\frac{\pi}{2}-\xi\right)+\sin\left(\frac{1}{\kappa}\right)}{2}
    \rightarrow\sin\left(\frac{\pi}{2}-\xi\right)
\end{equation}
in the limit $c\rightarrow0$ as expected. Hence, the ellipse is contained in $\mathbf{Disk}\left(0;\left[0,\sin\left(\frac{\pi}{2}-\xi+\mathbf{O}\left(\frac{1}{\kappa}\right)\right)\right]\right)$, so applying $\arcsin$ maps it into $\mathbf{Disk}\left(0;\left[0,\frac{\pi}{2}-\xi+\mathbf{O}\left(\frac{1}{\kappa}\right)\right]\right)$.

It remains to show that the region is actually an annulus gapped away from $0$ by $\mathbf{\Omega}\left(\frac{1}{\kappa}\right)$. To this end, let us increase $x_{\text{left}}$ by $\delta_x>0$, and solve for $\delta_y>0$ so that the new point $z=x_{\text{left}}+\delta_x+i\delta_y$ is on the ellipse:
\begin{equation}
    \delta_y=\sqrt{\frac{c}{\kappa}}
    \sqrt{\frac{2\delta_x\frac{\sin\left(\frac{\pi}{2}-\xi\right)-\sin\left(\frac{1}{\kappa}\right)}{2}}{1+\frac{c}{2\kappa}}
    -\frac{\delta_x^2}{\left(1+\frac{c}{2\kappa}\right)^2}}
    \leq\sqrt{\frac{c}{\kappa}}\sqrt{\frac{\delta_x\sin\left(\frac{\pi}{2}-\xi\right)}{1+\frac{c}{2\kappa}}}.
\end{equation}
This implies
\begin{equation}
\begin{aligned}
    \abs{\arcsin(z)-\arcsin\left(x_{\text{left}}+\delta_x\right)}
    &=\abs{\int_{x_{\text{left}}+\delta_x}^z\mathrm{d}u\ \frac{1}{\sqrt{1-u^2}}}
    \leq\int_{x_{\text{left}}+\delta_x}^z\abs{\mathrm{d}u}\frac{1}{\sqrt{1-\abs{u}^2}}\\
    &\leq\delta_y\frac{1}{\sqrt{1-\left(x_{\text{left}}+\delta_x\right)^2-\delta_y^2}}
    \rightarrow\delta_y\frac{1}{\sqrt{1-\left(\sin\left(\frac{1}{\kappa}\right)+\delta_x\right)^2}}
\end{aligned}
\end{equation}
for $c$ sufficiently small, and hence for a small value of $c$ and large value of $\kappa$
\begin{equation}
\begin{aligned}
    \abs{\arcsin(z)}
    &\geq\arcsin\left(x_{\text{left}}+\delta_x\right)-\delta_y\frac{1}{\sqrt{1-\left(\sin\left(\frac{1}{\kappa}\right)+\delta_x\right)^2}}\\
    &\geq x_{\text{left}}+\delta_x-\delta_y\frac{1}{\sqrt{1-\left(\sin\left(\frac{1}{\kappa}\right)+\delta_x\right)^2}}\\
    &\rightarrow\sin\left(\frac{1}{\kappa}\right)+\delta_x
    -\sqrt{\frac{c}{\kappa}}\sqrt{\delta_x\sin\left(\frac{\pi}{2}-\xi\right)}
    \frac{1}{\sqrt{1-\left(\sin\left(\frac{1}{\kappa}\right)+\delta_x\right)^2}}\\
    &\rightarrow\sin\left(\frac{1}{\kappa}\right)+\delta_x
    -\sqrt{\frac{c}{\kappa}}\sqrt{\delta_x\sin\left(\frac{\pi}{2}-\xi\right)}
    \frac{1}{\sqrt{1-\delta_x^2}}\\
    &\geq\sin\left(\frac{1}{\kappa}\right)+\delta_x
    -\sqrt{\frac{c}{\kappa}}\sqrt{\delta_x\sin\left(\frac{\pi}{2}-\xi\right)}
    \frac{1}{\sqrt{1-\sin^2\left(\frac{\pi}{2}-\xi\right)}}.\\
\end{aligned}
\end{equation}
This is quadratic in $\sqrt{\delta_x}$, and is thus minimized at $\delta_x=\mathbf{\Theta}\left(\frac{1}{\kappa}\right)$. This means
\begin{equation}
    \abs{\arcsin(z)}=\mathbf{\Omega}\left(\frac{1}{\kappa}\right)
\end{equation}
assuming $0<\xi\leq\frac{\pi}{2}$ is constant, $c$ is sufficiently small and $\kappa$ is sufficiently large.
\end{proof}

We now apply the normalized inverse function and $\cos$, obtaining
\begin{equation}
\begin{aligned}
    &\norm{\cos\left(\frac{1}{\kappa\arcsin\left(
    \sin\left(\frac{1}{\kappa}\right)
    +\frac{(\cdot)+1}{2}\left(\sin\left(\frac{\pi}{2}-\xi\right)-\sin\left(\frac{1}{\kappa}\right)\right)\right)}\right)}_{\max,\mathbf{Ellipse}\left(\rho=1+\sqrt{\frac{c}{\kappa}}\right)}\\
    &\leq\norm{\cos\left(\frac{1}{\kappa\arcsin\left(
    \sin\left(\frac{1}{\kappa}\right)
    +\frac{(\cdot)+1}{2}\left(\sin\left(\frac{\pi}{2}-\xi\right)-\sin\left(\frac{1}{\kappa}\right)\right)\right)}\right)}_{\max,\mathbf{Ellipse}\left([-1,1];a=1+\frac{c}{2\kappa},b=\sqrt{\frac{c}{\kappa}}\right)}\\
    &\leq\norm{\cos\left(\frac{1}{\kappa\arcsin(\cdot)}\right)}_{\max,\mathbf{Ellipse}\left(\left[\sin\left(\frac{1}{\kappa}\right),\sin\left(\frac{\pi}{2}-\xi\right)\right];
    a=\left(1+\frac{c}{2\kappa}\right)\frac{\sin\left(\frac{\pi}{2}-\xi\right)-\sin\left(\frac{1}{\kappa}\right)}{2},
    b=\sqrt{\frac{c}{\kappa}}\frac{\sin\left(\frac{\pi}{2}-\xi\right)-\sin\left(\frac{1}{\kappa}\right)}{2}\right)}\\
    &\leq\norm{\cos\left(\frac{1}{\kappa(\cdot)}\right)}_{\max,\mathbf{Disk}\left(0;\left[\mathbf{\Omega}\left(\frac{1}{\kappa}\right),\frac{\pi}{2}-\xi+\mathbf{O}\left(\frac{1}{\kappa}\right)\right]\right)}
    \leq\norm{\cos}_{\max,\mathbf{Disk}\left(0;\left[0,\mathbf{O}(1)\right]\right)}
    =\exp\left(\mathbf{O}(1)\right)
    =\mathbf{O}(1).
\end{aligned}
\end{equation}
Now suppose we start with constant $0<\xi<\xi_2\leq\frac{\pi}{2}$. By rescaling~\cor{approx_dominated_ellipse}, we obtain a dominated polynomial approximation $h_{\cos\mathbf{Inv}_+,\text{dom}}$, which is a real polynomial with the behavior
\begin{equation}
\begin{aligned}
    &\abs{h_{\cos\mathbf{Inv}_+,\text{dom}}(x)-\cos\left(\frac{\mathbf{Inv}_+\left(\arcsin(x)\right)}{\kappa}\right)}\leq\epsilon_{\text{dom}},\qquad&&\forall x\in\left[\sin\left(\frac{1}{\kappa}\right),\sin\left(\frac{\pi}{2}-\xi_2\right)\right],\\
    &\abs{h_{\cos\mathbf{Inv}_+,\text{dom}}(x)}\leq\abs{\cos\left(\frac{\mathbf{Inv}_+\left(\arcsin(x)\right)}{\kappa}\right)}+\epsilon_{\text{dom}},\qquad&&\forall x\in\left[0,\sin\left(\frac{\pi}{2}-\xi\right)\right],\\
    &\abs{h_{\cos\mathbf{Inv}_+,\text{dom}}(x)}\leq\epsilon_{\text{dom}},\qquad&&\forall x\in\left[-1,0\right]\bigcup\left[\sin\left(\frac{\pi}{2}-\xi\right),1\right],
\end{aligned}
\end{equation}
and an asymptotic degree of
\begin{equation}
    \mathbf{O}\left(\kappa\log\left(\frac{\kappa}{\epsilon_{\text{dom}}}\right)\right).
\end{equation}

To handle $x<0$, we use~\cor{approx_dominated_ellipse} or~\cite[Lemma 29]{Gilyen2018singular} to construct a dominated approximation of the constant function $h_{1,\text{dom}}$ with the behavior
\begin{equation}
\begin{aligned}
    &\abs{h_{1,\text{dom}}(x)-1}\leq\epsilon_{\text{dom}},\qquad&&\forall x\in\left[-\sin\left(\frac{\pi}{2}-\xi_2\right),-\sin\left(\frac{1}{\kappa}\right)\right],\\
    &\abs{h_{1,\text{dom}}(x)}\leq1+\epsilon_{\text{dom}},\qquad&&\forall x\in\left[-\sin\left(\frac{\pi}{2}-\xi\right),0\right],\\
    &\abs{h_{1,\text{dom}}(x)}\leq\epsilon_{\text{dom}},\qquad&&\forall x\in\left[-1,-\sin\left(\frac{\pi}{2}-\xi\right)\right]\bigcup[0,1],
\end{aligned}
\end{equation}
and an asymptotic degree of
\begin{equation}
    \mathbf{O}\left(\kappa\log\left(\frac{1}{\epsilon_{\text{dom}}}\right)\right).
\end{equation}
Now $h_{\cos\mathbf{Inv}_+,\text{dom}}+h_{1,\text{dom}}$ is the desired dominated polynomial approximation. For all $x\in\left[-\sin\left(\frac{\pi}{2}-\xi_2\right),-\sin\left(\frac{1}{\kappa}\right)\right]\bigcup\left[\sin\left(\frac{1}{\kappa}\right),\sin\left(\frac{\pi}{2}-\xi_2\right)\right]$,
\begin{equation}
    \abs{h_{\cos\mathbf{Inv}_+,\text{dom}}(x)+h_{1,\text{dom}}(x)-\cos\left(\frac{\mathbf{Inv}_+\left(\arcsin(x)\right)}{\kappa}\right)}\leq2\epsilon_{\text{dom}},
\end{equation}
whereas for $x\in\left[-\sin\left(\frac{\pi}{2}-\xi\right),\sin\left(\frac{\pi}{2}-\xi\right)\right]$,
\begin{equation}
    \abs{h_{\cos\mathbf{Inv}_+,\text{dom}}(x)+h_{1,\text{dom}}(x)}\leq\abs{\cos\left(\frac{\mathbf{Inv}_+\left(\arcsin(x)\right)}{\kappa}\right)}+2\epsilon_{\text{dom}},
\end{equation}
and for $x\in\left[-1,-\sin\left(\frac{\pi}{2}-\xi\right)\right]\bigcup\left[\sin\left(\frac{\pi}{2}-\xi\right),1\right]$,
\begin{equation}
    \abs{h_{\cos\mathbf{Inv}_+,\text{dom}}(x)+h_{1,\text{dom}}(x)}\leq2\epsilon_{\text{dom}}.
\end{equation}
This completes the analysis of $\cos\left(\frac{\mathbf{Inv}_+\left(\arcsin(x)\right)}{\kappa}\right)$. The remaining cases are treated similarly. Finally, to satisfy the odd and even parity constraint, we take the odd and even part of the constructed polynomial respectively.

\subsection{Dominated polynomial approximation for Hamiltonian fractional scaling}
\label{append:composite_frac}
Recall from~\sec{qsvt_inverse_frac} that in Hamiltonian fractional scaling, we aim to implement a scalar function. Specifically, for any $\epsilon>0$, $0<\tau<1$ and constant $0<\xi\leq\frac{\pi}{2}$, our goal is to find a real odd polynomial $p$ and even polynomial $q$ such that 
\begin{equation}
\begin{aligned}
    &\abs{p(x)-\sin\left(\tau\arcsin(x)\right)}\leq\epsilon,\qquad&&\forall x\in\left[-\sin\left(\frac{\pi}{2}-\xi\right),\sin\left(\frac{\pi}{2}-\xi\right)\right],\\
    &\abs{q(x)-\frac{\cos\left(\tau\arcsin(x)\right)}{\sqrt{1-x^2}}}\leq\epsilon,\qquad&&\forall x\in\left[-\sin\left(\frac{\pi}{2}-\xi\right),\sin\left(\frac{\pi}{2}-\xi\right)\right],\\
    &p^2(x)+(1-x^2)q^2(x)\leq1+\epsilon,\qquad&&\forall x\in[-1,1].
\end{aligned}
\end{equation}

Without loss of generality, let us consider $\sin\left(\tau\arcsin(x)\right)$. Setting $1-\xi_1=\sin\left(\frac{\pi}{2}-\xi\right)$, we have the rescaled function $\sin\left(\tau\arcsin((1-\xi_1)x)\right)$ over the unit interval. Suppose we start with the Bernstein ellipse $\mathbf{Ellipse}\left(\rho=1+\ell\right)$ for some $\ell>0$ sufficiently small. Then,
\begin{equation}
\begin{aligned}
    &\norm{\sin\left(\tau\arcsin((1-\xi_1)\cdot)\right)}_{\max,\mathbf{Ellipse}\left(\rho=1+\ell\right)}\\
    &\leq\norm{\sin\left(\tau\arcsin((1-\xi_1)\cdot)\right)}_{\max,\mathbf{Ellipse}\left(b=\ell\right)}\\
    &\leq\norm{\sin\left(\tau\arcsin((1-\xi_1)\cdot)\right)}_{\max,\mathbf{Stadium}\left([-1,1];\ell\right)}\\
    &\leq\norm{\sin\left(\tau(\cdot)\right)}_{\max,\mathbf{Stadium}\left(\left[-\arcsin\left(1-\xi_1\right),\arcsin\left(1-\xi_1\right)\right];
    \mathbf{O}(\ell)\right)}\\
    &\leq\norm{\sin\left(\tau(\cdot)\right)}_{\max,\mathbf{Stadium}\left(\left[-\frac{\pi}{2},\frac{\pi}{2}\right];
    \mathbf{O}(\ell)\right)}
    =\exp\left(\mathbf{O}(\ell)\right).\\
\end{aligned}
\end{equation}
Hence, the result is constant if $\ell=\mathbf{O}(1)$ is sufficiently small.
The remaining analysis follows from~\cor{approx_dominated_ellipse} similar to that of~\append{composite_qsvt}.

\subsection{Dominated polynomial approximation for Hamiltonian overlap estimation}
\label{append:composite_overlap}
Recall from~\sec{overlap_est} that in Hamiltonian-based overlap estimation, we aim to implement a scaled $\arcsin$ function. Specifically, for any $\epsilon>0$, and constant $0<\xi\leq1$, our goal is to find a real odd polynomial $p$ and even polynomial $q$ such that 
\begin{equation}
\begin{aligned}
    &\abs{p(x)-\sin\left(\frac{1}{2}\arcsin\arcsin(x)\right)}\leq\epsilon,\qquad&&\forall x\in\left[-\sin\left(1-\xi\right),\sin\left(1-\xi\right)\right],\\
    &\abs{q(x)-\frac{\cos\left(\frac{1}{2}\arcsin\arcsin(x)\right)}{\sqrt{1-x^2}}}\leq\epsilon,\qquad&&\forall x\in\left[-\sin\left(1-\xi\right),\sin\left(1-\xi\right)\right],\\
    &p^2(x)+(1-x^2)q^2(x)\leq1+\epsilon,\qquad&&\forall x\in[-1,1].
\end{aligned}
\end{equation}

Without loss of generality, let us consider $\sin\left(\frac{1}{2}\arcsin\arcsin(x)\right)$. Setting $1-\xi_1=\sin\left(1-\xi\right)$, we have the rescaled function $\sin\left(\frac{1}{2}\arcsin\arcsin((1-\xi_1)x)\right)$ over the unit interval. Suppose we start with the Bernstein ellipse $\mathbf{Ellipse}\left(\rho=1+\ell\right)$ for some $\ell>0$ sufficiently small. Then,
\begin{equation}
\begin{aligned}
    &\norm{\sin\left(\frac{1}{2}\arcsin\arcsin((1-\xi_1)(\cdot))\right)}_{\max,\mathbf{Ellipse}\left(\rho=1+\ell\right)}\\
    &\leq\norm{\sin\left(\frac{1}{2}\arcsin\arcsin((1-\xi_1)(\cdot))\right)}_{\max,\mathbf{Stadium}\left([-1,1];\ell\right)}\\
    &\leq\norm{\sin\left(\frac{1}{2}\arcsin(\cdot)\right)}_{\max,\mathbf{Stadium}\left([-1+\xi,1-\xi];\mathbf{O}(\ell)\right)}\\
    &=\norm{\sin\left(\frac{1}{2}\arcsin((1-\xi)(\cdot))\right)}_{\max,\mathbf{Stadium}\left([-1,1];\mathbf{O}(\ell)\right)}\\
    &\leq\norm{\sin\left(\frac{(\cdot)}{2}\right)}_{\max,\mathbf{Stadium}\left(\left[-\arcsin\left(1-\xi\right),\arcsin\left(1-\xi\right)\right];
    \mathbf{O}(\ell)\right)}\\
    &\leq\norm{\sin\left(\frac{(\cdot)}{2}\right)}_{\max,\mathbf{Stadium}\left(\left[-\frac{\pi}{2},\frac{\pi}{2}\right];
    \mathbf{O}(\ell)\right)}
    =\exp\left(\mathbf{O}(\ell)\right).\\
\end{aligned}
\end{equation}
Hence, the result is constant if $\ell=\mathbf{O}(1)$ is sufficiently small.
The remaining analysis follows from~\cor{approx_dominated_ellipse} similar to that of~\append{composite_qsvt}.

\subsection{Dominated polynomial approximation for Green's function estimation}
\label{append:composite_green}
Recall from~\sec{overlap_green} that in the Green's function estimation problem, we aim to implement a rational function. Let us consider the Hermitian component of the Green's function first. Specifically, for any $\epsilon>0$, $\eta>0$, and constant $0<\xi\leq\frac{\pi}{2}$, our goal is to find a real odd polynomial $p$ and even polynomial $q$ such that
\begin{equation}
\begin{aligned}
    &\abs{p(x)-\sin\left(\frac{1}{2}\arcsin\left(\frac{\eta\arcsin(x)}{\eta^2+\arcsin^2(x)}\right)\right)}\leq\epsilon,\quad&&\forall x\in\left[-\sin\left(\frac{\pi}{2}-\xi\right),\sin\left(\frac{\pi}{2}-\xi\right)\right],\\
    &\abs{q(x)-\frac{\cos\left(\frac{1}{2}\arcsin\left(\frac{\eta\arcsin(x)}{\eta^2+\arcsin^2(x)}\right)\right)}{\sqrt{1-x^2}}}\leq\epsilon,\quad&&\forall x\in\left[-\sin\left(\frac{\pi}{2}-\xi\right),\sin\left(\frac{\pi}{2}-\xi\right)\right],\\
    &p^2(x)+(1-x^2)q^2(x)\leq1+\epsilon,\quad&&\forall x\in[-1,1].
\end{aligned}
\end{equation}

Without loss of generality, let us consider $\sin\left(\frac{1}{2}\arcsin\left(\frac{\eta\arcsin(x)}{\eta^2+\arcsin^2(x)}\right)\right)$. Setting $1-\xi_1=\sin\left(1-\xi\right)$, we have the rescaled function $\sin\left(\frac{1}{2}\arcsin\left(\frac{\eta\arcsin((1-\xi_1)x)}{\eta^2+\arcsin^2((1-\xi_1)x)}\right)\right)$ over the unit interval. Suppose we start with the Bernstein ellipse $\mathbf{Ellipse}(\rho=1+c\eta)$ for some tunable constant $c>0$ sufficiently small. Then after the $\arcsin((1-\xi_1)(\cdot))$ map, we get
\begin{equation}
    \mathbf{Stadium}\left(\left[-\arcsin\left(1-\xi_1\right),\arcsin\left(1-\xi_1\right)\right];
    \frac{1-\xi_1}{\xi_1}c\eta\right).
\end{equation}
Let us choose $c$ so that $\frac{1-\xi_1}{\xi_1}c\leq\frac{1}{2}$. We now study how this region gets transformed under $\frac{\eta(\cdot)}{\eta^2+(\cdot)^2}$.

\begin{proposition}
For $\eta=\mathbf{O}(1)$ sufficiently small,
\begin{equation}
    \frac{\eta\cdot\mathbf{Stadium}\left(\left[-\frac{\pi}{2},\frac{\pi}{2}\right];
    \frac{\eta}{2}\right)}{\eta^2+\mathbf{Stadium}^2\left(\left[-\frac{\pi}{2},\frac{\pi}{2}\right];
    \frac{\eta}{2}\right)}
    \subseteq\mathbf{Disk}\left(0;\left[0,\frac{2}{3}\right]\right).
\end{equation}
\end{proposition}
\begin{proof}
Take any $z\in\mathbf{Stadium}\left(\left[-\frac{\pi}{2},\frac{\pi}{2}\right];
\frac{\eta}{2}\right)$ and consider the decomposition
\begin{equation}
    \abs{\frac{\eta z}{\eta^2+z^2}}
    =\frac{\eta \abs{z}}{\abs{z+i\eta}\abs{z-i\eta}}.
\end{equation}
If $z=x+iy$ with $x\in\left[-\frac{\pi}{2},\frac{\pi}{2}\right]$ and $y=\pm\frac{\eta}{2}$, then
\begin{equation}
    \frac{\eta \abs{z}}{\abs{z+i\eta}\abs{z-i\eta}}
    =\frac{\eta\sqrt{x^2+\frac{\eta^2}{4}}}{\sqrt{x^2+\frac{9\eta^2}{4}}\sqrt{x^2+\frac{\eta^2}{4}}}
    \leq\frac{2}{3}.
\end{equation}
On the other hand, if $z$ is on the semicircles, then $\abs{z}\leq\frac{\pi}{2}+\frac{\eta}{2}$ and $\abs{z\pm i\eta}\geq\sqrt{\frac{\pi^2}{4}+\frac{\eta^2}{4}}$, which gives
\begin{equation}
    \frac{\eta \abs{z}}{\abs{z+i\eta}\abs{z-i\eta}}
    \leq\frac{\eta\left(\frac{\pi}{2}+\frac{\eta}{2}\right)}{\frac{\pi^2}{4}+\frac{\eta^2}{4}}
    =\mathbf{O}\left(\eta\right),
\end{equation}
which is also smaller than $\frac{2}{3}$ for $\eta=\mathbf{O}(1)$ sufficiently small.
\end{proof}

Summarizing the above discussion, we obtain
\begin{equation}
\begin{aligned}
    &\norm{\sin\left(\frac{1}{2}\arcsin\left(\frac{\eta\arcsin((1-\xi_1)(\cdot))}{\eta^2+\arcsin^2((1-\xi_1)(\cdot))}\right)\right)}_{\max,\mathbf{Ellipse}(\rho=1+c\eta)}\\
    &\leq\norm{\sin\left(\frac{1}{2}\arcsin\left(\frac{\eta(\cdot)}{\eta^2+(\cdot)^2}\right)\right)}_{\max,\mathbf{Stadium}\left(\left[-\frac{\pi}{2},\frac{\pi}{2}\right];
    \frac{\eta}{2}\right)}\\
    &\leq\norm{\sin\left(\frac{1}{2}\arcsin(\cdot)\right)}_{\max,\mathbf{Disk}\left(0;\left[0,\frac{2}{3}\right]\right)}\\
    &\leq\norm{\sin}_{\max,\mathbf{Disk}\left(0;\left[0,\frac{1}{2}\arcsin\left(\frac{2}{3}\right)\right]\right)}
    =\mathbf{O}(1).
\end{aligned}
\end{equation}
Invoking~\lem{approx_ellipse}, we get a real odd polynomial $h_{\text{green}}$ with degree $d$ and error $\mathbf{O}\left(\frac{1}{\eta(1+\eta)^d}\right)$. To achieve accuracy $\epsilon_{\text{green}}$, we require that $(1+\eta)^d=\mathbf{\Theta}\left(\frac{1}{\eta\epsilon_{\text{green}}}\right)$, which can be attained by setting $d=\mathbf{\Theta}\left(\frac{1}{\eta}\log\left(\frac{1}{\eta\epsilon_{\text{green}}}\right)\right)$.
Therefore, we can choose $h_{\text{green}}$ with the behavior
\begin{equation}
    \norm{h_{\text{green}}-\sin\left(\frac{1}{2}\arcsin\left(\frac{\eta\arcsin(\cdot)}{\eta^2+\arcsin^2(\cdot)}\right)\right)}_{\max,\left[-\sin\left(\frac{\pi}{2}-\xi\right),\sin\left(\frac{\pi}{2}-\xi\right)\right]}
    \leq\epsilon_{\text{green}}
\end{equation}
and an asymptotic degree of
\begin{equation}
    \mathbf{O}\left(\frac{1}{\eta}\log\left(\frac{1}{\eta\epsilon_{\text{green}}}\right)\right).
\end{equation}

Note that $\sin\left(\frac{\pi}{2}-\xi\right)$ is still constant gapped away from $1$, but our Bernstein ellipse cannot be extended further without hitting the singularities of $\frac{\eta z}{\eta^2+z^2}$ at $z=\pm i\eta$. To address this, we apply the dominated extension of polynomials from~\prop{dominated_ext}.

Suppose we start with constant $0<\xi<\xi_2\leq\frac{\pi}{2}$. Applying the dominated extension gives a polynomial $h_{\text{green,dom}}$ with the behavior
\begin{footnotesize}
\begin{equation}
\newmaketag
\begin{aligned}
    &\abs{h_{\text{green,dom}}(x)-\sin\left(\frac{1}{2}\arcsin\left(\frac{\eta\arcsin(x)}{\eta^2+\arcsin^2(x)}\right)\right)}\leq\epsilon_{\text{dom}},\qquad&&\forall x\in\left[-\sin\left(\frac{\pi}{2}-\xi_2\right),\sin\left(\frac{\pi}{2}-\xi_2\right)\right],\\
    &\abs{h_{\text{green,dom}}(x)}\leq\abs{\sin\left(\frac{1}{2}\arcsin\left(\frac{\eta\arcsin(x)}{\eta^2+\arcsin^2(x)}\right)\right)}+\epsilon_{\text{dom}},\qquad&&\forall x\in\left[-\sin\left(\frac{\pi}{2}-\xi\right),\sin\left(\frac{\pi}{2}-\xi\right)\right],\\
    &\abs{h_{\text{green,dom}}(x)}\leq\epsilon_{\text{dom}},\qquad&&\forall x\in\left[-1,-\sin\left(\frac{\pi}{2}-\xi\right)\right]\bigcup\left[\sin\left(\frac{\pi}{2}-\xi\right),1\right]
\end{aligned}
\end{equation}
\end{footnotesize}%
and an asymptotic degree of
\begin{equation}
    \mathbf{O}\left(\frac{1}{\eta}\log\left(\frac{1}{\eta\epsilon_{\text{green}}}\right)+\log\left(\frac{1}{\epsilon_{\text{dom}}}\right)\right).
\end{equation}
The remaining analysis proceeds as in~\append{composite_qsvt}.

Now, we consider the anti-Hermitian component of the Green's function. Specifically, for any $\epsilon>0$, $\eta>0$, and constant $0<\xi\leq\frac{\pi}{2}$, our goal is to find a real even polynomial $h$ such that
\begin{equation}
\begin{aligned}
    &\abs{h(x)-\frac{\eta}{\sqrt{\eta^2+\arcsin^2(x)}}}\leq\epsilon,\qquad&&\forall x\in\left[-\sin\left(\frac{\pi}{2}-\xi\right),\sin\left(\frac{\pi}{2}-\xi\right)\right],\\
    &\abs{h(x)}\leq1+\epsilon,\qquad&&\forall x\in[-1,1].
\end{aligned}
\end{equation}
Note that different from the Hermitian case which requires dominated approximations, the anti-Hermitian case can be realized with a bounded polynomial approximation, as we aim for a unitary block encoding as the output.

Much of the analysis parallels that of the Hermitian case, so we omit the details, except that we need to study the behavior of stadiums under $\frac{\eta}{\sqrt{\eta^2+(\cdot)^2}}$.
\begin{proposition}
For $\eta=\mathbf{O}(1)$ sufficiently small,
\begin{equation}
    \frac{\eta}{\sqrt{\eta^2+\mathbf{Stadium}^2\left(\left[-\frac{\pi}{2},\frac{\pi}{2}\right];
    \frac{\eta}{2}\right)}}
    \subseteq\mathbf{Disk}\left(0;\left[0,\frac{2}{\sqrt{3}}\right]\right).
\end{equation}
\end{proposition}
\begin{proof}
Take any $z\in\mathbf{Stadium}\left(\left[-\frac{\pi}{2},\frac{\pi}{2}\right];
\frac{\eta}{2}\right)$.
If $z=x+iy$ with $x\in\left[-\frac{\pi}{2},\frac{\pi}{2}\right]$ and $y=\pm\frac{\eta}{2}$, then
\begin{equation}
    \frac{\eta^2}{\abs{z+i\eta}\abs{z-i\eta}}
    =\frac{\eta^2}{\sqrt{x^2+\frac{9\eta^2}{4}}\sqrt{x^2+\frac{\eta^2}{4}}}
    \leq\frac{4}{3}.
\end{equation}
On the other hand, if $z$ is on the semicircles, then $\abs{z\pm i\eta}\geq\sqrt{\frac{\pi^2}{4}+\frac{\eta^2}{4}}$, which gives
\begin{equation}
    \frac{\eta^2}{\abs{z+i\eta}\abs{z-i\eta}}
    \leq\frac{\eta^2}{\frac{\pi^2}{4}+\frac{\eta^2}{4}}
    =\mathbf{O}\left(\eta^2\right),
\end{equation}
which is also smaller than $\frac{4}{3}$ for $\eta=\mathbf{O}(1)$ sufficiently small.
\end{proof}

\subsection{Dominated polynomial approximation for Hamiltonian squaring}
\label{append:composite_square}
Recall from~\sec{sos_square} that in Hamiltonian squaring, we aim to implement a cubic function, which is an odd but can be linearly combined to yield the desired square function. Specifically, for any $\epsilon>0$ and constant $0<\xi\leq\frac{\pi}{2}$, our goal is to find a real odd polynomial $p$ and even polynomial $q$ such that
\begin{equation}
\begin{aligned}
    &\abs{p(x)-\sin\left(\arcsin^3(x)\right)}\leq\epsilon,\qquad&&\forall x\in\left[-\sin\left(\frac{\pi}{2}-\xi\right),\sin\left(\frac{\pi}{2}-\xi\right)\right],\\
    &\abs{q(x)-\frac{\cos\left(\arcsin^3(x)\right)}{\sqrt{1-x^2}}}\leq\epsilon,\qquad&&\forall x\in\left[-\sin\left(\frac{\pi}{2}-\xi\right),\sin\left(\frac{\pi}{2}-\xi\right)\right],\\
    &p^2(x)+(1-x^2)q^2(x)\leq1+\epsilon,\qquad&&\forall x\in[-1,1].
\end{aligned}
\end{equation}

Without loss of generality, let us consider $\sin\left(\arcsin^3(x)\right)$. Setting $1-\xi_1=\sin\left(\frac{\pi}{2}-\xi\right)$, we have the rescaled function $\sin\left(\arcsin^3((1-\xi_1)x)\right)$ over the unit interval. Suppose we start with the Bernstein ellipse $\mathbf{Ellipse}(\rho=1+\ell)$ for some $\ell>0$ sufficiently small. Then,
\begin{equation}
\begin{aligned}
    &\norm{\sin\left(\arcsin^3((1-\xi_1)(\cdot))\right)}_{\max,\mathbf{Ellipse}(\rho=1+\ell)}\\
    &\leq\norm{\sin\left(\arcsin^3((1-\xi_1)(\cdot))\right)}_{\max,\mathbf{Stadium}([-1,1];\ell)}\\
    &\leq\norm{\sin\left((\cdot)^3\right)}_{\max,\mathbf{Stadium}\left(\left[-\arcsin(1-\xi_1),\arcsin(1-\xi_1)\right];\mathbf{O}(\ell)\right)}\\
    &\leq\norm{\sin\left((\cdot)^3\right)}_{\max,\mathbf{Stadium}\left(\left[-\frac{\pi}{2},\frac{\pi}{2}\right];\mathbf{O}(\ell)\right)}\\
    &\leq\norm{\sin\left(\cdot\right)}_{\max,\mathbf{Stadium}\left(\left[-\frac{\pi^3}{8},\frac{\pi^3}{8}\right];\mathbf{O}(\ell^3)\right)}
    =\exp\left(\mathbf{O}(\ell^3)\right).
\end{aligned}
\end{equation}
Hence, the result is constant if $\ell=\mathbf{O}(1)$ is sufficiently small. The remaining analysis follows from~\cor{approx_dominated_ellipse} similar to that of~\append{composite_qsvt}.

%% file: fermionic.tex
In this appendix, we provide background on the analysis of fermionic systems. This includes a discussion on properties of fermionic operators in~\append{fermionic_property}, and an analysis of the fermionic $\eta$-seminorm in~\append{fermionic_seminorm}.

\subsection{Properties of fermionic operators}
\label{append:fermionic_property}

Let $A_j$ ($j=1,\ldots,n$) be fermionic operators satisfying the anticommutation relations $\left\{A_j,A_k\right\}=A_jA_k+A_kA_j=0$ and $\left\{A_j,A_k^\dagger\right\}=A_j A_k^\dagger+A_k^\dagger A_j=\pmb{\delta}_{j,k}I$, where $\pmb{\delta}_{j,k}=1$ if $j=k$ and $0$ otherwise.
Our analysis of the electronic structure Hamiltonian simulation in the main text makes extensive use of the following mapping of coefficient tensors to one-body free fermions:
\begin{equation}
    \mathbf{Quad}\left(J\right)
    =\sum_{j,k}\left[J\right]_{j,k}A_j^\dagger A_k.
\end{equation}
Alternatively, the mapping $\mathbf{Quad}$ acts on an arbitrary computational basis as
\begin{equation}
    \mathbf{Quad}:\ \ketbra{j}{k}
    \mapsto A_j^\dagger A_k
\end{equation}
and extends its action by linearity. This is in fact a restricted version of the more general correspondence between first and second quantized representation of fermionic systems, but our above definition suffices for the analysis of quantum simulation.
We refer the reader to~\cite{helgaker2014molecular,Otte2010Boundedness} for further discussions on fermionic operators not covered here.

To assist the analysis, we also introduce the mapping $\mathbf{Cre}$ between column vectors and fermionic creation operators, and $\mathbf{Ann}$ between row vectors and fermionic annihilation operators:
\begin{equation}
\begin{aligned}
    &\mathbf{Cre}:\ \ket{j}\mapsto A_j^\dagger,\qquad&&\mathbf{Cre}(\beta)=\sum_{j=1}^n\beta_jA_j^\dagger,\\
    &\mathbf{Ann}:\ \bra{j}\mapsto A_j,\qquad&&\mathbf{Ann}(\beta^\top)=\sum_{j=1}^n\beta_jA_j.
\end{aligned}
\end{equation}
These mappings of fermionic operators satisfy the following properties.

\begin{lemma}[Properties of fermionic operators]
\label{lem:fermionic_property}
Let $A_j$ ($j=1,\ldots,n$) be operators satisfying $\left\{A_j,A_k\right\}=A_jA_k+A_kA_j=0$ and $\left\{A_j,A_k^\dagger\right\}=A_j A_k^\dagger+A_k^\dagger A_j=\pmb{\delta}_{j,k}I$. Define
\begin{equation}
    \mathbf{Cre}(\beta)=\sum_{j=1}^n\beta_jA_j^\dagger,\qquad
    \mathbf{Ann}(\beta^\top)=\sum_{j=1}^n\beta_jA_j,\qquad
    \mathbf{Quad}(L)=\sum_{j,k=1}^n[L]_{j,k}A_j^\dagger A_k,
\end{equation}
for $\beta\in\mathbb{C}^n$ and $L\in\mathbb{C}^{n\times n}$. The following statements hold for scalars $b,c\in\mathbb{C}$, vectors $\beta,\gamma\in\mathbb{C}^n$ and matrices $L,M\in\mathbb{C}^{n\times n}$.
\begin{enumerate}
    \item
    \begin{enumerate}
        \item $\mathbf{Cre}(b\beta+c\gamma)=b\mathbf{Cre}(\beta)+c\mathbf{Cre}(\gamma)$; $\mathbf{Ann}(b\beta^\top+c\gamma^\top)=b\mathbf{Ann}(\beta^\top)+c\mathbf{Ann}(\gamma^\top)$;
        \item $\mathbf{Quad}(bL+cM)=b\mathbf{Quad}(L)+c\mathbf{Quad}(M)$.
    \end{enumerate}
    \item 
    \begin{enumerate}
        \item $\mathbf{Ann}(\beta^\dagger)=\mathbf{Cre}^\dagger(\beta)$;
        \item $\mathbf{Quad}(L^\dagger)=\mathbf{Quad}^\dagger(L)$.
    \end{enumerate}
    \item 
    \begin{enumerate}
        \item $\left\{\mathbf{Cre}(\beta),\mathbf{Cre}(\gamma)\right\}=\left\{\mathbf{Ann}(\beta^\top),\mathbf{Ann}(\gamma^\top)\right\}=0$;
        \item $\mathbf{Cre}(\beta)\mathbf{Ann}(\gamma^\top)=\mathbf{Quad}(\beta\gamma^\top)$;
        \item $\left\{\mathbf{Cre}(\beta),\mathbf{Ann}(\gamma^\top)\right\}=\left\{\mathbf{Ann}(\gamma^\top),\mathbf{Cre}(\beta)\right\}=\gamma^\top \beta I$;
        \item $\left[\mathbf{Quad}(L),\mathbf{Cre}(\beta)\right]=\mathbf{Cre}(L\beta)$; $\left[\mathbf{Quad}(L),\mathbf{Ann}(\beta^\top)\right]=-\mathbf{Ann}(\beta^\top L)$;
        \item $\left[\mathbf{Quad}(L),\mathbf{Quad}(M)\right]=\mathbf{Quad}([L,M])$.
    \end{enumerate}
    \item
    \begin{enumerate}
        \item $\mathbf{Cre}(\beta)=0\Leftrightarrow\mathbf{Ann}(\beta^\top)=0\Leftrightarrow\beta=0$;
        \item $\mathbf{Quad}(L)=0\Leftrightarrow L=0$.
    \end{enumerate}
\end{enumerate}
\end{lemma}
\begin{proof}
We will only prove the forward direction of Statement 4 as the remaining follow directly from the definition of $\mathbf{Cre}$, $\mathbf{Ann}$ and $\mathbf{Quad}$, along with the anticommutation relation $\left\{A_j,A_k\right\}=0$, $\left\{A_j,A_k^\dagger\right\}=\pmb{\delta}_{j,k}I$ and commutation relation
\begin{equation}
    \left[A_j^\dagger A_k,A_l^\dagger\right]
    =\pmb{\delta}_{l,k}A_j^\dagger,\qquad
    \left[A_j^\dagger A_k,A_m\right]
    =\pmb{\delta}_{m,j}A_k,
\end{equation}
of fermionic operators.

Suppose that $\mathbf{Cre}(\beta)=0$. Let us consider its anticommutator with $A_j=\mathbf{Ann}(e_j^\top)$, where $e_j$ is the standard basis vector with $1$ at position $j$ and $0$ elsewhere:
\begin{equation}
    0=\left\{A_j,\mathbf{Cre}(\beta)\right\}
    =\left\{\mathbf{Ann}\left(e_j^\top\right),\mathbf{Cre}(\beta)\right\}
    =e_j^\top\beta I.
\end{equation}
As $j$ goes through $1,\ldots,n$, this forces every component of $\beta$ to be zero and hence $\beta=0$. The same argument shows that $\mathbf{Ann}(\beta^\top)=0$ implies $\beta=0$.

Similarly, if $\mathbf{Quad}(L)=0$, we consider its commutator with $A_j^\dagger=\mathbf{Cre}(e_j)$:
\begin{equation}
    0=\left[\mathbf{Quad}(L),A_j^\dagger\right]=\left[\mathbf{Quad}(L),\mathbf{Cre}(e_j)\right]
    =\mathbf{Cre}(Le_j).
\end{equation}
This means $Le_j=0$ for all $j=1,\ldots,n$ and therefore $L=0$.
\end{proof}

As an immediate corollary, we have:
\begin{corollary}
Let $A_j$ ($j=1,\ldots,n$) be fermionic operators and define the mappings $\mathbf{Cre}$, $\mathbf{Ann}$, and $\mathbf{Quad}$ as in~\lem{fermionic_property}.
The following statements hold for vectors $\beta,\gamma\in\mathbb{C}^n$ and matrices $L,M\in\mathbb{C}^{n\times n}$.
\begin{enumerate}
    \item $\mathbf{Cre}\left(e^{L}\beta\right)=e^{\mathbf{Quad}(L)}\mathbf{Cre}(\beta)e^{-\mathbf{Quad}(L)}$;
    \item $\mathbf{Ann}\left(\gamma^\top e^{-L}\right)=e^{\mathbf{Quad}(L)}\mathbf{Ann}(\gamma^\top)e^{-\mathbf{Quad}(L)}$;
    \item $\mathbf{Quad}\left(e^{L}Me^{-L}\right)=e^{\mathbf{Quad}(L)}\mathbf{Quad}(M)e^{-\mathbf{Quad}(L)}$.
\end{enumerate}
\end{corollary}
\begin{proof}
The first claim follows from a Taylor expansion of the matrix exponential conjugation and an application of the mapping $\mathbf{Cre}$ term by term:
\begin{equation}
\begin{aligned}
    &\mathbf{Cre}\left(e^{L}\beta\right)\\
    &=\mathbf{Cre}\left(\beta+L\beta+\frac{1}{2!}L^2\beta+\cdots\right)\\
    &=\mathbf{Cre}(\beta)
    +\left[\mathbf{Quad}(L),\mathbf{Cre}(\beta)\right]
    +\frac{1}{2!}\left[\mathbf{Quad}(L),\left[\mathbf{Quad}(L),\mathbf{Cre}(\beta)\right]\right]
    +\cdots\\
    &=e^{\mathbf{Quad}(L)}\mathbf{Cre}(\beta)e^{-\mathbf{Quad}(L)},
\end{aligned}
\end{equation}
where passage to the limit is justified by the fact that $\mathbf{Cre}$ is necessarily continuous as a finite-dimensional linear mapping.
Similarly, for the second claim, we use properties of $\mathbf{Ann}$ to write:
\begin{equation}
\begin{aligned}
    &\mathbf{Ann}\left(\gamma^\top e^{-L}\right)\\
    &=\mathbf{Ann}\left(\gamma^\top-\gamma^\top L+\frac{1}{2!}\gamma^\top L^2+\cdots\right)\\
    &=\mathbf{Ann}\left(\gamma^\top\right)
    +\left[\mathbf{Quad}(L),\mathbf{Ann}\left(\gamma^\top\right)\right]
    +\frac{1}{2!}\left[\mathbf{Quad}(L),\left[\mathbf{Quad}(L),\mathbf{Ann}\left(\gamma^\top\right)\right]\right]
    +\cdots\\
    &=e^{\mathbf{Quad}(L)}\mathbf{Ann}(\gamma^\top)e^{-\mathbf{Quad}(L)}.
\end{aligned}
\end{equation}

Finally, to prove the third claim, suppose that $M=\sum_{j=1}^n\beta_j\gamma_j^\top$ is a decomposition of $M$ into rank-$1$ operators. Then,
\begin{equation}
\begin{aligned}
    \mathbf{Quad}\left(e^{L}Me^{-L}\right)
    &=\mathbf{Quad}\left(e^{L}\sum_{j=1}^n\beta_j\gamma_j^\top e^{-L}\right)
    =\sum_{j=1}^n\mathbf{Quad}\left(e^{L}\beta_j\gamma_j^\top e^{-L}\right)\\
    &=\sum_{j=1}^n\mathbf{Cre}\left(e^{L}\beta_j\right)\mathbf{Ann}\left(\gamma_j^\top e^{-L}\right)\\
    &=\sum_{j=1}^ne^{\mathbf{Quad}(L)}\mathbf{Cre}\left(\beta_j\right)e^{-\mathbf{Quad}(L)}
    e^{\mathbf{Quad}(L)}\mathbf{Ann}(\gamma_j^\top)e^{-\mathbf{Quad}(L)}\\
    &=e^{\mathbf{Quad}(L)}\sum_{j=1}^n\mathbf{Quad}\left(\beta_j\gamma_j^\top\right)e^{-\mathbf{Quad}(L)}
    =e^{\mathbf{Quad}(L)}\mathbf{Quad}(M)e^{-\mathbf{Quad}(L)}.
\end{aligned}
\end{equation}
This proves all the claimed identities.
\end{proof}

\subsection{Fermionic \texorpdfstring{$\eta$}{eta}-seminorm and its evaluation}
\label{append:fermionic_seminorm}
Let $A_j$ ($j=1,\ldots,n$) be fermionic operators satisfying the canonical anticommutation relations $\left\{A_j,A_k\right\}=A_jA_k+A_kA_j=0$ and $\left\{A_j,A_k^\dagger\right\}=A_j A_k^\dagger+A_k^\dagger A_j=\pmb{\delta}_{j,k}I$. Define the number operator 
\begin{equation}
    N=\mathbf{Quad}(I)=\sum_{j=1}^{n}A_j^\dagger A_j.
\end{equation}
We say $\ket{\psi_\eta}$ is an \emph{$\eta$-particle state} if it is an eigenstate of the number operator with eigenvalue $\eta$: $N\ket{\psi_\eta}=\sum_{j=1}^{n}A_j^\dagger A_j\ket{\psi_\eta}=\eta\ket{\psi_\eta}$. Note that
\begin{equation}
    \left(A_j^\dagger A_j\right)^2
    =A_j^\dagger\left(A_jA_j^\dagger\right)A_j
    =A_j^\dagger A_j-A_j^{\dagger2}A_j^2
    =A_j^\dagger A_j,
\end{equation}
whereas
\begin{equation}
    \left[A_j^\dagger A_j,A_k^\dagger A_k\right]
    =\left[\mathbf{Quad}\left(e_je_j^\top\right),\mathbf{Quad}\left(e_ke_k^\top\right)\right]
    =\mathbf{Quad}\left(\left[e_je_j^\top,e_ke_k^\top\right]\right)=0
\end{equation}
for $j\neq k$.
Hence, $\{A_j^\dagger A_j\}_{j=1}^n$ are pairwise commuting orthogonal projections that can be simultaneously diagonalized, so eigenvalues of their sum $N=\mathbf{Quad}(I)=\sum_{j=1}^{n}A_j^\dagger A_j$ are integers $0\leq\eta\leq n$.
We call an operator $U$ \emph{number preserving} if it commutes with the number operator: $\left[U,N\right]=0$.

As an example, since the identity matrix is invariant under any exponential conjugation
\begin{equation}
    I=e^{K}Ie^{-K},
\end{equation}
we have
\begin{equation}
    \mathbf{Quad}(I)=e^{\mathbf{Quad}(K)}\mathbf{Quad}(I)e^{-\mathbf{Quad}(K)}.
\end{equation}
This means $e^{\mathbf{Quad}(K)}=e^{\sum_{j,k}[K]_{j,k}A_j^\dagger A_k}$ commutes with $\mathbf{Quad}(I)=\sum_jA_j^\dagger A_j=N$, and must therefore be number-preserving.

In the following, we consider the \emph{fermionic $\eta$-seminorm} defined by
\begin{equation}
    \max_{\ket{\psi_\eta}}\abs{\bra{\psi_\eta}\mathbf{Quad}(J)\ket{\psi_\eta}},
\end{equation}
where $J\in\mathbb{C}^{n\times n}$ is a normal coefficient matrix and the maximization is taken over all $\eta$-particle states $\ket{\psi_\eta}$. This expectation value can be evaluated as follows~\cite[Eq.\ (16)]{McArdleCampbell22}.

\begin{lemma}[Fermionic $\eta$-seminorm with normal coefficient matrices]
\label{lem:normal_eta_norm}
Let $A_j$ ($j=1,\ldots,n$) be fermionic operators and define $\mathbf{Quad}$ as in~\lem{fermionic_property}.
Let $J\in\mathbb{C}^{n\times n}$ be a normal matrix with eigenvalues $\lambda_k(J)$, and $0\leq\eta\leq n$ be an integer. It holds that
\begin{equation}
    \max_{\ket{\psi_\eta}}\abs{\bra{\psi_\eta}\mathbf{Quad}(J)\ket{\psi_\eta}}=\max_{\#\mathcal{S}=\eta}\abs{\sum_{k\in \mathcal{S}}\lambda_k(J)},
\end{equation}
where the maximization on the left is taken over all $\eta$-particle states satisfying $\sum_jA_j^\dagger A_j\ket{\psi_\eta}=\eta\ket{\psi_\eta}$, and the maximization on the right is taken over all subsets $\mathcal{S}\subseteq\{1,2,\ldots,n\}$ of size $\#\mathcal{S}=\eta$.
\end{lemma}
\begin{proof}
Since $J$ is normal, it admits the spectral decomposition $J=U\Lambda U^\dagger$ where $U$ is unitary and $\Lambda$ is diagonal containing all the eigenvalues $\lambda_k(J)$. 
Taking the matrix logarithm,
we obtain
\begin{equation}
    \mathbf{Quad}(J)
    =\mathbf{Quad}\left(U\Lambda U^\dagger\right)
    =\mathbf{Quad}\left(e^{\log(U)}\Lambda e^{-\log(U)}\right)
    =e^{\mathbf{Quad}(\log(U))}\mathbf{Quad}(\Lambda)e^{-\mathbf{Quad}(\log(U))}.
\end{equation}
Hence,
\begin{equation}
\begin{aligned}
    \max_{\ket{\psi_\eta}}\abs{\bra{\psi_\eta}\mathbf{Quad}(J)\ket{\psi_\eta}}
    &=\max_{\ket{\psi_\eta}}\abs{\bra{\psi_\eta}e^{\mathbf{Quad}(\log(U))}\mathbf{Quad}(\Lambda)e^{-\mathbf{Quad}(\log(U))}\ket{\psi_\eta}}\\
    &=\max_{\ket{\psi_\eta}}\abs{\bra{\psi_\eta}\sum_k\lambda_k(J)A_k^\dagger A_k\ket{\psi_\eta}},
\end{aligned}
\end{equation}
where we have used the fact that $e^{\mathbf{Quad}(\log(U))}$ is a number-preserving unitary. 

To proceed, suppose that under the computational basis, quantum state $\ket{\psi_\eta}$ has the expansion
\begin{equation}
    \ket{\psi_\eta}=\sum_{\#\mathcal{S}=\eta}c_{\mathcal{S}}\ket{\mathcal{S}},
\end{equation}
with normalized coefficients $\norm{c}=1$ and occupied modes labeled by subsets $\mathcal{S}\subseteq\{1,2,\ldots,n\}$ of size $\eta$. This then gives
\begin{equation}
\begin{aligned}
    \max_{\ket{\psi_\eta}}\abs{\bra{\psi_\eta}\mathbf{Quad}(J)\ket{\psi_\eta}}
    &=\max_{\norm{c}=1}\abs{\sum_{\#\mathcal{S}_1=\eta}c_{\mathcal{S}_1}^*\bra{\mathcal{S}_1}\sum_k\lambda_k(J)A_k^\dagger A_k\sum_{\#\mathcal{S}_2=\eta}c_{\mathcal{S}_2}\ket{\mathcal{S}_2}}\\
    &=\max_{\norm{c}=1}\abs{\sum_{\#\mathcal{S}=\eta}\abs{c_{\mathcal{S}}}^2\sum_k\lambda_k(J)\bra{\mathcal{S}}A_k^\dagger A_k\ket{\mathcal{S}}}.
\end{aligned}
\end{equation}

It is clear from the triangle inequality that
\begin{equation}
    \max_{\ket{\psi_\eta}}\abs{\bra{\psi_\eta}\mathbf{Quad}(J)\ket{\psi_\eta}}
    \leq\max_{\#\mathcal{S}=\eta}\abs{\sum_k\lambda_k(J)\bra{\mathcal{S}}A_k^\dagger A_k\ket{\mathcal{S}}}
    =\max_{\#\mathcal{S}=\eta}\abs{\sum_{k\in \mathcal{S}}\lambda_k(J)}.
\end{equation}
But the above is in fact an equality. This is because for the specific $\mathcal{S}$ attaining the maximum on the right, we can simply let $c_{\mathcal{S}}=1$ and $c_{\mathcal{S}'}=0$ for any $\mathcal{S}'\neq\mathcal{S}$. This establishes the claimed formula for the maximum expectation value of fermionic operators.
\end{proof}

Whenever the context is clear, we will denote both the fermionic $\eta$-seminorm and its reduced version by $\norm{\cdot}_\eta$. That is, under the setting of~\lem{normal_eta_norm},
\begin{equation}
    \norm{\mathbf{Quad}(J)}_\eta
    =\max_{\ket{\psi_\eta}}\abs{\bra{\psi_\eta}\mathbf{Quad}(J)\ket{\psi_\eta}},\qquad
    \norm{J}_\eta
    =\max_{\#\mathcal{S}=\eta}\abs{\sum_{k\in \mathcal{S}}\lambda_k(J)}.
\end{equation}
We now explain how to efficiently evaluate the reduced fermionic $\eta$ seminorm $\max\limits_{\#\mathcal{S}=\eta}\abs{\sum_{k\in \mathcal{S}}\lambda_k(J)}$ on a classical computer. When the coefficient matrix $J$ is normal, we can find all its eigenvalues $\lambda_k(J)$ efficiently (with a complexity polynomial in $n$). Then the problem asks for the largest absolute value of an arbitrary sum of $0\leq\eta\leq n$ of them. This has the following geometric interpretation. Given vectors $x_1,\ldots,x_n\in\mathbb{R}^2$, our goal is to find a subset $\mathcal{S}\subseteq\{1,\ldots,n\}$ of size $\#\mathcal{S}=\eta$ that maximizes the Euclidean norm $\norm{\sum_{k\in\mathcal{S}}x_k}$. This is known as the \emph{longest $\eta$-vector sum} problem and can be solved in time $\mathbf{O}\left(n^2\log(n)\right)$~\cite{SHENMAIER202060} for any $0\leq\eta\leq n$.

In the special case where the coefficient matrix $J$ is Hermitian, it has $n$ real eigenvalues which can be efficiently sorted in time $\mathbf{O}\left(n\log(n)\right)$. Then the reduced fermionic seminorm $\max\limits_{\#\mathcal{S}=\eta}\abs{\sum_{k\in \mathcal{S}}\lambda_k(J)}$ is attained by either the largest $\eta$ or smallest $\eta$ eigenvalues $\lambda_k(J)$. The case where $J$ is anti-Hermitian can be handled similarly as $iJ$ is Hermitian and the magnitude of eigenvalues is not affected by the rescaling.